\documentclass[letterpaper, 10pt,conference,final]{IEEEtran} 
\pdfoutput=1

\newif\ifDiffForRevision
\newif\ifIncludeTechnicalReport
\newif\ifUseLinks

\UseLinkstrue
\IncludeTechnicalReporttrue
\DiffForRevisionfalse
\ifCLASSOPTIONdraft
 \def\baselinestretch{1}
\fi

\usepackage[utf8]{inputenc}
\usepackage[T1]{fontenc}

\usepackage{mathptmx}
\usepackage{xspace}
\usepackage[usenames,dvipsnames,svgnames,table]{xcolor}
\usepackage[final,stretch=10,shrink=10]{microtype}
\usepackage{ifdraft}
\usepackage{tikz}
\usetikzlibrary{fit, positioning, backgrounds,snakes,arrows, arrows.meta, decorations.pathreplacing, decorations.markings,calc}
\usepackage[noend]{algpseudocode}
\usepackage{algorithm}
\usepackage{algorithmicx}
\usepackage{varwidth} %
\usepackage{soul} %
\usepackage{printlen}

\usepackage{sectodo}

\usepackage[compress]{cite} %
\usepackage{placeins} %

\usepackage{calc}

\usepackage{captcont}
\usepackage{caption}
\DeclareCaptionLabelSeparator{dot}{.\enspace}
\clearcaptionsetup{figure}
\usepackage[thmmarks]{ntheorem}
\usepackage{macros}
\usepackage{amsmath}
\usepackage{amssymb}

\usepackage[english]{babel}
\usepackage[final,breaklinks=true]{hyperref}

\usepackage{etoolbox}

\appto\UrlBreaks{\do\-}  %

\usepackage{color}
\definecolor{figurehighlight1}{RGB}{190, 0, 0}
\definecolor{figurehighlight2}{RGB}{0, 0, 190}

\definecolor{anthrazit}{RGB}{ 50, 50, 50}
\definecolor{mittelblau}{RGB}{0,  65,  145}
\definecolor{hellblau}{RGB}{  0,  190, 255}
\definecolor{dunkelrot}{RGB}{  190,  0, 0}
\definecolor{kommentare}{RGB}{0, 99, 0}
\definecolor{gelb}{RGB}{222, 173, 0}
\definecolor{linkcolor}{rgb}{0,0,0.5}

\ifUseLinks
  \hypersetup{
    colorlinks=true,
    linkcolor=linkcolor,
    anchorcolor=linkcolor,
    citecolor=linkcolor,
    filecolor=linkcolor,
    menucolor=linkcolor,
    runcolor=linkcolor,
    urlcolor=linkcolor,
  }
\else
  \hypersetup{draft, colorlinks=false
  }
\fi

\usepackage{linegoal} %
\usepackage{style}

\usepackage[outline]{contour}
\usepackage{mathtools}
\usepackage{graphicx}

\def\protostep#1{\resizebox{!}{0.8\baselineskip}{\begin{tikzpicture}[baseline={([yshift=-0.5pt]O.base)}] \node (O) [sharp corners,fill=blue,inner sep=1ex]{\color{white}\textbf{\refstepcounter{protostep}\theprotostep\label{protostep:#1}}};\end{tikzpicture}}}
\def\refprotostep#1{\resizebox{!}{0.6\baselineskip}{\begin{tikzpicture}[baseline={([yshift=-1.5pt]O.base)}] \node (O) [draw,sharp corners]{\ref{protostep:#1}};\end{tikzpicture}}}

\ifdraft{
  \usepackage[scrtime]{prelim2e}
  \usepackage{showlabels}
  \pagestyle{plain}
  \ifdefined\fastmode
    \usepackage{comment}
    \excludecomment{figure}
    
  \fi
}{}

\begin{document}
\hyphenation{Brow-serID}
\hyphenation{in-fra-struc-ture}
\hyphenation{brow-ser}
\hyphenation{doc-u-ment}
\hyphenation{Chro-mi-um}
\hyphenation{meth-od}
\hyphenation{sec-ond-ary}
\hyphenation{Java-Script}
\hyphenation{Mo-zil-la}
\hyphenation{post-Mes-sage}

\title{An Extensive Formal Security Analysis of the OpenID Financial-grade API}

 \author{
   \IEEEauthorblockN{Daniel Fett}
   \IEEEauthorblockA{
     yes.com AG\\
     \small\texttt{mail@danielfett.de} 
   }
   \and
   \IEEEauthorblockN{Pedram Hosseyni}
   \IEEEauthorblockA{
     University of Stuttgart,
     Germany\\
      \small\texttt{pedram.hosseyni@sec.uni-stuttgart.de}
   }
   \and
   \IEEEauthorblockN{Ralf K\"usters}
   \IEEEauthorblockA{
     University of Stuttgart,
     Germany\\
      \small\texttt{ralf.kuesters@sec.uni-stuttgart.de}
   }
 }

\maketitle

\ifIncludeTechnicalReport
  \pagestyle{plain}
\fi

\ifdraft{
\listoftodos

\tableofcontents

\section*{Target Structure and Page Limits}
As discussed in skype meeting on Sep 7, but restored from memory since the document was not saved:
\begin{enumerate}
\item Introduction (2 Seiten mit Titel und Abstract)
\item OAuth and New Defense Mechanisms (2.5 Seiten)
\item
  \begin{enumerate}
  \item OAuth [short recap]
  \item PKCE
  \item Token Binding
  \end{enumerate}
\item FAPI (mit Varianten und Grafik) (2.5 Seiten)
\item Attacks (3.5 Seiten)
\item The FKS model (1.5 Seiten --- evtl. in Appendix?)
\item Analysis (2.5 Seiten)
\item
  \begin{enumerate}
  \item Model
  \item Security Properties
  \item Proof
  \end{enumerate}
\item Conclusion (0.5 Seiten) --- Bis hier max 13 Seiten
\item References ()
\item ---  max 18 Seiten

\end{enumerate}

}{ }

\insertExplanationForDiff

\begin{abstract}
  Forced by regulations and industry demand, banks worldwide are working to open their customers' online banking accounts to third-party services via web-based APIs. By using these so-called \emph{Open Banking} APIs, third-party companies, such as FinTechs, are able to read information about and initiate payments from their\highlightIfDiffMinor{users'}bank accounts. Such access to financial data and resources needs to meet particularly high security requirements to protect customers.

One of the most promising standards in this segment is the \emph{OpenID Financial-grade API (FAPI)}, currently under development in an open process by the OpenID Foundation and backed by large industry partners. The FAPI is a profile of OAuth~2.0 designed for high-risk scenarios and aiming to be secure against very strong attackers. To achieve this level of security, the FAPI employs a range of mechanisms that have been developed to harden OAuth~2.0, such as \emph{Code and Token Binding} (including mTLS and OAUTB), \emph{JWS Client Assertions}, and \emph{Proof Key for Code Exchange}. 

In this paper, we perform a rigorous, systematic formal analysis of the security of the FAPI, based on an existing comprehensive model of the web infrastructure---the \emph{Web Infrastructure Model (WIM)} proposed by Fett, K{\"u}sters, and Schmitz. To this end, we first develop a precise model of the FAPI in the WIM, including different profiles for read-only and read-write access, different flows, different types of clients, and different combinations of security features, capturing the complex interactions in a web-based environment. We then use our model of the FAPI to precisely define central security properties. In an attempt to prove these properties, we uncover partly severe attacks, breaking authentication, authorization, and session integrity properties. We develop mitigations against these attacks and finally are able to formally prove the security of\highlightIfDiffMinor{a fixed version of the FAPI.}

Although financial applications are high-stakes environments, this work is the first to formally analyze and, importantly, verify an Open Banking security profile.

By itself, this analysis is an important contribution to the development of the FAPI since it helps to define exact security properties and attacker models, and to avoid severe security risks before the first implementations of the standard go live.

Of independent interest, we also uncover weaknesses in the aforementioned security mechanisms for hardening OAuth 2.0. We illustrate that these mechanisms do not necessarily achieve the security properties they have been designed for.

\end{abstract}

\section{Introduction}

Delivering financial services has long been a field exclusive to
traditional banks. This has changed with the emergence of FinTech
companies that are expected to deliver more than 20\% of all financial
services\highlightIfDiffMinor{in}2020 \cite{pwc-fintech-report-2016}. Many FinTechs %
provide services that are based on access to a customers online
banking account information or on initiating payments from a customers
bank account.

For a long time, screen scraping has been the primary means of these
service providers to access the customer's data at the bank. Screen
scraping means that the customer enters online banking login
credentials at the service provider's website, which then uses this
data to log into the customer's online banking account by emulating a
web browser. The service provider then retrieves account information
(such as the balance or recent activities) and can trigger, for example,
a cash transfer, which may require the user to enter her second-factor
authentication credential (such as a TAN) at the service provider's web
interface.

Screen scraping is inherently insecure: first of all, the service
provider gets to know all login credentials, including the
second-factor authentication of the customer. Also, screen scraping is
prone to errors, for example, when the website of a bank changes.

Over the last years, the terms \emph{API banking} and \emph{Open
  Banking} have emerged to mark the introduction of standardized
interfaces to financial institutions' data. These interfaces enable
third parties, in particular FinTech companies, to access users' bank
account information and initiate payments through well-defined APIs.
All around the world, API banking is being promoted by law or by
industry demand: In Europe, the \emph{Payment Services Directive 2
  (PSD2)} regulation\highlightIfDiffMinor{mandates all}banks to introduce Open Banking
APIs\highlightIfDiffMinor{by}September 2019~\cite{eu-psd2}. The U.S.~Department of the %
Treasury recommends the implementation of such APIs as
well~\cite{us-treasury-nonbank-financials}. In South Korea, India,
Australia, and Japan, open banking is being pushed by large financial
corporations~\cite{openid-uk-open-banking}.

One important open banking standard currently under
development for this scenario is the \emph{OpenID Financial-grade API
  (FAPI)}.\footnote{In its current form, the FAPI does not (despite its
  name) define an API itself, but defines a security profile for the
  access to APIs.} The FAPI~\cite{fapi-ceb0f82} is a 
  profile\highlightIfDiffMinor{(i.e., a set of concrete protocol flows
  with extensions)}of the \emph{OAuth~2.0
  Authorization Framework} and the identity layer \emph{OpenID
  Connect} to provide a secure authorization and authentication scheme
for high-risk scenarios. The FAPI is under development at the OpenID
Foundation and supported by many large corporations, such as Microsoft
and the largest Japanese consulting firm, Nomura Research Institute.
The OpenID Foundation is also cooperating with other banking
standardization groups: The UK Open Banking Implementation Entity,
backed by nine major UK banks, has 
adopted the FAPI security profile.

The basic idea behind the FAPI is as follows: The owner of the bank
account (\emph{resource owner}, also called user in what follows)
visits some website or uses an app which provides some financial
service. The website or app is called a \emph{client} in the FAPI
terminology. The client redirects the user to the \emph{authorization
  server}, which is typically operated by the bank. The authorization
server asks for the user's bank account credentials. The user is
then redirected back to the client with some token. The client uses
this token to obtain bank account information or initiate a payment at
the \emph{resource server}, which is typically also operated by the
bank.

The FAPI aims to be\highlightIfDiffMinor{secure against}much stronger attackers than its %
foundations, OAuth~2.0 and OpenID Connect: the FAPI assumes that
sensitive tokens leak to an attacker through the user's browser or
operating system, and that endpoint URLs can be
misconfigured. On the one hand, both assumptions are well motivated by
real-world attacks and the high stakes nature of the environment where
the FAPI is to be used. On the other hand, they directly break the
security of OAuth~2.0 and OpenID Connect.

To provide security against such strong attackers, the FAPI employs a
range of OAuth 2.0 security extensions beyond those used in plain OAuth~2.0 and
OpenID Connect: the FAPI uses the so-called Proof Key for Code
Exchange (PKCE)\footnote{Pronounced \emph{pixie}, RFC~7636.} extension
to prevent unauthorized use of tokens. For client authentication
towards the authorization server, the FAPI employs \emph{JWS Client
  Assertions} or \emph{mutual TLS}. Additionally, \emph{OAuth token
  binding}\footnote{\url{https://tools.ietf.org/html/draft-ietf-oauth-token-binding-07}}
or \emph{certificate-bound access
  tokens}\footnote{\url{https://tools.ietf.org/html/draft-ietf-oauth-mtls-11}}
can be used as holder-of-key mechanisms. To introduce yet another new
feature, the FAPI is the first standard to make use of the so-called
JWT Secured Authorization Response Mode (JARM).

The FAPI consists of two main so-called \emph{parts}, here also called
modes, that stipulate different security profiles for read-only access
to resource servers (e.g., to retrieve bank account information) and
read-write access (e.g., for payment initiation). Both modes can be
used by \emph{confidential} clients, i.e., clients that can store and
protect secrets (such as web servers), and by \emph{public} clients
that cannot securely store secrets, such as JavaScript browser
applications. Combined with the new security features, this gives rise
to many different settings and configurations in which the FAPI can
run (see also Figure~\ref{fig:fapi-overview}).

This, the expected wide adoption, the exceptionally strong attacker
model, and the new security features make the FAPI a particularly
interesting, challenging, and important subject for a detailed
security analysis.  While the security of (plain) OAuth~2.0 and OpenID
Connect has been studied formally and informally many times before
\cite{FettKuestersSchmitz-CCS-2016,FettKuestersSchmitz-CSF-2017,Kumar-OAuth-2012,BansalBhargavanMaffeis-CSF-2012,BansalBhargavanetal-JCS-2014,Wangetal-USENIX-Explicating-SDKs-2013,PaiSharmaKumarPaiSingh-2011,ChariJutlaRoy-IACR-2011,SantsaiBeznosov-CCS-2012-OAuth,LiMitchell-ISC-2014,Yangetal-AsiaCCS-2016,Shernanetal-DIMVA-2015,Chenetal-2014,ShehabMohsen-2014,LiMitchell-DIMVA-2016,MladenovMainkaKrautwaldFeldmannSchwenk-OpenIDConnect-arXiv-2016},
there is no such analysis for the FAPI---or any other open banking
API---so far. In particular, there are no results in the strong
attacker model adopted for the FAPI, and there has been no formal
security analysis of the additional OAuth security mechanisms employed
by the FAPI (PKCE, JWS Client Assertions, mTLS Client Authentication,
OAuth Token Binding, Certificate-Bound Access Tokens, JARM), which is
of practical relevance in its own right.

In this paper, we therefore study the security of the FAPI in-depth,
including the OAuth security extensions. Based on a detailed formal
model of the web, we formalize the FAPI with its various 
configurations as well as its security properties. We discover four
previously unknown and severe attacks, propose fixes, and prove the
security of the fixed protocol based on our formal model of the FAPI,
again considering the various  configurations in which the
FAPI can run. Importantly, this also sheds light on  new OAuth
2.0 security extensions.  In detail, our contributions are as follows:

\subsubsection*{Contributions of this Paper} We build a
\textbf{detailed formal model of the FAPI} based on a comprehensive formal model of
the web infrastructure proposed by Fett et al.~in
\cite{FettKuestersSchmitz-SP-2014}, which we refer to as the Web
Infrastructure Model (WIM). The WIM has been successfully used to find
vulnerabilities in and prove the security of several web applications
and
standards~\cite{FettKuestersSchmitz-CSF-2017,FettKuestersSchmitz-CCS-2016,FettKuestersSchmitz-CCS-2015,FettKuestersSchmitz-ESORICS-BrowserID-Primary-2015,FettKuestersSchmitz-SP-2014}. 
It captures a wide set of web features from DNS to JavaScript in
unrivaled detail and comprehensiveness. In particular, it accounts for
the intricate inner workings of web browsers and their interactions
with the web environment. The WIM is ideally suited to identify
logical flaws in web protocols, detect a range of standard web
vulnerabilities (like cross-site request forgery, session fixation,
misuse of certain web browser features, etc.), and even to find new
classes of web attacks.

Based on the generic descriptions of web servers in the WIM, our
models for FAPI clients and authorization servers contain all
important features currently proposed in the FAPI standards. This
includes the flows from both parts of the FAPI, as well as the
different options for client authentication, holder-of-key mechanisms,
and token binding mentioned above. 

Using this model of the FAPI, we define precise \textbf{security
  properties} for authorization, authentication, and session
integrity. 
Roughly speaking, the 
authorization property requires that an attacker is unable to access
the resources of another user at a bank, or act on that user's behalf
towards the bank.
Authentication means that
an attacker is unable to log in at a client using the identity of
another user. 
Session integrity means that an attacker is unable
to force a user to be logged in at a client under the attackers
identity, or force a user to access (through the client) the
attacker's resources instead of the user's own resources (session
fixation).

During our first attempts to prove these properties, we
\textbf{discovered four unknown attacks} on the FAPI. With these
attacks, adversaries can gain access to the bank account of a user,
break session integrity,
and, interestingly, circumvent certain OAuth security extensions, such
as PKCE and Token Binding, employed by the FAPI.

We notified the OpenID FAPI Working Group of the attacks and
vulnerabilities found by our analysis and are working together with
them to fix the standard. To this end, we first \textbf{developed
  mitigations against the vulnerabilities}. We then, as another main
contribution of our work and to support design decisions during the
further development of the FAPI, implemented the fixes in our formal
model and provided the \textbf{first formal proof of the security of
  the FAPI}\highlightIfDiffMinor{(with our fixes
applied)}within our model of the FAPI, including all configurations
of the FAPI and the various ways in which the new OAuth security
extensions are employed in the FAPI (see
Figure~\ref{fig:fapi-overview}). This makes the FAPI the only open banking API to
enjoy a thorough and detailed formal security analysis. 

Our findings also show that (1) several \textbf{OAuth 2.0 security
  extensions} do not necessarily achieve the security properties they
have been designed for and that (2) combining these extensions in a
secure way is far from trivial.  These results are relevant for all
web applications and standards which employ such extensions.

\subsubsection*{Structure of this Paper} We first, in
Section~\ref{sec:oauth-and-new-defence-mechanisms}, recall OAuth~2.0
and OpenID Connect as the foundations of the FAPI. We also introduce
the new defense mechanisms that set the FAPI apart from
``traditional'' OAuth~2.0 and OpenID Connect flows. This sets the
stage for Section~\ref{sec:fapi} where we go into the details of the
FAPI and explain its design and features. In Section~\ref{sec:att}, we
present the attacks on the FAPI (and the new security mechanisms it
uses), which are the results of our initial proof attempts, and also
present our proposed fixes. The model of the FAPI and the analysis are
outlined in Section~\ref{sec:analysis}, along with a high-level
introduction to the Web Infrastructure Model we use as the basis for
our formal model and analysis of the FAPI. We conclude in
Section~\ref{sec:conclusion-outlook}.\chooseReferenceTRPaper{
  Full details and proofs are provided in the appendices.
}{
The appendix contains further
details. Full details and proofs are provided in our technical
report~\cite{FettHosseyniKuesters-TR-FAPI-2018}.
}

\section{OAuth and New Defense Mechanisms} \label{sec:oauth-and-new-defence-mechanisms}

The \emph{OpenID Financial-grade API} builds upon the OAuth~2.0
Authorization Framework~\cite{rfc6749-oauth2}. Compared to the
original OAuth~2.0 protocol, the FAPI aims at providing a much higher
degree of security. For achieving this, the FAPI security profiles
incorporate mechanisms defined in \emph{OpenID
  Connect}~\cite{openid-connect-core-1.0} (which itself builds upon
OAuth~2.0), and importantly, security extensions for OAuth
2.0 developed only recently by the IETF and the OpenID
Foundation.

In the following, we give a brief overview of both OAuth~2.0 and
OpenID Connect, and their security extensions used (among others) within the FAPI, namely
\emph{Proof Key for Code Exchange}, 
\emph{JWS Client Assertions},
\emph{OAuth 2.0 Mutual TLS for
  Client Authentication and Certificate Bound Access Tokens},
\emph{OAuth 2.0 Token Binding} and the \emph{JWT Secured Authorization
  Response Mode}.  The FAPI itself is presented in Section~\ref{sec:fapi}.

\subsection{Fundamentals of OAuth 2.0 and OpenID Connect}\label{sec:fundamentalsOAuthOIDC}

OAuth 2.0 and OpenID Connect are widely used for various
authentication and authorization tasks. In what follows, we first
explain OAuth 2.0 and then briefly OpenID Connect, which is based on
OAuth 2.0.

\subsubsection{OAuth 2.0} \label{subsection:oauth}

On a high level, OAuth~2.0 allows a \emph{resource owner}, or user, to
enable a \emph{client}, a website or an application, to access her
resources at some \emph{resource server}. In order for the user to
grant the client access to her resources, the user has to authenticate
herself at an \emph{authorization server}.

For example, in the context of the FAPI, resources include the user's
account information (like balance and previous transactions) at her
bank or the initiation of a payment transaction (cash transfer). The
client can be a FinTech company which wants to provide a 
financial service to the user via access to the user's bank
account. More specifically, the client might be the website of such a
company (\emph{web server client}) or the company's app on the
user's device. The resource and authorization servers would typically
be run by the user's bank. One client can make use of several
authorization and resource servers.

\begin{figure}[bt]
  \centering  
          \begin{tikzpicture}[]
        \pgfdeclarelayer{arrows}
        \pgfdeclarelayer{groups}
        \pgfdeclarelayer{markers}
        \pgfsetlayers{groups,arrows,main,markers}

        \matrix [column sep={0.15\textwidth,between origins}, row sep=0.1ex]
        {
        \node[annex_matrix_node,inner sep=0,outer sep=0](pos-0-0){}; &\node[annex_matrix_node,inner sep=0,outer sep=0](pos-1-0){}; &\node[annex_matrix_node,inner sep=0,outer sep=0](pos-2-0){};\node[annex_matrix_dummy_height,minimum height=5ex,anchor=center]{};\node[annex_matrix_dummy_height,minimum height=5ex,anchor=center]{};\node[annex_matrix_dummy_height,minimum height=5ex,anchor=center]{};\\
\node[annex_matrix_node,inner sep=0,outer sep=0](pos-0-1){}; &\node[annex_matrix_node,inner sep=0,outer sep=0](pos-1-1){}; &\node[annex_matrix_node,inner sep=0,outer sep=0](pos-2-1){};\node[annex_matrix_dummy_height,minimum height=4ex,anchor=south,yshift=-1ex]{};\\
\node[annex_matrix_node,inner sep=0,outer sep=0](pos-0-2){}; &\node[annex_matrix_node,inner sep=0,outer sep=0](pos-1-2){}; &\node[annex_matrix_node,inner sep=0,outer sep=0](pos-2-2){};\node[annex_matrix_dummy_height,minimum height=4ex+2ex+2ex,anchor=north,yshift=3ex]{};\\
\node[annex_matrix_node,inner sep=0,outer sep=0](pos-0-3){}; &\node[annex_matrix_node,inner sep=0,outer sep=0](pos-1-3){}; &\node[annex_matrix_node,inner sep=0,outer sep=0](pos-2-3){};\node[annex_matrix_dummy_height,minimum height=4ex+2ex,anchor=north,yshift=3ex]{};\\
\node[annex_matrix_node,inner sep=0,outer sep=0](pos-0-4){}; &\node[annex_matrix_node,inner sep=0,outer sep=0](pos-1-4){}; &\node[annex_matrix_node,inner sep=0,outer sep=0](pos-2-4){};\node[annex_matrix_dummy_height,minimum height=4ex,anchor=south,yshift=-1ex]{};\\
\node[annex_matrix_node,inner sep=0,outer sep=0](pos-0-5){}; &\node[annex_matrix_node,inner sep=0,outer sep=0](pos-1-5){}; &\node[annex_matrix_node,inner sep=0,outer sep=0](pos-2-5){};\node[annex_matrix_dummy_height,minimum height=4ex+2ex,anchor=north,yshift=3ex]{};\\
\node[annex_matrix_node,inner sep=0,outer sep=0](pos-0-6){}; &\node[annex_matrix_node,inner sep=0,outer sep=0](pos-1-6){}; &\node[annex_matrix_node,inner sep=0,outer sep=0](pos-2-6){};\node[annex_matrix_dummy_height,minimum height=4ex+2ex,anchor=north,yshift=3ex]{};\\
\node[annex_matrix_node,inner sep=0,outer sep=0](pos-0-7){}; &\node[annex_matrix_node,inner sep=0,outer sep=0](pos-1-7){}; &\node[annex_matrix_node,inner sep=0,outer sep=0](pos-2-7){};\node[annex_matrix_dummy_height,minimum height=4ex+2ex,anchor=north,yshift=3ex]{};\\
\node[annex_matrix_node,inner sep=0,outer sep=0](pos-0-8){}; &\node[annex_matrix_node,inner sep=0,outer sep=0](pos-1-8){}; &\node[annex_matrix_node,inner sep=0,outer sep=0](pos-2-8){};\node[annex_matrix_dummy_height,minimum height=4ex+2ex,anchor=north,yshift=3ex]{};\\
\node[annex_matrix_node,inner sep=0,outer sep=0](pos-0-9){}; &\node[annex_matrix_node,inner sep=0,outer sep=0](pos-1-9){}; &\node[annex_matrix_node,inner sep=0,outer sep=0](pos-2-9){};\node[annex_matrix_dummy_height,minimum height=5ex,anchor=center]{};\\
\node[annex_matrix_node,inner sep=0,outer sep=0](pos-0-10){}; &\node[annex_matrix_node,inner sep=0,outer sep=0](pos-1-10){}; &\node[annex_matrix_node,inner sep=0,outer sep=0](pos-2-10){};\node[annex_matrix_dummy_height,minimum height=5ex,anchor=center]{};\\
\node[annex_matrix_node,inner sep=0,outer sep=0](pos-0-11){}; &\node[annex_matrix_node,inner sep=0,outer sep=0](pos-1-11){}; &\node[annex_matrix_node,inner sep=0,outer sep=0](pos-2-11){};\node[annex_matrix_dummy_height,minimum height=4ex+2ex,anchor=north,yshift=3ex]{};\\
\node[annex_matrix_node,inner sep=0,outer sep=0](pos-0-12){}; &\node[annex_matrix_node,inner sep=0,outer sep=0](pos-1-12){}; &\node[annex_matrix_node,inner sep=0,outer sep=0](pos-2-12){};\node[annex_matrix_dummy_height,minimum height=4ex+2ex,anchor=north,yshift=3ex]{};\\
\node[annex_matrix_node,inner sep=0,outer sep=0](pos-0-13){}; &\node[annex_matrix_node,inner sep=0,outer sep=0](pos-1-13){}; &\node[annex_matrix_node,inner sep=0,outer sep=0](pos-2-13){};\node[annex_matrix_dummy_height,minimum height=5ex,anchor=center]{};\node[annex_matrix_dummy_height,minimum height=5ex,anchor=center]{};\node[annex_matrix_dummy_height,minimum height=5ex,anchor=center]{};\\
};

\node[name=StartParty_9_0,annex_start_party_box,] at (pos-0-0) {Browser (B)};

\node[name=StartParty_10_0,annex_start_party_box,] at (pos-1-0) {Client (C)};

\node[name=StartParty_11_0,annex_start_party_box,] at (pos-2-0) {Authorization Server (AS)};

\node[name=EndParty_20_0,annex_end_party_box,] at (pos-2-9) {Authorization Server (AS)};

\node[name=StartParty_21_0,annex_start_party_box,] at (pos-2-10) {Resource Server (RS)};

\node[name=EndParty_24_0,annex_end_party_box,] at (pos-1-13) {Client (C)};

\node[name=EndParty_25_0,annex_end_party_box,] at (pos-2-13) {Resource Server (RS)};

\node[name=EndParty_26_0,annex_end_party_box,] at (pos-0-13) {Browser (B)};

\begin{pgfonlayer}{arrows}%

\draw[annex_lifeline] (pos-0-0) -- (pos-0-13);

\draw[annex_lifeline] (pos-1-0) -- (pos-1-13);

\draw[annex_lifeline] (pos-2-0) -- (pos-2-9);

        \draw[annex_http_request] (pos-0-1) to node [annex_arrow_text,above=2.6pt,anchor=base](HTTPRequest_12_0){\setcounter{protostep}{0}\protostep{oauth-code:auth-req-1-start} \contour{white}{POST \nolinkurl{/start}}}  (pos-1-1);

        \draw[annex_http_response] (pos-1-2) to node [annex_arrow_text,above=2.6pt,anchor=base](HTTPResponse_13_0){\setcounter{protostep}{1}\protostep{oauth-code:auth-req-1} \contour{white}{Response}} node [annex_arrow_text,below=8pt,anchor=base](HTTPResponse_13_1){\contour{white}{Redirect to AS \nolinkurl{/authorization_endpoint}}} node [annex_arrow_text,below=8pt+8pt,anchor=base](HTTPResponse_13_2){\contour{white}{(client\_id, redirect\_uri, state)}}  (pos-0-2);

        \draw[annex_http_request] (pos-0-3) to node [annex_arrow_text,above=2.6pt,anchor=base](HTTPRequest_14_0){\setcounter{protostep}{2}\protostep{oauth-code:auth-req-2} \contour{white}{GET \nolinkurl{/authorization_endpoint} \textit{(Authorization Request)}}} node [annex_arrow_text,below=8pt,anchor=base](HTTPRequest_14_1){\contour{white}{(client\_id, redirect\_uri, state)}}  (pos-2-3);

        \draw[annex_http_response,transform canvas={yshift=0.25ex}] (pos-2-4) to node [annex_arrow_text,above=2.6pt,anchor=base](HTTPResponseRequest_15_0){\setcounter{protostep}{3}\protostep{oauth-code:ro-authN} \contour{white}{\textit{resource owner authenticates}}} (pos-0-4);
        \draw[annex_http_request,transform canvas={yshift=-0.25ex}] (pos-0-4) to  (pos-2-4);

        \draw[annex_http_response] (pos-2-5) to node [annex_arrow_text,above=2.6pt,anchor=base](HTTPResponse_16_0){\setcounter{protostep}{4}\protostep{oauth-code:auth-resp-1} \contour{white}{Response}} node [annex_arrow_text,below=8pt,anchor=base](HTTPResponse_16_1){\contour{white}{Redirect to C \nolinkurl{/redirect\_uri} (code, state)}}  (pos-0-5);

        \draw[annex_http_request] (pos-0-6) to node [annex_arrow_text,above=2.6pt,anchor=base](HTTPRequest_17_0){\setcounter{protostep}{5}\protostep{oauth-code:auth-resp-2} \contour{white}{GET \nolinkurl{/redirect_uri} \textit{(Authorization Response})}} node [annex_arrow_text,below=8pt,anchor=base](HTTPRequest_17_1){\contour{white}{(code, state)}}  (pos-1-6);

        \draw[annex_http_request] (pos-1-7) to node [annex_arrow_text,above=2.6pt,anchor=base](HTTPRequest_18_0){\setcounter{protostep}{6}\protostep{oauth-code:c-token-req} \contour{white}{POST \nolinkurl{/token_endpoint} \textit{(Token Request)}}} node [annex_arrow_text,below=8pt,anchor=base](HTTPRequest_18_1){\contour{white}{(code, client\_id, [client authentication])}}  (pos-2-7);

        \draw[annex_http_response] (pos-2-8) to node [annex_arrow_text,above=2.6pt,anchor=base](HTTPResponse_19_0){\setcounter{protostep}{7}\protostep{oauth-code:c-token-resp} \contour{white}{Response}} node [annex_arrow_text,below=8pt,anchor=base](HTTPResponse_19_1){\contour{white}{(access token)}}  (pos-1-8);

\draw[annex_lifeline] (pos-2-10) -- (pos-2-13);

        \draw[annex_http_request] (pos-1-11) to node [annex_arrow_text,above=2.6pt,anchor=base](HTTPRequest_22_0){\setcounter{protostep}{8}\protostep{oauth-code:resource-req} \contour{white}{GET \nolinkurl{/resource}}} node [annex_arrow_text,below=8pt,anchor=base](HTTPRequest_22_1){\contour{white}{(access token)}}  (pos-2-11);

        \draw[annex_http_response] (pos-2-12) to node [annex_arrow_text,above=2.6pt,anchor=base](HTTPResponse_23_0){\setcounter{protostep}{9}\protostep{oauth-code:resource-resp} \contour{white}{Response}} node [annex_arrow_text,below=8pt,anchor=base](HTTPResponse_23_1){\contour{white}{(resource)}}  (pos-1-12);

\end{pgfonlayer}

\begin{pgfonlayer}{markers}%

\end{pgfonlayer}
        \end{tikzpicture}
        
  \caption{Overview of the OAuth Authorization Code Flow}
  \label{fig:oauth-overview}
\end{figure}
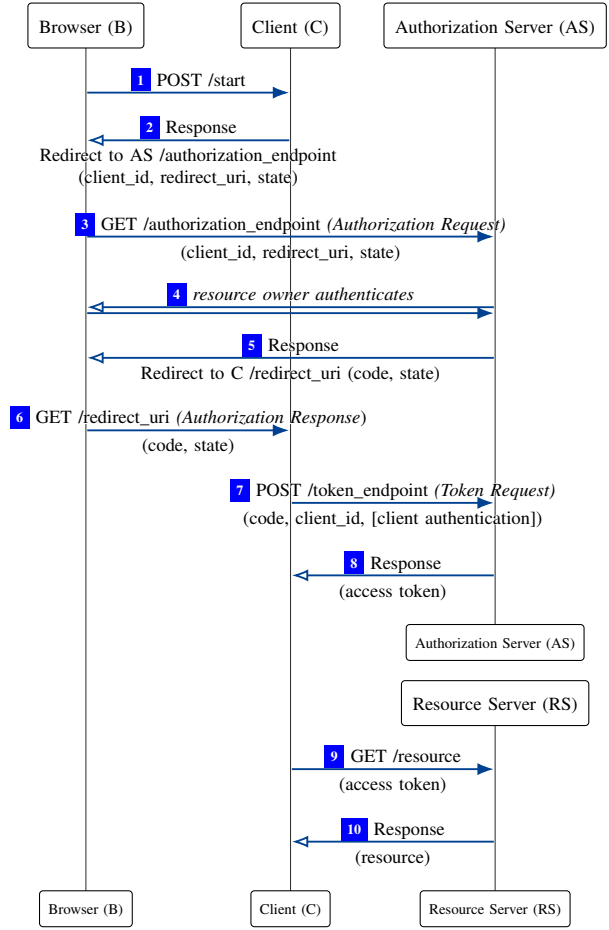

RFC 6749~\cite{rfc6749-oauth2} defines multiple modes of operation for
OAuth~2.0, so-called \emph{grant types}. We here focus on the
\emph{authorization code grant} since the other grant types are not
used in the FAPI. 

Figure~\ref{fig:oauth-overview} shows the authorization code grant,
which works as
follows: %
The user first visits the client's website or opens the client's app
on her smartphone and selects to log in or to give the client access
to her resources (Step~\refprotostep{oauth-code:auth-req-1-start}).
The client then redirects
the user to the so-called \emph{authorization endpoint} at the
authorization server (AS) in
Steps~\refprotostep{oauth-code:auth-req-1} and
\refprotostep{oauth-code:auth-req-2}. (Endpoints are URIs used in the
OAuth flow.) In this redirection, the client passes several parameters
to the AS, for example, the \emph{client id} which identifies the
client at the AS, a \emph{state} value that is used for CSRF
protection,\footnote{The state value is a nonce. The client later
  ensures that it receives the same nonce in the authorization
  response. Otherwise, an attacker could authenticate to the AS with
  his own identity and use the corresponding authorization response
  for logging in an honest user under the attacker's identity with a
  CSRF attack. This attack is also known as \emph{session
    swapping}.\label{footnote:CSRF}} a \emph{scope} parameter (not
shown in Figure~\ref{fig:oauth-overview}) that describes the
permissions requested by the client, and a redirection URI explained
below. Note that if the client's app is used, the redirection from the
app to the AS (Step~\refprotostep{oauth-code:auth-req-1}) is done by
opening the website of the AS in a browser window.
The AS authenticates the user (e.g., by the user entering username and
password) in Step~\refprotostep{oauth-code:ro-authN} and asks for her
consent to give the client access to her resources. The AS then
creates a so-called \emph{authorization code} (typically a nonce) and 
redirects the user back to the so-called \emph{redirection
  endpoint} of the client via the user's browser in
Steps~\refprotostep{oauth-code:auth-resp-1} and
\refprotostep{oauth-code:auth-resp-2}. (If the client's app is used, a
special redirect URI scheme, e.g., \nolinkurl{some-app://}, is used
which causes the operating system to forward the URI to the client's
app.) At the AS, one or more redirection endpoints for a client are
preregistered.\footnote{Without preregistration, a
  malicious client starting a login flow with the client id of an
  honest client could receive a code associated with the honest
  client.} In Step~\refprotostep{oauth-code:auth-req-1}, the
client chooses one of these preregistered URIs. The authorization
response (Step~\refprotostep{oauth-code:auth-resp-1}) is a redirection to this URI,
with the authorization code, the state value from the request, and
optionally further values appended as URI parameters.

When receiving the request resulting from the redirection
in Step~\refprotostep{oauth-code:auth-resp-2}, the client first checks that the
state value is the same as the one in the authorization request,
typically by looking it up in the user's session with the client. If
it is not the same, then the client suspects that an attacker tried to
inject an authorization code into the client's session (cross-site
request forgery, CSRF) and aborts the flow (see also
Footnote~\ref{footnote:CSRF}). Otherwise, the client now exchanges the
code for an \emph{access token} at the so-called \emph{token endpoint}
of the AS in Steps~\refprotostep{oauth-code:c-token-req} and
\refprotostep{oauth-code:c-token-resp}. For this purpose, the client
might be required to authenticate to the AS (see below). With this
access token, the client can finally access the resources at the
resource server (RS), as shown in
Steps~\refprotostep{oauth-code:resource-req} and
\refprotostep{oauth-code:resource-resp}.

The RS can use different methods to check the validity of
an access token presented by a client. The access token can, for
example, be a document signed by the AS containing all necessary information. Often, the access
token is not a structured document but a nonce. In this case, the
RS uses Token
Introspection~\cite{rfc7662-oauth-token-introspection}, i.e., it sends
the access token to the \emph{introspection endpoint} of the AS and
receives the information associated with the token from the AS.
An RS typically has only one (fixed) AS, which means that when
the RS receives an access token, it sends the introspection request to this
AS.

\subsubsection*{Public and Confidential
  Clients} \label{oauth-clients-pub-conf} Depending on whether a
client can keep long-term secrets, it is either called a \emph{public}
or a \emph{confidential} client. If the client is not able to maintain
secrets, as is typically the case for applications running on end-user
devices, the client is not required to authenticate itself at the
token endpoint of the AS. These kinds of clients are called
\emph{public} clients. Clients able to maintain secrets, such as web
server clients, must authenticate to the token endpoint (in
Step~\refprotostep{oauth-code:c-token-req} of
Figure~\ref{fig:oauth-overview}) and are called \emph{confidential}
clients.

For confidential clients, client authentication ensures that only a
legitimate client can exchange the authorization code for an access
token. OAuth~2.0 allows for several methods for client authentication
at the token endpoint, including sending a password or proving
possession of a secret \cite[Section 2.3]{rfc6749-oauth2}. 
For public clients, other measures are available, such as PKCE (see below), to obtain a sufficient level of security. 

\subsubsection{OpenID Connect}\label{sec:openidconnect}
OAuth~2.0 is built for \emph{authorization} only, i.e., the client
gets access to the resources of the user only if the user consented to
this access. It does not per se provide \emph{authentication}, i.e., proving
the identity of the user to the client. This is what OpenID Connect~\cite{openid-connect-core-1.0} was developed for. It adds an \emph{id
  token} to OAuth~2.0 which is issued by the AS and
contains identity information about the end-user. ID tokens can be
issued in the response from the authorization endpoint 
(Step~\refprotostep{oauth-code:auth-resp-1} of Figure~\ref{fig:oauth-overview})  and/or
at the token endpoint (Step~\refprotostep{oauth-code:c-token-resp} of Figure~\ref{fig:oauth-overview}).
They are signed by the AS and can be bound to other parameters of the response, such as the
hash of authorization codes or access
tokens. %
Therefore, they can also be used to protect responses against modification.

\subsection{Proof Key for Code Exchange}  \label{section-pkce}

The \emph{Proof Key for Code Exchange} (PKCE) extension (RFC 7636)\highlightIfDiffMinor{was}initially
created for OAuth public clients and
independently of the FAPI. Its goal is to protect against the
use of intercepted authorization codes. Before we explain how it works, we introduce the attack scenario against which PKCE should
protect according to RFC 7636.

This attack starts with the leakage of the
authorization code
after the browser receives it in the response from the
authorization endpoint (Step~\refprotostep{oauth-code:auth-resp-1}
of Figure~\ref{fig:oauth-overview}). A multitude of problems can lead to a leak of
the code, even if TLS is used to protect the network communication:
\begin{itemize}
\item On mobile operating systems, multiple apps can register
  themselves onto the same custom URI scheme (e.g.,
  \nolinkurl{some-app://redirection-response}). When receiving the
  authorization response, the operating system may forward the
  response (and the code) to a malicious app instead of the honest app
  (see \cite[Section~1]{rfc7636} and
  \cite[Section~8.1]{DennissBradley-RFC-2017}).
\item Mix-up attacks, in which a different AS is used than the client
  expects (see~\cite{FettKuestersSchmitz-CCS-2016} for details), can be
  used to leak an authorization code to a malicious server.
\item As highlighted in \cite{FettKuestersSchmitz-CSF-2017}, a
  Referer header can leak the code to an adversary.
\item The code can also appear in HTTP logs that can be disclosed
  (accidentally) to third parties or (intentionally) to
  administrators.
\end{itemize}
In a setting with a public client (i.e., without client authentication
at the token endpoint), an authorization code leaked to the attacker
can be redeemed directly by the attacker at the authorization server
to obtain an access token.

RFC 7636 aims to protect against such attacks even if not only the
authorization response leaks but also the authorization request as
well. Such leaks can happen, for example, from HTTP logs (Precondition
4b of Section 1 of RFC 7636) or unencrypted HTTP connections.

PKCE works as follows: Before sending the authorization request, the client creates a random value
called \emph{code verifier}. The client then creates the \emph{code
  challenge} by hashing the verifier\footnote{If it is assumed that the
  authorization request never leaks to the attacker, it is sufficient
  and allowed by RFC 7636 to use the verifier as the challenge, i.e.,
  without hashing.} and includes the challenge in the
authorization request (Step~\refprotostep{oauth-code:auth-req-1} of
Figure~\ref{fig:oauth-overview}). The AS associates the
generated authorization code with this challenge. Now, when the client redeems the
code in the request to the token endpoint
(Step~\refprotostep{oauth-code:c-token-req} of
Figure~\ref{fig:oauth-overview}), it includes the code verifier in the
token request. This message is sent directly to the AS and protected by TLS, which means that the verifier cannot be
intercepted. The idea is that if the authorization code leaked to the attacker, the attacker still cannot redeem the code to obtain the access token since he does not know the code verifier.

\subsection{Client Authentication using JWS Client Assertions} \label{cauthn-jws-client-assertion} %
As mentioned above, the goal of client authentication
 is to bind an authorization code to a certain
confidential client such that only this client can redeem the code at the AS. One
method for client authentication  is the use of JWS Client
Assertions \cite[Section 9]{openid-connect-core-1.0}, which requires proving possession of a key instead of
sending a password directly to the authorization server, as in plain
OAuth~2.0.

To this end, the client first generates a short document containing its
client identifier and the URI of the token endpoint. Now, depending on
whether the \emph{client secret} is a private (asymmetric) or a
symmetric key, the client either signs or MACs this document. It is then appended to the token request
(Step~\refprotostep{oauth-code:c-token-req} of
Figure~\ref{fig:oauth-overview}). As the document contains the URI of
the receiver, attacks in which the attacker tricks the client into
using a wrong URI are prevented, as the attacker cannot reuse the
document for the real endpoint (cf.
Section~\ref{sec:FAPImisconfiguredTEassumption}).
Technically, the short document is encoded as a JSON Web
Token (JWT)~\cite{rfc7519-jwt} to which its signature/MAC is attached to create
a so-called JSON Web Signature (JWS)~\cite{rfc7515-jws}.

\subsection{OAuth 2.0 Mutual TLS} \label{fapi:mTLS} %
\label{fapi:mTLS-CAuthN} \label{fapi:mTLS-TB} 

\emph{OAuth 2.0 Mutual TLS for Client Authentication and Certificate 
Bound Access Tokens} (mTLS)~\cite{ietf-oauth-mtls-09} provides a method for both client authentication
and token binding.

OAuth~2.0 Mutual TLS Client Authentication makes use of 
\emph{TLS
client authentication}\footnote{As noted in \cite{ietf-oauth-mtls-09}, Section 5.1  %
this extension supports all TLS versions with
certificate-based client authentication.} at the token endpoint (in
Step~\refprotostep{oauth-code:c-token-req} of
Figure~\ref{fig:oauth-overview}). In TLS client authentication, not
only the server authenticates to the client (as is common for TLS) but the client also authenticates to the server. To this
end, the client proves that it knows the private key belonging to a
certificate that is either (a) self-signed and preconfigured at the
respective AS or that is (b) issued for the respective client id by a predefined certificate authority within a public key infrastructure (PKI).

Token binding means binding an access token to a client such that only this
client is able to use the access token at the RS.
To achieve this, the AS associates the access token with
the certificate used by the client for the TLS connection to the token
endpoint. In the TLS connection  to the RS
(in Step~\refprotostep{oauth-code:resource-req} of
Figure~\ref{fig:oauth-overview}), the client then authenticates using the
same certificate. The RS
accepts the access token only if the client certificate is the one associated with the access token.\footnote{As mentioned above,
  the RS can read this information either directly from
  the access token if it is a signed document, or uses token
  introspection to retrieve the data from the AS.}

\subsection{OAuth 2.0 Token Binding} \label{fapi:OAUTB}

\emph{OAuth 2.0 Token Binding}
(OAUTB)~\cite{draft-ietf-oauth-token-binding-07} is used to
bind access tokens and/or authorization codes to certain TLS connections.
It is based on the \emph{Token Binding}
protocol~\cite{rfc8471,rfc8472,rfc8473,rfc5705} and\highlightIfDiffMinor{can be
used with all TLS versions.}In the following, we first sketch token binding in general before we explain  OAuth 2.0 Token Binding.

\subsubsection{Basics}
For simplicity of presentation, in the following, we assume that a
browser connects to a web server. The protocol remains the same if the
browser is replaced by another server. (In the context of OAuth 2.0,
in some settings in fact the client takes the role of the browser as explained below.)

At its core, token binding works as follows: When a web server
indicates (during TLS connection establishment) that it wants to use
token binding, the browser making the HTTP request over this TLS %
connection creates a public/private key pair for the web
server's origin. It then sends the public key to the server and proves
possession of the private key by using it to create a signature over a
value unique to the current TLS connection. Since the browser re-uses
the same key pair for future connections to the same origin, the web
server will be able to unambiguously recognize the browser in future
visits.

Central for the security of token binding is that the private key
remains secret inside the browser. To prevent replay attacks, the browser has to prove
possession of the private key by signing a value that is unique for
each TLS session. To this end, token binding uses the
\emph{Exported Keying Material} (EKM) of the TLS connection, a value
derived from data of the TLS handshake between the two participants,
as specified in \cite{rfc5705}. As long as at least one party follows
the protocol, the EKM will be unique for each TLS connection.

We can now illustrate the usage of token binding in the context of a
simplified protocol in which a browser $B$ requests a token from a
server $S$:
First, $B$ initiates a TLS connection to $S$, where $B$ and $S$ use TLS
extensions~\cite{rfc8472} to negotiate the use of token
binding and technical details thereof. Browser $B$ then creates a
public/private key pair $(k_{\mi{B},\mi{S}}, k'_{\mi{B},\mi{S}})$ for
the origin of $S$, unless such a key pair exists already. The public
key $k_{\mi{B},\mi{S}}$ (together with technical details about the key, such as its bit
length) is called \emph{Token Binding ID} (for the specific origin).

When sending the first HTTP request over the established TLS connection, $B$ includes in an HTTP header the
so-called \emph{Token Binding Message}:
\begin{align}\label{eq:firstTBMsg}
 \mathsf{TB\mhyphen{}Msg}[k_{\mi{B},\mi{S}}, \mathsf{sig}(\mi{EKM},k'_{\mi{B},\mi{S}})]
\end{align}
It contains both the Token Binding ID (i.e., essentially
$k_{\mi{B},\mi{S}}$) and the signed EKM value from the TLS connection,
as specified in~\cite{ietf-tokbind-https-17}. The server $S$ checks the
signature using $k_{\mi{B},\mi{S}}$ as included in this message and
then creates a token and associates it with the Token Binding ID as
the unique identifier of the browser.

When $B$ wants to redeem the token in a new TLS connection to $S$, $B$
creates a new Token Binding Message using the same Token Binding ID,
but signs the new EKM value:
\begin{align}\label{eq:secondTBMsg}
  \mathsf{TB\mhyphen{}Msg}[k_{\mi{B},\mi{S}}, \mathsf{sig}(\mi{\overline{EKM}}, k'_{\mi{B},\mi{S}})]
\end{align}
As the EKM values are unique to each TLS connection, $S$ concludes
that the sender of the message knows the private key of the Token
Binding ID, and as the sender used the same Token Binding ID as
before, the same party that requested the token in the first request
is using it now.

The above describes the simple situation that $B$
wants to redeem the token received from $S$ again at $S$, i.e., from the
same origin. In this case, we call the token binding message in
(\ref{eq:firstTBMsg}) a \emph{provided} token binding message. If $B$
wants to redeem the token received from $S$ at another origin, say at $C$,
then instead of just sending the provided token message in (\ref{eq:firstTBMsg}), $B$ would in
addition also send the so-called \emph{referred} token binding message, i.e., instead of (\ref{eq:firstTBMsg}) B would send 
\begin{equation}
  \begin{aligned}\label{eq:thirdTBMsg}
    \mathsf{TB\mhyphen{}prov\mhyphen{}Msg}[k_{\mi{B},\mi{S}}, \mathsf{sig}(\mi{EKM},k'_{\mi{B},\mi{S}})], \\ 
    \mathsf{TB\mhyphen{}ref\mhyphen{}Msg}[k_{\mi{B},\mi{C}}, \mathsf{sig}(\mi{EKM}, k'_{\mi{B},\mi{C}})].
  \end{aligned}
\end{equation}
Note that the EKM is the same in both messages, namely the EKM value
of the TLS connection between $B$ and $S$ (rather than between $B$ and
$C$, which has not happened yet anyway). Later when $B$ wants to redeem the token at $C$, $B$ would use
$k_{\mi{B},\mi{C}}$ in its (provided) token message to $C$.

\subsubsection{Token Binding for OAuth}

In the following, we explain how token binding is used in OAuth in the
case of app clients. The case of web server clients is discussed
below.

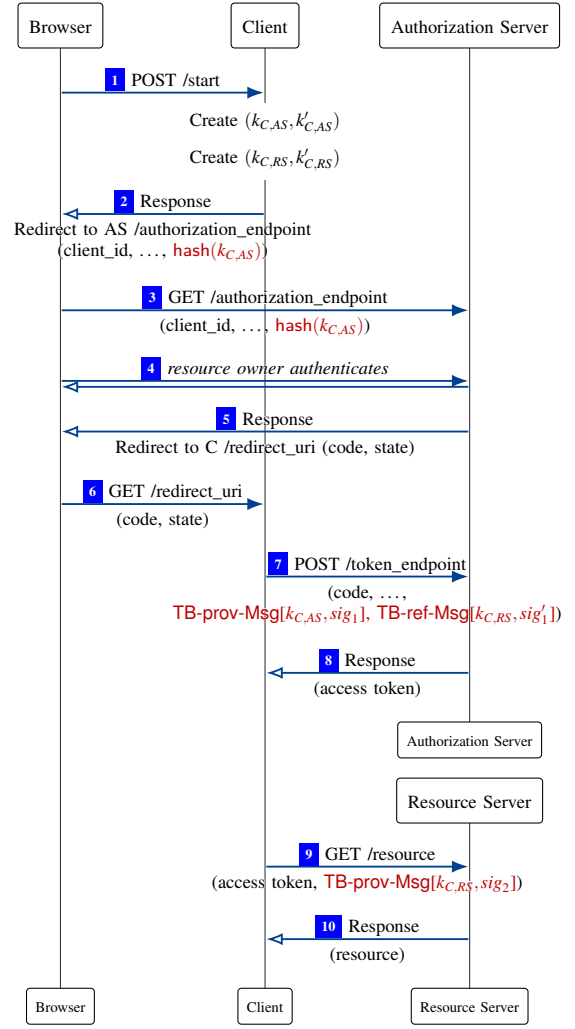
\begin{figure}[tbh]\label{fig:OAUTB}
   \centering  
           \begin{tikzpicture}[]
        \pgfdeclarelayer{arrows}
        \pgfdeclarelayer{groups}
        \pgfdeclarelayer{markers}
        \pgfsetlayers{groups,arrows,main,markers}

        \matrix [column sep={0.15\textwidth,between origins}, row sep=0.1ex]
        {
        \node[annex_matrix_node,inner sep=0,outer sep=0](pos-0-0){}; &\node[annex_matrix_node,inner sep=0,outer sep=0](pos-1-0){}; &\node[annex_matrix_node,inner sep=0,outer sep=0](pos-2-0){};\node[annex_matrix_dummy_height,minimum height=5ex,anchor=center]{};\node[annex_matrix_dummy_height,minimum height=5ex,anchor=center]{};\node[annex_matrix_dummy_height,minimum height=5ex,anchor=center]{};\\
\node[annex_matrix_node,inner sep=0,outer sep=0](pos-0-1){}; &\node[annex_matrix_node,inner sep=0,outer sep=0](pos-1-1){}; &\node[annex_matrix_node,inner sep=0,outer sep=0](pos-2-1){};\node[annex_matrix_dummy_height,minimum height=4ex,anchor=south,yshift=-1ex]{};\\
\node[annex_matrix_node,inner sep=0,outer sep=0](pos-0-2){}; &\node[annex_matrix_node,inner sep=0,outer sep=0](pos-1-2){}; &\node[annex_matrix_node,inner sep=0,outer sep=0](pos-2-2){};\node[annex_matrix_dummy_height,minimum height=3ex,anchor=center]{};\\
\node[annex_matrix_node,inner sep=0,outer sep=0](pos-0-3){}; &\node[annex_matrix_node,inner sep=0,outer sep=0](pos-1-3){}; &\node[annex_matrix_node,inner sep=0,outer sep=0](pos-2-3){};\node[annex_matrix_dummy_height,minimum height=3ex,anchor=center]{};\\
\node[annex_matrix_node,inner sep=0,outer sep=0](pos-0-4){}; &\node[annex_matrix_node,inner sep=0,outer sep=0](pos-1-4){}; &\node[annex_matrix_node,inner sep=0,outer sep=0](pos-2-4){};\node[annex_matrix_dummy_height,minimum height=4ex+2ex+2ex,anchor=north,yshift=3ex]{};\\
\node[annex_matrix_node,inner sep=0,outer sep=0](pos-0-5){}; &\node[annex_matrix_node,inner sep=0,outer sep=0](pos-1-5){}; &\node[annex_matrix_node,inner sep=0,outer sep=0](pos-2-5){};\node[annex_matrix_dummy_height,minimum height=4ex+2ex,anchor=north,yshift=3ex]{};\\
\node[annex_matrix_node,inner sep=0,outer sep=0](pos-0-6){}; &\node[annex_matrix_node,inner sep=0,outer sep=0](pos-1-6){}; &\node[annex_matrix_node,inner sep=0,outer sep=0](pos-2-6){};\node[annex_matrix_dummy_height,minimum height=4ex,anchor=south,yshift=-1ex]{};\\
\node[annex_matrix_node,inner sep=0,outer sep=0](pos-0-7){}; &\node[annex_matrix_node,inner sep=0,outer sep=0](pos-1-7){}; &\node[annex_matrix_node,inner sep=0,outer sep=0](pos-2-7){};\node[annex_matrix_dummy_height,minimum height=4ex+2ex,anchor=north,yshift=3ex]{};\\
\node[annex_matrix_node,inner sep=0,outer sep=0](pos-0-8){}; &\node[annex_matrix_node,inner sep=0,outer sep=0](pos-1-8){}; &\node[annex_matrix_node,inner sep=0,outer sep=0](pos-2-8){};\node[annex_matrix_dummy_height,minimum height=4ex+2ex,anchor=north,yshift=3ex]{};\\
\node[annex_matrix_node,inner sep=0,outer sep=0](pos-0-9){}; &\node[annex_matrix_node,inner sep=0,outer sep=0](pos-1-9){}; &\node[annex_matrix_node,inner sep=0,outer sep=0](pos-2-9){};\node[annex_matrix_dummy_height,minimum height=4ex+2ex+2ex,anchor=north,yshift=3ex]{};\\
\node[annex_matrix_node,inner sep=0,outer sep=0](pos-0-10){}; &\node[annex_matrix_node,inner sep=0,outer sep=0](pos-1-10){}; &\node[annex_matrix_node,inner sep=0,outer sep=0](pos-2-10){};\node[annex_matrix_dummy_height,minimum height=4ex+2ex,anchor=north,yshift=3ex]{};\\
\node[annex_matrix_node,inner sep=0,outer sep=0](pos-0-11){}; &\node[annex_matrix_node,inner sep=0,outer sep=0](pos-1-11){}; &\node[annex_matrix_node,inner sep=0,outer sep=0](pos-2-11){};\node[annex_matrix_dummy_height,minimum height=5ex,anchor=center]{};\\
\node[annex_matrix_node,inner sep=0,outer sep=0](pos-0-12){}; &\node[annex_matrix_node,inner sep=0,outer sep=0](pos-1-12){}; &\node[annex_matrix_node,inner sep=0,outer sep=0](pos-2-12){};\node[annex_matrix_dummy_height,minimum height=5ex,anchor=center]{};\\
\node[annex_matrix_node,inner sep=0,outer sep=0](pos-0-13){}; &\node[annex_matrix_node,inner sep=0,outer sep=0](pos-1-13){}; &\node[annex_matrix_node,inner sep=0,outer sep=0](pos-2-13){};\node[annex_matrix_dummy_height,minimum height=4ex+2ex,anchor=north,yshift=3ex]{};\\
\node[annex_matrix_node,inner sep=0,outer sep=0](pos-0-14){}; &\node[annex_matrix_node,inner sep=0,outer sep=0](pos-1-14){}; &\node[annex_matrix_node,inner sep=0,outer sep=0](pos-2-14){};\node[annex_matrix_dummy_height,minimum height=4ex+2ex,anchor=north,yshift=3ex]{};\\
\node[annex_matrix_node,inner sep=0,outer sep=0](pos-0-15){}; &\node[annex_matrix_node,inner sep=0,outer sep=0](pos-1-15){}; &\node[annex_matrix_node,inner sep=0,outer sep=0](pos-2-15){};\node[annex_matrix_dummy_height,minimum height=5ex,anchor=center]{};\node[annex_matrix_dummy_height,minimum height=5ex,anchor=center]{};\node[annex_matrix_dummy_height,minimum height=5ex,anchor=center]{};\\
};

\node[name=StartParty_9_0,annex_start_party_box,] at (pos-0-0) {Browser};

\node[name=StartParty_10_0,annex_start_party_box,] at (pos-1-0) {Client};

\node[name=StartParty_11_0,annex_start_party_box,] at (pos-2-0) {Authorization Server};

\node[annex_action,name=Action_13_0] at (pos-1-2) {\contour{white}{Create $(k_{\mi{C},\mi{AS}}, k'_{\mi{C},\mi{AS}})$}};

\node[annex_action,name=Action_14_0] at (pos-1-3) {\contour{white}{Create $(k_{\mi{C},\mi{RS}}, k'_{\mi{C},\mi{RS}})$}};

\node[name=EndParty_22_0,annex_end_party_box,] at (pos-2-11) {Authorization Server};

\node[name=StartParty_23_0,annex_start_party_box,] at (pos-2-12) {Resource Server};

\node[name=EndParty_26_0,annex_end_party_box,] at (pos-1-15) {Client};

\node[name=EndParty_27_0,annex_end_party_box,] at (pos-2-15) {Resource Server};

\node[name=EndParty_28_0,annex_end_party_box,] at (pos-0-15) {Browser};

\begin{pgfonlayer}{arrows}%

\draw[annex_lifeline] (pos-0-0) -- (pos-0-15);

\draw[annex_lifeline] (pos-1-0) -- (pos-1-15);

\draw[annex_lifeline] (pos-2-0) -- (pos-2-11);

        \draw[annex_http_request] (pos-0-1) to node [annex_arrow_text,above=2.6pt,anchor=base](HTTPRequest_12_0){\setcounter{protostep}{0}\protostep{oauth-code-oautb-app-client:auth-req-1-start} \contour{white}{POST \nolinkurl{/start}}}  (pos-1-1);

        \draw[annex_http_response] (pos-1-4) to node [annex_arrow_text,above=2.6pt,anchor=base](HTTPResponse_15_0){\setcounter{protostep}{1}\protostep{oauth-code-oautb-app-client:auth-req-1} \contour{white}{Response}} node [annex_arrow_text,below=8pt,anchor=base](HTTPResponse_15_1){\contour{white}{Redirect to AS \nolinkurl{/authorization_endpoint}}} node [annex_arrow_text,below=8pt+8pt,anchor=base](HTTPResponse_15_2){\contour{white}{(client\_id, \dots, $\textcolor{figurehighlight1}{\mathsf{hash}(k_{\mi{C},\mi{AS}})}$)}}  (pos-0-4);

        \draw[annex_http_request] (pos-0-5) to node [annex_arrow_text,above=2.6pt,anchor=base](HTTPRequest_16_0){\setcounter{protostep}{2}\protostep{oauth-code-oautb-app-client:auth-req-2} \contour{white}{GET \nolinkurl{/authorization_endpoint}}} node [annex_arrow_text,below=8pt,anchor=base](HTTPRequest_16_1){\contour{white}{(client\_id, \dots, $\textcolor{figurehighlight1}{\mathsf{hash}(k_{\mi{C},\mi{AS}})}$)}}  (pos-2-5);

        \draw[annex_http_request,transform canvas={yshift=0.25ex}] (pos-0-6) to node [annex_arrow_text,above=2.6pt,anchor=base](HTTPRequestResponse_17_0){\setcounter{protostep}{3}\protostep{oauth-code-oautb-app-client:ro-authN} \contour{white}{\textit{resource owner authenticates}}} (pos-2-6);
        \draw[annex_http_response,transform canvas={yshift=-0.25ex}] (pos-2-6) to  (pos-0-6);

        \draw[annex_http_response] (pos-2-7) to node [annex_arrow_text,above=2.6pt,anchor=base](HTTPResponse_18_0){\setcounter{protostep}{4}\protostep{oauth-code-oautb-app-client:auth-resp-1} \contour{white}{Response}} node [annex_arrow_text,below=8pt,anchor=base](HTTPResponse_18_1){\contour{white}{Redirect to C \nolinkurl{/redirect\_uri} (code, state)}}  (pos-0-7);

        \draw[annex_http_request] (pos-0-8) to node [annex_arrow_text,above=2.6pt,anchor=base](HTTPRequest_19_0){\setcounter{protostep}{5}\protostep{oauth-code-oautb-app-client:auth-resp-2} \contour{white}{GET \nolinkurl{/redirect_uri}}} node [annex_arrow_text,below=8pt,anchor=base](HTTPRequest_19_1){\contour{white}{(code, state)}}  (pos-1-8);

        \draw[annex_http_request] (pos-1-9) to node [annex_arrow_text,above=2.6pt,anchor=base](HTTPRequest_20_0){\setcounter{protostep}{6}\protostep{oauth-code-oautb-app-client:c-token-req} \contour{white}{POST \nolinkurl{/token_endpoint}}} node [annex_arrow_text,below=8pt,anchor=base](HTTPRequest_20_1){\contour{white}{(code, \dots,}} node [annex_arrow_text,below=8pt+8pt,anchor=base](HTTPRequest_20_2){\contour{white}{\textcolor{figurehighlight1}{\textsf{TB-prov-Msg}[$k_{\mi{C},\mi{AS}}, \mi{sig}_1$], \textsf{TB-ref-Msg}[$k_{\mi{C},\mi{RS}}, \mi{sig}'_1$]})}}  (pos-2-9);

        \draw[annex_http_response] (pos-2-10) to node [annex_arrow_text,above=2.6pt,anchor=base](HTTPResponse_21_0){\setcounter{protostep}{7}\protostep{oauth-code-oautb-app-client:c-token-resp} \contour{white}{Response}} node [annex_arrow_text,below=8pt,anchor=base](HTTPResponse_21_1){\contour{white}{(access token)}}  (pos-1-10);

\draw[annex_lifeline] (pos-2-12) -- (pos-2-15);

        \draw[annex_http_request] (pos-1-13) to node [annex_arrow_text,above=2.6pt,anchor=base](HTTPRequest_24_0){\setcounter{protostep}{8}\protostep{oauth-code-oautb-app-client:resource-req} \contour{white}{GET \nolinkurl{/resource}}} node [annex_arrow_text,below=8pt,anchor=base](HTTPRequest_24_1){\contour{white}{(access token, \textcolor{figurehighlight1}{\textsf{TB-prov-Msg}[$k_{\mi{C},\mi{RS}}, \mi{sig}_2$]})}}  (pos-2-13);

        \draw[annex_http_response] (pos-2-14) to node [annex_arrow_text,above=2.6pt,anchor=base](HTTPResponse_25_0){\setcounter{protostep}{9}\protostep{oauth-code-oautb-app-client:resource-resp} \contour{white}{Response}} node [annex_arrow_text,below=8pt,anchor=base](HTTPResponse_25_1){\contour{white}{(resource)}}  (pos-1-14);

\end{pgfonlayer}

\begin{pgfonlayer}{markers}%

\end{pgfonlayer}
        \end{tikzpicture}
        
   \caption{OAUTB for App Clients}
   \label{fig:oautb-app-client-full-flow}
\end{figure}

The flow is shown in Figure~\ref{fig:oautb-app-client-full-flow}.
Note that in this case, token binding is used between the OAuth client
and the authorization and resource servers; the browser in
Figure~\ref{fig:oauth-overview} is not involved.

The client has two token binding key pairs, one for the AS and one for
the RS (if these key pairs do not already exist, the client creates
them during the flow). When sending the authorization request
(Step~\refprotostep{oauth-code-oautb-app-client:auth-req-1} of
Figure~\ref{fig:oautb-app-client-full-flow}), the client includes the
hash of the Token Binding ID it uses for the AS as a PKCE challenge
(cf.~Section~\ref{section-pkce}). When exchanging the code for an
access token in
Step~\refprotostep{oauth-code-oautb-app-client:c-token-req}, the
client proves possession of the private key of this Token Binding ID,
and the AS only accepts the request when the hash of the Token Binding
ID is the same as the PKCE challenge. Therefore, the code can only be
exchanged by the participant that created the authorization request.
Note that for this purpose the AS only takes the \emph{provided} token
binding message sent to the AS in
Step~\refprotostep{oauth-code-oautb-app-client:c-token-req} into
account. However, the AS also checks the validity of the
\emph{referred} token binding message (using the same EKM value) and
associates $k_{C,RS}$
with the token issued by the AS in
Step~\refprotostep{oauth-code-oautb-app-client:c-token-resp}.

The token binding ID $k_{C,RS}$ is used in Step~\refprotostep{oauth-code-oautb-app-client:resource-req} by the client to redeem the token at the RS. The RS then checks if this is the same
token binding ID that is associated with the access token. This
information can be contained in the access token if it is structured
and  readable by the RS or via token
introspection.

Altogether, Token Binding for OAuth (in the case of app clients) is
supposed to bind both the authorization code and the access token to
the client. That is, only the client who initiated the flow (in
Step~\refprotostep{oauth-code-oautb-app-client:auth-req-1}) can redeem
the authorization code at the AS and the corresponding access token at
the RS, and hence, get access to the resource at the RS.

\subsubsection{Binding Authorization Codes for Web Server Clients}  %

In the case that the client is a web server, the binding of the
authorization code to the client is already done by client
authentication, as a web server client is always confidential
(cf. Section~\ref{subsection:oauth}).  Therefore, the client does
\emph{not} include the hash of a Token Binding ID in the authorization
request (Step~\refprotostep{oauth-code-oautb-app-client:auth-req-1} of
Figure~\ref{fig:oautb-app-client-full-flow}). Instead, the mechanism
defined in OAUTB aims at binding the authorization code to the browser/client
pair. (The binding of the access token to the client is done in the
same way as for an app client).

More precisely, for web server clients, the authorization code is
bound to the token binding ID that the browser uses for the
client. For this purpose, the client includes an additional HTTP
header\highlightIfDiffMinor{in the first response}to the browser %
(Step~\refprotostep{oauth-code-oautb-app-client:auth-req-1} of
Figure~\ref{fig:oautb-app-client-full-flow}), which signals the
browser that it should give the token binding ID it uses for the
client to the authorization server. When sending the authorization
request to the authorization server in
Step~\refprotostep{oauth-code-oautb-app-client:auth-req-2}, the
browser thus includes a provided and a referred token binding message,
where the referred message contains the token binding ID, that the browser later uses for the
client (say, $k_{B,C}$). When generating the authorization code, the authorization
server associates the code with $k_{B,C}$.

When redirecting the code to the client in
Step~\refprotostep{oauth-code-oautb-app-client:auth-resp-2}, the browser includes a
token binding message for $k_{B,C}$, thereby proving possession of
the private key.

When sending the token request in
Step~\refprotostep{oauth-code-oautb-app-client:c-token-req}, the
client includes $k_{B,C}$. We highlight that the client does not send a token binding message for $k_{B,C}$ since the client does not know the corresponding private key (only the browser does). 

The authorization server checks if this key is the same token binding ID
it associated the authorization code with, and therefore, can check if the code was
redirected to the client by the same browser that made the
authorization request. 
In other words, by this the authorization code is bound to the browser/client pair.

\subsection{JWT Secured Authorization Response Mode}\label{sec:JWT}

The recently developed \emph{JWT Secured Authorization Response Mode}
(JARM)~\cite{jarm-ceb0f82} aims at protecting the OAuth
authorization response (Step~\refprotostep{oauth-code:auth-resp-1} of
Figure~\ref{fig:oauth-overview}) by having the AS sign (and optionally
encrypt) the response. The authorization response is then encoded as a JWT 
(see Section~\ref{cauthn-jws-client-assertion}).
The JARM extension can be used with any OAuth~2.0 flow.

In addition to the regular parameters of the authorization response,
the JWT also contains its issuer (identifying the AS) and its audience (client id). For example, if
combined with the Authorization Code Flow, the response JWT contains
the issuer, audience, authorization code, and state values. 

By using JARM, the authorization response is integrity protected and injection of leaked authorization codes is prevented. 

\section{The OpenID Financial-grade API}\label{sec:fapi}

The OpenID Financial-grade API~\cite{fapi-ceb0f82} currently comprises two implementer's drafts.
One defines a profile for read-only access, the other one for read-write access.
Building on
Section~\ref{sec:oauth-and-new-defence-mechanisms}, here we
describe both profiles and the various configurations in which these
profiles can run (see Figure~\ref{fig:fapi-overview}). Furthermore,
we explain the assumptions made within the FAPI standard and the
underlying OAuth 2.0 extensions.

\begin{figure}[t]
  \centering  
  \vspace{-2em}
  \scalebox{0.9}{
    \input{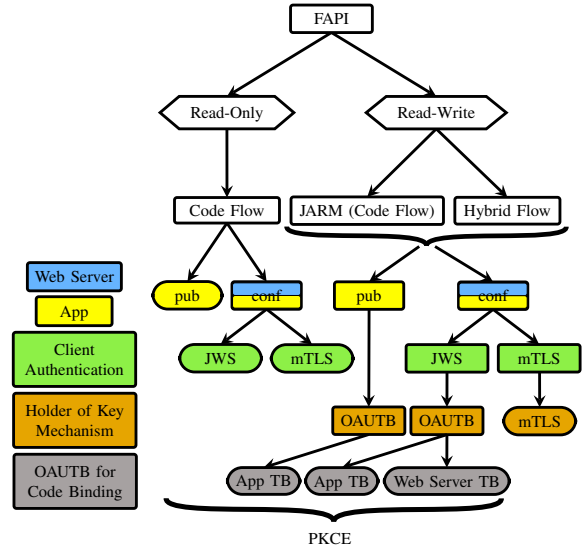}
  }
  \caption{Overview of the FAPI. One path (terminated by a box with
    rounded corners) describes one possible configuration of the FAPI.
    The paths marked with PKCE use PKCE. JARM and Hybrid flows both
    allow for the configurations shown.}
  \label{fig:fapi-overview}
\end{figure}

\subsection{Financial-grade API: Read-Only Profile}

In the following, we
explain the Read-Only flow as described in~\cite{fapi-ceb0f82-read-only}. %
The Read-Only profile %
aims at providing a secure way for accessing data that needs a higher
degree of protection than regular OAuth, e.g., for read access to
financial data. 

The Read-Only flow is essentially an \textbf{OAuth Authorization Code
  flow} (cf.~Section~\ref{sec:oauth-and-new-defence-mechanisms}).
Additionally, the client can request an ID Token (see Section~\ref{sec:openidconnect}) from the token
endpoint by adding a \emph{scope} parameter to the authorization
request (Step~\refprotostep{oauth-code:auth-req-1} of
Figure~\ref{fig:oauth-overview}) with the value $\str{openid}$.

In contrast to regular OAuth and OpenID Connect, the client is
required to have a different set of \textbf{redirection URIs} for each
authorization server. This separation prevents mix-up attacks, where
the authorization response (Step \refprotostep{oauth-code:auth-resp-2}  in 
Figure~\ref{fig:oauth-overview}) comes from a different AS than the client 
expects (see~\cite{FettKuestersSchmitz-CCS-2016} and
\cite{draft-ietf-oauth-security-topics}
for more details on mix-up attacks). 
When receiving the authorization response, the client checks if the
response was received at the redirection URI specified in the
authorization request (Step \refprotostep{oauth-code:auth-req-1} in Figure~\ref{fig:oauth-overview}).

One of the main additions to the regular OAuth flow is the use of
\textbf{PKCE} as explained in Section~\ref{section-pkce}. The PKCE
challenge is created by hashing a nonce.

The FAPI furthermore requires \textbf{confidential clients to
  authenticate} at the token endpoint  (in Step~\refprotostep{oauth-code:c-token-req} of
Figure~\ref{fig:oauth-overview}) using either \emph{JWS
  Client Assertions} (cf.~Section~\ref{cauthn-jws-client-assertion})
or \emph{Mutual TLS} (cf.~Section~\ref{fapi:mTLS-CAuthN}). 
Public clients do not use client authentication.

\subsection{Financial-grade API: Read-Write
  Profile} \label{subsection-fapi-read-write-profile}

The Read-Write profile~\cite{fapi-ceb0f82-read-write} aims at being secure under stronger assumptions
than the Read-Only profile, in order to be suitable for scenarios such
as write access to financial data. The full set of assumptions is
described in Section~\ref{sec-assumptions}.

The flow can be either an \textbf{OpenID Connect (OIDC) Hybrid flow}, which means that
both the authorization response (Step~\refprotostep{oauth-code:auth-resp-1} in Figure~\ref{fig:oauth-overview}) and the token 
response\highlightIfDiffMinor{(Step~\refprotostep{oauth-code:c-token-resp} in Figure~\ref{fig:oauth-overview})}contain an id %
token (see Section~\ref{sec:openidconnect}), or any other OAuth-based flow used together with \textbf{JARM} (see Section~\ref{sec:JWT}).
When using the Hybrid flow, the FAPI
profile also requires that the hash of the state value is included in
the first id token.

In addition to the parameters of the Read-Only flow, the authorization
request prepared by the client
(Step~\refprotostep{oauth-code:auth-req-1} of
Figure~\ref{fig:oauth-overview}) is required to contain a
\textbf{request JWS},\highlightIfDiffMinor{which is a JWT, signed by the client, containing all request
parameters together with the audience of the request (cf.~Section~\ref{cauthn-jws-client-assertion}).}

One of the main security features of the profile is the \textbf{binding of the
  authorization code and the access token} to the client, which is achieved by using
either 
mTLS (cf.~Section~\ref{fapi:mTLS})
or
OAUTB (OAuth~2.0 Token Binding, see Section~\ref{fapi:OAUTB}).
A public client is
required to use OAUTB, while a confidential client can use either
OAUTB or mTLS.

If the client is a confidential client using mTLS, the request does
not contain a PKCE challenge. When using OAUTB, the client uses a
\textbf{variant of PKCE}, depending on whether the client is a web
server client or an app client (cf.~Section~\ref{fapi:OAUTB}).

In the case of a confidential client, the \textbf{client
  authentication at the token endpoint} is done in the same way as for
the Read-Only flow, i.e., by using either JWS Client Assertions
(cf.~Section~\ref{cauthn-jws-client-assertion}) or Mutual TLS
(cf.~Section~\ref{fapi:mTLS-CAuthN}).

\subsection{Overview of Assumptions and Mitigations} \label{sec-assumptions}
In the following, we explain the conditions under which the FAPI
profiles and the OAuth extensions 
aim to be secure according to their specifications.

\subsubsection{Leak of Authorization Response}\label{sec:FAPIleakauthorizationresponseassumption} As described in
Section~\ref{section-pkce} in the context of PKCE, there are several
scenarios in which the authorization response
(Step~\refprotostep{oauth-code:auth-resp-2} of
Figure~\ref{fig:oauth-overview}), and hence, the authorization code,
can leak to the attacker (in clear), in particular in the case of app
clients. In our model of the FAPI, we therefore assume
that the authorization response is given to the attacker if the client
is an app. %
At first glance, leakage of the authorization code is indeed mitigated
by the use of PKCE since an attacker does not know the code verifier,
and hence, cannot redeem the code at the AS.  However, our attack
described in Section~\ref{attack:pkce} shows that the protection
provided by PKCE can be circumvented.

\subsubsection{Leak of Authorization
  Request} \label{sec:FAPIleakofauthorizationrequestassumption} The
Read-Only profile of the FAPI explicitly states that the PKCE
challenge should be created by hashing the verifier. The use of
hashing should protect the PKCE challenge even if the authorization
request leaks (e.g., by leaking HTTP logs,
cf.~Section~\ref{section-pkce}), and therefore, we assume in our model
that the authorization request
(Step~\refprotostep{oauth-code:auth-req-1} of
Figure~\ref{fig:oauth-overview}) leaks to the attacker.

\subsubsection{Leak of Access
  Token} \label{sec:FAPIleakofaccesstokenassumption} In the Read-Write
profile, it is  assumed that the access token might leak due to
phishing \cite[Section~8.3.5]{fapi-ceb0f82-read-write}. In our model, we therefore assume that the access token might leak in
Step~\refprotostep{oauth-code:auth-resp-1} of
Figure~\ref{fig:oauth-overview}. This problem is seemingly mitigated
by using either mTLS or OAUTB, which bind the access token to the
legitimate client, and hence, only the legitimate client should be
able to redeem the access token at the RS even if the access token
leaked. The FAPI specification states: ``When the FAPI client uses
MTLS or OAUTB, the access token is bound to the TLS channel, it is
access token phishing resistant as the phished access tokens cannot be
used.''~\cite[Section 8.3.5]{fapi-ceb0f82-read-write}. However, our
attack presented in Section~\ref{attack4-phished-at-malicious-as}
shows that this is not the case.

\subsubsection{Misconfigured Token
  Endpoint} \label{sec:FAPImisconfiguredTEassumption} An explicit
design decision by the FAPI working group was to make the Read-Write
profile secure even if the token request
(Step~\refprotostep{oauth-code:c-token-req} of
Figure~\ref{fig:oauth-overview}) leaks. The FAPI specification
describes this attack as follows: ``In this attack, the client
developer is social engineered into believing that the token endpoint
has changed to the URL that is controlled by the attacker. As the
result, the client sends the code and the client secret to the
attacker, which will be replayed subsequently.''~\cite[Section
8.3.2]{fapi-ceb0f82-read-write}.

Therefore, we make this assumption also in our FAPI model. Seemingly,
this problem is mitigated by code binding through client
authentication or OAUTB, which means that the attacker cannot use the
stolen code at the legitimate token endpoint. ``When the FAPI client
uses MTLS or OAUTB, the authorization code is bound to the TLS
channel, any phished client credentials and authorization codes
submitted to the token endpoint cannot be used since the authorization
code is bound to a particular TLS channel.''~\cite[Section
8.3.2]{fapi-ceb0f82-read-write}. Note that in the FAPI the client does
not authenticate by using the client secret as a password, but by
proving possession (either using JWS Client Assertions or mTLS), which
means that the attacker cannot reuse credentials.

However, our attack presented in Section~\ref{attack1-wrong-token-ep}
shows that this intuition is misleading.

\section{Attacks} \label{sec:att}

As already\highlightIfDiffMinor{mentioned}in the introduction, in Section~\ref{sec:analysis}
we present our rigorous formal analysis of the FAPI based on the Web
Infrastructure Model. Through this formal analysis of the FAPI with
the various OAuth 2.0 extensions it uses, we not only found attacks on
the FAPI but also on some of the OAuth 2.0 extensions, showing that
(1) these extensions do not achieve the security properties they have
been designed for and (2) that combining these extensions in a secure
way is far from trivial. Along with the attacks, we also propose fixes
to the standards. Our formal analysis presented in
Section~\ref{sec:analysis} considers the fixed versions.

We start by describing two attacks on Token Binding, followed by an
attack on PKCE, and one vulnerability hidden in the assumptions of
PKCE.

We emphasize that our attacks work even if all communication uses
TLS and even if the attacker is merely a web attacker, i.e., does
not control the network but only certain parties.

As already mentioned in the introduction, we notified the OpenID FAPI
Working Group of the attacks found by our analysis and are working
together with them to fix the standard.

\FloatBarrier
\subsection{Cuckoo's Token Attack } \label{attack4-phished-at-malicious-as}

As explained in Section~\ref{sec:FAPIleakofaccesstokenassumption},
the Read-Write profile of the FAPI aims at providing security even 
if the attacker obtains an access token, e.g., due to phishing.
Intuitively, this protection seems to be achieved by binding the access
token to the client via mTLS (see Section~\ref{fapi:mTLS}) 
or OAUTB (see Section~\ref{fapi:OAUTB}).

However, these mechanisms prevent the attacker only from directly
using the access token in the same flow. As illustrated next, in a
second flow, the attacker \emph{can} inject the bound access token and
let the client (to which the token is bound) use this token, which
enables the attacker to access resources belonging to an honest
identity.

This attack affects all configurations of the Read-Write
profile (see Figure~\ref{fig:fapi-overview}). Also, the Read-Only
profile is vulnerable to this attack; this profile is, however, not
meant to defend against stolen access tokens.

We note that the underlying principle of the attack should be relevant
to other use-cases of token binding as well, i.e., whenever a token is
bound to a participant, the involuntary use of a leaked token (by the
participant to which the token is bound) should be prevented.

\begin{figure}[tbh]
  \centering  
          \begin{tikzpicture}[]
        \pgfdeclarelayer{arrows}
        \pgfdeclarelayer{groups}
        \pgfdeclarelayer{markers}
        \pgfsetlayers{groups,arrows,main,markers}

        \matrix [column sep={0.15\textwidth,between origins}, row sep=0.1ex]
        {
        \node[annex_matrix_node,inner sep=0,outer sep=0](pos-0-0){}; &\node[annex_matrix_node,inner sep=0,outer sep=0](pos-1-0){}; &\node[annex_matrix_node,inner sep=0,outer sep=0](pos-2-0){};\node[annex_matrix_dummy_height,minimum height=5ex,anchor=center]{};\node[annex_matrix_dummy_height,minimum height=5ex,anchor=center]{};\node[annex_matrix_dummy_height,minimum height=5ex,anchor=center]{};\\
\node[annex_matrix_node,inner sep=0,outer sep=0](pos-0-1){}; &\node[annex_matrix_node,inner sep=0,outer sep=0](pos-1-1){}; &\node[annex_matrix_node,inner sep=0,outer sep=0](pos-2-1){};\node[annex_matrix_dummy_height,minimum height=4ex,anchor=south,yshift=-1ex]{};\\
\node[annex_matrix_node,inner sep=0,outer sep=0](pos-0-2){}; &\node[annex_matrix_node,inner sep=0,outer sep=0](pos-1-2){}; &\node[annex_matrix_node,inner sep=0,outer sep=0](pos-2-2){};\node[annex_matrix_dummy_height,minimum height=4ex+2ex,anchor=north,yshift=3ex]{};\\
\node[annex_matrix_node,inner sep=0,outer sep=0](pos-0-3){}; &\node[annex_matrix_node,inner sep=0,outer sep=0](pos-1-3){}; &\node[annex_matrix_node,inner sep=0,outer sep=0](pos-2-3){};\node[annex_matrix_dummy_height,minimum height=4ex+2ex,anchor=north,yshift=3ex]{};\\
\node[annex_matrix_node,inner sep=0,outer sep=0](pos-0-4){}; &\node[annex_matrix_node,inner sep=0,outer sep=0](pos-1-4){}; &\node[annex_matrix_node,inner sep=0,outer sep=0](pos-2-4){};\node[annex_matrix_dummy_height,minimum height=4ex+2ex,anchor=north,yshift=3ex]{};\\
\node[annex_matrix_node,inner sep=0,outer sep=0](pos-0-5){}; &\node[annex_matrix_node,inner sep=0,outer sep=0](pos-1-5){}; &\node[annex_matrix_node,inner sep=0,outer sep=0](pos-2-5){};\node[annex_matrix_dummy_height,minimum height=4ex+2ex,anchor=north,yshift=3ex]{};\\
\node[annex_matrix_node,inner sep=0,outer sep=0](pos-0-6){}; &\node[annex_matrix_node,inner sep=0,outer sep=0](pos-1-6){}; &\node[annex_matrix_node,inner sep=0,outer sep=0](pos-2-6){};\node[annex_matrix_dummy_height,minimum height=5ex,anchor=center]{};\\
\node[annex_matrix_node,inner sep=0,outer sep=0](pos-0-7){}; &\node[annex_matrix_node,inner sep=0,outer sep=0](pos-1-7){}; &\node[annex_matrix_node,inner sep=0,outer sep=0](pos-2-7){};\node[annex_matrix_dummy_height,minimum height=5ex,anchor=center]{};\\
\node[annex_matrix_node,inner sep=0,outer sep=0](pos-0-8){}; &\node[annex_matrix_node,inner sep=0,outer sep=0](pos-1-8){}; &\node[annex_matrix_node,inner sep=0,outer sep=0](pos-2-8){};\node[annex_matrix_dummy_height,minimum height=4ex+2ex,anchor=north,yshift=3ex]{};\\
\node[annex_matrix_node,inner sep=0,outer sep=0](pos-0-9){}; &\node[annex_matrix_node,inner sep=0,outer sep=0](pos-1-9){}; &\node[annex_matrix_node,inner sep=0,outer sep=0](pos-2-9){};\node[annex_matrix_dummy_height,minimum height=4ex+2ex,anchor=north,yshift=3ex]{};\\
\node[annex_matrix_node,inner sep=0,outer sep=0](pos-0-10){}; &\node[annex_matrix_node,inner sep=0,outer sep=0](pos-1-10){}; &\node[annex_matrix_node,inner sep=0,outer sep=0](pos-2-10){};\node[annex_matrix_dummy_height,minimum height=5ex,anchor=center]{};\node[annex_matrix_dummy_height,minimum height=5ex,anchor=center]{};\node[annex_matrix_dummy_height,minimum height=5ex,anchor=center]{};\\
};

\node[name=StartParty_9_0,annex_start_party_box,attacker] at (pos-0-0) {Attacker (User)};

\node[name=StartParty_10_0,annex_start_party_box,] at (pos-1-0) {Client};

\node[name=StartParty_11_0,annex_start_party_box,attacker] at (pos-2-0) {Attacker (AS)};

\node[name=EndParty_17_0,annex_end_party_box,attacker] at (pos-2-6) {Attacker (AS)};

\node[name=StartParty_18_0,annex_start_party_box,] at (pos-2-7) {Resource Server};

\node[name=EndParty_21_0,annex_end_party_box,attacker] at (pos-0-10) {Attacker (User)};

\node[name=EndParty_22_0,annex_end_party_box,] at (pos-1-10) {Client};

\node[name=EndParty_23_0,annex_end_party_box,] at (pos-2-10) {Resource Server};

\begin{pgfonlayer}{arrows}%

\draw[annex_lifeline] (pos-0-0) -- (pos-0-10);

\draw[annex_lifeline] (pos-1-0) -- (pos-1-10);

\draw[annex_lifeline] (pos-2-0) -- (pos-2-6);

        \draw[annex_http_request] (pos-0-1) to node [annex_arrow_text,above=2.6pt,anchor=base](HTTPRequest_12_0){\setcounter{protostep}{0}\protostep{cct-att:HTTPRequest_12} \contour{white}{POST \nolinkurl{/start}}}  (pos-1-1);

        \draw[annex_http_response] (pos-1-2) to node [annex_arrow_text,above=2.6pt,anchor=base](HTTPResponse_13_0){\setcounter{protostep}{1}\protostep{cct-att:auth-req-1} \contour{white}{Response}} node [annex_arrow_text,below=8pt,anchor=base](HTTPResponse_13_1){\contour{white}{Redirect to AS (client\_id, redirect\_uri, state)}}  (pos-0-2);

        \draw[annex_http_request] (pos-0-3) to node [annex_arrow_text,above=2.6pt,anchor=base](HTTPRequest_14_0){\setcounter{protostep}{2}\protostep{cct-att:auth-resp-2} \contour{white}{GET \nolinkurl{/redirect_uri}}} node [annex_arrow_text,below=8pt,anchor=base](HTTPRequest_14_1){\contour{white}{(code, state, id token$_1$)}}  (pos-1-3);

        \draw[annex_http_request] (pos-1-4) to node [annex_arrow_text,above=2.6pt,anchor=base](HTTPRequest_15_0){\setcounter{protostep}{3}\protostep{cct-att:c-token-req} \contour{white}{POST \nolinkurl{/token_endpoint}}} node [annex_arrow_text,below=8pt,anchor=base](HTTPRequest_15_1){\contour{white}{(code, client\_id)}}  (pos-2-4);

        \draw[annex_http_response] (pos-2-5) to node [annex_arrow_text,above=2.6pt,anchor=base](HTTPResponse_16_0){\setcounter{protostep}{4}\protostep{cct-att:c-token-resp} \contour{white}{Response}} node [annex_arrow_text,below=8pt,anchor=base](HTTPResponse_16_1){\contour{white}{(\textcolor{figurehighlight1}{access token}, id token$_2$)}}  (pos-1-5);

\draw[annex_lifeline] (pos-2-7) -- (pos-2-10);

        \draw[annex_http_request] (pos-1-8) to node [annex_arrow_text,above=2.6pt,anchor=base](HTTPRequest_19_0){\setcounter{protostep}{5}\protostep{cct-att:c-resource-req} \contour{white}{GET \nolinkurl{/resource}}} node [annex_arrow_text,below=8pt,anchor=base](HTTPRequest_19_1){\contour{white}{(access token)}}  (pos-2-8);

        \draw[annex_http_response] (pos-2-9) to node [annex_arrow_text,above=2.6pt,anchor=base](HTTPResponse_20_0){\setcounter{protostep}{6}\protostep{cct-att:c-resource-resp} \contour{white}{Response}} node [annex_arrow_text,below=8pt,anchor=base](HTTPResponse_20_1){\contour{white}{resource}}  (pos-1-9);

\end{pgfonlayer}

\begin{pgfonlayer}{markers}%

\end{pgfonlayer}
        \end{tikzpicture}
        
  \caption{Cuckoo's Token Attack}
  \label{fig:attack4}
\end{figure}
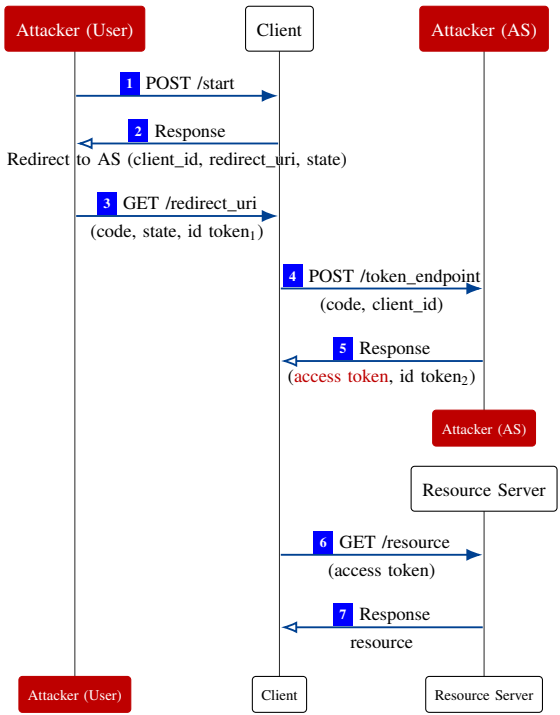

Figure~\ref{fig:attack4} depicts the attack for the OIDC Hybrid Flow,
i.e., when both responses of the AS contain id tokens (see
Section~\ref{subsection-fapi-read-write-profile}). The
attack works analogously for the code flow in combination with JARM
(see Section~\ref{subsection-fapi-read-write-profile}).

As explained, we assume that the attacker already obtained (phished)
an access token issued by an honest AS to an honest client for accessing resources of an
honest user. We also assume that the honest client supports the use of
several ASs (a common setting in practice, as already mentioned in
Section~\ref{sec:oauth-and-new-defence-mechanisms}), where in this
case one of the ASs is dishonest.\footnote{We highlight that we do not assume that the attacker  %
  controls the AS that issued the access token (i.e., the AS
  at which the honest user is registered).
  This means that the (honest) user uses an honest client
  and an honest authorization server.}

First, the attacker starts the flow
at the client and chooses his own AS. Since he is redirected to his own AS in
Step~\refprotostep{cct-att:auth-req-1}, he can skip the user
authentication step and return an authorization response immediately.
Apart from that, the flow continues normally until
Step~\refprotostep{cct-att:c-token-req}, where the client sends the
code to the attacker AS. In Step~\refprotostep{cct-att:c-token-resp},
the attacker AS returns the previously phished access token together
with the second id token.

Until here, all checks done by the client pass successfully, as the
attacker AS adheres to the protocol. The only
difference to an honest authorization server is that the attacker AS returns a
phished access token.
In Step~\refprotostep{cct-att:c-resource-req}, the resource server
receives the (phished) access token and provides the client access to
the honest resource owner's resources for the phished access
token,\footnote{Which RS
  is used in combination with an AS depends on the configuration of
  the client, which is acquired through means not defined in OAuth.
  Especially in scenarios where this configuration is done
  dynamically, a dishonest AS might be used in combination with an
  honest RS. But also if the client is configured manually, as is
  often the case today, it
  might be misconfigured or social engineered into using
  specific endpoints. Recall from
  Section~\ref{sec:fundamentalsOAuthOIDC} that the access token might
  be a document signed by the (honest) AS containing all information
  the RS needs to process the access token. Alternatively, and more
  common, the RS performs token introspection, if the access token is
  just a nonce. The RS typically uses only one AS (in this
  case, the honest AS) to which it will send the introspection
  request.} which implies that now the attacker has access to these
resources through the client.

To prevent the use of leaked access tokens, the client should include,
in the request to the RS, the identity of the AS the client received the access token
from. The client can take this value from the second id token. Now, the RS would only continue the flow if its belief is consistent with the one of the RS. We apply an analogous fix for flows
with JARM. These fixes are included in our model and shown to
work in Section~\ref{sec:analysis}.

\FloatBarrier
\subsection{Access Token Injection with ID Token Replay} \label{attack1-wrong-token-ep}

As described in Section~\ref{sec:FAPIleakofaccesstokenassumption}, the
Read-Write profile aims to be secure if an attacker acquires an access
token for an honest user. The profile also aims to be secure even if
the token endpoint URI is changed to an attacker-controlled URI (see
Section~\ref{sec:FAPImisconfiguredTEassumption}). Now, interestingly,
these two threat scenarios combined in this order are the base for the
attack described in the following. In this
attack, the attacker returns an access token at the misconfigured
token endpoint. While the attack looks similar to the previous attack
at first glance, here the attacker first interacts with the
\emph{honest} AS and later \emph{replays an id token} at the token
endpoint. Both attacks necessitate different fixes. The outcome,
however, is the same, and, just as the previous attack, this attack
affects all configurations of the Read-Write profile, even if JARM is
used. We explain the attack using
the Hybrid Flow.

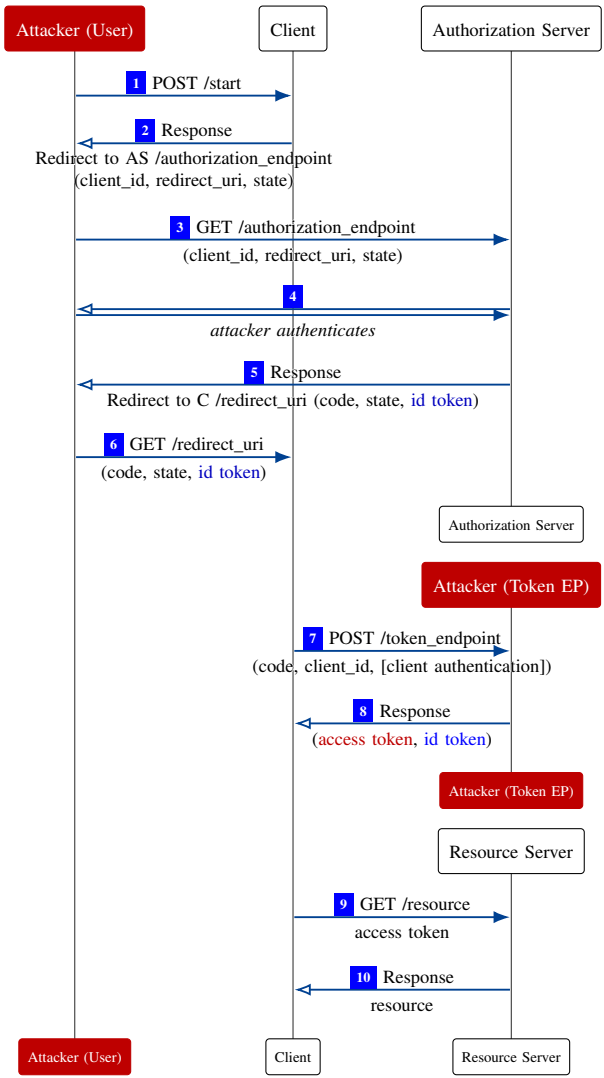
\begin{figure}[tbh]
  \centering  
          \begin{tikzpicture}[]
        \pgfdeclarelayer{arrows}
        \pgfdeclarelayer{groups}
        \pgfdeclarelayer{markers}
        \pgfsetlayers{groups,arrows,main,markers}

        \matrix [column sep={0.16\textwidth,between origins}, row sep=0.1ex]
        {
        \node[annex_matrix_node,inner sep=0,outer sep=0](pos-0-0){}; &\node[annex_matrix_node,inner sep=0,outer sep=0](pos-1-0){}; &\node[annex_matrix_node,inner sep=0,outer sep=0](pos-2-0){};\node[annex_matrix_dummy_height,minimum height=5ex,anchor=center]{};\node[annex_matrix_dummy_height,minimum height=5ex,anchor=center]{};\node[annex_matrix_dummy_height,minimum height=5ex,anchor=center]{};\\
\node[annex_matrix_node,inner sep=0,outer sep=0](pos-0-1){}; &\node[annex_matrix_node,inner sep=0,outer sep=0](pos-1-1){}; &\node[annex_matrix_node,inner sep=0,outer sep=0](pos-2-1){};\node[annex_matrix_dummy_height,minimum height=4ex,anchor=south,yshift=-1ex]{};\\
\node[annex_matrix_node,inner sep=0,outer sep=0](pos-0-2){}; &\node[annex_matrix_node,inner sep=0,outer sep=0](pos-1-2){}; &\node[annex_matrix_node,inner sep=0,outer sep=0](pos-2-2){};\node[annex_matrix_dummy_height,minimum height=4ex+2ex+2ex,anchor=north,yshift=3ex]{};\\
\node[annex_matrix_node,inner sep=0,outer sep=0](pos-0-3){}; &\node[annex_matrix_node,inner sep=0,outer sep=0](pos-1-3){}; &\node[annex_matrix_node,inner sep=0,outer sep=0](pos-2-3){};\node[annex_matrix_dummy_height,minimum height=4ex+2ex,anchor=north,yshift=3ex]{};\\
\node[annex_matrix_node,inner sep=0,outer sep=0](pos-0-4){}; &\node[annex_matrix_node,inner sep=0,outer sep=0](pos-1-4){}; &\node[annex_matrix_node,inner sep=0,outer sep=0](pos-2-4){};\node[annex_matrix_dummy_height,minimum height=4ex+2ex,anchor=north,yshift=3ex]{};\\
\node[annex_matrix_node,inner sep=0,outer sep=0](pos-0-5){}; &\node[annex_matrix_node,inner sep=0,outer sep=0](pos-1-5){}; &\node[annex_matrix_node,inner sep=0,outer sep=0](pos-2-5){};\node[annex_matrix_dummy_height,minimum height=4ex+2ex,anchor=north,yshift=3ex]{};\\
\node[annex_matrix_node,inner sep=0,outer sep=0](pos-0-6){}; &\node[annex_matrix_node,inner sep=0,outer sep=0](pos-1-6){}; &\node[annex_matrix_node,inner sep=0,outer sep=0](pos-2-6){};\node[annex_matrix_dummy_height,minimum height=4ex+2ex,anchor=north,yshift=3ex]{};\\
\node[annex_matrix_node,inner sep=0,outer sep=0](pos-0-7){}; &\node[annex_matrix_node,inner sep=0,outer sep=0](pos-1-7){}; &\node[annex_matrix_node,inner sep=0,outer sep=0](pos-2-7){};\node[annex_matrix_dummy_height,minimum height=5ex,anchor=center]{};\\
\node[annex_matrix_node,inner sep=0,outer sep=0](pos-0-8){}; &\node[annex_matrix_node,inner sep=0,outer sep=0](pos-1-8){}; &\node[annex_matrix_node,inner sep=0,outer sep=0](pos-2-8){};\node[annex_matrix_dummy_height,minimum height=5ex,anchor=center]{};\\
\node[annex_matrix_node,inner sep=0,outer sep=0](pos-0-9){}; &\node[annex_matrix_node,inner sep=0,outer sep=0](pos-1-9){}; &\node[annex_matrix_node,inner sep=0,outer sep=0](pos-2-9){};\node[annex_matrix_dummy_height,minimum height=4ex+2ex,anchor=north,yshift=3ex]{};\\
\node[annex_matrix_node,inner sep=0,outer sep=0](pos-0-10){}; &\node[annex_matrix_node,inner sep=0,outer sep=0](pos-1-10){}; &\node[annex_matrix_node,inner sep=0,outer sep=0](pos-2-10){};\node[annex_matrix_dummy_height,minimum height=4ex+2ex,anchor=north,yshift=3ex]{};\\
\node[annex_matrix_node,inner sep=0,outer sep=0](pos-0-11){}; &\node[annex_matrix_node,inner sep=0,outer sep=0](pos-1-11){}; &\node[annex_matrix_node,inner sep=0,outer sep=0](pos-2-11){};\node[annex_matrix_dummy_height,minimum height=5ex,anchor=center]{};\\
\node[annex_matrix_node,inner sep=0,outer sep=0](pos-0-12){}; &\node[annex_matrix_node,inner sep=0,outer sep=0](pos-1-12){}; &\node[annex_matrix_node,inner sep=0,outer sep=0](pos-2-12){};\node[annex_matrix_dummy_height,minimum height=5ex,anchor=center]{};\\
\node[annex_matrix_node,inner sep=0,outer sep=0](pos-0-13){}; &\node[annex_matrix_node,inner sep=0,outer sep=0](pos-1-13){}; &\node[annex_matrix_node,inner sep=0,outer sep=0](pos-2-13){};\node[annex_matrix_dummy_height,minimum height=4ex+2ex,anchor=north,yshift=3ex]{};\\
\node[annex_matrix_node,inner sep=0,outer sep=0](pos-0-14){}; &\node[annex_matrix_node,inner sep=0,outer sep=0](pos-1-14){}; &\node[annex_matrix_node,inner sep=0,outer sep=0](pos-2-14){};\node[annex_matrix_dummy_height,minimum height=4ex+2ex,anchor=north,yshift=3ex]{};\\
\node[annex_matrix_node,inner sep=0,outer sep=0](pos-0-15){}; &\node[annex_matrix_node,inner sep=0,outer sep=0](pos-1-15){}; &\node[annex_matrix_node,inner sep=0,outer sep=0](pos-2-15){};\node[annex_matrix_dummy_height,minimum height=5ex,anchor=center]{};\node[annex_matrix_dummy_height,minimum height=5ex,anchor=center]{};\node[annex_matrix_dummy_height,minimum height=5ex,anchor=center]{};\\
};

\node[name=StartParty_10_0,annex_start_party_box,attacker] at (pos-0-0) {Attacker (User)};

\node[name=StartParty_11_0,annex_start_party_box,] at (pos-1-0) {Client};

\node[name=StartParty_12_0,annex_start_party_box,] at (pos-2-0) {Authorization Server};

\node[name=EndParty_19_0,annex_end_party_box,] at (pos-2-7) {Authorization Server};

\node[name=StartParty_20_0,annex_start_party_box,attacker] at (pos-2-8) {Attacker (Token EP)};

\node[name=EndParty_23_0,annex_end_party_box,attacker] at (pos-2-11) {Attacker (Token EP)};

\node[name=StartParty_24_0,annex_start_party_box,] at (pos-2-12) {Resource Server};

\node[name=EndParty_27_0,annex_end_party_box,] at (pos-1-15) {Client};

\node[name=EndParty_28_0,annex_end_party_box,] at (pos-2-15) {Resource Server};

\node[name=EndParty_29_0,annex_end_party_box,attacker] at (pos-0-15) {Attacker (User)};

\begin{pgfonlayer}{arrows}%

\draw[annex_lifeline] (pos-0-0) -- (pos-0-15);

\draw[annex_lifeline] (pos-1-0) -- (pos-1-15);

\draw[annex_lifeline] (pos-2-0) -- (pos-2-7);

        \draw[annex_http_request] (pos-0-1) to node [annex_arrow_text,above=2.6pt,anchor=base](HTTPRequest_13_0){\setcounter{protostep}{0}\protostep{att-tep:auth-req-0-start} \contour{white}{POST \nolinkurl{/start}}}  (pos-1-1);

        \draw[annex_http_response] (pos-1-2) to node [annex_arrow_text,above=2.6pt,anchor=base](HTTPResponse_14_0){\setcounter{protostep}{1}\protostep{att-tep:auth-resp-0} \contour{white}{Response}} node [annex_arrow_text,below=8pt,anchor=base](HTTPResponse_14_1){\contour{white}{Redirect to AS \nolinkurl{/authorization_endpoint}}} node [annex_arrow_text,below=8pt+8pt,anchor=base](HTTPResponse_14_2){\contour{white}{(client\_id, redirect\_uri, state)}}  (pos-0-2);

        \draw[annex_http_request] (pos-0-3) to node [annex_arrow_text,above=2.6pt,anchor=base](HTTPRequest_15_0){\setcounter{protostep}{2}\protostep{att-tep:auth-req-2} \contour{white}{GET \nolinkurl{/authorization_endpoint}}} node [annex_arrow_text,below=8pt,anchor=base](HTTPRequest_15_1){\contour{white}{(client\_id, redirect\_uri, state)}}  (pos-2-3);

        \draw[annex_http_response,transform canvas={yshift=0.25ex}] (pos-2-4) to node [annex_arrow_text,above=2.6pt,anchor=base](HTTPResponseRequest_16_0){\setcounter{protostep}{3}\protostep{att-tep:ro-authN} } (pos-0-4);
        \draw[annex_http_request,transform canvas={yshift=-0.25ex}] (pos-0-4) to node [annex_arrow_text,below=8pt,anchor=base](HTTPResponseRequest_16_1){\contour{white}{\textit{attacker authenticates}}}  (pos-2-4);

        \draw[annex_http_response] (pos-2-5) to node [annex_arrow_text,above=2.6pt,anchor=base](HTTPResponse_17_0){\setcounter{protostep}{4}\protostep{att-tep:auth-resp-1} \contour{white}{Response}} node [annex_arrow_text,below=8pt,anchor=base](HTTPResponse_17_1){\contour{white}{Redirect to C \nolinkurl{/redirect\_uri} (code, state, \textcolor{figurehighlight2}{id token})}}  (pos-0-5);

        \draw[annex_http_request] (pos-0-6) to node [annex_arrow_text,above=2.6pt,anchor=base](HTTPRequest_18_0){\setcounter{protostep}{5}\protostep{att-tep:auth-resp-2} \contour{white}{GET \nolinkurl{/redirect_uri}}} node [annex_arrow_text,below=8pt,anchor=base](HTTPRequest_18_1){\contour{white}{(code, state, \textcolor{figurehighlight2}{id token})}}  (pos-1-6);

\draw[annex_lifeline] (pos-2-8) -- (pos-2-11);

        \draw[annex_http_request] (pos-1-9) to node [annex_arrow_text,above=2.6pt,anchor=base](HTTPRequest_21_0){\setcounter{protostep}{6}\protostep{att-tep:c-token-req} \contour{white}{POST \nolinkurl{/token_endpoint}}} node [annex_arrow_text,below=8pt,anchor=base](HTTPRequest_21_1){\contour{white}{(code, client\_id, [client authentication])}}  (pos-2-9);

        \draw[annex_http_response] (pos-2-10) to node [annex_arrow_text,above=2.6pt,anchor=base](HTTPResponse_22_0){\setcounter{protostep}{7}\protostep{att-tep:c-token-resp} \contour{white}{Response}} node [annex_arrow_text,below=8pt,anchor=base](HTTPResponse_22_1){\contour{white}{(\textcolor{figurehighlight1}{access token}, \textcolor{blue}{id token})}}  (pos-1-10);

\draw[annex_lifeline] (pos-2-12) -- (pos-2-15);

        \draw[annex_http_request] (pos-1-13) to node [annex_arrow_text,above=2.6pt,anchor=base](HTTPRequest_25_0){\setcounter{protostep}{8}\protostep{att-tep:c-resource-req} \contour{white}{GET \nolinkurl{/resource}}} node [annex_arrow_text,below=8pt,anchor=base](HTTPRequest_25_1){\contour{white}{access token}}  (pos-2-13);

        \draw[annex_http_response] (pos-2-14) to node [annex_arrow_text,above=2.6pt,anchor=base](HTTPResponse_26_0){\setcounter{protostep}{9}\protostep{att-tep:c-resource-resp} \contour{white}{Response}} node [annex_arrow_text,below=8pt,anchor=base](HTTPResponse_26_1){\contour{white}{resource}}  (pos-1-14);

\end{pgfonlayer}

\begin{pgfonlayer}{markers}%

\end{pgfonlayer}
        \end{tikzpicture}
        
  \caption{Access Token Injection with ID Token Replay Attack}
  \label{fig:attack1}
\end{figure}

Figure~\ref{fig:attack1} shows how the attack proceeds.
The attacker initiates the Read-Write flow at the client and follows the regular
flow until Step~\refprotostep{att-tep:auth-resp-2}. As the
authorization response was created by the honest AS, the state 
and all  values of the id token are correct and the client accepts
the authorization response.

In Step~\refprotostep{att-tep:c-token-req}, the client sends the token
request to the misconfigured token endpoint controlled by the
attacker. The value of the code and the checks regarding client
authentication and proof of possession of keys are not relevant for
the attacker.

In Step~\refprotostep{att-tep:c-token-resp}, the attacker sends the
token response containing the phished access token. As the flow is an
OIDC Hybrid Flow, the attacker is required to return an id token.
Here, he returns the same id token that he received in
Step~\refprotostep{att-tep:auth-resp-1}, which is signed by the honest AS.
The client is required to ensure that both id tokens have the same
subject and issuer values, which in this case holds true since they
are identical.

The client sends the access token to the honest resource server, by
which the attacker gets read-write access to the resource of the
honest resource owner through the client. %

As we show in our security analysis (see Section~\ref{sec:analysis}),
this scenario is prevented if the second id token is required to
contain the hash of the access token that is returned to the client,
as the attacker cannot create id tokens with a valid signature of the
AS. A similar fix also works for flows with JARM. 
The fixes are already included in our model.

\FloatBarrier
\subsection{PKCE Chosen Challenge Attack} \label{attack:pkce}

As detailed in
Section~\ref{sec:FAPIleakauthorizationresponseassumption}, the FAPI
uses PKCE in order to protect against leaked authorization codes. This
is particularly important for public clients as these clients, unlike confidential ones,  do not
authenticate to an AS when trying to exchange the code for an access
token.

Recall that the idea of PKCE is that a client creates a PKCE challenge
(hash of a nonce), gives it to the AS, and when redeeming the
authorization code at the AS, the client has to present the correct
PKCE verifier (the nonce). This idea works when just considering an
honest flow in which the code leaks to the attacker, who does not know
the PKCE verifier. However, our attack shows that the protection can
be circumvented by an attacker who pretends to be an honest client.

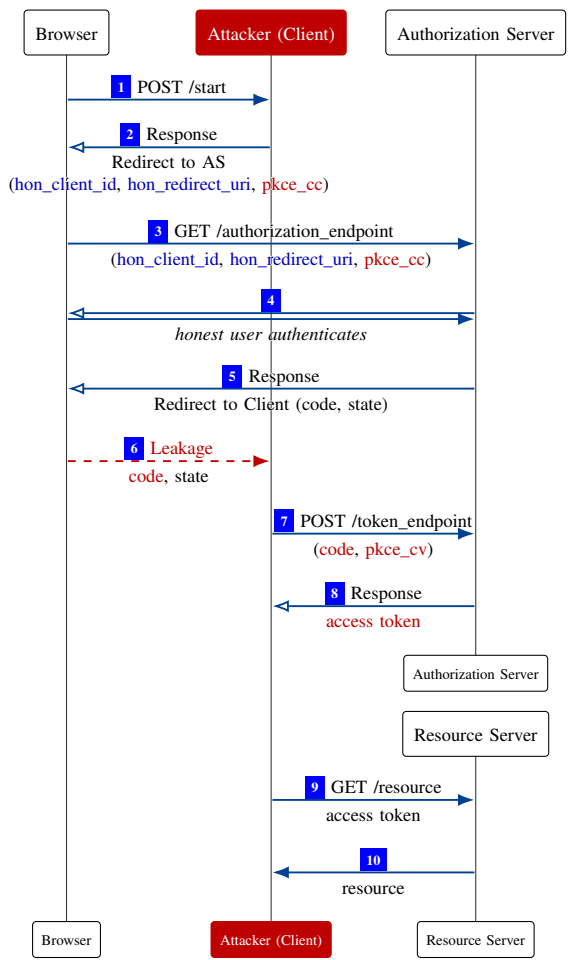
\begin{figure}[tb]
  \centering  
          \begin{tikzpicture}[]
        \pgfdeclarelayer{arrows}
        \pgfdeclarelayer{groups}
        \pgfdeclarelayer{markers}
        \pgfsetlayers{groups,arrows,main,markers}

        \matrix [column sep={0.15\textwidth,between origins}, row sep=0.1ex]
        {
        \node[annex_matrix_node,inner sep=0,outer sep=0](pos-0-0){}; &\node[annex_matrix_node,inner sep=0,outer sep=0](pos-1-0){}; &\node[annex_matrix_node,inner sep=0,outer sep=0](pos-2-0){};\node[annex_matrix_dummy_height,minimum height=5ex,anchor=center]{};\node[annex_matrix_dummy_height,minimum height=5ex,anchor=center]{};\node[annex_matrix_dummy_height,minimum height=5ex,anchor=center]{};\\
\node[annex_matrix_node,inner sep=0,outer sep=0](pos-0-1){}; &\node[annex_matrix_node,inner sep=0,outer sep=0](pos-1-1){}; &\node[annex_matrix_node,inner sep=0,outer sep=0](pos-2-1){};\node[annex_matrix_dummy_height,minimum height=4ex,anchor=south,yshift=-1ex]{};\\
\node[annex_matrix_node,inner sep=0,outer sep=0](pos-0-2){}; &\node[annex_matrix_node,inner sep=0,outer sep=0](pos-1-2){}; &\node[annex_matrix_node,inner sep=0,outer sep=0](pos-2-2){};\node[annex_matrix_dummy_height,minimum height=4ex+2ex+2ex,anchor=north,yshift=3ex]{};\\
\node[annex_matrix_node,inner sep=0,outer sep=0](pos-0-3){}; &\node[annex_matrix_node,inner sep=0,outer sep=0](pos-1-3){}; &\node[annex_matrix_node,inner sep=0,outer sep=0](pos-2-3){};\node[annex_matrix_dummy_height,minimum height=4ex+2ex,anchor=north,yshift=3ex]{};\\
\node[annex_matrix_node,inner sep=0,outer sep=0](pos-0-4){}; &\node[annex_matrix_node,inner sep=0,outer sep=0](pos-1-4){}; &\node[annex_matrix_node,inner sep=0,outer sep=0](pos-2-4){};\node[annex_matrix_dummy_height,minimum height=4ex+2ex,anchor=north,yshift=3ex]{};\\
\node[annex_matrix_node,inner sep=0,outer sep=0](pos-0-5){}; &\node[annex_matrix_node,inner sep=0,outer sep=0](pos-1-5){}; &\node[annex_matrix_node,inner sep=0,outer sep=0](pos-2-5){};\node[annex_matrix_dummy_height,minimum height=4ex+2ex,anchor=north,yshift=3ex]{};\\
\node[annex_matrix_node,inner sep=0,outer sep=0](pos-0-6){}; &\node[annex_matrix_node,inner sep=0,outer sep=0](pos-1-6){}; &\node[annex_matrix_node,inner sep=0,outer sep=0](pos-2-6){};\node[annex_matrix_dummy_height,minimum height=4ex+2ex,anchor=north,yshift=3ex]{};\\
\node[annex_matrix_node,inner sep=0,outer sep=0](pos-0-7){}; &\node[annex_matrix_node,inner sep=0,outer sep=0](pos-1-7){}; &\node[annex_matrix_node,inner sep=0,outer sep=0](pos-2-7){};\node[annex_matrix_dummy_height,minimum height=4ex+2ex,anchor=north,yshift=3ex]{};\\
\node[annex_matrix_node,inner sep=0,outer sep=0](pos-0-8){}; &\node[annex_matrix_node,inner sep=0,outer sep=0](pos-1-8){}; &\node[annex_matrix_node,inner sep=0,outer sep=0](pos-2-8){};\node[annex_matrix_dummy_height,minimum height=4ex+2ex,anchor=north,yshift=3ex]{};\\
\node[annex_matrix_node,inner sep=0,outer sep=0](pos-0-9){}; &\node[annex_matrix_node,inner sep=0,outer sep=0](pos-1-9){}; &\node[annex_matrix_node,inner sep=0,outer sep=0](pos-2-9){};\node[annex_matrix_dummy_height,minimum height=5ex,anchor=center]{};\\
\node[annex_matrix_node,inner sep=0,outer sep=0](pos-0-10){}; &\node[annex_matrix_node,inner sep=0,outer sep=0](pos-1-10){}; &\node[annex_matrix_node,inner sep=0,outer sep=0](pos-2-10){};\node[annex_matrix_dummy_height,minimum height=5ex,anchor=center]{};\\
\node[annex_matrix_node,inner sep=0,outer sep=0](pos-0-11){}; &\node[annex_matrix_node,inner sep=0,outer sep=0](pos-1-11){}; &\node[annex_matrix_node,inner sep=0,outer sep=0](pos-2-11){};\node[annex_matrix_dummy_height,minimum height=4ex+2ex,anchor=north,yshift=3ex]{};\\
\node[annex_matrix_node,inner sep=0,outer sep=0](pos-0-12){}; &\node[annex_matrix_node,inner sep=0,outer sep=0](pos-1-12){}; &\node[annex_matrix_node,inner sep=0,outer sep=0](pos-2-12){};\node[annex_matrix_dummy_height,minimum height=4ex+2ex,anchor=north,yshift=3ex]{};\\
\node[annex_matrix_node,inner sep=0,outer sep=0](pos-0-13){}; &\node[annex_matrix_node,inner sep=0,outer sep=0](pos-1-13){}; &\node[annex_matrix_node,inner sep=0,outer sep=0](pos-2-13){};\node[annex_matrix_dummy_height,minimum height=5ex,anchor=center]{};\node[annex_matrix_dummy_height,minimum height=5ex,anchor=center]{};\node[annex_matrix_dummy_height,minimum height=5ex,anchor=center]{};\\
};

\node[name=StartParty_9_0,annex_start_party_box,] at (pos-0-0) {Browser};

\node[name=StartParty_10_0,annex_start_party_box,attacker] at (pos-1-0) {Attacker (Client)};

\node[name=StartParty_11_0,annex_start_party_box,] at (pos-2-0) {Authorization Server};

\node[name=EndParty_20_0,annex_end_party_box,] at (pos-2-9) {Authorization Server};

\node[name=StartParty_21_0,annex_start_party_box,] at (pos-2-10) {Resource Server};

\node[name=EndParty_24_0,annex_end_party_box,] at (pos-0-13) {Browser};

\node[name=EndParty_25_0,annex_end_party_box,] at (pos-2-13) {Resource Server};

\node[name=EndParty_26_0,annex_end_party_box,attacker] at (pos-1-13) {Attacker (Client)};

\begin{pgfonlayer}{arrows}%

\draw[annex_lifeline] (pos-0-0) -- (pos-0-13);

\draw[annex_lifeline] (pos-1-0) -- (pos-1-13);

\draw[annex_lifeline] (pos-2-0) -- (pos-2-9);

        \draw[annex_http_request] (pos-0-1) to node [annex_arrow_text,above=2.6pt,anchor=base](HTTPRequest_12_0){\setcounter{protostep}{0}\protostep{att-pkce-cc:HTTPRequest_12} \contour{white}{POST \nolinkurl{/start}}}  (pos-1-1);

        \draw[annex_http_response] (pos-1-2) to node [annex_arrow_text,above=2.6pt,anchor=base](HTTPResponse_13_0){\setcounter{protostep}{1}\protostep{att-pkce-cc:auth-req-1} \contour{white}{Response}} node [annex_arrow_text,below=8pt,anchor=base](HTTPResponse_13_1){\contour{white}{Redirect to AS}} node [annex_arrow_text,below=8pt+8pt,anchor=base](HTTPResponse_13_2){\contour{white}{(\textcolor{figurehighlight2}{hon\_client\_id}, \textcolor{figurehighlight2}{hon\_redirect\_uri}, \textcolor{figurehighlight1}{pkce\_cc})}}  (pos-0-2);

        \draw[annex_http_request] (pos-0-3) to node [annex_arrow_text,above=2.6pt,anchor=base](HTTPRequest_14_0){\setcounter{protostep}{2}\protostep{att-pkce-cc:auth-req-2} \contour{white}{GET \nolinkurl{/authorization_endpoint}}} node [annex_arrow_text,below=8pt,anchor=base](HTTPRequest_14_1){\contour{white}{(\textcolor{figurehighlight2}{hon\_client\_id}, \textcolor{figurehighlight2}{hon\_redirect\_uri}, \textcolor{figurehighlight1}{pkce\_cc})}}  (pos-2-3);

        \draw[annex_http_response,transform canvas={yshift=0.25ex}] (pos-2-4) to node [annex_arrow_text,above=2.6pt,anchor=base](HTTPResponseRequest_15_0){\setcounter{protostep}{3}\protostep{att-pkce-cc:ro-authN} } (pos-0-4);
        \draw[annex_http_request,transform canvas={yshift=-0.25ex}] (pos-0-4) to node [annex_arrow_text,below=8pt,anchor=base](HTTPResponseRequest_15_1){\contour{white}{\textit{honest user authenticates}}}  (pos-2-4);

        \draw[annex_http_response] (pos-2-5) to node [annex_arrow_text,above=2.6pt,anchor=base](HTTPResponse_16_0){\setcounter{protostep}{4}\protostep{att-pkce-cc:auth-resp-1} \contour{white}{Response}} node [annex_arrow_text,below=8pt,anchor=base](HTTPResponse_16_1){\contour{white}{Redirect to Client (code, state)}}  (pos-0-5);

        \draw[annex_http_request,draw=figurehighlight1,dashed] (pos-0-6) to node [annex_arrow_text,above=2.6pt,anchor=base](HTTPRequest_17_0){\setcounter{protostep}{5}\protostep{att-pkce-cc:leaked-auth-resp} \contour{white}{\textcolor{figurehighlight1}{Leakage}}} node [annex_arrow_text,below=8pt,anchor=base](HTTPRequest_17_1){\contour{white}{\textcolor{figurehighlight1}{code}, state}}  (pos-1-6);

        \draw[annex_http_request] (pos-1-7) to node [annex_arrow_text,above=2.6pt,anchor=base](HTTPRequest_18_0){\setcounter{protostep}{6}\protostep{att-pkce-cc:c-token-req} \contour{white}{POST \nolinkurl{/token\_endpoint}}} node [annex_arrow_text,below=8pt,anchor=base](HTTPRequest_18_1){\contour{white}{(\textcolor{figurehighlight1}{code}, \textcolor{figurehighlight1}{pkce\_cv})}}  (pos-2-7);

        \draw[annex_http_response] (pos-2-8) to node [annex_arrow_text,above=2.6pt,anchor=base](HTTPResponse_19_0){\setcounter{protostep}{7}\protostep{att-pkce-cc:c-token-resp} \contour{white}{Response}} node [annex_arrow_text,below=8pt,anchor=base](HTTPResponse_19_1){\contour{white}{\textcolor{figurehighlight1}{access token}}}  (pos-1-8);

\draw[annex_lifeline] (pos-2-10) -- (pos-2-13);

        \draw[annex_http_request] (pos-1-11) to node [annex_arrow_text,above=2.6pt,anchor=base](HTTPRequest_22_0){\setcounter{protostep}{8}\protostep{att-pkce-cc:c-resource-req} \contour{white}{GET \nolinkurl{/resource}}} node [annex_arrow_text,below=8pt,anchor=base](HTTPRequest_22_1){\contour{white}{access token}}  (pos-2-11);

        \draw[annex_http_request] (pos-2-12) to node [annex_arrow_text,above=2.6pt,anchor=base](HTTPRequest_23_0){\setcounter{protostep}{9}\protostep{att-pkce-cc:c-resource-resp} } node [annex_arrow_text,below=8pt,anchor=base](HTTPRequest_23_1){\contour{white}{resource}}  (pos-1-12);

\end{pgfonlayer}

\begin{pgfonlayer}{markers}%

\end{pgfonlayer}
        \end{tikzpicture}
        
  \caption{PKCE Chosen Challenge Attack}
  \label{fig:attack2}
\end{figure}

This attack affects public clients who use the Read-Only profile of
the FAPI. It works as follows (see Figure~\ref{fig:attack2}): As in RFC
7636, two apps are installed on a user's device, an honest app and a
malicious app. The honest app is a client of an honest AS with the
client identifier $\mi{hon\_client\_id}$
and the redirection URI $\mi{hon\_redir\_uri}$.
The malicious app is not registered at the AS.

The Read-Only flow starts at the malicious app, which prompts the user
to log in. Now, the malicious app prepares an authorization request
containing the client id and a redirect URI of the honest client
(Step~\refprotostep{att-pkce-cc:auth-req-1}). At this point, the
malicious app also creates a PKCE verifier and includes the
corresponding challenge in the authorization request.

The flow continues until the browser receives the authorization
response in Step~\refprotostep{att-pkce-cc:auth-resp-1}. As the
redirection URIs are preregistered at the AS, the
redirection URI in the authorization request was chosen from the set
of redirect URIs of the honest app, and therefore, the authorization
response is redirected to the honest client after the browser receives
it.

As described in Sections~\ref{section-pkce} and
\ref{sec:FAPIleakauthorizationresponseassumption}, at this point, the
authorization response with the authorization code might leak
to the attacker (Step~\refprotostep{att-pkce-cc:leaked-auth-resp}).
The malicious app is now able to exchange the code (associated with
the honest client) at the token endpoint in
Steps~\refprotostep{att-pkce-cc:c-token-req} and~\refprotostep{att-pkce-cc:c-token-resp}, as it knows the correct PKCE
verifier and, as the honest app is a public client, without authenticating to the AS.

To prevent this scenario, an honest AS must ensure that the PKCE
challenge was created by the client with the id
$\mi{hon\_client\_id}$.
To achieve this, for public clients in the Read-Only flow we use the
same mechanism that the FAPI uses for public clients in the Read-Write
flow, namely the authorization request should contain a signed JWT
(see also Section~\ref{cauthn-jws-client-assertion}, although JWTs are
now used in a different way). This ensures that the client stated in
the request actually made the request, and hence, no other client
should know the PKCE verifier. Note that by using signed JWTs for public
clients the FAPI assumes that public clients can store some secrets
(which might, for example, be protected by user passwords). Our fix is
already included in the model and our analysis (Section~\ref{sec:analysis}) shows that it works.
\subsection{Authorization Request Leak Attacks} \label{attack:SI}

As explained in Section~\ref{sec:FAPIleakofauthorizationrequestassumption},
the PKCE challenge is created such that PKCE is supposed to work even if the authorization
request leaks (see also Section~\ref{section-pkce}). 

However, if a leak of the authorization request occurs not only the
PKCE challenge leaks to the attacker but also the state value, since
both values are contained in the authorization request. Our
attack shows that an attacker who knows the state value can circumvent
the CSRF protection the state value was supposed to provide. As a
result of the attack, the honest user is logged in under the identity
of the attacker and uses the resources of the attacker, which breaks
session integrity.   The details of this attack are presented in
Appendix~\ref{app:authorizationrequestleakattack}.

This is a well-known class of attacks for plain OAuth
flows~\cite{rfc6819-oauth2-security}, but it is important to highlight
that the protections designed into the FAPI do not sufficiently protect most flows against
such attacks, even though PKCE explicitly foresees the attack
vector.

To prevent this attack, one essentially has to prevent CSRF forgery in
this context. However, this is non-trivial because of the very strong
attacker model considered by the OpenID FAPI Working Group: leaks
and misconfigurations are assumed to occur at various places. As
further explained in
Appendix~\ref{app:authorizationrequestleakattack}, just assuming that
the authorization request does not leak to the attacker would not fix
the problem in general; one at least would have to assume that the
authorization response does not leak either.  Making these
assumptions, however, of course contradicts the OpenID FAPI Working
Group's intention, namely providing security even in the presence of very
strong attackers.

Fortunately, we can prove that regular FAPI web server clients which
use OAUTB are not vulnerable to this attack even in the presence of the
strong attackers assumed by the OpenID FAPI Working Group and
throughout this paper.  More specifically, we can prove session
integrity of the FAPI for such clients (and strong attackers), which
in particular excludes the above attack (see
Section~\ref{sec:analysis}). For all other types of clients, our attack
works, and there does not seem to be a fix which would not massively
change the flows, and hence, the standards, as argued in
Appendix~\ref{app:authorizationrequestleakattack}. In this sense, our
results for session integrity appear to be the best we can obtain for
the FAPI.

\FloatBarrier

\section{Formal Security Analysis}\label{sec:analysis}

In this section, we present our formal analysis of the FAPI. We start
by very briefly recalling the Web Infrastructure Model (WIM), followed
by a sketch of our formal model of the FAPI, which as already
mentioned uses the WIM as its basic web infrastructure model. We then
introduce central security properties the FAPI is supposed to satisfy,
along with our main theorem stating that these properties are
satisfied.

Since we cannot present the full formal details here, we provide\chooseReferenceTRPaper{the complete analysis in Appendices~\ref{appendix-additions-to-the-wim}--\ref{fapi-formal-proofs}.}
{some more details in the appendix, with full details and proofs provided in
  our technical report~\cite{FettHosseyniKuesters-TR-FAPI-2018}.}This includes the precise
formalization of clients, authorization servers, and resource servers,
as well as full detailed proofs.

\subsection{The Web Infrastructure Model} \label{fks-highlevel}

The Web Infrastructure Model (WIM) was
introduced by Fett, K{\"u}sters, and Schmitz in
\cite{FettKuestersSchmitz-SP-2014} (therefore also called the FKS
model) and further developed in subsequent work. The appendix of
\cite{FettKuestersSchmitz-TR-OIDC-2017} provides a detailed
description of the model; a comparison with other models and a
discussion of its scope and limitations can be found in
\cite{FettKuestersSchmitz-SP-2014,
  FettKuestersSchmitz-CCS-2015,FettKuestersSchmitz-ESORICS-BrowserID-Primary-2015}.
We here only give a brief overview of the WIM following the
description in \cite{FettKuestersSchmitz-CSF-2017}, with some more
details presented in Appendix~\ref{app:wim-sketch}. As explained
there, we slightly extend the WIM, among others to model OAUTB. We
choose the WIM for our work because, as mentioned in the
introduction, the WIM is the most comprehensive model of the web
infrastructure to date.

The WIM is designed independently of a specific web application and
closely mimics published (de-facto) standards and specifications for
the web, for example, the HTTP/1.1 and HTML5 standards and associated
(proposed) standards. Among others, HTTP(S) requests and 
responses,\footnote{We note that the WIM models %
TLS at a high level of abstraction such that messages are exchanged in 
a secure way.} including several headers, such as cookie, location, referer,
authorization, strict transport security (STS), and origin headers,
are modeled. The model of web browsers captures the concepts of
windows, documents, and iframes, including the complex navigation
rules, as well as modern technologies, such as web storage, web
messaging (via postMessage), and referrer policies. JavaScript is
modeled in an abstract way by so-called \emph{scripts} which can be
sent around and, among others, can create iframes, access other
windows, and initiate XMLHttpRequests.

The WIM defines a general communication
model, and, based on it, web systems consisting of web browsers, DNS
servers, and web servers as well as web and network attackers. 
The main entities in the model are \emph{(atomic) processes}, which
are used to model browsers, servers, and attackers. Each process
listens to one or more (IP) addresses.
Processes communicate via \emph{events}, which consist of a message as
well as a receiver and a sender address. In every step of a run, one
event is chosen non-deterministically from a ``pool'' of waiting
events and is delivered to one of the processes that listens to the
event's receiver address. The process can then handle the event and
output new events, which are added to the pool of events, and so on.
The WIM follows the Dolev-Yao approach (see, e.g.,
\cite{AbadiFournet-POPL-2001}). That is, messages are expressed as
formal terms over a signature $\Sigma$ which  contains
constants (for addresses, strings, nonces) as well as sequence,
projection, and function symbols (e.g., for encryption/decryption and
signatures).

A \emph{(Dolev-Yao) process} consists of a set of addresses the
process listens to, a set of states (terms), an initial state, and a
relation that takes an event and a state as input and
(non-deterministically) returns a new state and a sequence of events.
The relation models a computation step of the
process.
It is required that the output can be computed (formally, derived
in the usual Dolev-Yao style) from the input event and the state.

The so-called \emph{attacker process} 
records all messages it receives and outputs all events it can
possibly derive from its recorded messages. Hence, an attacker process
carries out all attacks any Dolev-Yao process could possibly perform.
Attackers can corrupt other parties, browsers, and servers.

A \emph{script} models JavaScript running in a browser. Scripts are
defined similarly to Dolev-Yao processes, but run in and interact with
the browser. Similar to an attacker process, an \emph{attacker script} can
(non-deterministically) perform every action a script can possibly
perform within a browser.

A \emph{system} is a set of processes. A \emph{configuration} of a
system is a tuple of the form $(S,E,N)$ where $S$ maps every process
of the system to its state, $E$ is the pool of waiting events, and $N$
is a sequence of unused nonces. In what follows, $s^p_0$ denotes the
initial state of process $p$. Systems induce \emph{runs}, i.e.,
sequences of configurations, where each configuration is obtained by
delivering one of the waiting events of the preceding configuration to
a process, which then performs a computation step.

A \emph{web system} formalizes the web infrastructure and web
applications. It contains a system consisting of honest and attacker
processes. Honest processes can be web browsers, web servers, or DNS
servers. Attackers can be either \emph{web attackers} (who can listen
to and send messages from their own addresses only) or \emph{network
  attackers} (who may listen to and spoof all addresses and therefore
are the most powerful attackers). A web system further contains a set
of scripts (comprising honest scripts and the attacker script).

In our FAPI model, we need to specify only the behavior of servers and
scripts. These are not defined by the WIM since they depend on the
specific application, unless they become corrupted, in which case they
behave like attacker processes and attacker scripts. We
assume the presence of a strong network attacker which also controls
all DNS servers (but we assume a working PKI).

\subsection{Sketch of the Formal FAPI Model} \label{fapi-informal-description-of-the-model}

A \textbf{FAPI web system (with a network attacker)}, denoted by
$\fapiwebsystem$, is a web system (as explained in
Section~\ref{fks-highlevel}) and can contain an unbounded finite
number of clients, authorization servers, resource servers, browsers,
and a network attacker. Note that a network attacker is the most
powerful attacker, which subsumes all other attackers. Except for the
attacker, all processes are initially honest and can become
(dynamically) corrupted by the attacker at any time.

In a FAPI web system, clients, authorization servers, and resource
servers act according to the specification of the FAPI presented in
Section~\ref{sec:fapi}. (As mentioned in Section~\ref{fks-highlevel},
the behavior of browsers is fixed by the standards. Their modeling is
independent of the FAPI and already contained in the WIM.) Our models
for clients and servers follow the latest recommendations regarding
the security of OAuth~2.0~\cite{draft-ietf-oauth-security-topics} to mitigate all previously known attacks. The model also contains the fixes pointed out in
Section~\ref{sec:att}, as otherwise, we would not be able to prove the
desired security properties (see below).

The primary goal of the FAPI is to provide a
high degree of security. Its flows are intended to be secure even if
information leaks to an attacker. As already outlined in Section~\ref{sec-assumptions}, we model this by
sending the authorization response (in the case of an app client), the
access token (in the case of a Read-Write flow), and the authorization request to an arbitrary
(non-deterministically chosen) IP address. Furthermore, in the
Read-Write profile, the token request can be sent to an arbitrary URI.

Importantly, one FAPI web system contains all possible settings in
which the FAPI can run, as depicted in Figure~\ref{fig:fapi-overview},
in particular, we consider all OAuth 2.0 extensions employed in the
FAPI.  More precisely, every client in a FAPI web system runs one of
the possible configurations (i.e., it implements on one path in
Figure~\ref{fig:fapi-overview}). Different clients may implement
different configurations. Every authorization and resource server in a
FAPI web system supports all configurations at once. When interacting
with a specific client, a server just chooses the configuration the
client supports. In our model, the various endpoints (authorization,
redirection, token), the information which client
supports which FAPI configuration, client credentials, etc.~are
preconfigured and contained in the initial states of the
processes. How this information is acquired is out of the scope of the
FAPI.

We emphasize that when proving security properties of the FAPI, we
prove these properties for all FAPI web systems, where different FAPI
web systems can differ in the number of clients and servers, and their
preconfigured information.

Furthermore, we note that there is no notion of time in the WIM,
hence, tokens do not expire. This is a safe overapproximation as 
it\highlightIfDiffMinor{gives the attacker more power.}

\highlightIfDiffAdded{
To give a feel for our formal FAPI model, an excerpt of the model is provided in Appendix~\ref{app:excerpt-of-client-model}.
}

\subsection{Security Properties and Main Theorem}\label{sec:secpropertiesandtheorem}

In the following, we define the security properties the FAPI should
fulfill, namely authorization, authentication, and session
integrity. These properties have been central to also OAuth 2.0 and
OpenID Connect
\cite{FettKuestersSchmitz-CCS-2016,FettKuestersSchmitz-CSF-2017}. But
as mentioned, the FAPI has been designed to fulfill these properties
under stronger adversaries, therefore using various OAuth
extensions. While our formulations of these properties are inspired by
those for OAuth 2.0 and OpenID Connect, they had to be adapted and
extended for the FAPI, e.g., to capture properties of resource
servers, which previously have not been modeled. We also state our main theorem.

We give an overview of each security property. For the
authorization property, we provide an in-depth explanation, together with
the formal definition.
  Appendix~\ref{app:proofsketch} contains a proof sketch for
  the authorization property. 
  Full details and proofs of all properties are given
  in\chooseReferenceTRPaper{Appendix~\ref{fapi-formal-properties}.}
  {our technical report~\cite{FettHosseyniKuesters-TR-FAPI-2018}.
  }\subsubsection{Authorization} 

Informally speaking, for authorization we require that an attacker
cannot access resources belonging to an honest user (browser).  A bit
more precise, we require that in all runs $\rho$ of a FAPI web system
$\fapiwebsystem$ if an honest resource server receives an access token
that is associated with an honest client, an honest authorization
server, and an identity of an honest user, then access to the
corresponding resource is not provided to the attacker in any way.  We
highlight that this does not only mean that the attacker cannot access
the resource directly at the resource server, but also that the
attacker cannot access the resource through a client.

\highlightIfDiffAdded{

  In order to formalize this property, we first need to define what it
  means for an access token to be associated with a client, an AS, and
  a user identity (see below for an explanation of this definition).
  
  \begin{definition}[Access Token associated with C, AS and ID]\label{informal-def:AT-associated-with-c}
    Let $c$ be a client with client id $\mi{clientId}$ issued to $c$
    by the authorization server $\mi{as}$, and let
    $\mi{id} \in \mathsf{ID}^\mi{as}$, where $\mathsf{ID}^\mi{as}$
    denotes the set of identities governed by as.  We say that an
    \emph{access token $t$ is associated with $c$, $\mi{as}$ and
      $\mi{id}$} in state $S$ of the configuration $(S, E,N)$ of a run
    $\rho$ of a FAPI web system, if there is a sequence
    $s \in S(\mi{as}).\str{accessTokens}$ such that
    $s \equiv \an{\mi{id}, \mi{clientId}, t, \str{r}}$,
    $s \equiv \an{\str{MTLS}, \mi{id}, \mi{clientId}, t, \mi{key},
      \str{rw}}$ or
    $s \equiv \an{\str{OAUTB}, \mi{id}, \mi{clientId}, t, \mi{key'},
      \str{rw}}$, for some $\mi{key}$ and $\mi{key'}$.
  \end{definition}
  
  Intuitively, an access token $t$ is associated with a client $c$, authorization server $\mi{as}$,
  and user identity $\mi{id}$, if $t$ was created by the authorization server
  $\mi{as}$ and if the AS has created $t$ for the client $c$ and the identity $\mi{id}$.
  
  More precisely, the access token is exchanged for an authorization code (at the token endpoint
  of the AS), which is issued for a specific client. This is also the client to which the access token
  is associated with.
  The user identity with which the access token is associated is the user identity that 
  authenticated at the AS (i.e., logged in at the website of the AS).
  In the model, the AS associates the access token with the client identifier and
  user identity by storing a sequence containing the identity, the client identifier
  and the access token 
  (i.e., $\an{\mi{id}, \mi{clientId}, t, \str{r}}$, 
    $\an{\str{MTLS}, \mi{id}, \mi{clientId}, t, \mi{key}, \str{rw}}$ or
    $\an{\str{OAUTB}, \mi{id}, \mi{clientId}, t, \mi{key'}, \str{rw}}$). Furthermore, the last entry of the sequence indicates if the client is using
  the Read-Only or the Read-Write flow. In addition to this, for the Read-Write flow, the AS
  stores whether the access token is bound via mTLS or OAUTB (along  with the corresponding key with which the 
  access token is associated).

  We can now define authorization formally, again the explanation of
  this definition follows below.
  
  \begin{definition}[Authorization Property]\label{informal-def:property-authz-a} 
    We say that the FAPI web system with a network attacker
    \emph{$\fapiwebsystem$ is secure w.r.t.~authorization} iff for
    every run $\rho$ of $\fapiwebsystem$, every configuration
    $(S, E, N)$ in $\rho$, every authorization server
    $\mi{as} \in \fAP{AS}$ that is honest in $S$ with
    $s^\mi{as}_0.\str{resource\_servers}$ being domains of honest
    resource servers used by $\mi{as}$, every identity $\mi{id} \in \mathsf{ID}^\mi{as}$
    for which the corresponding browser, say $b$, is honest in $S$,
    every client $c \in \fAP{C}$ that is honest in $S$ with client id
    $\mi{clientId}$ issued to $c$ by $\mi{as}$, every resource server
    $\mi{rs} \in \fAP{RS}$ that is honest in $S$ such that
    $\mi{id} \in s^\mi{rs}_0.\str{ids}$ (set of IDs handled by $\mi{rs}$),
    $s^\mi{rs}_0.\str{authServ} \in \mathsf{dom}(\mi{as})$ (set of domains controlled by $\mi{as}$) and with
    $\mi{dom}_\mi{rs} \in s_0^\mi{as}.\str{resource\_servers}$ (with
    $\mi{dom}_\mi{rs} \in \mathsf{dom}(\mi{rs})$), every access token
    $t$ associated with $c$, $\mi{as}$ and $\mi{id}$ and every
    resource access nonce
    $r \in s^\mi{rs}_0.\str{rNonce}[\mi{id}] \cup
    s^\mi{rs}_0.\str{wNonce}[\mi{id}]$ it holds true that:
  
    If $r$ is contained in a response to a request $m$ sent to $\mi{rs}$
    with $t \equiv m.\mi{header}[\str{Authorization}]$,
    then 
    $r$ is not derivable from the attackers knowledge in $S$.
  \end{definition}
  
  As outlined above, the authorization property states that if the honest resource server
  receives an access token associated with a client identifier, authorization server, and
  user identifier, then the corresponding resource access is not given to the attacker.
  Access to resources is modeled by nonces called \emph{resource
    access nonces}. For each user identity, there is one set of nonces
  representing read access, and another set representing write
  access. In our model of the FAPI, when a resource server receives an
  access token associated with a user from a client, the resource
  server returns to the client one of the resource access nonces of
  the user, which in turn the client forwards to the user's
  browser. The above security property requires that the attacker does
  not obtain such a resource access nonce (under the assumptions state
  in the property). This captures that there should be no direct or
  indirect way for the attacker to access the corresponding
  resource. In particular, the attacker should not be able to use a
  client such that he can access the resource through the client.
  
  For the authorization property to be meaningful, we require that the involved participants are
  honest. For example, we require that the authorization server at which the
  identity is registered is honest. If this is not the case (i.e., the
  attacker controls the AS), then the attacker could trivially access 
  resources.
  The same holds true for the client for which the access token is issued:
  If the user chooses a client that is controlled by the attacker, then
  the attacker can trivially access the resource (as the user authorized
  the attacker client to do so).
  In our model of the FAPI, the client (non-deterministically) chooses
  a resource server that the authorization server supports (this can
  be different for each login flow).  As in the Read-Only flow, the
  access token would trivially leak to the attacker if the resource
  server is controlled by the attacker, we require that the resource
  servers that the AS supports are honest.  Furthermore, in the WIM,
  the behavior of the user is subsumed in the browser model,
  therefore, we require that the browser that is responsible for the
  user identity that is involved in the flow should be
  honest. Otherwise, the attacker could trivially obtain the
  credentials of the user.

}

\subsubsection{Authentication}\label{sec:authenticationpropertyinformal}

Informally speaking, the authentication property states that an
attacker should not be able to log in at a client under the identity
of an honest user. More precisely, we require that in all runs $\rho$
of a FAPI web system $\fapiwebsystem$ if in $\rho$ a client considers an honest user
(browser) whose ID is governed by an honest AS to be logged in
(indicated by a service token which a user can use at the client),
then the adversary cannot obtain the service token.

\subsubsection{Session Integrity}\label{sec:sessionintegritypropertyinformal}

There are two session integrity properties that capture that an honest
user should not be logged in under the identity of the attacker and
should not use resources of the attacker. As shown  in
Section~\ref{attack:SI}, session integrity is not given for
all configurations available in the FAPI. Therefore, we show a limited
session integrity property that captures session integrity for
web server clients that use OAUTB.

Nonetheless, our session integrity property here is stronger than
those used
in~\cite{FettKuestersSchmitz-CCS-2016,FettKuestersSchmitz-CSF-2017} in
the sense that we define (and prove) session integrity not only in the
presence of web attackers, but also for the much stronger network
attacker.  (This is enabled by using the \emph{\_\_Secure-} prefix for
cookies.)

\emph{Session Integrity for Authorization for Web Server Clients with OAUTB:}
Intuitively, this property states that for all runs $\rho$ of a FAPI web system $\fapiwebsystem$,
if an honest user can access the resource of some identity $u$ (registered at AS $\mi{as}$)
through the honest 
web server client $c$, where $c$ uses OAUTB as the holder of key mechanism,
then (1) the user started the flow at $c$ and
(2) if $\mi{as}$ is honest, the user authenticated at the $\mi{as}$ using the identity
$u$.

\emph{Session Integrity for Authentication for Web Server Clients with OAUTB:}
Similar to the previous property, this property states that for all runs $\rho$ of a FAPI web system $\fapiwebsystem$,
if an honest user is logged in at the honest client $c$ under some identity $u$
(registered at AS $\mi{as}$),
with $c$ being a web server client using OAUTB as the holder of 
key mechanism,
then (1) the user started the flow at $c$ and
(2) if $\mi{as}$ is honest, the user authenticated at the $\mi{as}$ using the identity
$u$.

By \emph{Session Integrity for Web Server Clients with OAUTB} we denote
the conjunction of both properties.

Now, our main theorem says that these properties are satisfied for all
FAPI web systems.

\begin{theorem}\label{thm:theorem-1}
  Let $\fapiwebsystem$ be a FAPI web system with a network
  attacker. Then, $\fapiwebsystem$ is secure w.r.t. authorization and
  authentication. Furthermore, $\fapiwebsystem$ is secure w.r.t.
  session integrity for web server clients with OAUTB.
\end{theorem}

We emphasize that the FAPI web systems take
into account the strong attacker the FAPI is supposed to withstand as
explained in Section~\ref{sec-assumptions}. Such attackers immediately break plain
OAuth 2.0 and OpenID Connect. 
This, together with the various OAuth 2.0 security extensions which
the FAPI uses and combines in different ways, and which have not
formally been analyzed before, makes the proof challenging.

\section{Conclusion}\label{sec:conclusion-outlook}

In this paper, we performed the first formal
analysis of an Open Banking API, namely the OpenID Financial-grade
API. Based on the Web Infrastructure Model, we built a comprehensive
model comprising all protocol participants (clients, authorization
servers, and resource servers) and all important options employed in
the FAPI: clients can be app clients or web server clients and can
make use of either the Read-Only or the Read-Write profile. We modeled
all specified methods for authenticating at the authorization server
and both mechanisms for binding tokens to the client, namely, Mutual
TLS and OAuth 2.0 Token Binding. We also modeled PKCE, JWS Client
Assertions, and the JWT Secured Authorization Response Mode (JARM).

Based on this model, we then defined precise security properties for
the FAPI, namely authorization, authentication, and session integrity.
While trying to prove these properties for the FAPI, we found several
vulnerabilities that can enable an attacker to access protected
resources belonging to an honest user or
perform attacks on session integrity. We developed fixes against these
attacks and formally verified the security of the (fixed) OpenID FAPI.

This is an important result since the FAPI enjoys wide industry
support and is a promising candidate for the future lead in open
banking APIs. Financial-grade applications entail very high security
requirements that make a thorough formal security analysis, as 
performed in this paper, indispensable.

Our work also constitutes the very first analysis of various OAuth
security extensions, namely PKCE, OAuth mTLS, OAUTB, JARM, and JWS
Client Assertions.

 \smallskip\noindent\emph{Acknowledgements.}
 This work was partially supported by \textit{Deu\-tsche
 Forschungsgemeinschaft} (DFG) through  Grant KU\ 1434/10-2.

\newpage

\ifIncludeTechnicalReport
\FloatBarrier
\onecolumn
\fi
\appendices

\section{Authorization Request Leak Attack -- Details}\label{app:authorizationrequestleakattack}

We here provide further details about the authorization request leak
attack, which was only sketched in Section~\ref{attack:SI}.

A concrete instantiation of this attack is shown in
Figure~\ref{fig:app-attack-SI}, where the scenario is based on the
Read-Only flow of a public client. As explained
below, similar attacks
also work for all other configurations of the FAPI (except for web
server clients which use OAUTB, for which, as mentioned, we show that
they are not susceptible in Section~\ref{sec:analysis}).

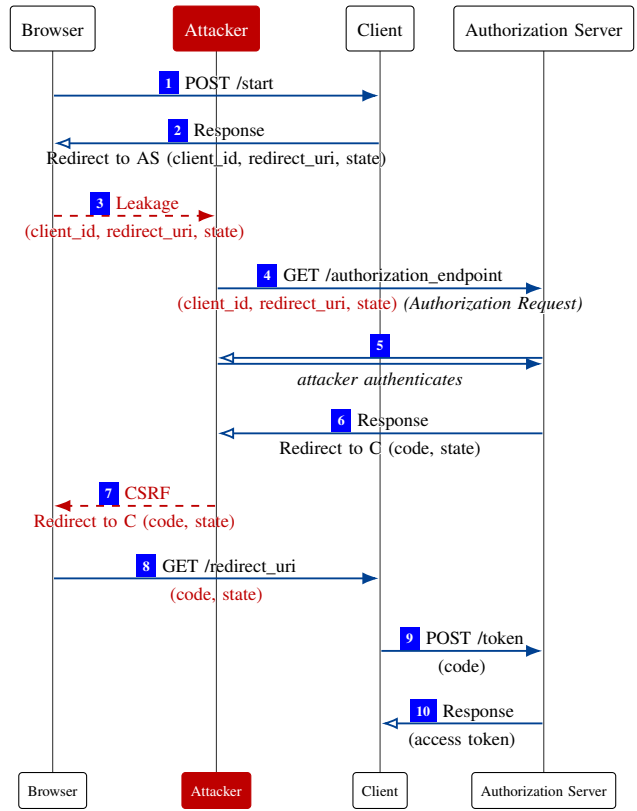
\begin{figure}[tb]
  \centering  
          \begin{tikzpicture}[]
        \pgfdeclarelayer{arrows}
        \pgfdeclarelayer{groups}
        \pgfdeclarelayer{markers}
        \pgfsetlayers{groups,arrows,main,markers}

        \matrix [column sep={0.12\textwidth,between origins}, row sep=0.1ex]
        {
        \node[annex_matrix_node,inner sep=0,outer sep=0](pos-0-0){}; &\node[annex_matrix_node,inner sep=0,outer sep=0](pos-1-0){}; &\node[annex_matrix_node,inner sep=0,outer sep=0](pos-2-0){}; &\node[annex_matrix_node,inner sep=0,outer sep=0](pos-3-0){};\node[annex_matrix_dummy_height,minimum height=5ex,anchor=center]{};\node[annex_matrix_dummy_height,minimum height=5ex,anchor=center]{};\node[annex_matrix_dummy_height,minimum height=5ex,anchor=center]{};\node[annex_matrix_dummy_height,minimum height=5ex,anchor=center]{};\\
\node[annex_matrix_node,inner sep=0,outer sep=0](pos-0-1){}; &\node[annex_matrix_node,inner sep=0,outer sep=0](pos-1-1){}; &\node[annex_matrix_node,inner sep=0,outer sep=0](pos-2-1){}; &\node[annex_matrix_node,inner sep=0,outer sep=0](pos-3-1){};\node[annex_matrix_dummy_height,minimum height=4ex,anchor=south,yshift=-1ex]{};\\
\node[annex_matrix_node,inner sep=0,outer sep=0](pos-0-2){}; &\node[annex_matrix_node,inner sep=0,outer sep=0](pos-1-2){}; &\node[annex_matrix_node,inner sep=0,outer sep=0](pos-2-2){}; &\node[annex_matrix_node,inner sep=0,outer sep=0](pos-3-2){};\node[annex_matrix_dummy_height,minimum height=4ex+2ex,anchor=north,yshift=3ex]{};\\
\node[annex_matrix_node,inner sep=0,outer sep=0](pos-0-3){}; &\node[annex_matrix_node,inner sep=0,outer sep=0](pos-1-3){}; &\node[annex_matrix_node,inner sep=0,outer sep=0](pos-2-3){}; &\node[annex_matrix_node,inner sep=0,outer sep=0](pos-3-3){};\node[annex_matrix_dummy_height,minimum height=4ex+2ex,anchor=north,yshift=3ex]{};\\
\node[annex_matrix_node,inner sep=0,outer sep=0](pos-0-4){}; &\node[annex_matrix_node,inner sep=0,outer sep=0](pos-1-4){}; &\node[annex_matrix_node,inner sep=0,outer sep=0](pos-2-4){}; &\node[annex_matrix_node,inner sep=0,outer sep=0](pos-3-4){};\node[annex_matrix_dummy_height,minimum height=4ex+2ex,anchor=north,yshift=3ex]{};\\
\node[annex_matrix_node,inner sep=0,outer sep=0](pos-0-5){}; &\node[annex_matrix_node,inner sep=0,outer sep=0](pos-1-5){}; &\node[annex_matrix_node,inner sep=0,outer sep=0](pos-2-5){}; &\node[annex_matrix_node,inner sep=0,outer sep=0](pos-3-5){};\node[annex_matrix_dummy_height,minimum height=4ex+2ex,anchor=north,yshift=3ex]{};\\
\node[annex_matrix_node,inner sep=0,outer sep=0](pos-0-6){}; &\node[annex_matrix_node,inner sep=0,outer sep=0](pos-1-6){}; &\node[annex_matrix_node,inner sep=0,outer sep=0](pos-2-6){}; &\node[annex_matrix_node,inner sep=0,outer sep=0](pos-3-6){};\node[annex_matrix_dummy_height,minimum height=4ex+2ex,anchor=north,yshift=3ex]{};\\
\node[annex_matrix_node,inner sep=0,outer sep=0](pos-0-7){}; &\node[annex_matrix_node,inner sep=0,outer sep=0](pos-1-7){}; &\node[annex_matrix_node,inner sep=0,outer sep=0](pos-2-7){}; &\node[annex_matrix_node,inner sep=0,outer sep=0](pos-3-7){};\node[annex_matrix_dummy_height,minimum height=4ex+2ex,anchor=north,yshift=3ex]{};\\
\node[annex_matrix_node,inner sep=0,outer sep=0](pos-0-8){}; &\node[annex_matrix_node,inner sep=0,outer sep=0](pos-1-8){}; &\node[annex_matrix_node,inner sep=0,outer sep=0](pos-2-8){}; &\node[annex_matrix_node,inner sep=0,outer sep=0](pos-3-8){};\node[annex_matrix_dummy_height,minimum height=4ex+2ex,anchor=north,yshift=3ex]{};\\
\node[annex_matrix_node,inner sep=0,outer sep=0](pos-0-9){}; &\node[annex_matrix_node,inner sep=0,outer sep=0](pos-1-9){}; &\node[annex_matrix_node,inner sep=0,outer sep=0](pos-2-9){}; &\node[annex_matrix_node,inner sep=0,outer sep=0](pos-3-9){};\node[annex_matrix_dummy_height,minimum height=4ex+2ex,anchor=north,yshift=3ex]{};\\
\node[annex_matrix_node,inner sep=0,outer sep=0](pos-0-10){}; &\node[annex_matrix_node,inner sep=0,outer sep=0](pos-1-10){}; &\node[annex_matrix_node,inner sep=0,outer sep=0](pos-2-10){}; &\node[annex_matrix_node,inner sep=0,outer sep=0](pos-3-10){};\node[annex_matrix_dummy_height,minimum height=4ex+2ex,anchor=north,yshift=3ex]{};\\
\node[annex_matrix_node,inner sep=0,outer sep=0](pos-0-11){}; &\node[annex_matrix_node,inner sep=0,outer sep=0](pos-1-11){}; &\node[annex_matrix_node,inner sep=0,outer sep=0](pos-2-11){}; &\node[annex_matrix_node,inner sep=0,outer sep=0](pos-3-11){};\node[annex_matrix_dummy_height,minimum height=5ex,anchor=center]{};\node[annex_matrix_dummy_height,minimum height=5ex,anchor=center]{};\node[annex_matrix_dummy_height,minimum height=5ex,anchor=center]{};\node[annex_matrix_dummy_height,minimum height=5ex,anchor=center]{};\\
};

\node[name=StartParty_9_0,annex_start_party_box,attacker] at (pos-1-0) {Attacker};

\node[name=StartParty_10_0,annex_start_party_box,] at (pos-0-0) {Browser};

\node[name=StartParty_11_0,annex_start_party_box,] at (pos-2-0) {Client};

\node[name=StartParty_12_0,annex_start_party_box,] at (pos-3-0) {Authorization Server};

\node[name=EndParty_23_0,annex_end_party_box,attacker] at (pos-1-11) {Attacker};

\node[name=EndParty_24_0,annex_end_party_box,] at (pos-2-11) {Client};

\node[name=EndParty_25_0,annex_end_party_box,] at (pos-3-11) {Authorization Server};

\node[name=EndParty_26_0,annex_end_party_box,] at (pos-0-11) {Browser};

\begin{pgfonlayer}{arrows}%

\draw[annex_lifeline] (pos-1-0) -- (pos-1-11);

\draw[annex_lifeline] (pos-0-0) -- (pos-0-11);

\draw[annex_lifeline] (pos-2-0) -- (pos-2-11);

\draw[annex_lifeline] (pos-3-0) -- (pos-3-11);

        \draw[annex_http_request] (pos-0-1) to node [annex_arrow_text,above=2.6pt,anchor=base](HTTPRequest_13_0){\setcounter{protostep}{0}\protostep{state-leak:HTTPRequest_13} \contour{white}{POST \nolinkurl{/start}}}  (pos-2-1);

        \draw[annex_http_response] (pos-2-2) to node [annex_arrow_text,above=2.6pt,anchor=base](HTTPResponse_14_0){\setcounter{protostep}{1}\protostep{state-leak:auth-req-1} \contour{white}{Response}} node [annex_arrow_text,below=8pt,anchor=base](HTTPResponse_14_1){\contour{white}{Redirect to AS (client\_id, redirect\_uri, state)}}  (pos-0-2);

        \draw[annex_http_request,draw=figurehighlight1,dashed] (pos-0-3) to node [annex_arrow_text,above=2.6pt,anchor=base](HTTPRequest_15_0){\setcounter{protostep}{2}\protostep{state-leak:auth-req-leak-1} \contour{white}{\textcolor{figurehighlight1}{Leakage}}} node [annex_arrow_text,below=8pt,anchor=base](HTTPRequest_15_1){\contour{white}{\textcolor{figurehighlight1}{(client\_id, redirect\_uri, state)}}}  (pos-1-3);

        \draw[annex_http_request] (pos-1-4) to node [annex_arrow_text,above=2.6pt,anchor=base](HTTPRequest_16_0){\setcounter{protostep}{3}\protostep{state-leak:auth-req-leak-2} \contour{white}{GET \nolinkurl{/authorization_endpoint}}} node [annex_arrow_text,below=8pt,anchor=base](HTTPRequest_16_1){\contour{white}{\textcolor{figurehighlight1}{(client\_id, redirect\_uri, state)} \textit{(Authorization Request)}}}  (pos-3-4);

        \draw[annex_http_response,transform canvas={yshift=0.25ex}] (pos-3-5) to node [annex_arrow_text,above=2.6pt,anchor=base](HTTPResponseRequest_17_0){\setcounter{protostep}{4}\protostep{state-leak:att-authN} } (pos-1-5);
        \draw[annex_http_request,transform canvas={yshift=-0.25ex}] (pos-1-5) to node [annex_arrow_text,below=8pt,anchor=base](HTTPResponseRequest_17_1){\contour{white}{\textit{attacker authenticates}}}  (pos-3-5);

        \draw[annex_http_response] (pos-3-6) to node [annex_arrow_text,above=2.6pt,anchor=base](HTTPResponse_18_0){\setcounter{protostep}{5}\protostep{state-leak:auth-resp-1} \contour{white}{Response}} node [annex_arrow_text,below=8pt,anchor=base](HTTPResponse_18_1){\contour{white}{Redirect to C (code, state)}}  (pos-1-6);

        \draw[annex_http_request,draw=figurehighlight1,dashed] (pos-1-7) to node [annex_arrow_text,above=2.6pt,anchor=base](HTTPRequest_19_0){\setcounter{protostep}{6}\protostep{state-leak:auth-resp-csrf} \contour{white}{\textcolor{figurehighlight1}{CSRF}}} node [annex_arrow_text,below=8pt,anchor=base](HTTPRequest_19_1){\contour{white}{\textcolor{figurehighlight1}{Redirect to C (code, state)}}}  (pos-0-7);

        \draw[annex_http_request] (pos-0-8) to node [annex_arrow_text,above=2.6pt,anchor=base](HTTPRequest_20_0){\setcounter{protostep}{7}\protostep{state-leak:auth-resp-2} \contour{white}{GET \nolinkurl{/redirect_uri}}} node [annex_arrow_text,below=8pt,anchor=base](HTTPRequest_20_1){\contour{white}{\textcolor{figurehighlight1}{(code, state)}}}  (pos-2-8);

        \draw[annex_http_request] (pos-2-9) to node [annex_arrow_text,above=2.6pt,anchor=base](HTTPRequest_21_0){\setcounter{protostep}{8}\protostep{state-leak:c-token-req} \contour{white}{POST \nolinkurl{/token}}} node [annex_arrow_text,below=8pt,anchor=base](HTTPRequest_21_1){\contour{white}{(code)}}  (pos-3-9);

        \draw[annex_http_response] (pos-3-10) to node [annex_arrow_text,above=2.6pt,anchor=base](HTTPResponse_22_0){\setcounter{protostep}{9}\protostep{state-leak:c-token-resp} \contour{white}{Response}} node [annex_arrow_text,below=8pt,anchor=base](HTTPResponse_22_1){\contour{white}{(access token)}}  (pos-2-10);

\end{pgfonlayer}

\begin{pgfonlayer}{markers}%

\end{pgfonlayer}
        \end{tikzpicture}
        
  \caption{Leakage of Authorization Request Attack}
  \label{fig:app-attack-SI}
\end{figure}

In the Authorization Request Leak Attack, the client sends the
authorization request to the browser in
Step~\refprotostep{state-leak:auth-req-1}, where it leaks to the
attacker in Step~\refprotostep{state-leak:auth-req-leak-1}. From here
on, the attacker behaves as the browser and logs himself in
(Step~\refprotostep{state-leak:att-authN}), hence, the authorization
code received in Step~\refprotostep{state-leak:auth-resp-1} is
associated with the identity of the attacker.

The state value used in the authorization request aims at preventing Cross-Site Request Forgery (CSRF) attacks.
However, as the state value leaks, this protection does not work. For showing that this is the case,
we assume that a CSRF attack happens. If, for example, the user is visiting a website that
is controlled by the attacker, then the attacker can 
send, from the browser of the user, a request to the AS containing the 
code and the state value (Step~\refprotostep{state-leak:auth-resp-2}).
As the state received by the client is the same that it included in the 
authorization request, the client continues the flow and uses the 
code to retrieve an access token in
Steps~\refprotostep{state-leak:c-token-req} and
\refprotostep{state-leak:c-token-resp}.

This access token is associated with the attacker, which means that
the honest user is accessing resources belonging to the attacker.

As a result, the honest user can be logged in under the
identity of the attacker if the authorization server returns an id
token. In the case of the Read-Write flow, the honest user can modify
resources of the attacker: for example, she might upload personal documents
to the account of the attacker.

As noted above, this attack might happen for all configurations,
except for the Read-Write flow when the client is a web server client using OAUTB
(see Figure~\ref{fig:fapi-overview}).

In all other configurations, this attack can happen as the  attacker can behave exactly like the browser of the honest user, i.e., 
after receiving the authorization request, the attacker can send this request to the AS, log in under his own identity, 
and would then receive a response that the client accepts.  The only 
flow in which this is different is the Read-Write flow where the client is a web server and uses OAUTB, 
as here, the browser (and therefore, also the attacker) needs to prove possession of a key pair 
(i.e., the key pair used for the client). As the attacker\highlightIfDiffMinor{cannot}prove possession of the private key of the 
key pair which the browser uses for the client, the AS would then stop the flow. 
(In the other flows, the AS does not check if the response was sent by the browser that logged in the user.)

If we say that the FAPI is not required to be secure if the authorization request leaks 
(i.e., if we remove the assumption that the authorization request leaks), then the flow is 
still not secure, as the authorization response might still leak to the attacker (see Section~\ref{sec:FAPIleakauthorizationresponseassumption}), 
which also contains the state value.  More precisely, the authorization response might leak in 
the case of app clients due to the operating system sending the response to the attacker app (for details, see Section~\ref{section-pkce}).
After receiving the authorization response, the attacker app knows the state value and can start a new flow using this value. 
The attacker can\highlightIfDiffMinor{then}continue from Step~\refprotostep{state-leak:auth-req-leak-1} (Figure~\ref{fig:app-attack-SI}),
and when receiving the authorization response 
(which is a URI containing the OAuth parameters), he could, using his own app that runs on the device of the victim, 
call the legitimate client app with this URI (i.e., with the code that is associated with the identity of the attacker and the state 
value with which the client started the flow).  The effect of this is that the legitimate app, 
at which the honest user started the flow, would continue the flow using an authorization code associated with the attacker.  
Therefore, the honest user would either be logged in with the identity of the attacker or use the resources of the attacker.

 We note that even encrypting the state value contained in the
 authorization request does not solve the problem, as the attacker is
 using the whole authorization request. (Strictly speaking, he acts as
 the browser of the honest user).

\ifIncludeTechnicalReport\clearpage\fi
\section{The WIM:\highlightIfDiffMinor{Some Background}}\label{app:wim-sketch}

We here provide more details about the Web Infrastructure Model.

\paragraph{Signature and Messages} As mentioned, the WIM follows the
Dolev-Yao approach where messages are expressed as formal terms over a
signature $\Sigma$.  For example, in the WIM an HTTP request is
represented as a term $r$ containing a nonce, an HTTP method, a domain
name, a path, URI parameters, request headers, and a message body. For
instance, an HTTP request for the URI \url{http://ex.com/show?p=1} is
represented as
$\mi{r} := \langle \cHttpReq, n_1, \mGet, \str{ex.com}, \str{/show},
\an{\an{\str{p},1}}, \an{}, \an{} \rangle$ where the body and the list
of request headers is empty. An HTTPS request for $r$ is of the form
$\ehreqWithVariable{r}{k'}{\pub(k_\text{ex.com})}$, where $k'$ is a
fresh symmetric key (a nonce) generated by the sender of the request
(typically a browser); the responder is supposed to use this key to
encrypt the response.

The \emph{equational theory} associated
with $\Sigma$
is defined as usual in Dolev-Yao models. The theory induces a congruence
relation $\equiv$
on terms, capturing the meaning of the function symbols in $\Sigma$.
For instance, the equation in the equational theory which captures
asymmetric decryption is $\dec{\enc x{\pub(y)}}{y}=x$.
With this, we have that, for example, $\dec{\ehreqWithVariable{r}{k'}{\pub(k_\text{ex.com})}}{k_\text{ex.com}}\equiv
  \an{r,k'}\,,$
i.e., these two terms are equivalent w.r.t.~the equational theory.

\paragraph{Scripts} A \emph{script} models JavaScript running in a browser. Scripts are
defined similarly to Dolev-Yao processes. When triggered by a browser, a
script is provided with state information. The script then outputs a
term representing a new internal state and a command to be interpreted
by the browser (see also the specification of browsers below). 
Similarly to an attacker process, the so-called
\emph{attacker script} outputs everything that is derivable from
the input.

\paragraph{Running a system} As mentioned, a run of a system is a
sequence of configurations.  The transition from one configuration to
the next configuration in a run is called a \emph{processing~step}. We
write, for example, $Q = (S, E, N) \xrightarrow[]{} (S', E', N')$ to
denote the transition from the configuration $(S, E, N)$ to the
configuration $(S', E', N')$, where $S$ and $S'$ are the states of the
processes in the system, $E$ and $E'$ are~pools of waiting events, and
$N$ and $N'$ are sequences of unused~nonces.

\paragraph{Web Browsers}
An honest browser is thought to be used by one honest user, who is
modeled as part of the browser. User actions, such as following a link, are
modeled as non-deterministic actions of the web browser. User
credentials are stored in the initial state of the browser and are
given to selected web pages when needed. Besides user credentials, the
state of a web browser contains (among others) a tree of windows and
documents, cookies, and web storage data (localStorage and
sessionStorage).

A \emph{window} inside a browser contains a set of
\emph{documents} (one being active at any time), modeling the
history of documents presented in this window. Each represents one
loaded web page and contains (among others) a script and a list of
subwindows (modeling iframes). The script, when triggered by the
browser, is provided with all data it has access to, such as a
(limited) view on other documents and windows, certain cookies, and
web storage data. Scripts then output a command and a new state. This
way, scripts can navigate or create windows, send XMLHttpRequests and
postMessages, submit forms, set/change cookies and web storage data,
and create iframes. Navigation and security rules ensure that scripts
can manipulate only specific aspects of the browser's state, according
to the relevant web standards.

A browser can output messages on the network of different types,
namely DNS and HTTP(S) (including XMLHttpRequests), and it processes
the responses. Several HTTP(S) headers are modeled, including, for
example, cookie, location, strict transport security (STS), and origin
headers. A browser, at any time, can also receive a so-called trigger
message upon which the browser non-de\-ter\-min\-is\-tically chooses an
action, for instance, to trigger a script in some document. The script
now outputs a command, as described above, which is then further
processed by the browser. Browsers can also become corrupted, i.e., be
taken over by web and network attackers. Once corrupted, a browser
behaves like an attacker process.

As detailed in\chooseReferenceTRPaper{Appendix~\ref{appendix-additions-to-the-wim},}
{our technical report~\cite{FettHosseyniKuesters-TR-FAPI-2018},}we extended the browser model of the WIM slightly in order
to incorporate OAUTB in the browser model.  We furthermore added the
behavior of the \emph{\_\_Secure-} prefix of cookies to the model,
which specifies that such cookies shall only be accepted when they are
transmitted over secure channels~\cite{rfc6265bis-httpwg}. Note that
for the FAPI, mTLS is only needed between clients and
servers. Therefore, mTLS has been modeled on top of the WIM, i.e., as
part of the modeling of FAPI clients and servers. The servers we
modeled for the FAPI of course also support OAUTB.

\highlightIfDiffRemoved{
 \subsection{Authorization Property}
  \input{appendix-authZ-from-initial-submission} %
}
\highlightIfDiffAdded{
  \ifIncludeTechnicalReport\clearpage\fi
\section{Excerpt of Client Model} \label{app:excerpt-of-client-model}

In this section, we provide a brief excerpt of the client model 
in order to give an impression of the formal
model.\chooseReferenceTRPaper{See Appendix~\ref{label:appendix-full-fapi-model}
for the full formal model of the FAPI.}{See our technical
report~\cite{FettHosseyniKuesters-TR-FAPI-2018} for the full formal model
of the FAPI.}

The excerpt given in Algorithm~\ref{excerpt:client-token-request}  
shows how the client prepares and
sends the token request to the authorization server, i.e.,
the part in which the client sends the authorization code in exchange for
an access token (and depending on the flow, also an id 
token).

This function is called by the client. The first two inputs are the
session identifier of the session (i.e., the session of the
resource owner at the client) and the authorization code
that the client wants to send to the AS. The value
$\mi{responseValue}$ contains information related to mTLS
or OAUTB (if used for the current flow). The last input is
the current state of the client.

In Lines~\ref{excerpt-str-check-if-misconf-tep} to
\ref{excerpt-line:r-client-chooses-right-tep}, the client chooses
either the token endpoint of the AS or some URL that was chosen non-deterministically.
This models the assumption shown in Section~\ref{sec:FAPImisconfiguredTEassumption}, 
which requires the Read-Write profile of
the FAPI to be secure even if the token endpoint is misconfigured. 

Starting from Line~\ref{excerpt-c-STR-first-if},
the function chooses the parameters of the request that depend
on the flow and configuration (see Figure~\ref{fig:fapi-overview}).

If the client uses the Read-Only profile, the token request always contains
the PKCE verifier (Line~\ref{excerpt-c-STR-first-if}).
For a confidential client (which means that the client has to authenticate 
at the token endpoint), the client either authenticates using JWS Client Assertions
(Line~\ref{excerpt-branch:r-client-creates-assertion}, see also Section~\ref{cauthn-jws-client-assertion}),
or with mTLS (Line~\ref{excerpt-branch:STR-client-using-MTLS}; 
for details on our model of mTLS refer to\chooseReferenceTRPaper{Appendix~\ref{appendix-mtls}).}{our technical
report~\cite{FettHosseyniKuesters-TR-FAPI-2018}).}

If the client uses the Read-Write profile, the client uses either
mTLS (again Line~\ref{excerpt-branch:STR-client-using-MTLS})
or OAUTB (Line~\ref{excerpt-client-str-branch-oautb-sig-2};
for details on our model of OAUTB refer to\chooseReferenceTRPaper{Appendix~\ref{appendix:OAUTB}).}{our technical
report~\cite{FettHosseyniKuesters-TR-FAPI-2018}).}

\begin{algorithm}
\caption{\label{excerpt:client-token-request} Relation of a Client $R^c$ -- Request to token endpoint.}
\begin{algorithmic}[1]
\algrenewcommand\algorithmicindent{0.8em}%
\Function{$\mathsf{SEND\_TOKEN\_REQUEST}$}{$\mi{sessionId}$, $\mi{code}$, $\mi{responseValue}$, $s'$}
\raggedright
  \Let{$\mi{session}$}{$s'.\str{sessions}[\mi{sessionId}]$}
  \Let{$\mi{identity}$}{$\mi{session}[\str{identity}]$} \label{excerpt-line:str-client-chooses-identity}
  \Let{$\mi{issuer}$}{$s'.\str{issuerCache}[\mi{identity}]$} \label{excerpt-line:str-client-chooses-issuer}
  \If{$\mi{session}[\str{misconfiguredTEp}] \equiv \True$} \label{excerpt-str-check-if-misconf-tep}
  	\Let{$\mi{url}$}{$\mi{session}[\str{token\_ep}] $} 
	\Else
    \Let{$\mi{url}$}{$s'.\str{oidcConfigCache}[\mi{issuer}][\str{token\_ep}]$} \label{excerpt-line:r-client-chooses-right-tep}
	\EndIf

	\Let{$\mi{credentials}$}{$s'.\str{clientCredentialsCache}[\mi{issuer}]$}
  \Let{$\mi{clientId}$}{$\mi{credentials}[\str{client\_id}]$}

  \Let{$\mi{clientType}$}{$\mi{credentials}[\str{client\_type}]$} 
 	\Let{$\mi{profile}$}{$\mi{credentials}[\str{profile}]$}  
 	\Let{$\mi{isApp}$}{$\mi{credentials}[\str{is\_app}]$}  

	\Let{$\mi{body}$}{$[\str{grant\_type}{:} \str{authorization\_code}, \str{code}{:}\mi{code}, \str{redirect\_uri}{:}\mi{session}[\str{redirect\_uri}], \str{client\_id}{:}\mi{clientId}	]$} \label{excerpt-c-STR-main-body}
  \If{$\mi{profile} \equiv \str{r}$} \label{excerpt-c-STR-first-if}
			\Let{$\mi{body}[\str{pkce\_verifier}]$}{$\mi{session}[\str{pkce\_verifier}] $} 
	\EndIf
  \If{$\mi{profile} \equiv \str{r} \wedge \mi{clientType} \equiv \str{pub} $}
		\Let{$\mi{message}$}{
      $\hreq{ nonce=\nu_2,
				method=\mPost,
				xhost=\mi{url}.\str{domain},
				path=\mi{url}.\str{path},
				parameters=\mi{url}.\str{parameters},
				headers=\bot,
				xbody=\mi{body}}$}  
    \CallFun{HTTPS\_SIMPLE\_SEND}{$[\str{responseTo}{:}\str{TOKEN}, \str{session}{:}\mi{sessionId}]$, $\mi{message}$, $s'$} \label{excerpt-line:r-client-sends-token-req-pub}
  	\ElsIf{$\mi{profile} \equiv \str{r} \wedge \mi{clientType} \equiv \str{conf\_JWS}$} \label{excerpt-branch:r-client-creates-assertion}
			\Let{$\mi{clientSecret}$}{$\mi{credentials}[\str{client\_secret}]$} 
			\Let{$\mi{jwt}$}{$[ \str{iss}{:}\mi{clientId}, 
													\str{aud}{:}\mi{url}.\str{domain} 
												]$} 
			\Let{$\mi{body}[\str{assertion}]$}{$\mathsf{mac}(\mi{jwt}, \mi{clientSecret})$}
			\Let{$\mi{message}$}{$\hreq{ nonce=\nu_2,
					method=\mPost,
					xhost=\mi{url}.\str{domain},
					path=\mi{url}.\str{path},
					parameters=\mi{url}.\str{parameters},
					headers=\bot,
          xbody=\mi{body}}$}   \label{excerpt-line:r-client-assertion-rec}
			\CallFun{HTTPS\_SIMPLE\_SEND}{$[\str{responseTo}{:}\str{TOKEN}, \str{session}{:}\mi{sessionId}]$, $\mi{message}$, $s'$} \label{excerpt-line:r-client-sends-assertion} \label{excerpt-line:r-client-sends-token-req-jws} %

  \ElsIf{$\mi{clientType} \equiv \str{conf\_MTLS} $} \Comment{both profiles} \label{excerpt-branch:STR-client-using-MTLS}
    \If{$\mi{responseValue}[\str{type}] \not \equiv \str{MTLS}$} \label{excerpt-c-str-mtls-check-resp-val-type}
         \Stop{\DefStop} 
    \EndIf
		\Let{$\mi{body}[\str{TLS\_AuthN}]$}{$\mi{responseValue}[\str{mtls\_nonce}]$} 
		\Let{$\mi{message}$}{$\hreq{ nonce=\nu_2,
			method=\mPost,
			xhost=\mi{url}.\str{domain},
			path=\mi{url}.\str{path},
			parameters=\mi{url}.\str{parameters},
			headers=\bot,
			xbody=\mi{body}}$}  
		\CallFun{HTTPS\_SIMPLE\_SEND}{$[\str{responseTo}{:}\str{TOKEN}$, $\str{session}{:}\mi{sessionId}]$, $\mi{message}$, $s'$} \label{excerpt-line:client-sends-token-req-mtls}
	\Else \Comment{rw with OAUTB} \label{excerpt-client-str-branch-oautb-sig-2}
    \If{$\mi{responseValue}[\str{type}] \not \equiv \str{OAUTB}$} \label{excerpt-client-str-resp-value-oautb}
         \Stop{\DefStop} 
    \EndIf
    \Let{$\mi{ekm}$}{$\mi{responseValue}[\str{ekm}]$} \label{excerpt-line:str-oautb-respVal1}

		\Let{$\mi{TB\_AS}$}{$s'.\str{TBindings}[\mi{url}.\str{host}]$} \Comment{priv. key}
		\Let{$\mi{TB\_RS}$}{$s'.\str{TBindings}[\mi{session}[\str{RS}]]$} \Comment{priv. key} \label{excerpt-c-str-retrieve-rs-tb}
	 	\Let{$\mi{TB\_Msg\_prov}$}{$[\str{id}{:} \mathsf{pub}(\mi{TB\_AS}), \str{sig}{:} \mathsf{sig}( \mi{ekm} , \mi{TB\_AS})]$} 
		\Let{$\mi{TB\_Msg\_ref}$}{$[\str{id}{:} \mathsf{pub}(\mi{TB\_RS}),  \str{sig}{:} \mathsf{sig}( \mi{ekm} , \mi{TB\_RS})]$} 
			
		\Let{$\mi{headers}$}{$[\str{Sec\mhyphen{}Token\mhyphen{}Binding}{:}[\str{prov}{:}\mi{TB\_Msg\_prov}$, $\str{ref}{:}\mi{TB\_Msg\_ref}]]$}

    \If{$\mi{clientType} \equiv \str{conf\_OAUTB}$} \Comment{client authentication}
				\Let{$\mi{clientSecret}$}{$\mi{credentials}[\str{client\_secret}]$} 
        \Let{$\mi{jwt}$}{$[ \str{iss}{:}\mi{clientId},  \label{excerpt-line:jws-oautb-set-aud}
														\str{aud}{:}\mi{url}.\str{domain},  
													]$} 
			\Let{$\mi{body}[\str{assertion}]$}{$\mathsf{mac}(\mi{jwt}, \mi{clientSecret})$} \label{excerpt-line:rw-client-creates-assertion} %
		\EndIf
    \If{$\mi{isApp} \equiv \bot$} \Comment{W.S. client: TBID used by browser}
			\Let{$\mi{body}[\str{pkce\_verifier}]$}{$\mi{session}[\str{browserTBID}] $}  \label{excerpt-line:rw-client-pkce-cv-browserTBID}
		\EndIf
    \Let{$\mi{message}$}{$\hreq{ nonce=\nu_2, 
  		method=\mPost,
  		xhost=\mi{url}.\str{domain},
  		path=\mi{url}.\str{path},
  		parameters=\mi{url}.\str{parameters},
  		headers=\mi{headers},
  		xbody=\mi{body}}$}  \label{excerpt-line:jws-oautb-msg}
		\CallFun{HTTPS\_SIMPLE\_SEND}{$[\str{responseTo}{:}\str{TOKEN}$, $\str{session}{:}\mi{sessionId}]$, $\mi{message}$, $s'$} \label{excerpt-line:jws-oautb-cont-send}
	\EndIf
\EndFunction
\end{algorithmic} %
\end{algorithm}

}
\ifIncludeTechnicalReport
\FloatBarrier\clearpage\fi
\section{Proof Sketch of Theorem~\ref{thm:theorem-1}, Authorization}\label{app:proofsketch}

We here provide a proof sketch of Theorem~\ref{thm:theorem-1} that is
concerned with the authorization property. The complete formal proof
of this theorem is given in\chooseReferenceTRPaper{Appendix~\ref{fapi-formal-proofs}.}
{the technical report \cite{FettHosseyniKuesters-TR-FAPI-2018}.}

For proving the authorization property, we show that\highlightIfDiffMinor{when a participant}provides
access to a resource, i.e., by sending a resource access nonce, this
access is not provided to the attacker:

\paragraph{Resource server does not provide the attacker access to resources}
We show that the resource server does not provide the attacker access to resources of an honest user.

In case of the Read-Only flow, we show that an access token associated with an honest client, an honest authorization
server, and an honest identity
does not leak to the attacker, and therefore, the attacker cannot obtain access to resources.

In case of the Read-Write flow, such an access token might leak to the attacker, but this
token cannot be used by the attacker at the resource server due to Token Binding, either via OAUTB or mTLS.

\paragraph{Web server client does not provide the attacker access to resources}

App clients are only usable via
the device they are running on, i.e., they are not usable over the network 
(by which we mean that if, for example, the user wants to view one of her documents
with an app client, she does this directly using the device).
Therefore, we only look at the case
of web server clients, as such a client can be used over the network,
e.g., by the browser of the end-user or by the attacker.

In the following, we show that honest web server clients do not provide the attacker
access to resources belonging to an honest
identity. We show this for all possible configurations that could trick the
client into doing so, e.g., with a misconfigured token endpoint or with an authorization server
controlled by the attacker that returns a leaked access token.

The access to the resource is provided to the sender of the redirection request. To access
a resource, this means that the attacker must have sent the request to the redirection
endpoint of the client.

For a Read-Only flow, the token endpoint 
is configured correctly.
This means that the attacker must include a code in the request
such that the client can exchange it for an access token.
We show that such a code (associated with an honest identity and the client) does not
leak to an attacker.

For a Read-Write flow, the token endpoint 
can be misconfigured such that it is controlled by the attacker, and
we also assume that access tokens leak to the attacker (see 
Section~\ref{sec-assumptions}).

We show that a leaked access token cannot be used at the client by the
attacker.  If only the token endpoint is controlled by the attacker,
he must include an id token (when using the OIDC Hybrid flow, see
below for the Authorization Code flow 
with JARM) in the token response
such that it contains the hash of the access token and be signed by
the honest authorization 
server (the hash of the access token was not
included in the original draft and was included by us as a mitigation
in Section~\ref{attack1-wrong-token-ep}). However, such an id token
does 
not leak to the attacker, which prevents the use of leaked access 
tokens at misconfigured token endpoints.
For the Authorization  
Code flow with JARM, the attacker would need a
response JWS. As in the case of the Hybrid flow, we show that the
response JWS needed by the client for accessing resources of an honest
identity does not leak.

A leaked access token can also be used by the attacker if
the client chooses an authorization server under the control of the 
attacker. Here, the id tokens are created by the attacker and accepted by the 
client. For preventing the use of this access token, the client
includes the 
issuer of the second id token (or of the response JWS defined by JARM)
in the request to the resource server,
as detailed in Section~\ref{attack4-phished-at-malicious-as}.
As each resource server has one preconfigured authorization server,
the resource server does not provide access to a resource in this case.

The only remaining case is that the attacker includes a code associated with
the honest user in the request to the redirection endpoint of the client. For the Hybrid flow, both id tokens
contained in the authorization response and in the token response are required to 
have the same subject attribute and the same issuer value, which means
that they are both signed by the authorization server. However, such an id
token does not leak to the attacker, which means that the client will stop
the flow when receiving the second id token contained in the token response.
When using JARM, this would require the attacker to send a response JWS signed
by the authorization server that contains the code that belongs to an honest client
and an honest user identity. In the technical report,
we show that such a response JWS does not leak to the attacker.

\ifIncludeTechnicalReport
  
  \FloatBarrier\clearpage
  \section{Modeling mTLS} \label{appendix-mtls}  
The WIM models TLS at a high level of abstraction.
An HTTP request is encrypted with the public key of the recipient and contains a
symmetric key, which is used for encrypting the HTTP response. 
Furthermore, the model contains no certificates or public key infrastructures but uses a
function that maps domains to their public key.

Figure~\ref{fig:mtls-overview1} shows an overview of how we modeled mTLS.
The basic idea is that the server sends a nonce encrypted with the public key
of the client. The client proves possession of the private key by
decrypting this message.
In Step~\refprotostep{mtls1:mtls-init-req}, the client sends its client identifier to 
the authorization server. The authorization server then looks up the
public key associated with the client identifier, chooses a nonce
and encrypts it with the public key. As depicted in Step~\refprotostep{mtls1:mtls-init-resp},
the server additionally includes its public key. When the client decrypts the message,
it checks if the public key belongs to the server it wants to send the original 
message to. This prevents man-in-the-middle attacks, as only the honest client can
decrypt the response and as the public key of the server cannot be changed
by an attacker.
In Step~\refprotostep{mtls1:mtls-second-req}, the client sends the original request
with the decrypted nonce. When the server receives this message, 
it knows that the nonce was decrypted by the honest client
(as only the client knows the corresponding private key)
and that the client had chosen to send the nonce to the server
(due to the public key included in the response). Therefore,
the server can conclude that the message was sent by the honest client.

In effect, this resembles the behavior of the TLS handshake, as the verification of the client
certificate in TLS is done by signing all handshake messages
\cite[Section 7.4.8]{rfc5246}, which also includes information
about the server certificate, which means that the signature cannot
be reused for another server.
Instead of signing a sequence that contains information about the receiver,
in our model, the client checks the sender of the nonce, and only sends the decrypted
nonce to the creator of the nonce. In other words, a nonce decrypted by an honest
server that gets decrypted by the honest client is never sent to the attacker.

As explained in Section~\ref{fapi:mTLS}, the client uses the same certificate it used for the token request
when sending the access token to the resource server. 
While the resource server has to check the possession of corresponding private keys,
the validity of the certificate was already checked at the authorization server and can be ignored by the resource server. 
Therefore, in our model of the FAPI, the client does not send its client id
to the resource server, but its public key, and the resource server
encrypts the message with this public key.

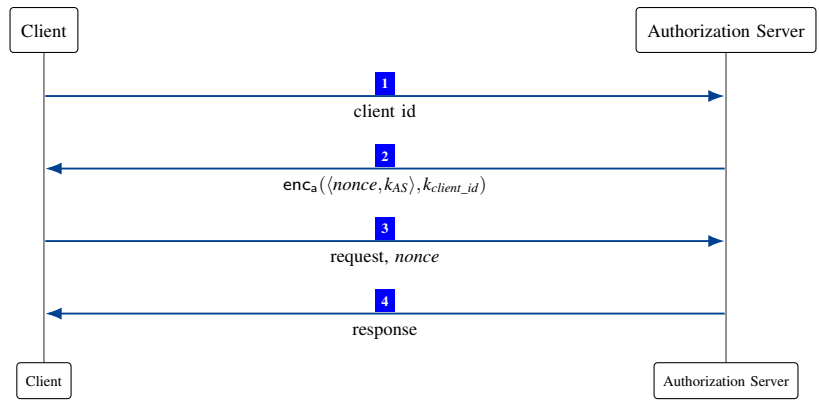
\begin{figure}[h!]
  \centering  
          \begin{tikzpicture}[]
        \pgfdeclarelayer{arrows}
        \pgfdeclarelayer{groups}
        \pgfdeclarelayer{markers}
        \pgfsetlayers{groups,arrows,main,markers}

        \matrix [column sep={0.5\textwidth,between origins}, row sep=0.1ex]
        {
        \node[annex_matrix_node,inner sep=0,outer sep=0](pos-0-0){}; &\node[annex_matrix_node,inner sep=0,outer sep=0](pos-1-0){};\node[annex_matrix_dummy_height,minimum height=5ex,anchor=center]{};\node[annex_matrix_dummy_height,minimum height=5ex,anchor=center]{};\\
\node[annex_matrix_node,inner sep=0,outer sep=0](pos-0-1){}; &\node[annex_matrix_node,inner sep=0,outer sep=0](pos-1-1){};\node[annex_matrix_dummy_height,minimum height=4ex+2ex,anchor=north,yshift=3ex]{};\\
\node[annex_matrix_node,inner sep=0,outer sep=0](pos-0-2){}; &\node[annex_matrix_node,inner sep=0,outer sep=0](pos-1-2){};\node[annex_matrix_dummy_height,minimum height=4ex+2ex,anchor=north,yshift=3ex]{};\\
\node[annex_matrix_node,inner sep=0,outer sep=0](pos-0-3){}; &\node[annex_matrix_node,inner sep=0,outer sep=0](pos-1-3){};\node[annex_matrix_dummy_height,minimum height=4ex+2ex,anchor=north,yshift=3ex]{};\\
\node[annex_matrix_node,inner sep=0,outer sep=0](pos-0-4){}; &\node[annex_matrix_node,inner sep=0,outer sep=0](pos-1-4){};\node[annex_matrix_dummy_height,minimum height=4ex+2ex,anchor=north,yshift=3ex]{};\\
\node[annex_matrix_node,inner sep=0,outer sep=0](pos-0-5){}; &\node[annex_matrix_node,inner sep=0,outer sep=0](pos-1-5){};\node[annex_matrix_dummy_height,minimum height=5ex,anchor=center]{};\node[annex_matrix_dummy_height,minimum height=5ex,anchor=center]{};\\
};

\node[name=StartParty_7_0,annex_start_party_box,] at (pos-0-0) {Client};

\node[name=StartParty_8_0,annex_start_party_box,] at (pos-1-0) {Authorization Server};

\node[name=EndParty_13_0,annex_end_party_box,] at (pos-0-5) {Client};

\node[name=EndParty_14_0,annex_end_party_box,] at (pos-1-5) {Authorization Server};

\begin{pgfonlayer}{arrows}%

\draw[annex_lifeline] (pos-0-0) -- (pos-0-5);

\draw[annex_lifeline] (pos-1-0) -- (pos-1-5);

        \draw[annex_http_request] (pos-0-1) to node [annex_arrow_text,above=2.6pt,anchor=base](HTTPRequest_9_0){\setcounter{protostep}{0}\protostep{mtls1:mtls-init-req} } node [annex_arrow_text,below=8pt,anchor=base](HTTPRequest_9_1){\contour{white}{client id}}  (pos-1-1);

        \draw[annex_http_request] (pos-1-2) to node [annex_arrow_text,above=2.6pt,anchor=base](HTTPRequest_10_0){\setcounter{protostep}{1}\protostep{mtls1:mtls-init-resp} } node [annex_arrow_text,below=8pt,anchor=base](HTTPRequest_10_1){\contour{white}{$\enc{\an{\mi{nonce}, \mi{k}_\mi{AS}}}{\mi{k}_\mi{client\_id}}$}}  (pos-0-2);

        \draw[annex_http_request] (pos-0-3) to node [annex_arrow_text,above=2.6pt,anchor=base](HTTPRequest_11_0){\setcounter{protostep}{2}\protostep{mtls1:mtls-second-req} } node [annex_arrow_text,below=8pt,anchor=base](HTTPRequest_11_1){\contour{white}{request, $\mi{nonce}$}}  (pos-1-3);

        \draw[annex_http_request] (pos-1-4) to node [annex_arrow_text,above=2.6pt,anchor=base](HTTPRequest_12_0){\setcounter{protostep}{3}\protostep{mtls1:mtls-second-resp} } node [annex_arrow_text,below=8pt,anchor=base](HTTPRequest_12_1){\contour{white}{response}}  (pos-0-4);

\end{pgfonlayer}

\begin{pgfonlayer}{markers}%

\end{pgfonlayer}
        \end{tikzpicture}
        
  \caption{Overview of mTLS}
  \label{fig:mtls-overview1}
\end{figure}

All messages are sent by the generic HTTPS server model (Appendix~\ref{sec:generic-https-server-model}),
which means that each request is encrypted asymmetrically, and 
the responses are encrypted symmetrically with a key that was included 
in the request. 
For completeness, Figure~\ref{fig:mtls-details} shows the complete messages, i.e., with
the encryption used for transmitting the messages.

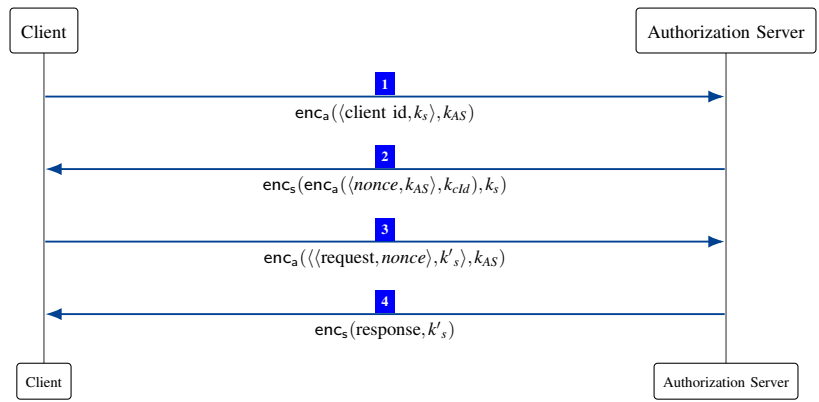
\begin{figure}[h!]
  \centering  
          \begin{tikzpicture}[]
        \pgfdeclarelayer{arrows}
        \pgfdeclarelayer{groups}
        \pgfdeclarelayer{markers}
        \pgfsetlayers{groups,arrows,main,markers}

        \matrix [column sep={0.5\textwidth,between origins}, row sep=0.1ex]
        {
        \node[annex_matrix_node,inner sep=0,outer sep=0](pos-0-0){}; &\node[annex_matrix_node,inner sep=0,outer sep=0](pos-1-0){};\node[annex_matrix_dummy_height,minimum height=5ex,anchor=center]{};\node[annex_matrix_dummy_height,minimum height=5ex,anchor=center]{};\\
\node[annex_matrix_node,inner sep=0,outer sep=0](pos-0-1){}; &\node[annex_matrix_node,inner sep=0,outer sep=0](pos-1-1){};\node[annex_matrix_dummy_height,minimum height=4ex+2ex,anchor=north,yshift=3ex]{};\\
\node[annex_matrix_node,inner sep=0,outer sep=0](pos-0-2){}; &\node[annex_matrix_node,inner sep=0,outer sep=0](pos-1-2){};\node[annex_matrix_dummy_height,minimum height=4ex+2ex,anchor=north,yshift=3ex]{};\\
\node[annex_matrix_node,inner sep=0,outer sep=0](pos-0-3){}; &\node[annex_matrix_node,inner sep=0,outer sep=0](pos-1-3){};\node[annex_matrix_dummy_height,minimum height=4ex+2ex,anchor=north,yshift=3ex]{};\\
\node[annex_matrix_node,inner sep=0,outer sep=0](pos-0-4){}; &\node[annex_matrix_node,inner sep=0,outer sep=0](pos-1-4){};\node[annex_matrix_dummy_height,minimum height=4ex+2ex,anchor=north,yshift=3ex]{};\\
\node[annex_matrix_node,inner sep=0,outer sep=0](pos-0-5){}; &\node[annex_matrix_node,inner sep=0,outer sep=0](pos-1-5){};\node[annex_matrix_dummy_height,minimum height=5ex,anchor=center]{};\node[annex_matrix_dummy_height,minimum height=5ex,anchor=center]{};\\
};

\node[name=StartParty_7_0,annex_start_party_box,] at (pos-0-0) {Client};

\node[name=StartParty_8_0,annex_start_party_box,] at (pos-1-0) {Authorization Server};

\node[name=EndParty_13_0,annex_end_party_box,] at (pos-0-5) {Client};

\node[name=EndParty_14_0,annex_end_party_box,] at (pos-1-5) {Authorization Server};

\begin{pgfonlayer}{arrows}%

\draw[annex_lifeline] (pos-0-0) -- (pos-0-5);

\draw[annex_lifeline] (pos-1-0) -- (pos-1-5);

        \draw[annex_http_request] (pos-0-1) to node [annex_arrow_text,above=2.6pt,anchor=base](HTTPRequest_9_0){\setcounter{protostep}{0}\protostep{mtls2:mtls-init-req} } node [annex_arrow_text,below=8pt,anchor=base](HTTPRequest_9_1){\contour{white}{$\enc{\an{\text{client id}, \mi{k}_s}}{\mi{k}_\mi{AS}}$}}  (pos-1-1);

        \draw[annex_http_request] (pos-1-2) to node [annex_arrow_text,above=2.6pt,anchor=base](HTTPRequest_10_0){\setcounter{protostep}{1}\protostep{mtls2:mtls-init-resp} } node [annex_arrow_text,below=8pt,anchor=base](HTTPRequest_10_1){\contour{white}{$\encs{\enc{\an{\mi{nonce}, \mi{k}_\mi{AS}}}{\mi{k}_\mi{cId}}}{\mi{k}_s}$}}  (pos-0-2);

        \draw[annex_http_request] (pos-0-3) to node [annex_arrow_text,above=2.6pt,anchor=base](HTTPRequest_11_0){\setcounter{protostep}{2}\protostep{mtls2:mtls-second-req} } node [annex_arrow_text,below=8pt,anchor=base](HTTPRequest_11_1){\contour{white}{$\enc{\an{\an{\text{request}, \mi{nonce}}, \mi{k'}_s}}{\mi{k}_\mi{AS}}$}}  (pos-1-3);

        \draw[annex_http_request] (pos-1-4) to node [annex_arrow_text,above=2.6pt,anchor=base](HTTPRequest_12_0){\setcounter{protostep}{3}\protostep{mtls2:mtls-second-resp} } node [annex_arrow_text,below=8pt,anchor=base](HTTPRequest_12_1){\contour{white}{$\encs{\text{response}}{\mi{k'}_s}$}}  (pos-0-4);

\end{pgfonlayer}

\begin{pgfonlayer}{markers}%

\end{pgfonlayer}
        \end{tikzpicture}
        
  \caption{Detailed view on mTLS}
  \label{fig:mtls-details}
\end{figure}

  \FloatBarrier\clearpage
  \section{More Details on OAuth 2.0 Token Binding} \label{appendix:OAUTB}

In the following, we describe more details of OAuth 2.0 Token Binding (see Section~\ref{fapi:OAUTB})
and the modeling within the Web Infrastructure Model.

\subsection{Exported Keying Material}

The proof of possession of a private key is done by signing
a so-called \emph{Exported Keying Material} (EKM). This value
is created with parameters of the TLS connection such that it is
unique to the connection. Therefore, it is not possible for an attacker
to reuse signed EKM values for an honest server. 
Before explaining how the EKM value is generated, we give a brief explanation
of the relevant mechanisms used in TLS \cite{rfc5246}:

\begin{itemize}
\item Client and Server Random:  The client random and server random
are values chosen by the client and server and transmitted in the TLS handshake.
\item Pseudorandom Function:  TLS specifies a pseudorandom function (PRF)
that is used within the TLS protocol \cite[Section 5]{rfc5246}.
\item Premaster Secret:  The premaster secret is a secret shared between
both participants of the TLS connection, for example, created by a Diffie-Hellman key exchange.
The length of the premaster secret varies depending on the method used for its creation.
\item Master Secret:  The master secret is created by applying the pseudorandom function
to the premaster secret to generate a secret of a fixed length. It essentially has the value
$\mathsf{PRF}(\mi{premaster\_secret}, 
\mi{client\_random}, \mi{server\_random})$ \cite[Section 8.1]{rfc5246}; we omitted constant
values).
\end{itemize}

Using these building blocks, the EKM value \cite{rfc5705} is essentially defined as

\[ \mathsf{PRF}(\mi{master\_secret}, \mi{client\_random}, \mi{server\_random}) \]

As noted in Section 7.5 of the Token Binding Protocol \cite{ietf-tokbind-protocol-19},
the use of the Extended Master Secret TLS extension \cite{BhargavanDelignat-LavaudPirontiLangleyRay-RFC-2015}
is mandatory. This extension redefines the master secret
as  $ \mathsf{PRF}(\mi{pre\_master\_secret}, \mi{session\_hash}) $,
where $\mi{session\_hash}$ is the hash of all TLS handshake messages
of the session
(again, omitting constant values).
Without this extension, the \emph{Triple Handshake Attack}
\cite{BhargavanDelignat-LavaudFournetPirontiStrub-SP-2014}
can be applied, which eventually leads to the creation of two
TLS sessions with the same master secret. %
The basic idea is that if an honest client establishes a TLS connection
to a malicious server
(for example, when the token endpoint is misconfigured),
the attacker can relay the random values chosen by the honest client and server, hence,
creating TLS connections to both the client and server with the same master secret
and therefore, with the same EKM value.
By including the hash over the TLS handshake messages, relevant information about
the session, like the certificate used by the server (which is exchanged in the handshake),
influence the value of the master secret.

 \subsection{Modeling Token Binding}
 
 The main difficulty of modeling OAuth 2.0 Token Binding is the high level
 of abstraction of TLS within the WIM, as already explained in Appendix~\ref{appendix-mtls}.
 
 An overview of how we modeled OAUTB is shown in Figure~\ref{fig:oautb-overview}.
 For compensating the absence of the TLS handshake, both participants choose
 nonces, as shown in Steps~\refprotostep{oautb1:oautb-init-req} 
 and \refprotostep{oautb1:oautb-init-resp}.
 When the client sends the actual request, 
 it creates and includes a Token Binding message, as shown in  
 Step~\refprotostep{oautb1:oautb-second-req}. As explained above, 
 the client includes not only the two nonces, but also the 
 public key of the authorization server in the EKM value
 for modeling the extended master secret.

 \begin{figure}[h!]
   \centering  
           \begin{tikzpicture}[]
        \pgfdeclarelayer{arrows}
        \pgfdeclarelayer{groups}
        \pgfdeclarelayer{markers}
        \pgfsetlayers{groups,arrows,main,markers}

        \matrix [column sep={0.5\textwidth,between origins}, row sep=0.1ex]
        {
        \node[annex_matrix_node,inner sep=0,outer sep=0](pos-0-0){}; &\node[annex_matrix_node,inner sep=0,outer sep=0](pos-1-0){};\node[annex_matrix_dummy_height,minimum height=5ex,anchor=center]{};\node[annex_matrix_dummy_height,minimum height=5ex,anchor=center]{};\\
\node[annex_matrix_node,inner sep=0,outer sep=0](pos-0-1){}; &\node[annex_matrix_node,inner sep=0,outer sep=0](pos-1-1){};\node[annex_matrix_dummy_height,minimum height=4ex+2ex,anchor=north,yshift=3ex]{};\\
\node[annex_matrix_node,inner sep=0,outer sep=0](pos-0-2){}; &\node[annex_matrix_node,inner sep=0,outer sep=0](pos-1-2){};\node[annex_matrix_dummy_height,minimum height=4ex+2ex,anchor=north,yshift=3ex]{};\\
\node[annex_matrix_node,inner sep=0,outer sep=0](pos-0-3){}; &\node[annex_matrix_node,inner sep=0,outer sep=0](pos-1-3){};\node[annex_matrix_dummy_height,minimum height=4ex+2ex,anchor=north,yshift=3ex]{};\\
\node[annex_matrix_node,inner sep=0,outer sep=0](pos-0-4){}; &\node[annex_matrix_node,inner sep=0,outer sep=0](pos-1-4){};\node[annex_matrix_dummy_height,minimum height=4ex+2ex,anchor=north,yshift=3ex]{};\\
\node[annex_matrix_node,inner sep=0,outer sep=0](pos-0-5){}; &\node[annex_matrix_node,inner sep=0,outer sep=0](pos-1-5){};\node[annex_matrix_dummy_height,minimum height=5ex,anchor=center]{};\node[annex_matrix_dummy_height,minimum height=5ex,anchor=center]{};\\
};

\node[name=StartParty_7_0,annex_start_party_box,] at (pos-0-0) {Client};

\node[name=StartParty_8_0,annex_start_party_box,] at (pos-1-0) {Authorization Server};

\node[name=EndParty_13_0,annex_end_party_box,] at (pos-0-5) {Client};

\node[name=EndParty_14_0,annex_end_party_box,] at (pos-1-5) {Authorization Server};

\begin{pgfonlayer}{arrows}%

\draw[annex_lifeline] (pos-0-0) -- (pos-0-5);

\draw[annex_lifeline] (pos-1-0) -- (pos-1-5);

        \draw[annex_http_request] (pos-0-1) to node [annex_arrow_text,above=2.6pt,anchor=base](HTTPRequest_9_0){\setcounter{protostep}{0}\protostep{oautb1:oautb-init-req} } node [annex_arrow_text,below=8pt,anchor=base](HTTPRequest_9_1){\contour{white}{$\mi{nonce}_1$}}  (pos-1-1);

        \draw[annex_http_request] (pos-1-2) to node [annex_arrow_text,above=2.6pt,anchor=base](HTTPRequest_10_0){\setcounter{protostep}{1}\protostep{oautb1:oautb-init-resp} } node [annex_arrow_text,below=8pt,anchor=base](HTTPRequest_10_1){\contour{white}{$\mi{nonce}_2$}}  (pos-0-2);

        \draw[annex_http_request] (pos-0-3) to node [annex_arrow_text,above=2.6pt,anchor=base](HTTPRequest_11_0){\setcounter{protostep}{2}\protostep{oautb1:oautb-second-req} } node [annex_arrow_text,below=8pt,anchor=base](HTTPRequest_11_1){\contour{white}{request, $\mi{TB\_ID}_{\text{C}, \text{AS}}$, $\mathsf{sig}(\mathsf{hash}(\an{ \mi{nonce}_1, \mi{nonce}_2, \mi{k}_\mi{AS} }), \mi{TB\_priv\_key}_{\text{C}, \text{AS}})$}}  (pos-1-3);

        \draw[annex_http_request] (pos-1-4) to node [annex_arrow_text,above=2.6pt,anchor=base](HTTPRequest_12_0){\setcounter{protostep}{3}\protostep{oautb1:oautb-second-resp} } node [annex_arrow_text,below=8pt,anchor=base](HTTPRequest_12_1){\contour{white}{response}}  (pos-0-4);

\end{pgfonlayer}

\begin{pgfonlayer}{markers}%

\end{pgfonlayer}
        \end{tikzpicture}
        
   \caption{Overview of OAUTB}
   \label{fig:oautb-overview}
 \end{figure}
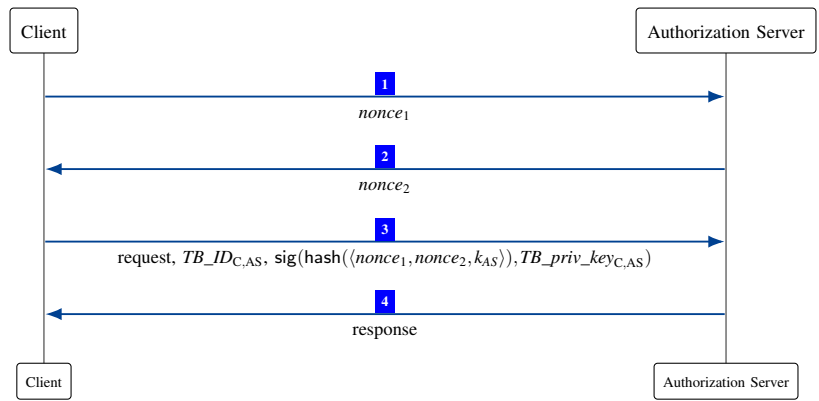
 
 For completeness, Figure~\ref{fig:oautb-details}
 shows the entire messages, i.e., with 
 the encryption done by the generic HTTPS server
 (Appendix~\ref{sec:generic-https-server-model}).
 
 \begin{figure}[h!]
   \centering  
           \begin{tikzpicture}[]
        \pgfdeclarelayer{arrows}
        \pgfdeclarelayer{groups}
        \pgfdeclarelayer{markers}
        \pgfsetlayers{groups,arrows,main,markers}

        \matrix [column sep={0.5\textwidth,between origins}, row sep=0.1ex]
        {
        \node[annex_matrix_node,inner sep=0,outer sep=0](pos-0-0){}; &\node[annex_matrix_node,inner sep=0,outer sep=0](pos-1-0){};\node[annex_matrix_dummy_height,minimum height=5ex,anchor=center]{};\node[annex_matrix_dummy_height,minimum height=5ex,anchor=center]{};\\
\node[annex_matrix_node,inner sep=0,outer sep=0](pos-0-1){}; &\node[annex_matrix_node,inner sep=0,outer sep=0](pos-1-1){};\node[annex_matrix_dummy_height,minimum height=4ex+2ex,anchor=north,yshift=3ex]{};\\
\node[annex_matrix_node,inner sep=0,outer sep=0](pos-0-2){}; &\node[annex_matrix_node,inner sep=0,outer sep=0](pos-1-2){};\node[annex_matrix_dummy_height,minimum height=4ex+2ex,anchor=north,yshift=3ex]{};\\
\node[annex_matrix_node,inner sep=0,outer sep=0](pos-0-3){}; &\node[annex_matrix_node,inner sep=0,outer sep=0](pos-1-3){};\node[annex_matrix_dummy_height,minimum height=4ex+2ex,anchor=north,yshift=3ex]{};\\
\node[annex_matrix_node,inner sep=0,outer sep=0](pos-0-4){}; &\node[annex_matrix_node,inner sep=0,outer sep=0](pos-1-4){};\node[annex_matrix_dummy_height,minimum height=4ex+2ex,anchor=north,yshift=3ex]{};\\
\node[annex_matrix_node,inner sep=0,outer sep=0](pos-0-5){}; &\node[annex_matrix_node,inner sep=0,outer sep=0](pos-1-5){};\node[annex_matrix_dummy_height,minimum height=5ex,anchor=center]{};\node[annex_matrix_dummy_height,minimum height=5ex,anchor=center]{};\\
};

\node[name=StartParty_7_0,annex_start_party_box,] at (pos-0-0) {Client};

\node[name=StartParty_8_0,annex_start_party_box,] at (pos-1-0) {Authorization Server};

\node[name=EndParty_13_0,annex_end_party_box,] at (pos-0-5) {Client};

\node[name=EndParty_14_0,annex_end_party_box,] at (pos-1-5) {Authorization Server};

\begin{pgfonlayer}{arrows}%

\draw[annex_lifeline] (pos-0-0) -- (pos-0-5);

\draw[annex_lifeline] (pos-1-0) -- (pos-1-5);

        \draw[annex_http_request] (pos-0-1) to node [annex_arrow_text,above=2.6pt,anchor=base](HTTPRequest_9_0){\setcounter{protostep}{0}\protostep{oautb2:oautb-init-req} } node [annex_arrow_text,below=8pt,anchor=base](HTTPRequest_9_1){\contour{white}{$\enc{\an{\mi{nonce}_1, \mi{k}_s}}{\mi{k}_\mi{AS}}$}}  (pos-1-1);

        \draw[annex_http_request] (pos-1-2) to node [annex_arrow_text,above=2.6pt,anchor=base](HTTPRequest_10_0){\setcounter{protostep}{1}\protostep{oautb2:oautb-init-resp} } node [annex_arrow_text,below=8pt,anchor=base](HTTPRequest_10_1){\contour{white}{$\encs{\mi{nonce}_2}{\mi{k}_s}$}}  (pos-0-2);

        \draw[annex_http_request] (pos-0-3) to node [annex_arrow_text,above=2.6pt,anchor=base](HTTPRequest_11_0){\setcounter{protostep}{2}\protostep{oautb2:oautb-second-req} } node [annex_arrow_text,below=8pt,anchor=base](HTTPRequest_11_1){\contour{white}{$\enc{\an{ \an{request, \mi{TB\_ID}_{\text{C}, \text{AS}}, \mathsf{sig}(\mathsf{hash}(\an{ \mi{nonce}_1, \mi{nonce}_2, \mi{k}_\mi{AS} }), \mi{TB\_priv\_key}_{\text{C}, \text{AS}})}, \mi{k'}_s}}{\mi{k}_\mi{AS}}$}}  (pos-1-3);

        \draw[annex_http_request] (pos-1-4) to node [annex_arrow_text,above=2.6pt,anchor=base](HTTPRequest_12_0){\setcounter{protostep}{3}\protostep{oautb2:oautb-second-resp} } node [annex_arrow_text,below=8pt,anchor=base](HTTPRequest_12_1){\contour{white}{$\encs{\text{response}}{\mi{k'}_s}$}}  (pos-0-4);

\end{pgfonlayer}

\begin{pgfonlayer}{markers}%

\end{pgfonlayer}
        \end{tikzpicture}
        
   \caption{Detailed view on OAUTB}
   \label{fig:oautb-details}
 \end{figure}
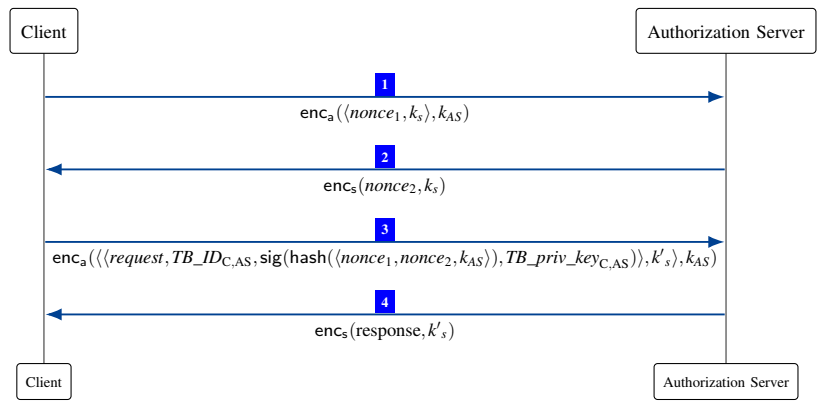

 \subsection{Access Token issued from Authorization Endpoint} \label{expl-single-at}
 In the Read-Write profile, an access token issued from the authorization endpoint is required to be token bound.
 When using OAUTB, this means that the token must be bound to the Token Binding
 ID which the client uses for the resource server. However,
 the authorization request is sent to the authorization 
 server by the browser, which cannot send Token Binding messages
 with an ID used by the client.
 Therefore, our model does not include access tokens being
 issued from the authorization endpoint, as this would
 mean that additional communication between the 
 client and the authorization server is needed,
 which is not fully specified yet.

 We note that this issue is also present in the case of confidential clients using mTLS.
 Here, the access token cannot be bound to the certificate of the client,
 as the authorization request is not sent directly by the client to the
 authorization server, but by the browser. More details can be found in 
 Section 4.5 of \cite{ietf-oauth-mtls-09}.

  \FloatBarrier\clearpage
  \clearpage
\section{Additions to the Web Infrastructure Model} \label{appendix-additions-to-the-wim}
Within the scope of this technical report, we adhere to the WIM as
defined in \cite{FettKuestersSchmitz-TR-OIDC-2017}, where it was used for modeling and analyzing
OpenID Connect. In the following, we describe the additions to the model for the analysis
of the FAPI. 

\subsection{Functions} 
In addition to the function symbols defined in Appendix B of \cite{FettKuestersSchmitz-TR-OIDC-2017},
we add the function symbol for hashing $\mathsf{hash}(.)$ 
to the signature $\Sigma$.
For computing and verifying message authentication codes, we
add
$\mathsf{mac}(.,.)$
and $\mathsf{checkmac}(.,.)$.
Regarding the equational theory, we additionally define
$\mathsf{mac}(x,y) = \mathsf{hash}(\an{x,y})$
and 
$\mathsf{checkmac}(\mathsf{mac}(x,y),y) = \True$.
Furthermore, we extend the definition of $\mathsf{extractmsg}$
such that it also extracts messages out of 
$\mathsf{mac}(x,y)$, i.e.,
$\mathsf{extractmsg}(\mathsf{mac}(x,y)) = x$.
As a short form of the mapping from identities to their governor,
we define $\mi{gov}(.) = \mi{governor}(.)$.

\subsection{Cookies} 
In Appendix C of \cite{FettKuestersSchmitz-TR-OIDC-2017}, a cookie is defined as a
term of the form $\an{\mi{name}, \mi{content}}$ with $\mi{name} \in \terms$. 
As the name is a term, it may also be a sequence consisting of two part. 
If the name it consists of two parts, we call the 
first part of the sequence (i.e., $\mi{name}.1$) the \emph{prefix} of the name.

In the following, define the \emph{Secure} prefix (see~\cite{rfc6265bis-httpwg}):
  When the $\mi{\_\_Secure}$ prefix is set, the browser accepts the cookie
  only if the $\mi{secure}$ attribute is set. As such cookies are only transferred over
  secure channels (i.e., with TLS), the cookie cannot be set by a network attacker.

For modeling this, we require that the $\mathsf{AddCookie}$
function should be called with (as an additional argument) 
the protocol with which the corresponding request was sent,
i.e., we modify the function
$\mathsf{PROCESSRESPONSE}$ (using the same number for the algorithm as in
\cite{FettKuestersSchmitz-TR-OIDC-2017}) :

\setcounter{algorithm}{7}

\begin{algorithm}[h]
\caption{\label{alg:processresponse-browser} Web Browser Model: Process an HTTP response.}
\begin{algorithmic}[1]
\Function{$\mathsf{PROCESSRESPONSE}$}{$\mi{response}$, $\mi{reference}$, $\mi{request}$, $\mi{requestUrl}$, $s'$}
  \If{$\mathtt{Set{\mhyphen}Cookie} \in
    \comp{\mi{response}}{headers}$}
    \For{\textbf{each} $c \inPairing \comp{\mi{response}}{headers}\left[\mathtt{Set{\mhyphen}Cookie}\right]$, $c \in \mathsf{Cookies}$}
      \Let{$\comp{s'}{cookies}\left[\mi{request}.\str{host}\right]$\breakalgohook{3}}{$\mathsf{AddCookie}(\comp{s'}{cookies}\left[\mi{request}.\str{host}\right], c, \mi{requestUrl}.\str{protocol})$} \label{line:set-cookie-browser}
    \EndFor
  \EndIf  
\dots
\EndFunction
\end{algorithmic} %
\end{algorithm}

\setcounter{algorithm}{1}

The modified function $\mathsf{AddCookie}$ looks as follows
(again using the same number as in
\cite{FettKuestersSchmitz-TR-OIDC-2017}):

\setcounter{definition}{42}

\begin{definition} \label{def:addcookie} For a sequence of cookies (with pairwise different
  names) $\mi{oldcookies}$ and a cookie $c$, the sequence
  $\mathsf{AddCookie}(\mi{oldcookies}, c, \mi{protocol})$ is defined by the
  following algorithm: 
  If ($(c.\str{name}.1 \equiv \str{\_\_Secure}) \Rightarrow (\mi{protocol} \equiv \mi{S})$), then: 
  Let $m := \mi{oldcookies}$. Remove
  any $c'$ from $m$ that has $\comp{c}{name} \equiv
  \comp{c'}{name}$. Append $c$ to $m$ and return $m$.
\end{definition}

\setcounter{definition}{2}

\subsection{Browser} \label{lab:appendix-fapi-browser}
We use the same browser model as defined in Appendix D of \cite{FettKuestersSchmitz-TR-OIDC-2017},
with small modifications. Therefore, we do not show the full model here
but limit the description to the algorithms that differ from the original browser model.

\subsubsection*{Main Differences}

\begin{itemize}
\item Algorithm~\ref{alg:send}: New input argument $\mi{refTB}$. This value is saved along with the other 
information needed for the current DNS request.
\item Algorithm~\ref{alg:cont-send}: Sends messages directly to a receiver, encrypted with a symmetric key. The receiver and key are
given as an input. This algorithm is used for sending the follow-up request when using OAUTB.
\item Algorithm~\ref{alg:processresponse}: If the response contains $\str{tb\_nonce}$, the original message is sent together
with the OAUTB message. Furthermore, the actions specified when the
$\str{Include\mhyphen{}Referred\mhyphen{}Token\mhyphen{}Binding\mhyphen{}ID}$
header is set are implemented here. 
\item Algorithm~\ref{alg:browsermain}: When receiving a DNS response, the browser checks if OAUTB needs to be applied. In this case,
the browser first sends a request to the OAUTB endpoint to obtain a new nonce. The original message is saved and send together with
the OAUTB message in a second request.
\end{itemize}

When sending the first OAUTB request, the nonce $\nu_{n1}$ is used as a reference (in Algorithm~\ref{alg:browsermain}).

\begin{algorithm}[H] %
\caption{\label{alg:send} Web Browser Model: Prepare headers, do DNS resolution, save message. }
\begin{algorithmic}[1]
  \Function{$\mathsf{HTTP\_SEND}$}{$\mi{reference}$, $\mi{message}$, $\mi{url}$, $\mi{origin}$, $\mi{referrer}$, $\mi{referrerPolicy}$, $s'$, $\mi{refTB}$}
    \If{$\comp{\mi{message}}{host} \inPairing \comp{s'}{sts}$}
      \Let{$\mi{url}.\str{protocol}$}{$\https$}
    \EndIf
    \Let{ $\mi{cookies}$}{$\langle\{\an{\comp{c}{name}, \comp{\comp{c}{content}}{value}} | c\inPairing \comp{s'}{cookies}\left[\comp{\mi{message}}{host}\right]$} \label{line:assemble-cookies-for-request} \breakalgohook{1} $\wedge \left(\comp{\comp{c}{content}}{secure} \implies \left(\mi{url}.\str{protocol} = \https\right)\right) \}\rangle$ \label{line:cookie-rules-http}
    \Let{$\comp{\mi{message}}{headers}[\str{Cookie}]$}{$\mi{cookies}$}
    \If{$\mi{origin} \not\equiv \bot$}
      \Let{$\comp{\mi{message}}{headers}[\str{Origin}]$}{$\mi{origin}$}
    \EndIf
    \If{$\mi{referrerPolicy} \equiv \str{noreferrer}$} 
      \Let{$\mi{referrer}$}{$\bot$}
    \EndIf
    \If{$\mi{referrer} \not\equiv \bot$}
      \If{$\mi{referrerPolicy} \equiv \str{origin}$} 
        \Let{$\mi{referrer}$}{$\an{\cUrl, \mi{referrer}.\str{protocol}, \mi{referrer}.\str{host}, \str{/}, \an{}, \bot}$} \Comment{Referrer stripped down to origin.}
      \EndIf
      \Let{$\mi{referrer}.\str{fragment}$}{$\bot$} \Comment{Browsers do not send fragment identifiers in the Referer header.}
      \Let{$\comp{\mi{message}}{headers}[\str{Referer}]$}{$\mi{referrer}$}
    \EndIf
    \Let{$\comp{s'}{pendingDNS}[\nu_8]$}{$\an{\mi{reference},
        \mi{message}, \mi{url}, \mi{refTB}}$} \label{line:add-to-pendingdns}
    \State \textbf{stop} $\an{\an{\comp{s'}{DNSaddress},a,
    \an{\cDNSresolve, \mi{message}.\str{host}, \nu_8}}}$, $s'$
  \EndFunction
\end{algorithmic} %
\end{algorithm}

\FloatBarrier

\begin{algorithm}[H] %
\caption{\label{alg:cont-send} Web Browser Model: Send message to IP encrypted with given sym.~key. }
\begin{algorithmic}[1]
  \Function{$\mathsf{OAUTB\_CONT\_SEND}$}{$\mi{reference}$, $\mi{message}$, $\mi{url}$, $\mi{key}$, $\mi{f}$, $s'$}
        \AppendBreak{2}{$\langle\mi{reference}$, $\mi{message}$, $\mi{url}$, $\mi{key}$, $\mi{f} \rangle$}{$\comp{s'}{pendingRequests}$} \label{line:cont-send-add-to-pendingrequests-https}
        \Let{$\mi{message}$}{$\encs{\mi{message}}{\mi{key}}$} \label{line:cont-send-select-enc-key}
      \State \textbf{stop} $\an{\an{f, a, \mi{message}}}$, $s'$
  \EndFunction
\end{algorithmic} %
\end{algorithm}

\FloatBarrier
\begin{algorithm}[H] %
\caption{\label{alg:processresponse} Web Browser Model: Process an HTTP response.}
\begin{algorithmic}[1]
\Function{$\mathsf{PROCESSRESPONSE}$}{$\mi{response}$, $\mi{reference}$, $\mi{request}$, $\mi{requestUrl}$, $s'$, $\mi{key}$,
	$\mi{f}$}
  \If{$\str{tb\_nonce} \in \comp{\mi{response}}{body}$}
			\Let{$\an{\mi{orig\_reference}, \mi{orig\_message}, \mi{orig\_url}, \mi{refTB}}$}{$s'.\str{tokenBindingRequests}[\mi{request}.\str{nonce}]$}

				\Let{$\mi{ekm}$}{$\mathsf{hash}(\an{\mi{request}.\str{nonce}, \mi{response}.\str{body}
						[\str{tb\_nonce}], \mi{keyMapping}[\mi{request}.\str{host}] })$} %

				\Let{$\mi{tb\_prov\_priv}$}{$s'.\mi{tokenBindings}[\mi{orig\_url}.\str{host}]$} 
        \Let{$\mi{TB\_Msg\_provided}$}{$[\str{id}{:} \mathsf{pub}(\mi{tb\_prov\_priv}),\str{sig}{:} \mathsf{sig}(\mi{ekm}, \mi{tb\_prov\_priv})]$} \label{browser-TB-prov-msg}

			\If{$\mi{refTB} \not \equiv \bot$}
				\Let{$\mi{tb\_ref\_priv}$}{$s'.\mi{tokenBindings}[\mi{refTB}]$} %
   			\Let{$\mi{TB\_Msg\_referred}$}{$[\str{id}{:} \mathsf{pub}(\mi{tb\_ref\_priv}), \str{sig}{:}\mathsf{sig}( \mi{ekm}, \mi{tb\_ref\_priv})]$}
			\Else
   			\Let{$\mi{TB\_Msg\_referred}$}{$\an{}$}
			\EndIf
			
			\Let{$\mi{orig\_message}.\mi{headers}[\str{Sec\mhyphen{}Token\mhyphen{}Binding}]$}{$[\str{prov}{:}\mi{TB\_Msg\_provided},\str{ref}{:}\mi{TB\_Msg\_referred}]$}

      \CallFun{TB\_CONT\_SEND}{$\mi{orig\_reference}$, $\mi{orig\_message}$, $\mi{orig\_url}$, $\mi{key}$, $\mi{f}$,  $s'$}
	\EndIf
  \If{$\mathtt{Set{\mhyphen}Cookie} \in
    \comp{\mi{response}}{headers}$}
    \For{\textbf{each} $c \inPairing \comp{\mi{response}}{headers}\left[\mathtt{Set{\mhyphen}Cookie}\right]$, $c \in \mathsf{Cookies}$}
      \Let{$\comp{s'}{cookies}\left[\mi{request}.\str{host}\right]$\breakalgohook{3}}{$\mathsf{AddCookie}(\comp{s'}{cookies}\left[\mi{request}.\str{host}\right], c, \mi{requestUrl}.\str{protocol})$} \label{line:set-cookie}
    \EndFor
  \EndIf  
  \If{$\mathtt{Strict{\mhyphen}Transport{\mhyphen}Security} \in \comp{\mi{response}}{headers}$ $\wedge$ $\mi{requestUrl}.\str{protocol} \equiv \https$}
    \Append{$\comp{\mi{request}}{host}$}{$\comp{s'}{sts}$}
  \EndIf
  \If{$\str{Referer} \in \comp{request}{headers}$} 
    \Let{$\mi{referrer}$}{$\comp{request}{headers}[\str{Referer}]$}
  \Else
    \Let{$\mi{referrer}$}{$\bot$}
  \EndIf
  \If{$\mathtt{Location} \in \comp{\mi{response}}{headers} \wedge \comp{\mi{response}}{status} \in \{303, 307\}$} \label{line:location-header} 
    \Let{$\mi{url}$}{$\comp{\mi{response}}{headers}\left[\mathtt{Location}\right]$}
    \If{$\mi{url}.\str{fragment} \equiv \bot$}
      \Let{$\mi{url}.\str{fragment}$}{$\mi{requestUrl}.\str{fragment}$}
    \EndIf
	  \If{$\mathtt{Include{\mhyphen}Referred{\mhyphen}Token{\mhyphen}Binding{\mhyphen}ID} \in \comp{\mi{response}}{headers}$}
	  	\Comment{Always use TB } %
	  	\If{$\comp{\mi{response}}{headers}[\mathtt{Include{\mhyphen}Referred{\mhyphen}Token{\mhyphen}Binding{\mhyphen}ID}] \equiv \True $}
	 			\Let{$s'.\str{useTB}[\mi{response}.\str{host}]$}{$\True$} 
	 			\Let{$s'.\str{useTB}[\mi{response}.\str{headers}[\str{Location}]]$}{$\True$} 
	  	\EndIf
	  \EndIf
    \Let{$\mi{method}'$}{$\comp{\mi{request}}{method}$} 
\algstore{proc-resp0}
\end{algorithmic}
\end{algorithm}

 \begin{algorithm}[h!]
 \begin{algorithmic}[1]
 \algrestore{proc-resp0}
    \Let{$\mi{body}'$}{$\comp{\mi{request}}{body}$} 
    \If{$\str{Origin} \in \comp{request}{headers}$}
      \Let{$\mi{origin}$}{$\an{\comp{request}{headers}[\str{Origin}], \an{\comp{request}{host}, \mi{url}.\str{protocol}}}$}
    \Else
      \Let{$\mi{origin}$}{$\bot$}
    \EndIf
    \If{$\comp{\mi{response}}{status} \equiv 303 \wedge \comp{\mi{request}}{method} \not\in \{\mGet, \mHead\}$}
      \Let {$\mi{method}'$}{$\mGet$}
      \Let{$\mi{body}'$}{$\an{}$} \label{browser-remove-body}
    \EndIf
    \If{$\exists\, \ptr{w} \in \mathsf{Subwindows}(s')$ \textbf{such that} $\comp{\compn{s'}{\ptr{w}}}{nonce} \equiv \mi{reference}$} \Comment{Do not redirect XHRs.}
      \Let{$\mi{req}$}{$\hreq{
            nonce=\nu_6,
            method=\mi{method'},
            host=\comp{\mi{url}}{host},
            path=\comp{\mi{url}}{path},
            headers=\an{},
            parameters=\comp{\mi{url}}{parameters},
            xbody=\mi{body}'
          }$}
      \Let{$\mi{referrerPolicy}$}{$\mi{response}.\str{headers}[\str{ReferrerPolicy}]$}
		\If{$\mathtt{Include{\mhyphen}Referred{\mhyphen}Token{\mhyphen}Binding{\mhyphen}ID} \in \comp{\mi{response}}{headers}$}
	  	\If{$\comp{\mi{response}}{headers}[\mathtt{Include{\mhyphen}Referred{\mhyphen}Token{\mhyphen}Binding{\mhyphen}ID}] \equiv \True $}
      \Let{$\mi{refTBID}$}{$\mi{response}.\str{host}$}
	  	\EndIf
		\Else
      \Let{$\mi{refTBID}$}{$\bot$}
		\EndIf
      \CallFun{HTTP\_SEND}{$\mi{reference}$, $\mi{req}$, $\mi{url}$, $\mi{origin}$, $\mi{referrer}$, $\mi{referrerPolicy}$, $s'$, $\mi{refTBID}$ }\label{line:send-redirect}
    \EndIf
  \EndIf
  \If{$\exists\, \ptr{w} \in \mathsf{Subwindows}(s')$ \textbf{such that} $\comp{\compn{s'}{\ptr{w}}}{nonce} \equiv \mi{reference}$} \Comment{normal response}
    \If{$\mi{response}.\str{body} \not\sim \an{*,*}$}
      \State \textbf{stop} $\{\}$, $s'$
    \EndIf
    \Let{$\mi{script}$}{$\proj{1}{\comp{\mi{response}}{body}}$}
    \Let{$\mi{scriptstate}$}{$\proj{2}{\comp{\mi{response}}{body}}$}
    \Let{$\mi{referrer}$}{$\mi{request}.\str{headers}[\str{Referer}]$}
    \Let{$d$}{$\an{\nu_7, \mi{requestUrl}, \mi{response}.\str{headers}, \mi{referrer}, \mi{script}, \mi{scriptstate}, \an{}, \an{}, \True}$} \label{line:take-script} \label{line:set-origin-of-document}
    \If{$\comp{\compn{s'}{\ptr{w}}}{documents} \equiv \an{}$}
      \Let{$\comp{\compn{s'}{\ptr{w}}}{documents}$}{$\an{d}$}
    \Else
      \LetND{$\ptr{i}$}{$\mathbb{N}$ \textbf{such that} $\comp{\compn{\comp{\compn{s'}{\ptr{w}}}{documents}}{\ptr{i}}}{active} \equiv \True$} %
      \Let{$\comp{\compn{\comp{\compn{s'}{\ptr{w}}}{documents}}{\ptr{i}}}{active}$}{$\bot$}
      \State \textbf{remove} $\compn{\comp{\compn{s'}{\ptr{w}}}{documents}}{(\ptr{i}+1)}$ and all following documents \breakalgohook{3} from $\comp{\compn{s'}{\ptr{w}}}{documents}$
      \Append{$d$}{$\comp{\compn{s'}{\ptr{w}}}{documents}$}
    \EndIf
    \State \textbf{stop} $\{\}$, $s'$
  \ElsIf{$\exists\, \ptr{w} \in \mathsf{Subwindows}(s')$, $\ptr{d}$ \textbf{such that} $\comp{\compn{s'}{\ptr{d}}}{nonce} \equiv \proj{1}{\mi{reference}} $ \breakalgohook{1}  $\wedge$  $\compn{s'}{\ptr{d}} = \comp{\compn{s'}{\ptr{w}}}{activedocument}$} \label{line:process-xhr-response} \Comment{process XHR response}
    \Let{$\mi{headers}$}{$\mi{response}.\str{headers} - \str{Set\mhyphen{}Cookie}$}
    \Append{\breakalgo{3}$\an{\tXMLHTTPRequest, \mi{headers}, \comp{\mi{response}}{body}, \proj{2}{\mi{reference}}}$}{$\comp{\compn{s'}{\ptr{d}}}{scriptinputs}$}
  \EndIf
\EndFunction
\end{algorithmic} %
\end{algorithm}

\FloatBarrier
\begin{algorithm}[H] %
\caption{\label{alg:browsermain} Web Browser Model: Main Algorithm}
\begin{algorithmic}[1]
\Statex[-1] \textbf{Input:} $\an{a,f,m},s$
  \Let{$s'$}{$s$}

  \If{$\comp{s}{isCorrupted} \not\equiv \bot$}
    \Let{$\comp{s'}{pendingRequests}$}{$\an{m, \comp{s}{pendingRequests}}$} \Comment{Collect incoming messages}
    \LetND{$m'$}{$d_{V}(s')$} %
    \LetND{$a'$}{$\addresses$} %
    \State \textbf{stop} $\an{\an{a',a,m'}}$, $s'$
  \EndIf
  \If{$m \equiv \trigger$} \Comment{A special trigger message. }
    \LetND{$\mi{switch}$}{$\{\str{script},\str{urlbar},\str{reload},\str{forward}, \str{back}\}$} \label{line:browser-switch}  %
    \LetNDST{$\ptr{w}$}{$\mathsf{Subwindows}(s')$}{$\comp{\compn{s'}{\ptr{w}}}{documents} \neq \an{}$\breakalgohook{2}}{\textbf{stop}}%
    \Comment{Pointer to some window.}
    \LetNDST{$\ptr{tlw}$}{$\mathbb{N}$}{$\comp{\compn{s'}{\ptr{tlw}}}{documents} \neq \an{}$\breakalgohook{2}}{\textbf{stop}}%
    \Comment{Pointer to some top-level window.}
    \If{$\mi{switch} \equiv \str{script}$} \Comment{Run some script.}
      \Let{$\ptr{d}$}{$\ptr{w} \plusPairing \str{activedocument}$} \label{line:browser-trigger-document}  
      \CallFun{RUNSCRIPT}{$\ptr{w}$, $\ptr{d}$, $s'$}
    \ElsIf{$\mi{switch} \equiv \str{urlbar}$} \Comment{Create some new request.}
      \LetND{$\mi{newwindow}$}{$\{\True, \bot \}$}
      \If{$\mi{newwindow} \equiv \True$} \Comment{Create a new window.}
        \Let{$\mi{windownonce}$}{$\nu_1$}
        \Let{$w'$}{$\an{\mi{windownonce}, \an{}, \bot}$}
        \Append{$w'$}{$\comp{s'}{windows}$}
      \Else \Comment{Use existing top-level window.}
        \Let{$\mi{windownonce}$}{$s'.\ptr{tlw}.nonce$}
      \EndIf
      \LetND{$\mi{protocol}$}{$\{\http, \https\}$} \label{line:browser-choose-url} %
\algstore{myalg-browsermain-3}
\end{algorithmic}
\end{algorithm}

\begin{algorithm}                     
\begin{algorithmic} [1]                   %
\algrestore{myalg-browsermain-3}
      \LetND{$\mi{host}$}{$\dns$} %
      \LetND{$\mi{path}$}{$\mathbb{S}$} %
      \LetND{$\mi{fragment}$}{$\mathbb{S}$} %
      \LetND{$\mi{parameters}$}{$\dict{\mathbb{S}}{\mathbb{S}}$} %
      \Let{$\mi{url}$}{$\an{\cUrl, \mi{protocol}, \mi{host}, \mi{path}, \mi{parameters}, \mi{fragment}}$}
      \Let{$\mi{req}$}{$\hreq{
          nonce=\nu_2,
          method=\mGet,
          host=\mi{host},
          path=\mi{path},
          headers=\an{},
          parameters=\mi{parameters},
          body=\an{}
        }$}
      \CallFun{HTTP\_SEND}{$\mi{windownonce}$, $\mi{req}$, $\mi{url}$, $\bot$, $\bot$, $\bot$, $s'$}\label{line:send-random}

    \ElsIf{$\mi{switch} \equiv \str{reload}$} \Comment{Reload some document.}
      \LetNDST{$\ptr{w}$}{$\mathsf{Subwindows}(s')$}{$\comp{\compn{s'}{\ptr{w}}}{documents} \neq \an{}$\breakalgohook{2}}{\textbf{stop}} \label{line:browser-reload-window}%
      \Let{$\mi{url}$}{$s'.\ptr{w}.\str{activedocument}.\str{location}$}
      \Let{$\mi{req}$}{$\hreq{ nonce=\nu_2, 
          method=\mGet, host=\comp{\mi{url}}{host},
          path=\comp{\mi{url}}{path},
          headers=\an{},
          parameters=\comp{\mi{url}}{parameters}, body=\an{}
        }$}
      \Let{$\mi{referrer}$}{$s'.\ptr{w}.\str{activedocument}.\str{referrer}$}
      \Let{$s'$}{$\mathsf{CANCELNAV}(\comp{\compn{s'}{\ptr{w}}}{nonce}, s')$}
      \CallFun{HTTP\_SEND}{$\comp{\compn{s'}{\ptr{w}}}{nonce}$, $\mi{req}$, $\mi{url}$, $\bot$, $\mi{referrer}$, $\bot$, $s'$}
    \ElsIf{$\mi{switch} \equiv \str{forward}$}
      \State $\mathsf{NAVFORWARD}$($\ptr{w}$, $s'$)
    \ElsIf{$\mi{switch} \equiv \str{back}$}
      \State $\mathsf{NAVBACK}$($\ptr{w}$, $s'$)
    \EndIf
  \ElsIf{$m \equiv \fullcorrupt$} \Comment{Request to corrupt browser}
    \Let{$\comp{s'}{isCorrupted}$}{$\fullcorrupt$}
    \State \textbf{stop} $\an{}$, $s'$
  \ElsIf{$m \equiv \closecorrupt$} \Comment{Close the browser}
    \Let{$\comp{s'}{secrets}$}{$\an{}$}  
    \Let{$\comp{s'}{windows}$}{$\an{}$}
    \Let{$\comp{s'}{pendingDNS}$}{$\an{}$}
    \Let{$\comp{s'}{pendingRequests}$}{$\an{}$}
    \Let{$\comp{s'}{sessionStorage}$}{$\an{}$}
    \State \textbf{let} $\comp{s'}{cookies} \subsetPairing \cookies$ \textbf{such that} \breakalgohook{1} $(c \inPairing \comp{s'}{cookies}) {\iff} (c \inPairing \comp{s}{cookies} \wedge \comp{\comp{c}{content}}{session} \equiv \bot$)
    \Let{$\comp{s'}{isCorrupted}$}{$\closecorrupt$}
    \State \textbf{stop} $\an{}$, $s'$

  \ElsIf{$\exists\, \an{\mi{reference}, \mi{request}, \mi{url}, \mi{key}, f}$
      $\inPairing \comp{s'}{pendingRequests}$ \breakalgohook{0}
      \textbf{such that} $\proj{1}{\decs{m}{\mi{key}}} \equiv \cHttpResp$ } %
    \Comment{Encrypted HTTP response}
    \Let{$m'$}{$\decs{m}{\mi{key}}$}
    \If{$\comp{m'}{nonce} \not\equiv \comp{\mi{request}}{nonce}$}
      \State \textbf{stop}
    \EndIf
    \State \textbf{remove} $\an{\mi{reference}, \mi{request}, \mi{url}, \mi{key}, f}$ \textbf{from} $\comp{s'}{pendingRequests}$
    \CallFun{PROCESSRESPONSE}{$m'$, $\mi{reference}$, $\mi{request}$, $\mi{url}$, $s'$}
  \ElsIf{$\proj{1}{m} \equiv \cHttpResp$ $\wedge$ $\exists\, \an{\mi{reference}, \mi{request}, \mi{url}, \bot, f}$ $\inPairing \comp{s'}{pendingRequests}$ \breakalgohook{0}\textbf{such that} $\comp{m'}{nonce} \equiv \comp{\mi{request}}{key}$ } %
    \State \textbf{remove} $\an{\mi{reference}, \mi{request}, \mi{url}, \bot, f}$ \textbf{from} $\comp{s'}{pendingRequests}$
    \CallFun{PROCESSRESPONSE}{$m$, $\mi{reference}$, $\mi{request}$, $\mi{url}$, $s'$}
  \ElsIf{$m \in \dnsresponses$} \Comment{Successful DNS response}
      \If{$\comp{m}{nonce} \not\in \comp{s}{pendingDNS} \vee \comp{m}{result} \not\in \addresses \vee \comp{m}{domain} \not\equiv \comp{\proj{2}{\comp{s}{pendingDNS}}}{host}$}
        \State \textbf{stop} \label{line:browser-dns-response-stop}
      \EndIf
      \Let{$\an{\mi{reference}, \mi{message}, \mi{url}, \mi{refTB}}$}{$\comp{s}{pendingDNS}[\comp{m}{nonce}]$}
      \If{$s'.\str{useTB}[\mi{url}.\str{host}] \equiv \True$}
							\Let{$\mi{TB\_req}$}{$\hreq{ nonce=\nu_{n1},
								method=\mGet, %
		            host=\comp{\mi{url}}{host},
		            path=\str{/OAUTB\mhyphen{}prepare}, %
								parameters=\mi{url}.\str{parameters}, %
		            headers=\an{},
		            xbody=\an{}
                } $} %
	 			\Let{$s'.\str{tokenBindingRequests}$}{$s'.\str{tokenBindingRequests} 
            \cup $\breakalgohook{7}$
						[\nu_{n1} {:} \an{\mi{reference}, \mi{message}, \mi{url}, \mi{refTB}}]$}
        \AppendBreak{2}{$\langle\mi{reference}$, $\mi{TB\_req}$, $\mi{url}$, $\nu_3$, $\comp{m}{result}\rangle$}{$\comp{s'}{pendingRequests}$} 
				\Comment{TB only with TLS}
        \Let{$\mi{message}$}{$\enc{\an{\mi{TB\_req},\nu_3}}{\comp{s'}{keyMapping}\left[\comp{\mi{message}}{host}\right]}$}

      \ElsIf{$\mi{url}.\str{protocol} \equiv \https$}
        \AppendBreak{2}{$\langle\mi{reference}$, $\mi{message}$, $\mi{url}$, $\nu_3$, $\comp{m}{result}\rangle$}{$\comp{s'}{pendingRequests}$} \label{line:add-to-pendingrequests-https}
        \Let{$\mi{message}$}{$\enc{\an{\mi{message},\nu_3}}{\comp{s'}{keyMapping}\left[\comp{\mi{message}}{host}\right]}$} \label{line:select-enc-key}
      \Else
        \AppendBreak{2}{$\langle\mi{reference}$, $\mi{message}$, $\mi{url}$, $\bot$, $\comp{m}{result}\rangle$}{$\comp{s'}{pendingRequests}$} \label{line:add-to-pendingrequests}
      \EndIf
      \Let{$\comp{s'}{pendingDNS}$}{$\comp{s'}{pendingDNS} - \comp{m}{nonce}$}
      \State \textbf{stop} $\an{\an{\comp{m}{result}, a, \mi{message}}}$, $s'$
  \Else \Comment{Some other message}
    \CallFun{PROCESS\_OTHER}{$m$, $a$, $f$, $s'$}
  \EndIf
  \State \textbf{stop}

\end{algorithmic} %
\end{algorithm}
\clearpage
\FloatBarrier %

\section{Generic HTTPS Server Model} 
\label{sec:generic-https-server-model}

The generic HTTPS server is used in the concrete instantiations
of clients, authorization servers and resource servers. The placeholder
algorithms defined in this section are replaced by the 
algorithms in the corresponding processes. Here, we use the same
model as defined in Appendix E of \cite{FettKuestersSchmitz-TR-OIDC-2017}.

In the following, we give the complete definition of the generic HTTPS
server, as it is used in the 
instantiations of the FAPI participants
and also referenced within the proof.

\begin{definition}[Base state for an HTTPS server.]\label{def:generic-https-state}
  The state of each HTTPS server that is an instantiation of this
  relation must contain at least the following subterms:
  $\mi{pendingDNS} \in \dict{\nonces}{\terms}$,
  $\mi{pendingRequests} \in \dict{\nonces}{\terms}$ (both containing
  arbitrary terms), $\mi{DNSaddress} \in \addresses$ (containing the
  IP address of a DNS server),
  $\mi{keyMapping} \in \dict{\dns}{\terms}$ (containing a mapping from
  domains to public keys), $\mi{tlskeys} \in \dict{\dns}{\nonces}$
  (containing a mapping from domains to private keys) and
  $\mi{corrupt}\in\terms$ (either $\bot$ if the server is not
  corrupted, or an arbitrary term otherwise). 
\end{definition}

Table~\ref{tab:gen-https-serv-placeholder-list} shows a list of placeholders
for nonces 
used in these algorithms.

\begin{table}[tbh]
  \centering
  \begin{tabular}{|@{\hspace{1ex}}l@{\hspace{1ex}}|@{\hspace{1ex}}l@{\hspace{1ex}}|}\hline 
    \hfill Placeholder\hfill  &\hfill  Usage\hfill  \\\hline\hline
    $\nu_0$ & new nonce for DNS requests  \\\hline
    $\nu_1$ & new symmetric key  \\\hline
  \end{tabular}
  
  \caption{List of placeholders used in the generic HTTPS server algorithm.}
  \label{tab:gen-https-serv-placeholder-list}
\end{table}

We now define the default functions of
the generic web server in
Algorithms~\ref{alg:simple-send}--\ref{alg:default-other-handler}, and
the main relation in Algorithm~\ref{alg:generic-server-main}.

\begin{algorithm}[thp]
\caption{\label{alg:simple-send} Generic HTTPS Server Model: Sending a DNS message (in preparation for sending an HTTPS message). }
\begin{algorithmic}[1]
  \Function{$\mathsf{HTTPS\_SIMPLE\_SEND}$}{$\mi{reference}$, $\mi{message}$, $s'$}
    \Let{$\comp{s'}{pendingDNS}[\nu_0]$}{$\an{\mi{reference}, 
        \mi{message}}$} \label{line:generic-add-pendingDNS}
    \State \textbf{stop} $\an{\an{\comp{s'}{DNSaddress},a,
    \an{\cDNSresolve, \mi{message}.\str{host}, \nu_0}}}$, $s'$
  \EndFunction
\end{algorithmic} %
\end{algorithm}

\begin{algorithm}[thp]
\caption{\label{alg:default-https-response-handler} Generic HTTPS Server Model: Default HTTPS response handler.}
\begin{algorithmic}[1]
  \Function{$\mathsf{PROCESS\_HTTPS\_RESPONSE}$}{$m$, $\mi{reference}$, $\mi{request}$, $\mi{key}$, $a$, $f$, $s'$}
    \Stop{\DefStop}
  \EndFunction
\end{algorithmic} %
\end{algorithm}

\begin{algorithm}[thp]
\caption{\label{alg:default-trigger} Generic HTTPS Server Model: Default trigger event handler.}
\begin{algorithmic}[1]
  \Function{$\mathsf{TRIGGER}$}{$s'$}
    \Stop{\DefStop}
  \EndFunction
\end{algorithmic} %
\end{algorithm}

\begin{algorithm}[thp]
\caption{\label{alg:default-https-request-handler} Generic HTTPS Server Model: Default HTTPS request handler.}
\begin{algorithmic}[1]
  \Function{$\mathsf{PROCESS\_HTTPS\_REQUEST}$}{$m$, $k$, $a$, $f$, $s'$}
    \Stop{\DefStop}
  \EndFunction
\end{algorithmic} %
\end{algorithm}

\begin{algorithm}[h!]
\caption{\label{alg:default-other-handler} Generic HTTPS Server Model: Default handler for other messages.}
\begin{algorithmic}[1]
  \Function{$\mathsf{PROCESS\_OTHER}$}{$m$, $a$, $f$, $s'$}
    \Stop{\DefStop}
  \EndFunction
\end{algorithmic} %
\end{algorithm}

\FloatBarrier

\begin{algorithm}[t!]
\caption{\label{alg:generic-server-main} Generic HTTPS Server Model: Main relation of a generic HTTPS server}
\begin{algorithmic}[1]
\Statex[-1] \textbf{Input:} $\an{a,f,m},s$
\Let{$s'$}{$s$} 
  \If{$s'.\str{corrupt} \not\equiv \bot \vee m \equiv \corrupt$}
    \Let{$s'.\str{corrupt}$}{$\an{\an{a, f, m}, s'.\str{corrupt}}$}
    \LetND{$m'$}{$d_{V}(s')$}
    \LetND{$a'$}{$\addresses$}
    \Stop{$\an{\an{a',a,m'}}$, $s'$}
  \EndIf

   \If{$\exists\, m_{\text{dec}}$, $k$, $k'$, $\mi{inDomain}$ \textbf{such that} $\an{m_{\text{dec}}, k} \equiv \dec{m}{k'}$ $\wedge$ $\an{inDomain,k'} \in s.\str{tlskeys}$} %
     \LetST{$n$, $\mi{method}$, $\mi{path}$, $\mi{parameters}$, $\mi{headers}$, $\mi{body}$}{\breakalgohook{0}$\an{\cHttpReq, n, \mi{method}, \mi{inDomain}, \mi{path}, \mi{parameters}, \mi{headers}, \mi{body}} \equiv m_{\text{dec}}$\breakalgohook{0}}{\textbf{stop} \DefStop} \label{line:gen-server-asym-check}

     \CallFun{PROCESS\_HTTPS\_REQUEST}{$m_\text{dec}$, $k$, $a$, $f$, $s'$} \label{line:gen-server-first-req}

  \ElsIf{$m \in \dnsresponses$} \Comment{Successful DNS response}
      \If{$\comp{m}{nonce} \not\in \comp{s}{pendingDNS} \vee \comp{m}{result} \not\in \addresses \vee \comp{m}{domain} \not\equiv s.\str{pendingDNS}[m.\str{nonce}].2.\str{host}$}
      \Stop{\DefStop}
      \EndIf
      \Let{$\an{\mi{reference}, \mi{request}}$}{$\comp{s}{pendingDNS}[\comp{m}{nonce}]$} \label{line:generic-set-request-from-pendingDNS}
      \AppendBreak{2}{$\langle\mi{reference}$, $\mi{request}$, $\nu_1$, $\comp{m}{result}\rangle$}{$\comp{s'}{pendingRequests}$} \label{line:generic-move-reference-to-pending-request} 
      \Let{$\mi{message}$}{$\enc{\an{\mi{request},\nu_1}}{\comp{s'}{keyMapping}\left[\comp{\mi{request}}{host}\right]}$} \label{line:generic-select-enc-key}
      \Let{$\comp{s'}{pendingDNS}$}{$\comp{s'}{pendingDNS} - \comp{m}{nonce}$}\label{line:generic-remove-pendingdns}
      \Stop{$\an{\an{\comp{m}{result}, a, \mi{message}}}$, $s'$} \label{line:generic-send-https-request}

  \ElsIf{$\exists\, \an{\mi{reference}, \mi{request}, \mi{key}, f}$
      $\inPairing \comp{s'}{pendingRequests}$ \breakalgohook{0}
      \textbf{such that} $\proj{1}{\decs{m}{\mi{key}}} \equiv \cHttpResp$ } \label{line:generic-https-response}
    \Comment{Encrypted HTTP response}

    \Let{$m'$}{$\decs{m}{\mi{key}}$}
    \If{$\comp{m'}{nonce} \not\equiv \comp{\mi{request}}{nonce}$}
      \Stop{\DefStop}
    \EndIf
    \State \textbf{remove} $\an{\mi{reference}, \mi{request}, \mi{key}, f}$ \textbf{from} $\comp{s'}{pendingRequests}$\label{line:generic-remove-pending-request}
    \CallFun{PROCESS\_HTTPS\_RESPONSE}{$m'$, $\mi{reference}$, $\mi{request}$, $\mi{key}$, $a$, $f$, $s'$} \label{https-server-calls-PHRESP}
    \Stop{\DefStop}

  \ElsIf{$m \equiv \str{TRIGGER}$} \Comment{Process was triggered}
    \CallFun{PROCESS\_TRIGGER}{$s'$}

  \EndIf
  \Stop{\DefStop}
  
\end{algorithmic} %
\end{algorithm}

\FloatBarrier %

\clearpage
\section{FAPI Model} \label{label:appendix-full-fapi-model}

In the following, we will give the formal application-specific model 
of the participants of the FAPI, i.e., the clients, authorization servers and resource servers.
Then, building upon these models, we will give the definition of a FAPI
web system with a network attacker.

\subsection{Clients} \label{app:clients-fapi}

Similar to Section H of Appendix F of \cite{FettKuestersSchmitz-TR-OIDC-2017}, a client $c \in \fAP{C}$ 
is a web server modeled as an atomic
DY process $(I^c, Z^c, R^c, s^c_0)$ with the addresses
$I^c := \mapAddresstoAP(c)$.

\begin{definition}\label{def:clients}
  A \emph{state $s\in Z^c$ of client $c$} is a term of the form
  $\langle\mi{DNSaddress}$, $\mi{pendingDNS}$, $\mi{pendingRequests}$,
  $\mi{corrupt}$, $\mi{keyMapping}$,
  $\mi{tlskeys}$, %
  $\mi{sessions}$, $\mi{issuerCache}$, $\mi{oidcConfigCache}$,
  $\mi{jwksCache}$, 
  $\mi{clientCredentialsCache}$, 
  $\mi{oautbEKM}$,
  $\mi{authReqSigKey}$,
  $\mi{tokenBindings} \rangle$ with 
  $\mi{DNSaddress} \in \addresses$,
  $\mi{pendingDNS} \in \dict{\nonces}{\terms}$,
  $\mi{pendingRequests} \in \dict{\nonces}{\terms}$,
  $\mi{corrupt}\in\terms$, $\mi{keyMapping} \in \dict{\dns}{\terms}$,
  $\mi{tlskeys} \in \dict{\dns}{K_\text{TLS}}$ (all former components
  as in Definition~\ref{def:generic-https-state}),
  $\mi{sessions} \in \dict{\nonces}{\terms}$,
  $\mi{issuerCache} \in \dict{\terms}{\terms}$,
  $\mi{oidcConfigCache} \in \dict{\terms}{\terms}$,
  $\mi{jwksCache} \in \dict{\terms}{\terms}$, 
  $\mi{clientCredentialsCache} \in \dict{\terms}{\terms}$, 
  $\mi{oautbEKM} \in \terms$, %
  $\mi{authReqSigKey} \in \nonces$ and %
  $\mi{tokenBindings} \in \dict{\dns}{\nonces} $.

  An \emph{initial state $s^c_0$ of $c$} is a state of $c$ with
  $s^c_0.\str{pendingDNS} \equiv \an{}$, $s^c_0.\str{pendingRequests} \equiv 
  \an{}$, $s^c_0.\str{corrupt} \equiv \bot$,  $s^c_0.\str{keyMapping}$ being the
  same as the keymapping for browsers, $s^c_0.\str{tlskeys} \equiv \mi{tlskeys}^c$,
  $s^c_0.\str{sessions} \equiv \an{}$ and
  $s^c_0.\str{oautbEKM} \equiv \an{}$.  
\end{definition}

We require the initial state to contain preconfigured values for 
$s^c_0.\str{issuerCache}$ (mapping from identities to domains of
the corresponding authorization server), $s^c_0.\str{oidcConfigCache}$ 
(authorization endpoint and token endpoint for each domain of an authorization server) and
$s^c_0.\str{jwksCache}$ (mapping from issuers to the public key used to validate their signatures).
More precisely, we require that for all authorization servers
$\mi{as}$ and all domains $\mi{dom_{\mi{as}}} \in \mathsf{dom}(\mi{as})$ it holds true that
\begin{align}
& s^c_0.\str{oidcConfigCache}[\mi{dom_{\mi{as}}}][\str{token\_ep}] \equiv 
  \an{\tUrl, \https, \mi{dom_{\mi{as}}'}, \str{/token}, \an{}, \an{}} \\
& s^c_0.\str{oidcConfigCache}[\mi{dom_{\mi{as}}}][\str{auth\_ep}] \equiv 
  \an{\tUrl, \https, \mi{dom_{\mi{as}}'}, \str{/auth}, \an{}, \an{}} \\
& s^c_0.\str{jwksCache}[\mi{dom_{\mi{as}}}] \equiv 
  \mathsf{pub}(s^\mi{as}_0.\str{jwk})\\
& s^c_0.\str{issuerCache}[\mi{id}] \in \mathsf{dom}(\mi{as}) \Leftrightarrow \mi{id} \in \IDs^\mi{as},
\end{align}

for $\mi{dom_{\mi{as}}'} \in \mathsf{dom}(\mi{as})$.
These properties hold true if OIDC Discovery \cite{openid-connect-discovery-1.0} and Dynamic 
Client Registration \cite{openid-connect-dynamic-client-registration-1.0} are used,
as already shown in Lemmas 1 to 4 of \cite{FettKuestersSchmitz-TR-OIDC-2017}.  

Furthermore, $s^c_0.\str{oidcConfigCache}[\mi{issuer}][\str{resource\_servers}] =
s^\mi{as}_0.\str{resource\_servers}$ shall contain the domains of all resource
servers that the authorization server $\mi{issuer}$ supports. 

For signing authorization requests, we require that $s^c_0.\str{authReqSigKey}$
contains a key that is only known to $c$. 

For each identity $\mi{id} \in \IDs^\mi{as}$ governed by authorization server $\mi{as}$,
the credentials of the client are stored in 
$s^c_0.\str{clientCredentialsCache}[s^c_0.\str{issuerCache}[\mi{id}]]$,
which shall be equal to $s^\mi{as}_0.\str{clients}[\mi{clientId}]$, where
$\mi{clientId}$ is the client id that was issued to $c$ by $\mi{as}$. 

For OAuth 2.0 Token Binding, we require that $s^c_0.\str{tokenBindings}[d]$ contains a
different nonce for each $d \in \dns$. 

The relation $R^c$ is based on the model of generic HTTPS servers (see
Section~\ref{sec:generic-https-server-model}). The algorithms that differ from
or do not exist in the generic server model are defined in
Algorithms~\ref{alg:client-fapi-http-request}--\ref{alg:client-check-id-token}.
Table~\ref{tab:client-placeholder-list} shows a list of all placeholders
used in these algorithms.

The scripts that are used by the clients are described in
Algorithms~\ref{alg:script-client-index}
and~\ref{alg:script-c-get-fragment}. 
As in \cite{FettKuestersSchmitz-TR-OIDC-2017}, the current URL of a document
is extracted by using the function
$\mathsf{GETURL}(\mi{tree},\mi{docnonce})$,
which searches for the document with 
the identifier $\mi{docnonce}$
in the (cleaned) tree $\mi{tree}$
of the browser's windows and documents. It then returns the URL $u$
of that document. If no document with nonce $\mi{docnonce}$
is found in the tree $\mi{tree}$, $\notdef$ is returned.

\begin{table}[tbh]
  \centering
  \begin{tabular}{|@{\hspace{1ex}}l@{\hspace{1ex}}|@{\hspace{1ex}}l@{\hspace{1ex}}|}\hline 
    \hfill Placeholder\hfill  &\hfill  Usage\hfill  \\\hline\hline
    $\nu_1$ & new login session id  \\\hline
    $\nu_2$ & new HTTP request nonce  \\\hline
    $\nu_3$ & new HTTP request nonce  \\\hline
    $\nu_4$ & new service session id  \\\hline
    $\nu_5$ & new HTTP request nonce  \\\hline
    $\nu_6$ & new state value  \\\hline
    $\nu_7$ & new \emph{nonce} value  \\\hline
    $\nu_8$ & new nonce for OAUTB  \\\hline
    $\nu_9$ & new PKCE verifier  \\\hline
    $\nu_{10}$ & new HTTP request nonce  \\\hline
    $\nu_{11}$ & new HTTP request nonce\\\hline
  \end{tabular}
  
  \caption{List of placeholders used in the client algorithm.}
  \label{tab:client-placeholder-list}
\end{table}

Furthermore, we require that $\mi{leak} \in \addresses$ is an arbitrary 
IP address. By sending messages to this address, we model
the leakage of tokens or messages, as this address 
can be the address of an attacker.

\subsection*{Description} 
In the following, we will describe the basic functionality of the
algorithms, focusing on the main differences to the algorithms
used in \cite{FettKuestersSchmitz-TR-OIDC-2017}.

\begin{itemize}
\item Algorithm~\ref{alg:client-fapi-http-request} handles incoming requests.
The browser can request a nonce for OAUTB in the $\str{OAUTB\mhyphen{}prepare}$ path.
The redirection endpoint requires and checks the corresponding OAUTB header if the client is a
web server client using OAUTB. 
In case of a read-write client, the authorization response
contains an id token, which the client checks in Algorithm~\ref{alg:client-check-first-id-token}
before continuing the flow.
\item Algorithm~\ref{alg:client-fapi-http-response} handles responses. If the client is a read-write
client, it checks the second id token contained in the token request before continuing the flow.
Furthermore, when using mTLS or OAUTB, the client processes the initial response needed for
the actual request that contains either the decrypted mTLS nonce or a signed EKM value.
When receiving a resource access nonce, the web server client sends the nonce to the process
that sent the request to the redirection endpoint.
\item Algorithm~\ref{alg:client-fapi-start-login-flow} starts a new login flow. 
The client prepares the authorization request, which is redirected by the browser.
The request differs depending on the client type and whether the client is 
an app client or a web server client. In all cases, it creates a signed request JWS
for preventing the attack described in Section~\ref{attack:pkce}.
Furthermore, the request leaks at this point, modeling leaking browser
logs.
The resource server used at the end of the flow is chosen from
the set of resource servers provided by the authorization server.
\item Algorithms~\ref{alg:client-prepare-token-request}
and \ref{alg:client-token-request} handle the token request. 
If the client uses mTLS or OAUTB, it sends an initial request 
with Algorithm~\ref{alg:client-prepare-token-request}.
For modeling a misconfigured token endpoint, the read-write client may choose between
the preconfigured token endpoint or an arbitrary URL. 
\item The request to the resource server is made similarly in 
Algorithms~\ref{alg:client-prepare-use-access-token} and \ref{alg:client-use-access-token}.
\item Algorithm~\ref{alg:client-check-first-id-token} checks the first id token
and continues the flow, whereas Algorithm~\ref{alg:client-check-id-token}
uses the id token for logging in the identity.
\end{itemize}

\vfill

\begin{algorithm}[H] %
\caption{\label{alg:client-fapi-http-request} Relation of a Client
  $R^c$ -- Processing HTTPS Requests}
\begin{algorithmic}[1]
\Function{$\mathsf{PROCESS\_HTTPS\_REQUEST}$}{$m$, $k$, $a$, $f$, $s'$}
\Comment{\textbf{Process an incoming HTTPS request.} Other message types are handled in 
separate functions. $m$ is the incoming message, $k$ is the encryption key for the 
response, $a$ is the receiver, $f$ the sender of the message. $s'$ is the current state 
of the atomic DY process $c$.}
     \If{$m.\str{path} \equiv \str{/}$} \label{line:client-serve-index} \Comment{Serve index page.}
      \Let{$\mi{headers}$}{$[\str{ReferrerPolicy}{:}\str{origin}]$}
       \ParboxComment{Set the Referrer Policy for the index page of the client.} %
       \Let{$m'$}{$\encs{\an{\cHttpResp, m.\str{nonce}, 200, \mi{headers}, \an{\str{script\_client\_index}, \an{}}}}{k}$}

       \Comment{Send $\mi{script\_client\_index}$ in HTTP response.}
       \Stop{$\an{\an{f,a,m'}}$, $s'$} \label{line:client-send-index}
     \State
     \ElsIf{$m.\str{path} \equiv \str{/startLogin} \wedge m.\str{method} \equiv \mPost$} \label{line:client-start-login-endpoint} \Comment{\textbf{Serve the request to start a new login.}}
       \If{$m.\str{headers}{[\str{Origin}]}  \not\equiv \an{m.\str{host}, \https}$} \label{c-startLogin-check-origin-header}
         \Stop{\DefStop} \Comment{Check the Origin header for CSRF protection.}
       \EndIf
       \Let{$\mi{id}$}{$m.\str{body}$}
       \Let{$\mi{sessionId}$}{$\nu_1$} \label{line:client-choose-lsid} \ParboxComment{Session id is a freshly chosen nonce.}
       \Let{$s'.\str{sessions}[\mi{sessionId}]$}{$[\str{startRequest}{:}[\str{message}{:}m,\str{key}{:}k,\str{receiver}{:}a,\str{sender}{:}f],
					$\breakalgohook{9}$ 
          \str{identity}:\mi{id}]$} \label{c-create-initial-session}
       \ParboxComment{Create new session record.}
       \CallFun{START\_LOGIN\_FLOW}{$\mi{sessionId}$, $s'$}
\ParboxComment{Call the function that starts a login flow.} 

  \State
  \ElsIf{$m.\str{path} \equiv \str{/OAUTB\mhyphen{}prepare}$}\label{line:client-oautb-ep} 
    \Comment{For OAUTB between the user-agent and the client}
    \Let{$\mi{headers}$}{$[\str{ReferrerPolicy}{:}\str{origin}]$}
		\Let{$\mi{tbNonce}$}{$\nu_8$}
		\Let{$s'.\str{oautbEKM}$} 
				{$s'.\str{oautbEKM} \plusPairing \mathsf{hash}(\an{m.\str{nonce}, \mi{tbNonce}, 
							\mi{keyMapping}[m.\str{host}] }) $} 
		\Let{$m'$}{$\encs{\an{\cHttpResp, m.\str{nonce}, 200, \mi{headers}, [\str{tbNonce}{:} \mi{tbNonce}]	}}{k}$}      
		\Stop{\StopWithMPrime} \label{client-send-oautb-ep}

 \algstore{https-req}
 \end{algorithmic}
 \end{algorithm}
 
 \begin{algorithm}
 \begin{algorithmic}[1]
 \algrestore{https-req}

     \ElsIf{$m.\str{path} \equiv \str{/redirect\_ep}$} \label{line:client-redir-endpoint} 
     \ParboxComment{\textbf{User is being redirected after authentication to the AS.}}
       \Let{$\mi{sessionId}$}{$m.\str{headers}[\str{Cookie}][\an{\str{\_\_Secure},\str{sessionId}}][\str{value}]$} \label{c-redir-ep-get-sid-from-cookie} \label{c-redir-ep-check-secure-prefix}
       \If{$\mi{sessionId} \not\in s'.\str{sessions}$}
         \Stop
       \EndIf
       \Let{$\mi{session}$}{$s'.\str{sessions}[\mi{sessionId}]$}
       \ParboxComment{Retrieve session data.}

			 \If{$\an{\https, \mi{m}.\str{host},  \mi{m}.\str{path}, \mi{m}.\str{parameters}, \bot} 
			 				\not\equiv \mi{session}[\str{redirectUri}]$}
			 		\Stop  \ParboxComment{Check if response was received at the right redirect uri.}
			 \EndIf
       \Let{$\mi{issuer}$}{$s'.\str{issuerCache}[\mi{session}[\str{identity}]]$} 
       \ParboxComment{Mappings from identites to issuers.} 
       \Let{$\mi{credentials}$}{$s'.\str{clientCredentialsCache}[\mi{issuer}]$} %
       \Let{$\mi{responseType}$}{$\mi{session}[\str{response\_type}]$}

       \ParboxComment{Determines the flow to use, e.g., \texttt{code id\_token} for an OIDC hybrid flow.}

       \If{$\mi{responseType} \in \{ \an{\str{code}}, \an{\str{JARM\_code}}\}$} \ParboxComment{Authorization code mode: Take data from URL parameters.} %
         \Let{$\mi{data}$}{$m.\str{parameters}$} \label{c-redirect-ep-data-from-param}
       \Else \ParboxComment{Hybrid mode: Send $\str{script\_client\_get\_fragment}$ to browser to retrieve data from URL fragment}
         \If{$m.\str{method} \equiv \mGet$}
           \Let{$\mi{headers}$}{$\an{\an{\str{ReferrerPolicy}, \str{origin}}}$}
           \Let{$m'$}{$\encs{\an{\cHttpResp, m.\str{nonce}, 200, \mi{headers}, \an{\str{script\_client\_get\_fragment}, \bot}}}{k}$} \label{line:c-send-script-get-fragment}
           \Stop{$\an{\an{f,a,m'}}$, $s'$}
         \Else
           \ParboxComment{ POST request: script $\str{script\_client\_get\_fragment}$ is sending the data from URL fragment.}
           \Let{$\mi{data}$}{$m.\str{body}$} \label{c-redirect-ep-data-from-body}
         \EndIf
       \EndIf 
    \If{$\mi{credentials}[\str{profile}] \equiv \str{rw} \wedge \mi{credentials}[\str{client\_type}] \equiv \str{conf\_OAUTB} \wedge \mi{credentials}[\str{is\_app}] \equiv \bot$}\\
    \hspace{1cm} \Comment{Check Provided TB-ID (for PKCE)} %
  		\If{$\str{Sec\mhyphen{}Token\mhyphen{}Binding} \not \in m.\str{headers}$}
    		\Stop{\DefStop}
			\EndIf
  		\LetND{$\mi{ekmInfo}$}{$s'.\str{oautbEKM}$} 
	     		\Let{$\mi{TB\_Msg\_provided}$}{$m.\str{headers}[\str{Sec\mhyphen{}Token\mhyphen{}Binding}][\str{prov}]$} 
	     		\Let{$\mi{TB\_provided\_pub}$}{$\mi{TB\_Msg\_provided}[\str{id}]$} 
	     		\Let{$\mi{TB\_provided\_sig}$}{$\mi{TB\_Msg\_provided}[\str{sig}]$} 

  			\If{$\mathsf{checksig}(\mi{TB\_provided\_sig}, \mi{TB\_provided\_pub}) \not \equiv \True 
          $\breakalgohook{5}$
  			  \vee \mathsf{extractmsg}(\mi{TB\_provided\_sig}) \not \equiv \mi{ekmInfo} $}
								\Stop{\DefStop}
				\EndIf

	      \Let{$s'.\str{session}[\str{browserTBID}]$}{$\mi{TB\_provided\_pub}$} \label{c-set-browserTBID}
  			\Let{$s'.\str{oautbEKM}$}{$s'.\str{oautbEKM} - \mi{ekmInfo}$}
		\EndIf

    \If{$\mi{data}[\str{state}] \not\equiv \mi{session}[\str{state}]$} \label{c-redirect-ep-check-state}
      \Stop \ParboxComment{Check $\mi{state}$ value.}
    \EndIf

    \If{$\mi{data}[\str{state}] \equiv \bot$} \label{c-check-if-state-equiv-bot} %
      \Stop \ParboxComment{$\mi{state}$ value is not valid.}
    \EndIf

    \Let{$s'.\str{sessions}[\mi{sessionId}][\str{state}]$}{$\bot$} \Comment{Invalidate state} \label{c-invalidates-state}

    \Let{$s'.\str{sessions}[\mi{sessionId}][\str{redirectEpRequest}]$}{\breakalgohook{1}$
        [\str{message}{:}m,\str{key}{:}k,\str{receiver}{:}a,\str{sender}{:}f,\str{data}{:}\mi{data}]$}
        \Comment{\parbox[t]{5cm}{Store incoming request for 
        sending a resource nonce or the service session id}}
        \label{line:client-set-redirect-ep-request-record} 

    \If{$\mi{credentials}[\str{profile}] \equiv \str{r}$} 
         \CallFun{PREPARE\_TOKEN\_REQUEST}{$\mi{sessionId}$, $\mi{data}[\str{code}]$, $s'$}  \label{c-redirection-ep-call-PTR}
    \ElsIf{$\mi{responseType} \equiv \an{\str{code}, \str{id\_token}}$} \Comment{Check id token}
         \CallFun{CHECK\_FIRST\_ID\_TOKEN}{$\mi{sessionId}$, $\mi{data}{[\str{id\_token}]}$, 
        $\mi{data}[\str{code}]$, $s'$} \label{line:c-call-check-id-token-immediately} 
    \Else \Comment{JARM: Check Response JWS}
         \CallFun{CHECK\_RESPONSE\_JWS}{$\mi{sessionId}$, $\mi{data}{[\str{responseJWS}]}$, 
           $\mi{data}[\str{code}]$, $s'$} \label{line:c-call-check-response-jws} 
    \EndIf
   \Stop
  \EndIf
 \EndFunction
 \end{algorithmic} %
 \end{algorithm}

\begin{algorithm} %
\caption{\label{alg:client-fapi-http-response} Relation of a Client $R^c$ -- Processing HTTPS Responses}
\begin{algorithmic}[1]
\Function{$\mathsf{PROCESS\_HTTPS\_RESPONSE}$}{$m$, $\mi{reference}$, $\mi{request}$, $\mi{key}$, $a$, $f$, $s'$}
  \Let{$\mi{session}$}{$s'.\str{sessions}[\mi{reference}[\str{session}]]$}
  \Let{$\mi{sessionId}$}{$\mi{reference}[\str{session}]$}
  \Let{$\mi{id}$}{$\mi{session}[\str{identity}]$} 
  \Let{$\mi{issuer}$}{$s'.\str{issuerCache}[\mi{id}]$}
  \Let{$\mi{profile}$}{$s'.\str{clientCredentialsCache}[\mi{issuer}][\str{profile}]$}
  \Let{$\mi{isApp}$}{$s'.\str{clientCredentialsCache}[\mi{issuer}][\str{is\_app}]$}
  \Let{$\mi{clientId}$}{$s'.\str{clientCredentialsCache}[\mi{issuer}][\str{client\_id}]$}
  \If{$\mi{reference}[\str{responseTo}] \equiv \str{TOKEN}$}
  	\If{$\mi{session}[\str{scope}] \equiv \an{} \wedge \mi{profile} \equiv \str{r}$}
    	\Let{$\mi{useAccessTokenNow}$}{$\True$}
    \ElsIf{$\mi{session}[\str{scope}] \equiv \an{\str{openid}} \wedge \mi{profile} \equiv \str{r}$}
    	\LetND{$\mi{useAccessTokenNow}$}{$\{\True, \bot\}$}
    \ElsIf{$\mi{profile} \equiv \str{rw} \wedge \mi{session}[\str{response\_type}] \equiv
          \an{\str{code}, \str{id\_token}}$}  \Comment{OIDC Hybrid Flow}
        \Let{$\mi{firstIdToken}$}{$\mi{session}[\str{redirectEpRequest}][\str{data}][\str{id\_token}]$}
        \If{$\mathsf{checksig}(m.\str{body}[\str{id\_token}], s'.\str{jwksCache}[\mi{issuer}]) \not\equiv \top$} \label{c-checksig-second-idt}
          \Stop \ParboxComment{Check the signature of the id token.}
        \EndIf
        \If{$\mathsf{extractmsg}(m.\str{body}[\str{id\_token}])[\str{sub}] \not\equiv \mathsf{extractmsg}(\mi{firstIdToken})[\str{sub}]$} \label{c-compare-sub-of-both-id-tokens}
          \Stop \ParboxComment{Check if $\str{sub}$ is the same as in the first id token ([OIDC], 3.3.3.6).}
        \EndIf
        \If{$\mathsf{extractmsg}(m.\str{body}[\str{id\_token}])[\str{iss}] \not\equiv \mathsf{extractmsg}(\mi{firstIdToken})[\str{iss}]$} \label{c-compare-iss-of-both-id-tokens} %
          \Stop \ParboxComment{Check if $\str{iss}$ is the same as in the first id token ([OIDC], 3.3.3.6).}
        \EndIf
        \If{$\mathsf{extractmsg}(m.\str{body}[\str{id\_token}])[\str{iss}] \not\equiv \mi{issuer}$} \label{c-token-resp-check-issuer}
			    \Stop \ParboxComment{Check the issuer.}
			  \EndIf
        \If{$\mathsf{extractmsg}(m.\str{body}[\str{id\_token}])[\str{at\_hash}] \not\equiv \mathsf{hash}(m.\str{body}[\str{access\_token}])$} \label{c-check-at-hash}
          \Stop \ParboxComment{Check $\str{at\_hash}$ of second id token (protection against reuse of phised access token).}
        \EndIf
        \If{$\mathsf{extractmsg}(m.\str{body}[\str{id\_token}])[\str{aud}] \not\equiv \mi{clientId}$} \label{c-hybrid-token-resp-check-aud}
          \Stop \ParboxComment{Check $\str{aud}$ of second id token.}
        \EndIf
				\Let{$s'.\str{sessions}[\mi{reference}[\str{session}]][\str{idt2\_iss}]$}{$\mathsf{extractmsg}(m.\str{body}[\str{id\_token}])[\str{iss}]$} \label{c-set-idt2-iss}
    	  \LetND{$\mi{useAccessTokenNow}$}{$\{\True, \bot\}$}
    \ElsIf{$\mi{profile} \equiv \str{rw} \wedge \mi{session}[\str{response\_type}] \equiv
          \an{\str{JARM\_code}}$}  \Comment{Code flow (JARM)}
        \Let{$\mi{requestJWS}$}{$\mi{session}[\str{redirectEpRequest}][\str{data}][\str{responseJWS}]$} \label{JARM-get-reqJWS-from-redireEpReq-1}
        \If{$\mathsf{extractmsg}(\mi{requestJWS})[\str{at\_hash}] \not\equiv \mathsf{hash}(m.\str{body}[\str{access\_token}])$} \label{JARM-c-check-at-hash}
          \Stop \ParboxComment{Check $\str{at\_hash}$ of requestJWS (protection against reuse of phised access token).}
        \EndIf
  	    \If{$\mi{session}[\str{scope}] \equiv \an{}$}
          \Let{$\mi{useAccessTokenNow}$}{$\True$}
        \Else
          \LetND{$\mi{useAccessTokenNow}$}{$\{\True, \bot\}$}
        \EndIf
		\EndIf
    \If{$\mi{useAccessTokenNow} \equiv \True$}
      \CallFun{PREPARE\_USE\_ACCESS\_TOKEN}{$\mi{reference}[\str{session}]$, $m.\str{body}[\str{access\_token}]$, $s'$} \label{line:client-call-use-access-token-2}
    \EndIf
    \CallFun{CHECK\_ID\_TOKEN}{$\mi{reference}[\str{session}]$, $m.\str{body}[\str{id\_token}]$, $s'$} \label{line:client-call-check-id-token-after-code}

  \ElsIf{$\mi{reference}[\str{responseTo}] \equiv \str{MTLS\_AS}$} \label{c-PHTTPSResp-ref-mtls-as}
  	\Let{$\mi{code}$}{$s'.\str{sessions}[\mi{sessionId}][\str{code}]$} %
   \If{$\exists\, m_{\text{dec}}$, $k'$, $\mi{dom}$ \textbf{such that} $m_{\text{dec}} \equiv \dec{m.\str{body}}{k'}$ $\wedge$ $\an{dom,k'} \in s'.\str{tlskeys}$} 
   \LetST{$\mi{mtlsNonce}, \mi{pubKey}$}{$\an{\mi{mtlsNonce}, \mi{pubKey}} \equiv  m_{\text{dec}}$}{\textbf{stop} \DefStop} \label{line:client-decrypt-mtls}
   \EndIf
    
    \If{$\mi{pubKey} \in \mi{keyMapping}(\mi{request}.\str{host})$} \ParboxComment{Send nonce only to the process that created it}
      \CallFun{SEND\_TOKEN\_REQUEST}{$\mi{sessionId}$, $\mi{code}$, 
          $[\str{type}:\str{MTLS},
          \str{mtls\_nonce}{:}\mi{mtlsNonce}]$,
          $s'$} \label{line:client-https-response-mtls-as-send-token-req}
      \Else
         \Stop{\DefStop} 
      \EndIf
  \ElsIf{$\mi{reference}[\str{responseTo}] \equiv \str{OAUTB\_AS}$}  \label{line:ref-oautb-as}
  	\Let{$\mi{code}$}{$s'.\str{sessions}[\mi{sessionId}][\str{code}]$} 
		\Let{$\mi{ekm}$}{$\mathsf{hash}(\an{ m.\str{nonce},  m.\str{body}[\str{tb\_nonce}], 
              \mi{keyMapping}[\mi{request}.\str{host}] })$}   \label{line:PHResp-STR-oautb-sign-ekm} \\ $\hspace{1cm}$\ParboxComment{Include public key to model Extended Master Secret}
    \CallFun{SEND\_TOKEN\_REQUEST}{$\mi{sessionId}$, $\mi{code}$, $[\str{type}:\str{OAUTB}, \str{ekm}{:} \mi{ekm}]$, $s'$} \label{line:PHResp-STR-oautb-call}
  \ElsIf{$\mi{reference}[\str{responseTo}] \equiv \str{MTLS\_RS}$} \label{branch-ref-mtls-rs}
  		\Let{$\mi{token}$}{$s'.\str{sessions}[\mi{sessionId}][\str{token}]$}  \label{PHResp-mtls-rs-retrieve-token}

      \If{$\exists\, m_{\text{dec}}$, $k'$, $\mi{dom}$ \textbf{such that} $m_{\text{dec}} \equiv \dec{m.\str{body}}{k'}$ $\wedge$ $\an{dom,k'} \in s'.\str{tlskeys}$} 
      \LetST{$\mi{mtlsNonce}, \mi{pubKey}$}{$\an{\mi{mtlsNonce}, \mi{pubKey}} \equiv  m_{\text{dec}}$}{\textbf{stop} \DefStop} \label{line:client-decrypt-mtls-rs}
   \EndIf

    \If{$\mi{pubKey} \in \mi{keyMapping}(\mi{request}.\str{host})$} \ParboxComment{Send nonce only to the process that created it}\label{c-PHResp-check-mtls-rs-pubkey}
      \CallFun{USE\_ACCESS\_TOKEN}{$\mi{sessionId}$, $\mi{token}$, \label{mtls-authn-rs-use-access-token}
        $[\str{type}:\str{MTLS},
        \str{mtls\_nonce}{:}\mi{mtlsNonce}]$,
        $s'$}
      \Else
         \Stop{\DefStop} 
      \EndIf

  \ElsIf{$\mi{reference}[\str{responseTo}] \equiv \str{OAUTB\_RS}$}
          \Let{$\mi{token}$}{$s'.\str{sessions}[\mi{sessionId}][\str{token}]$} \label{PHResp-oautb-rs-retrieve-token}
				  \Let{$\mi{ekm}$}{$\mathsf{hash}(\an{ m.\str{nonce},  m.\str{body}[\str{tb\_nonce}], 
                  \mi{keyMapping}[\mi{request}.\str{host}] })$}   \label{line:PHResp-UAT-oautb-sign-ekm}
         \CallFun{USE\_ACCESS\_TOKEN}{$\mi{sessionId}$, $\mi{token}$, 
                $[\str{type}:\str{OAUTB},
                  \str{ekm}{:} \mi{ekm}]$, $s'$} \label{line:PHResp-UAT-oautb-call}

 \algstore{https-req-PHRESP}
 \end{algorithmic}
 \end{algorithm}
 
 \begin{algorithm}
 \begin{algorithmic}[1]
 \algrestore{https-req-PHRESP}

  \ElsIf{$\mi{reference}[\str{responseTo}] \equiv \str{RESOURCE\_USAGE}$} \label{client-PHResp-case-resource-usage}
      \Let{$\mi{resource}$}{$m.\str{body}[\str{resource}]$} \label{client-receives-resource}
      \Let{$s'.\str{sessions}[\mi{sessionId}][\str{resource}]$}{$\mi{resource}$} \label{client-saves-resource} 

      \Let{$\mi{request}$}{$\mi{session}[\str{redirectEpRequest}]$}
      \Comment{\parbox[t]{8cm}{Retrieve stored meta data of the request from the browser 
      to the redir. endpoint stored in Algorithm~\ref{alg:client-fapi-http-request}.}}

    \If{$\mi{isApp} \equiv \bot$} \Comment{Send resource access nonce to browser} \label{ws-client-send-r-to-browser}
      \Let{$\mi{headers}$}{$[\str{ReferrerPolicy}{:}\str{origin}]$}
      \Let{$m'$}{$\encs{\an{\cHttpResp, \mi{request}[\str{message}].\str{nonce}, 200,   
          \mi{headers}, \mi{resource}}}{\mi{request}[\str{key}]}$}
      \Stop{$\an{\an{\mi{request}[\str{sender}],\mi{request}[\str{receiver}],m'}}$, $s'$} \label{line:client-send-resource}
    \Else \Comment{$\str{isApp} \equiv \True$}
      \Stop{$s'$} \label{line:app-c-do-not-send-resource} 
    \EndIf
  \EndIf

\EndFunction
\end{algorithmic} %
\end{algorithm}

\begin{algorithm}[h] %
  \caption{\label{alg:client-fapi-start-login-flow} Relation of a Client $R^c$ -- Starting the login flow.}
\begin{algorithmic}[1]
\Function{$\mathsf{START\_LOGIN\_FLOW}$}{$\mi{sessionId}$, $s'$}
  \Let{$\mi{session}$}{$s'.\str{sessions}[\mi{sessionId}]$} \label{c-SLF-start}
  \Let{$\mi{identity}$}{$\mi{session}[\str{identity}]$}
  \Let{$\mi{issuer}$}{$s'.\str{issuerCache}[\mi{identity}]$}
  \Let{$\mi{oidcConfig}$}{$s'.\str{oidcConfigCache}[\mi{issuer}]$}
  \Let{$\mi{authEndpoint}$}{$\mi{oidcConfig}[\str{auth\_ep}]$} \label{c-SLF-authep}
  \Let{$\mi{redirectUris}$}{$\mi{oidcConfig}[\str{redirect\_uris}]$} \label{line:client-choose-redirect-uris} 
  \ParboxComment{Set of redirect URIs for the AS.}
  \LetND{$\mi{resourceServer}$}{$\mi{oidcConfig}[\str{resource\_servers}]$} \Comment{Choose resource server} %
  \Let{$s'.\str{sessions}[\mi{sessionId}]$}{$s'.\str{sessions}[\mi{sessionId}] \cup [\str{RS}{:}\mi{resourceServer}]$} \label{line:c-sets-session-RS}

  \Let{$\mi{credentials}$}{$s'.\str{clientCredentialsCache}[\mi{issuer}]$} 

 \Let{$\mi{headers}$}{$[\str{ReferrerPolicy}{:}\str{origin}]$}
  \Let{$\mi{headers}[\cSetCookie]$}{$[\str{sessionId}{:}\an{\mi{sessionId},\top,\top,\top}]$} \label{client-setCookie-lsid}

 \Let{$\mi{profile}$}{$\mi{credentials}[\str{profile}]$} 
 \Comment{either $\str{r}$ or $\str{rw}$}

 \Let{$\mi{isApp}$}{$\mi{credentials}[\str{is\_app}]$}
 \Comment{either $\True$ or $\bot$}

 \Let{$\mi{clientType}$}{$\mi{credentials}[\str{client\_type}]$} 
  \Comment{$\str{pub}$, $\str{conf\_JWS}$, $\str{conf\_MTLS}$ or $\str{conf\_OAUTB}$}
  \If{$\mi{profile} \equiv \str{r}$}
  	\Let{$\mi{responseType}$}{$\an{\str{code}}$} 
	\Else 
    \LetND{$\mi{responseType}$}{$\{ \an{\str{code}, \str{id\_token}}, \an{\str{JARM\_code}} \}$}  %
	\EndIf

  \LetND{$\mi{redirectUri}$}{$\mi{redirectUris}$} 
  \Comment{Auth.~response must be received here.}
 	\Let{$s'.\str{sessions}[\mi{sessionId}]$}{$s'.\str{sessions}[\mi{sessionId}] \cup [\str{redirect\_uri}{:}\mi{redirectUri}]$}

 \LetND{$\mi{scope}$}{$\{\an{},  \an{\str{openid}} \}$}
  \If{$\mi{scope} \equiv \an{\str{openid}} \vee (\mi{profile} \equiv \str{rw} 
    \wedge \mi{responseType} \equiv \an{\str{code}, \str{id\_token}})$} \Comment{Nonce for obtaining an id token.} %
  			\Let{$\mi{nonce}$}{$\nu_7$} 
	\Else
  			\Let{$\mi{nonce}$}{$\an{}$}
	\EndIf
 	\If{$\mi{profile} \equiv \str{r}$}
  		\Let{$\mi{pkceVerifier}$}{$\nu_9$ }\label{r-client-chooses-pkce-cv}
  		\Let{$\mi{pkceChallenge}$}{$\mathsf{hash}(\mi{pkceVerifier})$} \label{r-client-creates-pkce-cc}
     	\Let{$s'.\str{sessions}[\mi{sessionId}]$}{$s'.\str{sessions}[\mi{sessionId}] \cup [\str{pkce\_verifier}{:}\mi{pkceVerifier}]$} \label{r-c-creates-pkce-cv}
 	\ElsIf{$(\mi{clientType} \equiv \str{pub}) \vee (\mi{clientType} \equiv \str{conf\_OAUTB}
        \wedge \mi{isApp} \equiv \True)$} 	\Comment{OAUTB (app client)} 
  	  \Let{$\mi{TB\mhyphen{}Id}$}{$\mathsf{pub}(s'.\str{tokenBindings}[\mi{authEndpoint}.\str{host}])$} 
   		\Let{$\mi{pkceChallenge}$}{$\mathsf{hash}(\mi{TB\mhyphen{}Id})$} \label{pub-rw-client-creates-pkce-cc}
  \ElsIf{$\mi{clientType} \equiv \str{conf\_OAUTB}$}\Comment{web server client} 
   		\Let{$\mi{pkceChallenge}$}{$\str{referred\_tb}$}
 			\Let{$\mi{headers}[\str{Include{\mhyphen}Referred{\mhyphen}Token{\mhyphen}Binding{\mhyphen}ID}]$}{$\True$}
	\Else \Comment{rw + conf\_MTLS}
   		\Let{ $\mi{pkceChallenge}$}{ $\an{}$}
	\EndIf 
	\Let{$\mi{data}$}{$[\str{response\_type}{:}\mi{responseType}, \str{redirect\_uri}{:}\mi{redirectUri},
    \str{client\_id}{:}\mi{credentials}[\str{client\_id}],
    $\breakalgohook{6}$
		\str{scope}{:}\mi{scope},
		\str{nonce}{:}\mi{nonce},
		\str{pkce\_challenge}{:}\mi{pkceChallenge},
    \str{state}{:}\nu_6 ]$} \label{c-slf-chooses-state-value-here}

		\Let{$\mi{requestJWT}$}{$\mi{data} \cup [\str{aud}{:}\mi{authEndpoint}.\str{host}]$} \label{client-slf-reqJWS-aud}
		\Let{$\mi{requestJWS}$}{$\mathsf{sig}(\mi{requestJWT}, s'.\str{authReqSigKey})$} \label{client-slf-requestJWS-sign-1}
		\Let{$\mi{data}[\str{request\_jws}]$}{$\mi{requestJWS}$}
  \Let{$s'.\str{sessions}[\mi{sessionId}]$}{$s'.\str{sessions}[\mi{sessionId}] \cup \mi{data}$}

  \Let{$\mi{authEndpoint}.\str{parameters}$}{$ \mi{data}$}
  \Let{$\mi{headers}[\str{Location}]$}{$\mi{authEndpoint}$}
  \Let{$\mi{request}$}{$s'.\str{sessions}[\mi{sessionId}][\str{startRequest}]$} \label{c-redir-ep-retrieve-startReq}
  \Let{$m'$}{$\encs{\an{\cHttpResp, \mi{request}[\str{message}].\str{nonce}, 303, 
      \mi{headers}, \bot}}{\mi{request}[\str{key}]}$} 

	\Let{$m_{leak}$}{$\an{\str{LEAK}, \mi{authEndpoint}}$}

  \Stop{$\an{
							\an{\mi{leak}, \mi{request}[\str{receiver}], m_{leak}},
  					  \an{\mi{request}[\str{sender}],\mi{request}[\str{receiver}],m'}
				    }$, $s'$} \label{line:client-send-authorization-redir}
\EndFunction
\end{algorithmic} %
\end{algorithm}

\begin{algorithm}
\caption{\label{alg:client-prepare-token-request} Relation of a Client $R^c$ -- Prepare request to token endpoint.}
\begin{algorithmic}[1]
\Function{$\mathsf{PREPARE\_TOKEN\_REQUEST}$}{$\mi{sessionId}$, $\mi{code}$, $s'$}

  \Let{$\mi{session}$}{$s'.\str{sessions}[\mi{sessionId}]$}
  \Let{$\mi{identity}$}{$\mi{session}[\str{identity}]$}
  \Let{$\mi{issuer}$}{$s'.\str{issuerCache}[\mi{identity}]$}

	\Let{$\mi{credentials}$}{$s'.\str{clientCredentialsCache}[\mi{issuer}]$}
  \Let{$\mi{clientId}$}{$\mi{credentials}[\str{client\_id}]$}
  \Let{$\mi{clientType}$}{$\mi{credentials}[\str{client\_type}]$} 
  \Let{$\mi{profile}$}{$\mi{credentials}[\str{profile}]$}

	\If{$\mi{profile} \equiv \str{rw}$} \label{c-chooses-misconfTEp}
  	\LetND{$\mi{misconfiguredTEp}$}{$\{\True, \bot \}$}
	\Else
  	\Let{$\mi{misconfiguredTEp}$}{$ \bot $} \label{ptr-r-client-set-misconf-tep-false}
	\EndIf
  \Let{$s'.\str{sessions}[\str{sessionId}][\str{misconfiguredTEp}]$}{$\mi{misconfiguredTEp}$} 
	\If{$\mi{misconfiguredTEp} \equiv \True$}
	  \ParboxComment{Choose wrong token endpoint.}
	  \LetND{$\mi{host}$}{$\dns$}
	  \LetND{$\mi{path}$}{$\mathbb{S}$}
    \LetND{$\mi{parameters}$}{$\dict{\mathbb{S}}{\mathbb{S}}$} 
	  \Let{$\mi{url}$}{$\an{\cUrl, \https , \mi{host}, \mi{path}, \mi{parameters}, \bot }$}
  	\Let{$s'.\str{sessions}[\str{sessionId}][\str{token\_ep}]$}{$\mi{url}$} 
	\Else
  	\Let{$\mi{url}$}{$s'.\str{oidcConfigCache}[\mi{issuer}][\str{token\_ep}]$}
	\EndIf
  \Let{$s'.\str{sessions}[\str{sessionId}][\str{code}]$}{$\mi{code}$}   \label{c-PTR-save-code}

  \If{$\mi{profile} \equiv \str{r}$} \Comment{Send token request}
 		\If{$\mi{clientType} \equiv \str{pub} \vee \mi{clientType} \equiv \str{conf\_JWS}$}
        \CallFun{SEND\_TOKEN\_REQUEST}{$\mi{sessionId}$, $\mi{code}$, $\an{}$, $s'$} \label{c-PTR-call-STR}
		\EndIf
	\EndIf

  \If{$ \mi{clientType} \equiv \str{conf\_MTLS} $} \Comment{both profiles}
	    \Let{$\mi{body}$}{$[\str{client\_id}{:}\mi{clientId}]$}
      \Let{$\mi{message}$}{$\hreq{nonce=\nu_{10},
	  	   method=\mGet, 
	  	   xhost=\mi{url}.\str{domain},
	  	   path=\str{/MTLS\mhyphen{}prepare}, 
	  	   parameters=\mi{url}.\str{parameters}, 
	  	   headers=\an{}, 
         xbody=\mi{body} }$} \label{line-client-https-simple-mtls-as-msg}
	  	\CallFun{HTTPS\_SIMPLE\_SEND}{$[\str{responseTo}{:}\str{MTLS\_AS}, \str{session}{:}\mi{sessionId}]$, $\mi{message}$, $s'$} \label{line:client-https-simple-mtls-as}
    \Else \Comment{OAUTB} %
			\Let{$\mi{message}$}{$\hreq{ nonce=\nu_{10},
				method=\mGet, 
				xhost=\mi{url}.\str{domain},
	   		path=\str{/OAUTB\mhyphen{}prepare}, 
				parameters=\mi{url}.\str{parameters}, 
				headers=\an{}, 
        xbody=\an{}}$}  \label{line-client-https-simple-oautb-at-msg} 
			\CallFun{HTTPS\_SIMPLE\_SEND}{$[\str{responseTo}{:}\str{OAUTB\_AS},
                         \str{session}{:}\mi{sessionId}]$, $\mi{message}$, $s'$}  \label{line:PTR-oautb-as-send}
		\EndIf
\EndFunction
\end{algorithmic} %
\end{algorithm}

\begin{algorithm}
\caption{\label{alg:client-token-request} Relation of a Client $R^c$ -- Request to token endpoint.}
\begin{algorithmic}[1]
\Function{$\mathsf{SEND\_TOKEN\_REQUEST}$}{$\mi{sessionId}$, $\mi{code}$, $\mi{responseValue}$, $s'$}
  \Let{$\mi{session}$}{$s'.\str{sessions}[\mi{sessionId}]$}
  \Let{$\mi{identity}$}{$\mi{session}[\str{identity}]$} \label{line:str-client-chooses-identity}
  \Let{$\mi{issuer}$}{$s'.\str{issuerCache}[\mi{identity}]$} \label{line:str-client-chooses-issuer}
  \If{$\mi{session}[\str{misconfiguredTEp}] \equiv \True$} \label{str-check-if-misconf-tep}
  	\Let{$\mi{url}$}{$\mi{session}[\str{token\_ep}] $} 
	\Else
    \Let{$\mi{url}$}{$s'.\str{oidcConfigCache}[\mi{issuer}][\str{token\_ep}]$} \label{line:r-client-chooses-right-tep}
	\EndIf

	\Let{$\mi{credentials}$}{$s'.\str{clientCredentialsCache}[\mi{issuer}]$}
  \Let{$\mi{clientId}$}{$\mi{credentials}[\str{client\_id}]$}

  \Let{$\mi{clientType}$}{$\mi{credentials}[\str{client\_type}]$} 
 	\Let{$\mi{profile}$}{$\mi{credentials}[\str{profile}]$}  
 	\Let{$\mi{isApp}$}{$\mi{credentials}[\str{is\_app}]$}  

	\Let{$\mi{body}$}{$[\str{grant\_type}{:} \str{authorization\_code}, \str{code}{:}\mi{code}, 
						$\breakalgohook{5} $
						\str{redirect\_uri}{:}\mi{session}[\str{redirect\_uri}],
						\str{client\_id}{:}\mi{clientId}	]$} \label{c-STR-main-body}

  \If{$\mi{profile} \equiv \str{r}$}
			\Let{$\mi{body}[\str{pkce\_verifier}]$}{$\mi{session}[\str{pkce\_verifier}] $} 
	\EndIf
  \If{$\mi{profile} \equiv \str{r} \wedge \mi{clientType} \equiv \str{pub} $}
		\Let{$\mi{message}$}{$\hreq{ nonce=\nu_2,
				method=\mPost,
				xhost=\mi{url}.\str{domain},
				path=\mi{url}.\str{path},
				parameters=\mi{url}.\str{parameters},
				headers=\bot,
				xbody=\mi{body}}$}  
    \CallFun{HTTPS\_SIMPLE\_SEND}{$[\str{responseTo}{:}\str{TOKEN},\str{session}{:}\mi{sessionId}]$, $\mi{message}$, $s'$} \label{line:r-client-sends-token-req-pub}
  	\ElsIf{$\mi{profile} \equiv \str{r} \wedge \mi{clientType} \equiv \str{conf\_JWS}$} \label{branch:r-client-creates-assertion}
			\Let{$\mi{clientSecret}$}{$\mi{credentials}[\str{client\_secret}]$} 
			\Let{$\mi{jwt}$}{$[ \str{iss}{:}\mi{clientId}, 
													\str{aud}{:}\mi{url}.\str{domain} 
												]$} 
			\Let{$\mi{body}[\str{assertion}]$}{$\mathsf{mac}(\mi{jwt}, \mi{clientSecret})$}
			\Let{$\mi{message}$}{$\hreq{ nonce=\nu_2,
					method=\mPost,
					xhost=\mi{url}.\str{domain},
					path=\mi{url}.\str{path},
					parameters=\mi{url}.\str{parameters},
					headers=\bot,
          xbody=\mi{body}}$}   \label{line:r-client-assertion-rec}
			\CallFun{HTTPS\_SIMPLE\_SEND}{$[\str{responseTo}{:}\str{TOKEN},\str{session}{:}\mi{sessionId}]$, $\mi{message}$, $s'$} \label{line:r-client-sends-assertion} \label{line:r-client-sends-token-req-jws} %
  \ElsIf{$\mi{clientType} \equiv \str{conf\_MTLS} $} \Comment{both profiles} 
    \If{$\mi{responseValue}[\str{type}] \not \equiv \str{MTLS}$} \label{c-str-mtls-check-resp-val-type}
         \Stop{\DefStop} 
    \EndIf
		\Let{$\mi{body}[\str{TLS\_AuthN}]$}{$\mi{responseValue}[\str{mtls\_nonce}]$} 
		\Let{$\mi{message}$}{$\hreq{ nonce=\nu_2,
			method=\mPost,
			xhost=\mi{url}.\str{domain},
			path=\mi{url}.\str{path},
			parameters=\mi{url}.\str{parameters},
			headers=\bot,
			xbody=\mi{body}}$}  
		\CallFun{HTTPS\_SIMPLE\_SEND}{$[\str{responseTo}{:}\str{TOKEN},\str{session}{:}\mi{sessionId}]$, $\mi{message}$, $s'$} \label{line:client-sends-token-req-mtls}
	\Else \Comment{rw with OAUTB} \label{client-str-branch-oautb-sig-2}
    \If{$\mi{responseValue}[\str{type}] \not \equiv \str{OAUTB}$} \label{client-str-resp-value-oautb}
         \Stop{\DefStop} 
    \EndIf
    \Let{$\mi{ekm}$}{$\mi{responseValue}[\str{ekm}]$} \label{line:str-oautb-respVal1}

		\Let{$\mi{TB\_AS}$}{$s'.\str{tokenBindings}[\mi{url}.\str{host}]$} \Comment{private key}
		\Let{$\mi{TB\_RS}$}{$s'.\str{tokenBindings}[\mi{session}[\str{RS}]]$} \Comment{private key} \label{c-str-retrieve-rs-tb}
			
	 	\Let{$\mi{TB\_Msg\_prov}$}{$[\str{id}{:} \mathsf{pub}(\mi{TB\_AS}), \str{sig}{:} \mathsf{sig}( \mi{ekm} , \mi{TB\_AS})]$} \label{client-str-sig-prov}
		\Let{$\mi{TB\_Msg\_ref}$}{$[\str{id}{:} \mathsf{pub}(\mi{TB\_RS}),  \str{sig}{:} \mathsf{sig}( \mi{ekm} , \mi{TB\_RS})]$} \label{c-str-compose-ref-tb-msg}
			
		\Let{$\mi{headers}$}{$[\str{Sec\mhyphen{}Token\mhyphen{}Binding}{:}[\str{prov}{:}\mi{TB\_Msg\_prov},\str{ref}{:}\mi{TB\_Msg\_ref}]]$} \label{c-str-sec-tb-header}

    \If{$\mi{clientType} \equiv \str{conf\_OAUTB}$} \Comment{client authentication}
				\Let{$\mi{clientSecret}$}{$\mi{credentials}[\str{client\_secret}]$} 
        \Let{$\mi{jwt}$}{$[ \str{iss}{:}\mi{clientId},  \label{line:jws-oautb-set-aud}
														\str{aud}{:}\mi{url}.\str{domain},  
													]$} 
			\Let{$\mi{body}[\str{assertion}]$}{$\mathsf{mac}(\mi{jwt}, \mi{clientSecret})$} \label{line:rw-client-creates-assertion} %
		\EndIf
    \If{$\mi{isApp} \equiv \bot$} \Comment{web server client: send TBID used by browser}
			\Let{$\mi{body}[\str{pkce\_verifier}]$}{$\mi{session}[\str{browserTBID}] $}  \label{line:rw-client-pkce-cv-browserTBID}
		\EndIf
    \Let{$\mi{message}$}{$\hreq{ nonce=\nu_2, 
  		method=\mPost,
  		xhost=\mi{url}.\str{domain},
  		path=\mi{url}.\str{path},
  		parameters=\mi{url}.\str{parameters},
  		headers=\mi{headers},
  		xbody=\mi{body}}$}  \label{line:jws-oautb-msg}
		\CallFun{HTTPS\_SIMPLE\_SEND}{$[\str{responseTo}{:}\str{TOKEN},\str{session}{:}\mi{sessionId}]$, $\mi{message}$, $s'$} \label{line:jws-oautb-cont-send}
	\EndIf
\EndFunction
\end{algorithmic} %
\end{algorithm}

\begin{algorithm}[tbp]
\caption{\label{alg:client-prepare-use-access-token} Relation of a Client $R^c$ -- Prepare for using the access token.}
\begin{algorithmic}[1]
\Function{$\mathsf{PREPARE\_USE\_ACCESS\_TOKEN}$}{$\mi{sessionId}$, $\mi{token}$, $s'$}
  \Let{$\mi{session}$}{$s'.\str{sessions}[\mi{sessionId}]$} \label{line:c-PUAT-session}
  \Let{$\mi{identity}$}{$\mi{session}[\str{identity}]$}
  \Let{$\mi{issuer}$}{$s'.\str{issuerCache}[\mi{identity}]$}
	\Let{$\mi{credentials}$}{$s'.\str{clientCredentialsCache}[\mi{issuer}]$}
  \Let{$\mi{clientType}$}{$\mi{credentials}[\str{client\_type}]$} 
  \Let{$\mi{profile}$}{$\mi{credentials}[\str{profile}]$}
  \Let{$s'.\str{sessions}[\str{sessionId}][\str{token}]$}{$\mi{token}$} \label{c-save-access-token-in-session}

  \Let{$\mi{rsHost}$}{$\mi{session}[\str{RS}]$} \label{line:c-PUAT-rsHost}

  \If{$\mi{profile} \equiv \str{r}$}
         \CallFun{USE\_ACCESS\_TOKEN}{$\mi{sessionId}$, $\mi{token}$, $\an{}$, $s'$} 
  \ElsIf{$\mi{clientType} \equiv \str{conf\_MTLS}$} \Comment{rw + mTLS: AT is bound to client}
	  		\Let{$\mi{message}$}{$\hreq{ nonce=\nu_{11},
	  	    method=\mGet, 
	  	    xhost=\mi{rsHost}, 
	  	    path=\str{/MTLS\mhyphen{}prepare}, 
	  	    parameters=\bot, 
	  	    headers=\bot,
          xbody=\mi{body} }$} \label{line:c-PUAT-message}
        \CallFun{HTTPS\_SIMPLE\_SEND}{$[\str{responseTo}{:}\str{MTLS\_RS},\str{session}{:}\mi{sessionId}]$, $\mi{message}$, $s'$} \label{line:c-sends-ref-mtls-rs}
  \Else \Comment{OAUTB}
	  		\Let{$\mi{message}$}{$\hreq{ nonce=\nu_{11},
	  	    method=\mGet, 
	  	    xhost=\mi{rsHost}, 
	  	    path=\str{/OAUTB\mhyphen{}prepare}, 
	  	    parameters=\bot, 
	  	    headers=\bot,
	  	    xbody=\mi{body} }$}
	  		\CallFun{HTTPS\_SIMPLE\_SEND}{$[\str{responseTo}{:}\str{OAUTB\_RS},
										\str{session}{:}\mi{sessionId}]$, $\mi{message}$, $s'$} 
	\EndIf
\EndFunction
\end{algorithmic} %
\end{algorithm}

\begin{algorithm}[tbp]
\caption{\label{alg:client-use-access-token} Relation of a Client $R^c$ -- Using the access token.}
\begin{algorithmic}[1]
\Function{$\mathsf{USE\_ACCESS\_TOKEN}$}{$\mi{sessionId}$, $\mi{token}$, $\mi{responseValue}$, $s'$}
  \Let{$\mi{session}$}{$s'.\str{sessions}[\mi{sessionId}]$} \label{c:UAT-retrieve-session}
  \Let{$\mi{identity}$}{$\mi{session}[\str{identity}]$}
  \Let{$\mi{issuer}$}{$s'.\str{issuerCache}[\mi{identity}]$}

	\Let{$\mi{credentials}$}{$s'.\str{clientCredentialsCache}[\mi{issuer}]$}
  \Let{$\mi{clientType}$}{$\mi{credentials}[\str{client\_type}]$} 
  \Let{$\mi{profile}$}{$\mi{credentials}[\str{profile}]$}

  \Let{$\mi{headers}$}{$[\str{Authorization}:\an{\str{Bearer},\mi{token}}]$} \label{c:UAT-set-authorization-header}

 	\Let{$\mi{body}$}{$[]$} 
         
  \Let{$\mi{rsHost}$}{$\mi{session}[\str{RS}]$} \label{c:UAT-retrieve-rsHost}
  \If{$\mi{profile} \equiv \str{r}$}
 		 \Let{$\mi{message}$}{$\hreq{ nonce=\nu_3,
 		     method=\mGet, %
 		     xhost=\mi{rsHost}, 
 		     path=\str{/resource\mhyphen{}r},
 		     parameters=\an{}, 
 		     headers=\mi{headers},
 		     xbody=\mi{body}}$}  
  	\CallFun{HTTPS\_SIMPLE\_SEND}{$[\str{responseTo}{:}\str{RESOURCE\_USAGE},\str{session}{:}\mi{sessionId}]$, $\mi{message}$, $s'$} \label{r-c-sends-req-to-rs}
  \Else
    \If{$ \mi{session}[\str{response\_type}] \equiv \an{\str{code}, \str{id\_token}}$}  \Comment{OIDC Hybrid Flow}
  		\Let{$\mi{body}[\str{at\_iss}]$}{$\mi{session}[\str{idt2\_iss}]$} \Comment{Send issuer of second id token} \label{c-sets-at-iss}
    \Else \Comment{JARM Code Flow}
      \Let{$\mi{body}[\str{at\_iss}]$}{$\mi{session}[\str{JARM\_iss}]$} \Comment{Send issuer request JWS} \label{c-sets-at-iss-JARM}
    \EndIf

      \If{$\mi{clientType} \equiv \str{conf\_MTLS}$} \label{client-uat-if-client-type-mtls}
        \If{$\mi{responseValue}[\str{type}] \not \equiv \str{MTLS}$}
            \Stop{\DefStop} 
        \EndIf
  			\Let{$\mi{body}[\str{MTLS\_AuthN}]$}{$\mi{responseValue}[\str{mtls\_nonce}]$}
  		\Else \Comment{OAUTB} 
        \If{$\mi{responseValue}[\str{type}] \not \equiv \str{OAUTB}$} \label{client-uat-resp-value-oautb}
            \Stop{\DefStop} 
        \EndIf
				\Let{$\mi{ekm}$}{$\mi{responseValue}[\str{ekm}]$}
				\Let{$\mi{TB\_RS}$}{$s'.\str{tokenBindings}[\mi{rsHost}]$} 
			 	\Let{$\mi{TB\_Msg\_prov}$}{$[\str{id}{:} \mathsf{pub}(\mi{TB\_RS}), \str{sig}{:} 
            \mathsf{sig}(\mi{ekm}, \mi{TB\_RS})]$} \label{client-uat-sig-3}
			
				\Let{$\mi{headers}[\str{Sec\mhyphen{}Token\mhyphen{}Binding}]$}{$[\str{prov}{:}\mi{TB\_Msg\_prov}]$}
			\EndIf
		  \Let{$\mi{message}$}{$\hreq{ nonce=\nu_3,
		      method=\mPost,
		      xhost=\mi{rsHost}, 
 		     	path=\str{/resource\mhyphen{}rw},
		      parameters=\an{}, 
		      headers=\mi{headers},
          xbody=\mi{body}}$}  \label{c:UAT-create-rw-message}
  		\CallFun{HTTPS\_SIMPLE\_SEND}{$[\str{responseTo}{:}\str{RESOURCE\_USAGE},\str{session}{:}\mi{sessionId}]$,
          $\mi{message}$, $s'$} \label{c-sends-msg-to-resource-rw}
  \EndIf
\EndFunction
\end{algorithmic} %
\end{algorithm}

\begin{algorithm}[tbp]
\caption{\label{alg:client-check-first-id-token} Relation of a Client $R^c$ -- Check first ID token (without Login).}
\begin{algorithmic}[1]
\Function{$\mathsf{CHECK\_FIRST\_ID\_TOKEN}$}{$\mi{sessionId}$, $\mi{id\_token}$, $\mi{code}$, $s'$}
  \Comment{\textbf{Check id token validity.}}
  \Let{$\mi{session}$}{$s'.\str{sessions}[\mi{sessionId}]$}
  \Let{$\mi{identity}$}{$\mi{session}[\str{identity}]$}
  \Let{$\mi{issuer}$}{$s'.\str{issuerCache}[\mi{identity}]$}
  \Let{$\mi{oidcConfig}$}{$s'.\str{oidcConfigCache}[\mi{issuer}]$}
  \Let{$\mi{credentials}$}{$s'.\str{clientCredentialsCache}[\mi{issuer}]$}
  \Let{$\mi{jwks}$}{$s'.\str{jwksCache}[\mi{issuer}]$}
  \Let{$\mi{data}$}{$\mathsf{extractmsg}(\mi{id\_token})$}
  \If{$\mi{data}[\str{s\_hash}] \not \equiv \mathsf{hash}(\mi{session}[\str{state}])$} 
    \Stop \ParboxComment{Check state hash.} 
  \EndIf
  \If{$\mi{data}[\str{c\_hash}] \not \equiv \mathsf{hash}(\mi{code})$}  \label{c-check-first-idt-check-c-hash}
    \Stop \ParboxComment{Check code hash.}
  \EndIf
  \label{line:client-start-first-id-token-checks}
  \If{$\mi{data}[\str{iss}] \not\equiv \mi{issuer} 
      \vee \mi{data}[\str{aud}] \not\equiv \mi{credentials}[\str{client\_id}]$} \label{c-check-first-idt-aud}
    \Stop \ParboxComment{Check the issuer and audience.}
  \EndIf
  \If{$\mathsf{checksig}(\mi{id\_token},\mi{jwks}) \not\equiv \top$} \label{c-check-first-idt-sig}
    \Stop \ParboxComment{Check the signature of the id token.}
  \EndIf
  \If{$\mi{data}[\str{nonce}] \not\equiv \mi{session}[\str{nonce}]$}
      \Stop    \ParboxComment{Check nonce.} 
  \EndIf

  \CallFun{PREPARE\_TOKEN\_REQUEST}{$\mi{sessionId}$, $\mi{code}$, $s'$} \label{c-second-PTR}
\EndFunction
\end{algorithmic} %
\end{algorithm}

\begin{algorithm}[tbp]
\caption{\label{alg:client-check-id-token} Relation of a Client $R^c$ -- Check ID token.}
\begin{algorithmic}[1]
\Function{$\mathsf{CHECK\_ID\_TOKEN}$}{$\mi{sessionId}$, $\mi{id\_token}$, $s'$}
  \Comment{\textbf{Check id token and create service session.}}
  \Let{$\mi{session}$}{$s'.\str{sessions}[\mi{sessionId}]$} \label{c-CID-retrieve-session}
  \Let{$\mi{identity}$}{$\mi{session}[\str{identity}]$} \label{c-cit-chooses-identity}
  \Let{$\mi{issuer}$}{$s'.\str{issuerCache}[\mi{identity}]$} \label{c-cit-chooses-issuer}
  \Let{$\mi{oidcConfig}$}{$s'.\str{oidcConfigCache}[\mi{issuer}]$}
  \Let{$\mi{credentials}$}{$s'.\str{clientCredentialsCache}[\mi{issuer}]$}
  \Let{$\mi{jwks}$}{$s'.\str{jwksCache}[\mi{issuer}]$} \label{c-check-id-token-choose-key}
  \Let{$\mi{data}$}{$\mathsf{extractmsg}(\mi{id\_token})$}
  \label{line:client-start-id-token-checks}
  \If{$\mi{data}[\str{iss}] \not\equiv \mi{issuer} 
       \vee \mi{data}[\str{aud}] \not\equiv \mi{credentials}[\str{client\_id}]$} \label{c-check-id-token-aud}
    \Stop \ParboxComment{Check the issuer and audience.}
  \EndIf
  \If{$\mathsf{checksig}(\mi{id\_token},\mi{jwks}) \not\equiv \top$} \label{c-CIT-sig}
    \Stop \ParboxComment{Check the signature of the id token.}
  \EndIf
  \If{$%
    \mi{data}[\str{nonce}] \not\equiv \mi{session}[\str{nonce}]$}
      \Stop    \ParboxComment{Check nonce.}  %
  \EndIf \\
  \ParboxComment{User is now logged in. Store user identity and issuer.}
  \Let{$s'.\str{sessions}[\mi{sessionId}][\str{loggedInAs}]$}{$\an{\mi{issuer}, \mi{data}[\str{sub}]}$}  \label{c-CIT-retrieve-sub}
  \If{$\mi{credentials}[\str{is\_app}] \equiv \bot$} \Comment{Send service session id to browser} \label{c-ssid-check-isApp}
      \Let{$s'.\str{sessions}[\mi{sessionId}][\str{serviceSessionId}]$}{$\nu_4$} \label{line:client-choose-service-session-id} %
      \Let{$\mi{request}$}{$\mi{session}[\str{redirectEpRequest}]$} \label{c-CIT-retrieves-request}
      \Comment{\parbox[t]{7cm}{Send to sender of request to the redirection endpoint. The request's meta data was stored in $\mathsf{PROCESS\_HTTPS\_REQUEST}$ (Algorithm~\ref{alg:client-fapi-http-request}). }}
      \Let{$\mi{headers}$}{$[\str{ReferrerPolicy}{:}\str{origin}]$}\\
      \ParboxComment{Create a cookie containing the service session id.}
      \Let{$\mi{headers}[\cSetCookie]$}{$[\an{\str{\_\_Secure}, \str{serviceSessionId}}{:}\an{\nu_4,\top,\top,\top}]$}\\ %
      \ParboxComment{Respond to browser's request to the redirection endpoint.}
      \Let{$m'$}{$\encs{\an{\cHttpResp, \mi{request}[\str{message}].\str{nonce}, 200, \mi{headers}, \str{ok}}}{\mi{request}[\str{key}]}$}
      \Stop{$\an{\an{\mi{request}[\str{sender}],\mi{request}[\str{receiver}],m'}}$, $s'$} \label{line:client-send-set-service-session}
  \Else \Comment{app client}
      \Stop{s'} \label{line:app-c-do-not-send-ssid}
   \EndIf

\EndFunction
\end{algorithmic} %
\end{algorithm}

\begin{algorithm}[tbp]
\caption{\label{alg:client-check-response-jws} Relation of a Client $R^c$ -- Check Response JWS.}
\begin{algorithmic}[1]
\Function{$\mathsf{CHECK\_RESPONSE\_JWS}$}{$\mi{sessionId}$, $\mi{responseJWS}$, $\mi{code}$, $s'$}
  \Comment{\textbf{Check validity of response JWS.}}
  \Let{$\mi{session}$}{$s'.\str{sessions}[\mi{sessionId}]$}
  \Let{$\mi{identity}$}{$\mi{session}[\str{identity}]$} \label{c-check-respJWS-select-identity-from-session}
  \Let{$\mi{issuer}$}{$s'.\str{issuerCache}[\mi{identity}]$} \label{c-check-respJWS-select-get-issuer-with-identity}
  \Let{$\mi{oidcConfig}$}{$s'.\str{oidcConfigCache}[\mi{issuer}]$}
  \Let{$\mi{credentials}$}{$s'.\str{clientCredentialsCache}[\mi{issuer}]$}
  \Let{$\mi{jwks}$}{$s'.\str{jwksCache}[\mi{issuer}]$} \label{c-check-respJWS-retrieve-jwks}
  \Let{$\mi{data}$}{$\mathsf{extractmsg}(\mi{responseJWS})$} \label{c-check-respJWS-retrieve-data-from-resp-jws}
  \If{$\mi{data}[\str{state}] \not \equiv \mathsf \mi{session}[\str{state}]$} 
    \Stop \ParboxComment{Check state hash.} 
  \EndIf
  \If{$\mi{data}[\str{code}] \not \equiv \mi{code}$} 
    \Stop \ParboxComment{Check code.}
  \EndIf
  \If{$\mi{data}[\str{iss}] \not\equiv \mi{issuer} 
      \vee \mi{data}[\str{aud}] \not\equiv \mi{credentials}[\str{client\_id}]$} \label{c-check-response-jws-aud}
    \Stop \ParboxComment{Check the issuer and audience.}
  \EndIf
  \If{$\mathsf{checksig}(\mi{responseJWS},\mi{jwks}) \not\equiv \top$} \label{c-check-response-jws-check-signature}
    \Stop \ParboxComment{Check the signature of the JWS.} 
  \EndIf

  \Let{$s'.\str{sessions}[\mi{sessionId}][\str{JARM\_iss}]$}{$\mi{data}[\str{iss}] $} \label{c-set-JARM-iss}

  \CallFun{PREPARE\_TOKEN\_REQUEST}{$\mi{sessionId}$, $\mi{code}$, $s'$} \label{c-second-PTR-JARM}
\EndFunction
\end{algorithmic} %
\end{algorithm}

\begin{algorithm}[h!]
\caption{\label{alg:script-client-index}Relation of $\mi{script\_client\_index}$ }
\begin{algorithmic}[1]
\Statex[-1] \textbf{Input:} $\langle\mi{tree}$, $\mi{docnonce}$, $\mi{scriptstate}$, $\mi{scriptinputs}$, $\mi{cookies}$, $\mi{localStorage}$, $\mi{sessionStorage}$, $\mi{ids}$, $\mi{secrets}\rangle$
\Comment{\textbf{Script that models the index page of a client.} Users can initiate the login flow or follow arbitrary links. The script receives various information about the current browser state, filtered according to the access rules (same origin policy and others) in the browser.} 

\LetND{$\mi{switch}$}{$\{\str{auth},\str{link}\}$}
\Comment{\parbox[t]{9cm}{Non-deterministically decide whether to start a login flow or to follow some link.}}

\If{$\mi{switch} \equiv \str{auth}$}
\ParboxComment{\textbf{Start login flow.}}

\Let{$\mi{url}$}{$\mathsf{GETURL}(\mi{tree},\mi{docnonce})$}
\ParboxComment{Retrieve own URL.}
\LetND{$\mi{id}$}{$\mi{ids}$}\label{line:script-client-index-select-id} \ParboxComment{Retrieve one of user's identities.}
\Let{$\mi{url'}$}{$\an{\tUrl, \https, \mi{url}.\str{host}, \str{/startLogin},
    \an{}, \an{}}$}
\ParboxComment{Assemble URL.}
\Let{$\mi{command}$}{$\an{\tForm, \mi{url'}, \mi{\mPost}, \mi{id}, \bot}$}

\ParboxComment{Post a form including the identity to the client.}

\State \textbf{stop} $\an{s,\mi{cookies},\mi{localStorage},\mi{sessionStorage},\mi{command}}$ \label{line:script-client-index-start-fapi-session} \Comment{\parbox[t]{4cm}{Finish script's run and instruct the browser to follow the command (follow post).}}

\Else  \ParboxComment{\textbf{Follow link.}}

  \LetND{$\mi{protocol}$}{$\{\http, \https\}$} 
  \ParboxComment{Non-deterministically select protocol (HTTP or HTTPS).}
  \LetND{$\mi{host}$}{$\dns$}
  \ParboxComment{Non-det. select host.}
  \LetND{$\mi{path}$}{$\mathbb{S}$}
  \ParboxComment{Non-det. select path.}
  \LetND{$\mi{fragment}$}{$\mathbb{S}$}
  \ParboxComment{Non-det. select fragment part.}
  \LetND{$\mi{parameters}$}{$\dict{\mathbb{S}}{\mathbb{S}}$} 
  \ParboxComment{Non-det. select parameters.}
  \Let{$\mi{url}$}{$\an{\cUrl, \mi{protocol}, \mi{host}, \mi{path}, \mi{parameters}, \mi{fragment}}$}
  \Comment{Assemble URL.}

  \Let{$\mi{command}$}{$\an{\tHref, \mi{url}, \bot, \bot}$}
  \ParboxComment{Follow link to the selected URL.}
  \State \textbf{stop} $\an{s,\mi{cookies},\mi{localStorage},\mi{sessionStorage},\mi{command}}$
      \Comment{\parbox[t]{4cm}{Finish script's run and instruct the browser to follow the command (follow link).}}
\EndIf

\end{algorithmic} %
\end{algorithm}

\begin{algorithm}[h!]
\caption{Relation of $\mi{script\_c\_get\_fragment}$ }\label{alg:script-c-get-fragment}
\begin{algorithmic}[1]
\Statex[-1] \textbf{Input:} $\langle\mi{tree}$, $\mi{docnonce}$, $\mi{scriptstate}$, $\mi{scriptinputs}$, $\mi{cookies}$, $\mi{localStorage}$, $\mi{sessionStorage}$, $\mi{ids}$, $\mi{secrets}\rangle$
\Let{$\mi{url}$}{$\mathsf{GETURL}(\mi{tree},\mi{docnonce})$}
\Let{$\mi{url'}$}{$\an{\tUrl, \https, \mi{url}.\str{host}, \str{/redirect\_ep},
    \an{},
    \an{}}$}
\Let{$\mi{command}$}{$\an{\tForm, \mi{url'}, \mi{\mPost}, \mi{url}.\str{fragment}, \bot}$}

\State \textbf{stop} $\an{s,\mi{cookies},\mi{localStorage},\mi{sessionStorage},\mi{command}}$

\end{algorithmic} %
\end{algorithm}

\clearpage
\FloatBarrier

\subsection{Authorization Servers}  \label{app:authorization-server}

Similar to Section I of Appendix F of \cite{FettKuestersSchmitz-TR-OIDC-2017}, an authorization server $\mi{as} \in \fAP{AS}$ is a web server modeled as
an atomic process $(I^\mi{as}, Z^\mi{as}, R^\mi{as}, s_0^\mi{as})$ with the addresses
$I^\mi{as} := \mapAddresstoAP(\mi{as})$. 

\begin{definition}\label{def:initial-state-as}
  A \emph{state $s\in Z^\mi{as}$ of an authorization server $\mi{as}$} is a term of the form
  $\langle\mi{DNSaddress}$, $\mi{pendingDNS}$, $\mi{pendingRequests}$,
  $\mi{corrupt}$, $\mi{keyMapping}$,
  $\mi{tlskeys}$, %
  $\mi{clients}$  (dict from nonces to terms),
  $\mi{records}$ (sequence of terms),
  $\mi{jwk}$ (signing key (only one)),
  $\mi{oautbEKM}$ (sequence of terms),
  $\mi{accessTokens}$ (sequence of issued access tokens)
  $\rangle$
  with
  $\mi{DNSaddress} \in \addresses$,
  $\mi{pendingDNS} \in \dict{\nonces}{\terms}$,
  $\mi{pendingRequests} \in \dict{\nonces}{\terms}$,
  $\mi{corrupt}\in\terms$, $\mi{keyMapping} \in \dict{\dns}{\terms}$ ,
  $\mi{tlskeys} \in \dict{\dns}{K_\text{TLS}}$ (all former components
  as in Definition~\ref{def:generic-https-state}),
  $\mi{clients}\in\dict{\nonces}{\terms}$, %
  $\mi{records}\in\terms$,
  $\mi{jwk}\in K_\text{sign}$,
  $\mi{oautbEKM} \in \terms$ and 
  $\mi{accessTokens}\in \terms$.

  An \emph{initial state $s^\mi{as}_0$ of $\mi{as}$} is a state of $\mi{as}$ with
  $s^\mi{as}_0.\str{pendingDNS} \equiv \an{}$,
  $s^\mi{as}_0.\str{pendingRequests} \equiv \an{}$,
  $s^\mi{as}_0.\str{corrupt} \equiv \bot$, $s^\mi{as}_0.\str{keyMapping}$ being
  the same as the keymapping for browsers,
  $s^\mi{as}_0.\str{tlskeys} \equiv \mi{tlskeys}^\mi{as}$,
  $s^\mi{as}_0.\str{records} \equiv \an{}$, 
  $s^\mi{as}_0.\str{jwk} \equiv \mathsf{signkey}(\mi{as})$,
  $s^c_0.\str{oautbEKM} \equiv \an{}$ and  
  $s^\mi{as}_0.\str{accessTokens} \equiv \an{}$.
\end{definition}

We require 
$s^\mi{as}_0.\str{clients}[\mi{clientId}]$ to be preconfigured and to contain information 
about the client with client id $\mi{clientId}$
regarding the profile ($\str{r}$ or $\str{rw}$, accessible via the key $\str{profile}$),
the type of the client ($\str{pub}$, $\str{conf\_JWS}$,$\str{conf\_MTLS}$ or $\str{conf\_OAUTB}$, 
accessible via $\str{client\_type}$), the client secret
(if the client type is $\str{conf\_JWS}$ or $\str{conf\_OAUTB}$,
accessible via $\str{client\_secret}$)
and
whether the client is an app client or a web server client
(either $\True$ or $\bot$, accessible via $\str{is\_app}$).  

For checking the signature of signed request JWTs,
we require that
$s^\mi{as}_0.\str{clients}[\mi{clientId}][\str{jws\_key}] \equiv \mathsf{pub}(s^c_0.\str{authReqSigKey})$. 

The resource servers that are supported by the authorization server
shall be contained in 
$s^\mi{as}_0.\str{resource\_servers} \subset \mi{dom}(\fAP{RS})$. %
For the corresponding resource servers $\mi{rs}$, we require 
that $s^\mi{rs}_0.\str{authServ} \in \mi{dom}(\mi{as})$.

The relation $R^\mi{as}$ is based on the model of generic HTTPS servers 
(see Section~\ref{sec:generic-https-server-model}). 
Algorithm~\ref{alg:as-fapi} specifies the algorithm responsible for processing
HTTPS requests. The script $\mathit{script\_as\_form}$ is defined
in Algorithm~\ref{alg:fapi-script-as-form} 

Furthermore, we require that $\mi{leak} \in \addresses$ is an arbitrary 
IP address. 

Table~\ref{tab:as-placeholder-list} shows a list of all placeholders
used in the algorithms of the authorization server.

 \begin{table}[tbh]
   \centering
   \begin{tabular}{|@{\hspace{1ex}}l@{\hspace{1ex}}|@{\hspace{1ex}}l@{\hspace{1ex}}|}\hline 
     \hfill Placeholder\hfill  &\hfill  Usage\hfill  \\\hline\hline
     $\nu_1$ & new authorization code  \\\hline
     $\nu_2$ & new access token \\\hline
     $\nu_3$ & nonce for mTLS \\\hline
     $\nu_4$ & nonce for OAUTB \\\hline
   \end{tabular}
   
   \caption{List of placeholders used in the authorization server algorithm.}
   \label{tab:as-placeholder-list}
 \end{table} ~

\vspace{1ex}
\noindent
\textbf{Differences to \cite{FettKuestersSchmitz-TR-OIDC-2017}:}
In the $\str{/auth2}$ path, the authorization server requires a
signed request JWS in all cases to prevent the attack described in Section~\ref{attack:pkce}.
The authorization server also checks and handles the OAUTB message send by the browser
if the client is a confidential client using OAUTB.
The authorization response always contains an authorization code. It also contains 
an id token if the client is a read-write client.
The leakage of the authorization response is modeled by sending the 
response to an arbitrary IP address in case of an app client.

The $\str{/token}$ path handles the token request. In case of a read-write client,
the access token is bound to the client either via mTLS or OAUTB.
To model the leakage of an access token due to phishing, the authorization server sends 
the access token to an arbitrary IP address in case of a read-write client.

The $\str{/MTLS\mhyphen{}prepare}$ and $\str{/OAUTB\mhyphen{}prepare}$ 
paths send the initial response needed for mTLS and OAUTB.

\vspace{1ex}
\noindent
\textbf{Remarks:} 
To model mTLS, the authorization server sends an encrypted nonce to the client. The client decrypts the message and
includes the nonce in the actual request.
This way, the client proves possession of the corresponding key.

As we do not model public key infrastructures, the client does not send a certificate to the authorization server
(as specified in Section 7.4.6 of \cite{rfc5246}). In general, the model uses the function $\mi{keyMapping}$,
which maps domains to the public key that would be contained in the certificate.

We require that the authorization server has a mapping from client ids to the corresponding public
keys in $s^\mi{as}_0.\str{clients}[\mi{clientId}][\str{mtls\_key}]$ (only if the client type is $\str{conf\_MTLS}$),
with the value $\mi{keyMapping}[\mi{dom_{c}}]$ with $\mi{dom_c} \in \mathsf{dom}(c)$.

As explained in Section~\ref{expl-single-at}, we do not model access tokens being issued from the authorization
endpoint. 

\begin{algorithm}
  \caption{\label{alg:as-fapi} Relation of AS $R^\mi{as}$ -- Processing HTTPS Requests}
\begin{algorithmic}[1]
\Function{$\mathsf{PROCESS\_HTTPS\_REQUEST}$}{$m$, $k$, $a$, $f$, $s'$}
  \If{$m.\str{path} \equiv \str{/auth}$}
    \If{$m.\str{method} \equiv \mGet$}
      \Let{$\mi{data}$}{$ m.\str{parameters}$}  %
    \EndIf
    \Let{$m'$}{$\encs{\an{\cHttpResp, m.\str{nonce}, 200, \an{\an{\str{ReferrerPolicy}, \str{origin}}}, \an{\str{script\_as\_form}, \mi{data}}}}{k}$}
    \Stop{\StopWithMPrime}\label{line:as-send-form}
  \State
  \ElsIf{$m.\str{path} \equiv \str{/auth2} \wedge m.\str{method} \equiv \mPost \wedge m.\str{headers}[\str{Origin}]  
  \equiv \an{m.\str{host}, \https}$}    \label{line:as-receive-auth2}
    \Let{$\mi{identity}$}{$m.\str{body}[\str{identity}]$} \label{line:as-auth2-id1}
    \Let{$\mi{password}$}{$m.\str{body}[\str{password}]$}
    \If{$\mi{identity}.\str{domain} \not\in \mathsf{dom}(\mi{as})$}
      \Stop{\DefStop}
    \EndIf
    \If{$\mi{password} \not\equiv \mapIDtoPLI(\mi{identity})$} \label{line:as-auth2-checks-secretofid}
      \Stop{\DefStop}
    \EndIf

    \Let{$\mi{clientId}$}{$m.\str{body}[\str{client\_id}]$} \label{as-auth2-client-id-from-body}
    \If{$\mi{clientId} \not\in s'.\str{clients}$}
      \Stop{\DefStop}
    \EndIf    

    \Let{$\mi{clientInfo}$}{$s'.\str{clients}[\mi{clientId}]$}
    \Let{$\mi{profile}$}{$\mi{clientInfo}[\str{profile}]$}
    \Let{$\mi{clientType}$}{$\mi{clientInfo}[\str{client\_type}]$}
    \Let{$\mi{isApp}$}{$\mi{clientInfo}[\str{is\_app}]$}

		\If{$ \str{request\_jws} \in m.\str{body}$} \label{as-auth2-check-jws-start}
		\Comment{Request must be as JWS}
	   	\Let{$\mi{requestJWS}$}{$m.\str{body}[\str{request\_jws}]$}
      \If{$\mathsf{checksig}(\mi{requestJWS}, \mi{clientInfo}[\str{jws\_key}] ) \not \equiv \True $ } \label{as-check-sig-of-jws}
	      \Stop{\DefStop}\Comment{wrong signature}
			\EndIf
	   	\Let{$\mi{requestData}$}{$\mathsf{extractmsg}(\mi{requestJWS})$} \label{as-retrieves-requestData-from-JWS}
	   	\If{$\mi{requestData}[\str{aud}] \not \equiv m.\str{host} $} %
	       \Stop{\DefStop}\Comment{wrong audience}
			\EndIf
	   	\If{$\mi{requestData}[\str{client\_id}] \not \equiv \mi{clientId}$} 
	       \Stop{\DefStop}\Comment{clientId not the same as in body} \label{as-auth2-check-jws-end}
			\EndIf
		\Else 
      \Stop{\DefStop}
		\EndIf
		\If{$\mi{profile} \equiv \str{rw} \wedge \mi{clientType} \equiv \str{conf\_OAUTB} \wedge \mi{isApp} \equiv \bot $} \Comment{Check both Token Binding messages}
        \If{$\str{Sec\mhyphen{}Token\mhyphen{}Binding} \not \in m.\str{headers}$} %
      		\Stop{\DefStop}
				\EndIf
  			\LetND{$\mi{ekmInfo}$}{$s'.\str{oautbEKM}$} 
		     		\Let{$\mi{TB\_Msg\_provided}$}{$m.\str{headers}[\str{Sec\mhyphen{}Token\mhyphen{}Binding}][\str{prov}]$} 
            \Let{$\mi{TB\_Msg\_referred}$}{$m.\str{headers}[\str{Sec\mhyphen{}Token\mhyphen{}Binding}][\str{ref}]$}  \label{as-auth2-get-ref-msg}

		     		\Let{$\mi{TB\_provided\_pub}$}{$\mi{TB\_Msg\_provided}[\str{id}]$} 
		     		\Let{$\mi{TB\_provided\_sig}$}{$\mi{TB\_Msg\_provided}[\str{sig}]$} 

            \Let{$\mi{TB\_referred\_pub}$}{$\mi{TB\_Msg\_referred}[\str{id}]$}  \label{as-auth2-get-ref-id}
		     		\Let{$\mi{TB\_referred\_sig}$}{$\mi{TB\_Msg\_referred}[\str{sig}]$} 

    			\If{$\mathsf{checksig}(\mi{TB\_provided\_sig}, \mi{TB\_provided\_pub}) \not \equiv \True $}
									\Stop{\DefStop}
					\EndIf
    			\If{$\mathsf{extractmsg}(\mi{TB\_provided\_sig}) \not \equiv \mi{ekmInfo}$} %
									\Stop{\DefStop}
					\EndIf

    			\If{$\mathsf{checksig}(\mi{TB\_referred\_sig}, \mi{TB\_referred\_pub}) \not \equiv \True $}
									\Stop{\DefStop}
					\EndIf

    			\If{$\mathsf{extractmsg}(\mi{TB\_referred\_sig}) \not \equiv \mi{ekmInfo} $} %
									\Stop{\DefStop}
					\EndIf

    			\Let{$s'.\str{oautbEKM}$}{$s'.\str{oautbEKM} - \mi{ekmInfo}$} 
		\EndIf

    \Let{$\mi{responseType}$}{$\mi{requestData}[\str{response\_type}]$}
    \Let{$\mi{redirectUri}$}{$\mi{requestData}[\str{redirect\_uri}]$}
    \Let{$\mi{state}$}{$\mi{requestData}[\str{state}]$}
    \Let{$\mi{nonce}$}{$\mi{requestData}[\str{nonce}]$}

    \If{$\mi{state} \equiv \an{}$}
      \Stop{\DefStop} \Comment{state must be included}
    \EndIf

    \If{$\mi{redirectUri} \not\inPairing \mi{clientInfo}[\str{redirect\_uris}]$}
      \Stop{\DefStop}
    \EndIf

    \Let{$\mi{record}$}{$[ \str{client\_id} : \mi{clientId} ] $} \Comment{Save data in record}
    \Let{$\mi{record}[\str{redirect\_uri}]$}{$\mi{redirectUri}$}
    \Let{$\mi{record}[\str{subject}]$}{$\mi{identity}$} \label{line:as-auth2-id2}
    \Let{$\mi{record}[\str{issuer}]$}{$m.\str{host}$} \label{as-rec-set-iss}

\algstore{myalg}
\end{algorithmic}
\end{algorithm}

\begin{algorithm}                     
\begin{algorithmic} [1]                   %
\algrestore{myalg}

    \Let{$\mi{record}[\str{nonce}]$}{$\mi{nonce}$}
		\Let{$\mi{record}[\str{scope}]$}{$\mi{requestData[\str{scope}]}$}	
    \Let{$\mi{record}[\str{response\_type}]$}{$\mi{responseType}$}	
    \Let{$\mi{record}[\str{code}]$}{$\nu_1$} \label{as-creates-code}
    \Let{$\mi{record}[\str{access\_token}]$}{$\nu_2$} \Comment{Access token for token request} \label{as-creates-at}

    \If{$\mi{profile} \equiv \str{r} \vee \mi{clientType} \equiv \str{pub} \vee 
            (\mi{clientType} \equiv \str{conf\_OAUTB} \wedge \mi{isApp} \equiv \True)$} 
            \Comment{\parbox[t]{4cm}{Save PKCE challenge (unless in rw profile with mTLS)}}
    		\Let{$\mi{record}[\str{pkce\_challenge}]$}{$\mi{requestData}[\str{pkce\_challenge}]$} \label{line:as-auth2-requestData-read-2}
		\ElsIf{$\mi{clientType} \equiv \str{conf\_OAUTB} \wedge \mi{isApp} \equiv \bot $} 
			\Let{$\mi{record}[\str{pkce\_challenge}]$}{$\mi{TB\_referred\_pub}$} \label{line:as-auth2-conf-OAUTB-pkce-cc}
		\EndIf

    \Append{$\mi{record}$}{$s'.\str{records}$} \label{as-adds-record}

    \Let{$\mi{responseData}$}{$[\str{code}{:} \nu_1]$} \Comment{Always send code}

    \If{$(\mi{profile} \equiv \str{rw} \wedge \mi{responseType} \not \in \{\an{\str{code}, \str{id\_token}},
        \an{\str{JARM\_code}} \}) 
          \vee (\mi{profile} \equiv \str{r} \wedge \mi{responseType} \not \equiv \an{code}) $} \label{auth-ep-check-response-type}
       \Stop{\DefStop}
 		\EndIf

    \If{$\mi{responseType} \equiv \an{\str{JARM\_code}}$} \label{auth-ep-JARM-branch}
      \Let{$\mi{responseJWT}$}{$[ \str{iss}: \mi{record}[\str{issuer}], \str{aud} : \mi{record}[\str{client\_id}], 
        \str{code}{:}\mi{record}[\str{code}], 
				$ \breakalgohook{8} $
        \str{at\_hash}{:}\mathsf{hash}(\mi{record}[\str{access\_token}]),
				\str{state}{:}\mi{state}  ]$} \label{line:as-create-response-JWT-auth2}
      \Let{$\mi{responseData}[\str{responseJWS}]$}{$\sig{\mi{responseJWT}}{s'.\str{jwk}}$} \label{as-creates-response-jws}
    \EndIf
    \If{$\str{id\_token} \inPairing \mi{responseType}$} \label{auth-ep-id-token-branch}
      \Let{$\mi{idTokenBody}$}{$[ \str{iss}: \mi{record}[\str{issuer}], \str{sub} :\mi{record}[\str{subject}],$ \breakalgohook{9} $\str{aud} : \mi{record}[\str{client\_id}], \str{nonce} : \mi{record}[\str{nonce}],
				$ \breakalgohook{9} $
        \str{c\_hash}{:}\mathsf{hash}(\mi{record}[\str{code}]), 
				\str{s\_hash}{:}\mathsf{hash}(\mi{state})  ]$} \label{line:as-create-id-token-auth2}
      \Let{$\mi{responseData}[\str{id\_token}]$}{$\sig{\mi{idTokenBody}}{s'.\str{jwk}}$} 

    \EndIf

    \Let{$\mi{responseData}[\str{state}]$}{$\mi{state}$}
    \If{$\mi{responseType} \in \{ \an{\str{code}}, \an{\str{JARM\_code}}  \} $} \Comment{Authorization code mode}
      \Let{$\mi{redirectUri}.\str{parameters}$}{$\mi{redirectUri}.\str{parameters} \cup \mi{responseData}$}
    \Else \Comment{Hybrid Mode}
      \Let{$\mi{redirectUri}.\str{fragment}$}{$\mi{redirectUri}.\str{fragment} \cup \mi{responseData}$}
    \EndIf
    \Let{$m'$}{$\encs{\an{\cHttpResp, m.\str{nonce}, 303, \an{\an{\str{Location}, \mi{redirectUri}}}, \an{}}}{k}$}

    \If{$\mi{clientInfo}[\str{is\_app}] \equiv \True$} \Comment{Leakage of authorization response}
        \Stop{$\an{
          \an{\mi{leak}, a, \an{\str{LEAK}, \mi{clientId}, \an{\str{Location}, \mi{redirectUri}}}},
          \an{f, a, m'}
        }$, $s'$}   \label{line:as-send-auth-resp-leak} 
    \Else \Comment{$\str{is\_app} \equiv \bot$}
      \Stop{\StopWithMPrime} \label{line:as-send-auth-resp}
    \EndIf

  \ElsIf{$m.\str{path} \equiv \str{/token} \wedge m.\str{method} \equiv \mPost$}\label{line:as-token-ep} %
   	\Let{$\mi{clientId}$}{$m.\str{body}[\str{client\_id}]$} \label{as-token-ep-cid-from-body}

    \Let{$\mi{code}$}{$m.\str{body}[\str{code}]$} \label{as-token-ep-get-code-from-body}
    \LetST{$\mi{record}$, $\mi{ptr}$}{$\mi{record} \equiv s'.\str{records}.\mi{ptr} \wedge \mi{record}[\str{code}] \equiv \mi{code}  
 							$\breakalgohook{5} $
              \wedge \mi{code} \not\equiv \bot$}{\textbf{stop}\ \DefStop} \label{as-token-ep-retrieves-record}
    \If{$\mi{record}[\str{client\_id}] \not\equiv \mi{clientId}$} \label{line:tokenep-checks-clientId}
      \Stop{\DefStop}
    \EndIf

    \Let{$\mi{clientInfo}$}{$s'.\str{clients}[\mi{clientId}]$}
    \Let{$\mi{profile}$}{$\mi{clientInfo}[\str{profile}]$}
   	\Let{$\mi{clientType}$}{$\mi{clientInfo}[\str{client\_type}]$} 

    \If{$\mi{profile} \equiv \str{rw} \wedge (\mi{clientType} \equiv \str{pub} \vee  \mi{clientType} \equiv \str{conf\_OAUTB})$}  \Comment{\parbox[t]{3cm}{Check both Token Binding messages}}
  				\LetND{$\mi{ekmInfo}$}{$s'.\str{oautbEKM}$} 

		     		\Let{$\mi{TB\_Msg\_provided}$}{$m.\str{headers}[\str{Sec\mhyphen{}Token\mhyphen{}Binding}][\str{prov}]$} 
		     		\Let{$\mi{TB\_provided\_pub}$}{$\mi{TB\_Msg\_provided}[\str{id}]$} \label{as-token-ep-get-prov-id}
		     		\Let{$\mi{TB\_provided\_sig}$}{$\mi{TB\_Msg\_provided}[\str{sig}]$} \label{as-token-ep-get-prov-sig}

    			\If{$\mathsf{checksig}(\mi{TB\_provided\_sig}, \mi{TB\_provided\_pub}) \not \equiv \True  
 					     $\breakalgohook{3} $
    			     \vee \mathsf{extractmsg}(\mi{TB\_provided\_sig}) \not \equiv \mi{ekmInfo} $} \label{as-token-ep-check-prov-tb-msg}
                  \Stop{\DefStop} \Comment{Wrong signature or ekm value}
					\EndIf

	     		\Let{$\mi{TB\_Msg\_referred}$}{$m.\str{headers}[\str{Sec\mhyphen{}Token\mhyphen{}Binding}][\str{ref}]$} \label{as-token-ep-ref-tb-key-1}
	     		\Let{$\mi{TB\_referred\_pub}$}{$\mi{TB\_Msg\_referred}[\str{id}]$} \label{as-token-ep-ref-tb-key-2}
	     		\Let{$\mi{TB\_referred\_sig}$}{$\mi{TB\_Msg\_referred}[\str{sig}]$}

    			\If{$\mathsf{checksig}(\mi{TB\_referred\_sig}, \mi{TB\_referred\_pub}) \not \equiv \True 
 					     $\breakalgohook{3} $
    			     \vee \mathsf{extractmsg}(\mi{TB\_referred\_sig}) \not \equiv \mi{ekmInfo} $}
                  \Stop{\DefStop} \Comment{Wrong signature or ekm value}
					\EndIf

					\Let{$s'.\str{oautbEKM}$}{$s'.\str{oautbEKM} - \mi{ekmInfo}$} 
		\EndIf

    \If{$\mi{clientType} \equiv \str{conf\_JWS} \vee \mi{clientType} \equiv \str{conf\_OAUTB}$} \Comment{Check JWS} \label{branch:token-ep-check-jws}
    	\Let{$\mi{clientSecret}$}{$\mi{clientInfo}[\str{client\_secret}]$}
      \If{$\mathsf{checkmac}(m.\str{body}[\str{assertion}], \mi{clientSecret}) \not \equiv \True$} \label{line:token-ep-check-jws-sig}
					\Stop{\DefStop} \Comment{Invalid MAC} \label{line:stop-tokenep-jws-mac}
			\EndIf

			\Let{$\mi{assertion}$}{$\mathsf{extractmsg}(m.\str{body}[\str{assertion}])$} 

			\If{$\mi{assertion}[\str{aud}]  \not \equiv m.\str{host}
          \vee \mi{assertion}[\str{iss}]  \not \equiv \mi{clientId}$} \label{line:token-ep-check-jws-values}
            \Stop{\DefStop} \Comment{Invalid audience or clientId} \label{line:stop-tokenep-jws-aud}
			\EndIf

\algstore{myalg2}
\end{algorithmic}
\end{algorithm}

\begin{algorithm}                     
\begin{algorithmic} [1]                   %
\algrestore{myalg2}

    \ElsIf{$\mi{clientType} \equiv \str{conf\_MTLS} $} 
      \LetND{$\mi{mtlsInfo}$}{$s'.\str{mtlsRequests}[\mi{clientId}]$}  \label{as-token-ep-retrieve-mtlsInfo}
      \If{$\mi{mtlsInfo}.1 \not \equiv m.\str{body}[\str{TLS\_AuthN}]$} \label{as-token-ep-check-mtls}
					\Stop{\DefStop}
			\EndIf
			\Let{$s'.\str{mtlsRequests}[\mi{clientId}]$}{$s'.\str{mtlsRequests}[\mi{clientId}] - \mi{mtlsInfo}$}
		\EndIf

   	\If{$\mi{profile} \equiv \str{r}$} %
      \If{$\mathsf{hash}(m.\str{body}[\str{pkce\_verifier}]) \not \equiv \mi{record}[\str{pkce\_challenge}]$} \label{line:as-checks-r-pkce}
  			\Stop{\DefStop}
 			\EndIf
  	\ElsIf{$\mi{profile} \equiv \str{rw} \wedge \mi{clientType} \equiv \str{conf\_OAUTB} \wedge \mi{isApp} \equiv \bot $}
			\If{$m.\str{body}[\str{pkce\_verifier}] \not \equiv 
 								\mi{record}[\str{pkce\_challenge}]$} \Comment{Sec. 5.2 of [OAUTB]} \label{line:as-checks-conf-OAUTB-pkce}
 			\Stop{\DefStop}
 		\EndIf

    	\ElsIf{$\mi{profile} \equiv \str{rw} \wedge (\mi{clientType} \equiv \str{pub} \vee 
        (\mi{clientType} \equiv \str{conf\_OAUTB} \wedge \mi{isApp} \equiv \True ))$}
							\If{$\mathsf{hash}(\mi{TB\_provided\_pub}) \not \equiv \mi{record}[\str{pkce\_challenge}]$} \Comment{Sec. 5.1 of [OAUTB]} \label{as-rw-pub-check-pkce}
									\Stop{\DefStop}
							\EndIf
    	\EndIf %

    \If{\textbf{not} $( \mi{record}[\str{redirect\_uri}] \equiv m.\str{body}[\str{redirect\_uri}] \vee 
   					$\breakalgohook{5}$
            (|\mi{clientInfo}[\str{redirect\_uris}]| = 1 \wedge \str{redirect\_uri} \not\in m.\str{body}))$} \label{as-check-redir-uri-at-token-endpoint}
      \Stop{\DefStop} \Comment{If only one redirect URI is registered, it can be omitted.}
    \EndIf

    \Let{$s'.\str{records}.\mi{ptr}[\str{code}]$}{$\bot$} \Comment{Invalidate code} \label{as-invalidates-code}

    \Let{$\mi{accessToken}$}{$\mi{record}[\str{access\_token}]$} \label{as-tep-get-at-from-record}

		\If{$\mi{profile} \equiv \str{rw}$} \label{as-adds-sequence-for-rw-client}
		\Comment{Create token binding}
     	\If{$\mi{clientType} \equiv \str{conf\_MTLS}$} 
  		     	\Let{$s'.\str{accessTokens}$}{$s'.\str{accessTokens} \plusPairing 
   					$\breakalgohook{5}$
            \an{\str{MTLS}, \mi{record}[\str{subject}], \mi{clientId}, \mi{accessToken}, \mi{mtlsInfo}.2 , \str{rw}}$} \label{as-creates-mtls-token-binding}
     	\Else \Comment{OAUTB}
  		     	\Let{$s'.\str{accessTokens}$}{$s'.\str{accessTokens} \plusPairing  
   					$\breakalgohook{5}$
            \an{\str{OAUTB}, \mi{record}[\str{subject}], \mi{clientId}, \mi{accessToken}, \mi{TB\_referred\_pub}, \str{rw}}$} \label{as-creates-oautb-token-binding}
  					
      \EndIf %

		\Else \label{as-adds-sequence-for-r-client}
            \Let{$s'.\str{accessTokens}$}{$s'.\str{accessTokens} \plusPairing \an{\mi{record}[\str{subject}], \mi{clientId}, \mi{accessToken}, \str{r}}$} \label{line:as-creates-r-at-sequence}
		\EndIf

		\If{$\str{openid} \inPairing \mi{record}[\str{scope}] \vee 
        (\mi{record}[\str{response\_type}] \equiv 
        \an{\str{code}, \str{id\_token}})$} \hspace{-0.3cm} \Comment{\parbox[t]{3cm}{return id token}}
		    \Let{$\mi{idTokenBody}$}{$[ \str{iss}: \mi{record}[\str{issuer}] ]$} \label{line:as-create-id-token-from-code}
		    \Let{$\mi{idTokenBody}[\str{sub}]$}{$\mi{record}[\str{subject}]$} \label{as-token-ep-add-sub-to-id-token}
		    \Let{$\mi{idTokenBody}[\str{aud}]$}{$\mi{record}[\str{client\_id}]$} \label{as-token-ep-add-client-id-to-id-token}
		    \Let{$\mi{idTokenBody}[\str{nonce}]$}{$\mi{record}[\str{nonce}]$}
		    \Let{$\mi{idTokenBody}[\str{at\_hash}]$}{$\mathsf{hash}(\mi{accessToken})$} \Comment{Mitigate reuse of phished AT}
		
		    \Let{$\mi{idToken}$}{$\sig{\mi{idTokenBody}}{s'.\str{jwk}}$} \label{as-token-ep-create-id-token}
    		\Let{$m'$}{$\encs{\an{\cHttpResp, m.\str{nonce}, 200, \an{}, [\str{access\_token}{:} \mi{accessToken}, %
						\str{id\_token}{:} \mi{idToken}]}}{k}$}      
		\Else
    		\Let{$m'$}{$\encs{\an{\cHttpResp, m.\str{nonce}, 200, \an{}, [\str{access\_token}{:} \mi{accessToken}] %
				}}{k}$}      
		\EndIf

		\If{$\mi{profile} \equiv \str{rw}$}
        \Stop{$\an{
          \an{\mi{leak}, a, \an{\str{LEAK}, \mi{clientId}, \mi{accessToken}}},
          \an{f, a, m'}
        }$, $s'$} \Comment{Leakage of access token} \label{as-send-token-ep-readwrite}
    \Else \Comment{$\mi{profile} \equiv \str{r}$}
      \Stop{\StopWithMPrime} \label{as-send-token-ep-read}
    \EndIf

	\State
  \ElsIf{$m.\str{path} \equiv \str{/MTLS\mhyphen{}prepare}$}\label{line:as-mtls-ep} 
    	\Let{$\mi{clientId}$}{$m.\str{body}[\str{client\_id}]$} 
      \Let{$\mi{mtlsNonce}$}{$\mi{\nu_3}$} \label{line:as-creates-mtls-nonce}
      \Let{$\mi{clientKey}$}{$s'.\str{clients}[\str{clientId}][\str{mtls\_key}]$} \label{line:as-choses-mtls-key}
			\Let{$s'.\str{mtlsRequests}[\mi{clientId}]$}
              {$s'.\str{mtlsRequests}[\mi{clientId}] \plusPairing \an{\mi{mtlsNonce}, \mi{clientKey} }$} \label{as-add-mtlsInfo}
      \Let{$m'$}{$\encs{\an{\cHttpResp, m.\str{nonce}, 200, \an{}, \mathsf{enc_a}(\an{\mi{mtlsNonce}, \mi{keyMapping}(m.\str{host})},
      \mi{clientKey})}}{k}$}      
			\Stop{\StopWithMPrime} \label{as-send-mtls-prepare-response}
	\State
  \ElsIf{$m.\str{path} \equiv \str{/OAUTB\mhyphen{}prepare}$}
            \Let{$\mi{tbNonce}$}{$\mi{\nu_4}$}
						\Let{$s'.\str{oautbEKM}$} 
								{$s'.\str{oautbEKM} \plusPairing \mathsf{hash}(\an{m.\str{nonce}, \mi{tbNonce},
											\mi{keyMapping}[m.\str{host}] }) $} \label{as-create-and-add-oautbEKM}
            \State \Comment{Own public key is needed for modeling Extended Master Secret}
						\Let{$m'$}{$\encs{\an{\cHttpResp, m.\str{nonce}, 200, \an{}, [\str{tb\_nonce}{:}\mi{tbNonce}]}}{k}$}      
						\Stop{\StopWithMPrime} \label{as-send-oautb-prepare-response}
  \EndIf
\EndFunction
\end{algorithmic} %
\end{algorithm}

\begin{algorithm}
\caption{\label{alg:fapi-script-as-form} Relation of $\mi{script\_as\_form}$ }
\begin{algorithmic}[1]
\Statex[-1] \textbf{Input:} $\langle\mi{tree}$, $\mi{docnonce}$, $\mi{scriptstate}$, $\mi{scriptinputs}$, $\mi{cookies}$, $\mi{localStorage}$, $\mi{sessionStorage}$, $\mi{ids}$, $\mi{secrets}\rangle$
\Let{$\mi{url}$}{$\mathsf{GETURL}(\mi{tree},\mi{docnonce})$} 
\Let{$\mi{url}'$}{$\an{\tUrl, \https, \mi{url}.\str{host}, \str{/auth2},
    \an{}, \an{}}$}
\Let{$\mi{formData}$}{$\mi{scriptstate}$}
\LetND{$\mi{identity}$}{$\mi{ids}$} \label{line:script-as-form-select-id}
\LetND{$\mi{secret}$}{$\mi{secrets}$}
\Let{$\mi{formData}[\str{identity}]$}{$\mi{identity}$}
\Let{$\mi{formData}[\str{password}]$}{$\mi{secret}$}
\Let{$\mi{command}$}{$\an{\tForm, \mi{url}', \mi{\mPost}, \mi{formData}, \bot}$}
\State \textbf{stop} $\an{s,\mi{cookies},\mi{localStorage},\mi{sessionStorage},\mi{command}}$

\end{algorithmic} %
\end{algorithm}

\FloatBarrier %

\clearpage
\FloatBarrier

\subsection{Resource Servers}  \label{app:rs}
 A resource server $\mi{rs} \in \fAP{RS}$ is a web server modeled as
 an atomic process $(I^\mi{rs}, Z^\mi{rs}, R^\mi{rs}, s_0^\mi{rs})$ with the addresses
 $I^\mi{rs} := \mapAddresstoAP(\mi{rs})$. The set of states $Z^\mi{rs}$ and the initial state
 $s^\mi{rs}_0$ of~$\mi{rs}$ are defined in the following.

\begin{definition}\label{def:initial-state-rs}
  A \emph{state $s\in Z^\mi{rs}$ of a resource server $\mi{rs}$} is a term of the form
  $\langle\mi{DNSaddress}$, $\mi{pendingDNS}$, $\mi{pendingRequests}$,
  $\mi{corrupt}$, $\mi{keyMapping}$,
  $\mi{tlskeys}$, %
  $\mi{mtlsRequests}$ (sequence of terms),
  $\mi{oautbEKM}$ (sequence of terms), 
  $\mi{rNonce}$ (dict from $\mi{\IDs}$ to sequence of nonces),
  $\mi{wNonce}$ (dict from $\mi{\IDs}$ to sequence of nonces),
	$\mi{ids}$ (sequence of ids),
	$\mi{authServ}$ (domain)$\rangle$ with
  $\mi{DNSaddress} \in \addresses$,
  $\mi{pendingDNS} \in \dict{\nonces}{\terms}$,
  $\mi{pendingRequests} \in \dict{\nonces}{\terms}$,
  $\mi{corrupt}\in\terms$, 
  $\mi{keyMapping} \in \dict{\dns}{\terms}$,
  $\mi{tlskeys} \in \dict{\dns}{K_\text{TLS}}$ (all former components
  as in Definition~\ref{def:generic-https-state}),
  $\mi{mtlsRequests}\in \terms$,
  $\mi{oautbEKM}\in \terms$, 
  $\mi{rNonce}\in \dict{\IDs}{N_\text{r}}$, 
  $\mi{wNonce}\in \dict{\IDs}{N_\text{w}}$,
  $\mi{ids} \subset \IDs$ and
  $\mi{authServ}\in \dns$.
 
   An \emph{initial state $s^\mi{rs}_0$ of $\mi{rs}$} is a state of $\mi{rs}$ with
   $s^\mi{rs}_0.\str{pendingDNS} \equiv \an{}$, 
   $s^\mi{rs}_0.\str{pendingRequests} \equiv \an{}$,
   $s^\mi{rs}_0.\str{corrupt} \equiv \bot$, 
   $s^\mi{rs}_0.\str{keyMapping}$ being
   the same as the keymapping for browsers,
   $s^\mi{rs}_0.\str{tlskeys} \equiv \mi{tlskeys}^\mi{rs}$,
   $s^\mi{rs}_0.\str{mtlsRequests} \equiv \an{}$,
   $s^\mi{rs}_0.\str{oautbEKM} \equiv \an{}$, 
   $s^\mi{rs}_0.\str{ids}$ being the sequence of identities for which the resource server
   manages resources,
   $s^\mi{rs}_0.\str{rNonce}$ and 
   $s^\mi{rs}_0.\str{wNonce}$ being the set of nonces representing read and 
   write access to resources, where each set contains an infinite sequence of nonces for each
   $\mi{id} \in s^\mi{rs}_0.\str{ids}$,
   $s^\mi{rs}_0.\str{authServ}$ being the domain of the authorization server that the resource server
   supports.
\end{definition}

The relation $R^\mi{rs}$ is again based on the generic HTTPS
server model (see Section~\ref{sec:generic-https-server-model}), for which the algorithm
used for processing HTTP requests is defined in 
Algorithm~\ref{alg:rs-oidc}.

Table~\ref{tab:rs-placeholder-list} shows a list of placeholders used in the resource server
algorithm.

\begin{table}[tbh] %
  \centering
  \begin{tabular}{|@{\hspace{1ex}}l@{\hspace{1ex}}|@{\hspace{1ex}}l@{\hspace{1ex}}|}\hline 
    \hfill Placeholder\hfill  &\hfill  Usage\hfill  \\\hline\hline
    $\nu_1$ & new nonce for mTLS  \\\hline
    $\nu_2$ & new nonce for OAUTB \\\hline
  \end{tabular}
  
  \caption{List of placeholders used in the resource server algorithm.}
  \label{tab:rs-placeholder-list}
\end{table}

\vspace{1ex}
\noindent
\textbf{Description and Remarks:}
 A resource server has two paths for requesting access to resources, depending on the
 profile used in the authorization process. To simplify the protected resource access,
 we use two disjunct set of nonces, where the set $\mi{rNonce}$ represents read access to 
 a resource, and the set $\mi{wNonce}$ represents write access.

 As before, there are paths for requesting nonces required for mTLS and OAUTB, 
 as access tokens are bound to read-write clients.

 For checking the token binding of an access token when using mTLS, the nonce chosen by the resource server
 is encrypted with a public key sent by the client. In contrast to an authorization server, a resource server is not required to
 check the validity of certificates (e.g., by checking the certificate chain), as this was already done by the 
 authorization server that created the bound access token (Section 4.2 of \cite{ietf-oauth-mtls-09}). 

 Furthermore, we require that 
 $\forall \mi{id}\in s_0^{\mi{rs}}.\str{ids}: \mathsf{governor}(\mi{id}) \equiv s_0^{\mi{rs}}.\str{authServ}$,
 i.e., the resource server contains only resources of identities that are governed by the authorization
 server $s_0^{\mi{rs}}.\str{authServ}$.

\FloatBarrier

\begin{algorithm}[h!]
\caption{\label{alg:rs-oidc} Relation of RS $R^\mi{rs}$ -- Processing HTTPS Requests}
\begin{algorithmic}[1]
\Function{$\mathsf{PROCESS\_HTTPS\_REQUEST}$}{$m$, $k$, $a$, $f$, $s'$}
  \If{$m.\str{path} \equiv \str{/MTLS\mhyphen{}prepare}$}\label{line:rs-mtls} 
						\Let{$\mi{mtlsNonce}$}{$\nu_1$} \label{rs-chooses-mtls-nonce}
            \Let{$\mi{clientKey}$}{$m.\str{body}[\str{pub\_key}]$} \label{rs-get-pub-key-mtls} \ParboxComment{Certificate is not required to be checked \cite[Section 4.2]{ietf-oauth-mtls-09}}
						\Let{$s'.\str{mtlsRequests}$}{$s'.\str{mtlsRequests} \plusPairing \an{\mi{mtlsNonce}, \mi{clientKey}} $} \label{rs-adds-mtlsreq}  %
						\Let{$m'$}{$\encs{\an{\cHttpResp, m.\str{nonce}, 200, \an{}, \mathsf{enc_a}(\an{\mi{mtlsNonce}, \mi{keyMapping}(m.\str{host})},\mi{clientKey})}}{k}$}  \label{rs-encrypts-mtlsnonce}
						\Stop{\StopWithMPrime} \label{rs-send-mtls-nonce}
     \State
  \ElsIf{$m.\str{path} \equiv \str{/OAUTB\mhyphen{}prepare}$}\label{line:rs-oautb} 
						\Let{$\mi{tbNonce}$}{$\nu_2$} 
						\Let{$s'.\str{oautbEKM}$} 
								{$s'.\str{oautbEKM} \plusPairing \mathsf{hash}(\an{m.\str{nonce}, \mi{tbNonce}, 
										\mi{keyMapping}[m.\str{host}] }) $}  \label{rs-creates-and-adds-oautbEKM-to-state}
						\Let{$m'$}{$\encs{\an{\cHttpResp, m.\str{nonce}, 200, \an{}, [\str{tb\_nonce}{:} \mi{tbNonce}]	}}{k}$}      
						\Stop{\StopWithMPrime} 
     \State
  \ElsIf{$m.\str{path} \equiv \str{/resource\mhyphen{}r}$}
    \LetST{$\mi{id} \in s'.\mi{ids}$}{$\mathsf{check\_read\_AT}(\mi{id}, m.\mi{header}[\str{Authorization}], s'.\str{authServ}) \equiv \True$
						\breakalgohook{2}   }
						{\textbf{stop}\ \DefStop}  \label{rs-resource-r-check-at}
    \LetND{$\mi{resource}$}{$s'.\str{rNonce}[\mi{id}]$}
    \Let{$m'$}{$\encs{\an{\cHttpResp, m.\str{nonce}, 200, \an{}, [\str{resource}{:}\mi{resource}]}}{k}$}      
    \Stop{\StopWithMPrime} \label{rs-send-rNonce}
     \State
  \ElsIf{$m.\str{path} \equiv \str{/resource\mhyphen{}rw}$}
    \If{$\str{at\_iss} \not \in m.\str{body}$} \label{rs:check-if-at-iss-in-body}
        \Stop{\DefStop} 
      \DefStop
    \EndIf
		\If{$m.\str{body}[\str{at\_iss}] \not \equiv s'.\str{authServ}$} \label{rs-check-at-iss}
        \Stop{\DefStop} 
    \EndIf
		\If{$\str{MTLS\_AuthN} \in m.\str{body}$}
    	\LetST{$\mi{mtlsInfo}$}{$\mi{mtlsInfo} \inPairing s'.\str{mtlsRequests}
 					$\breakalgohook{5} $
          \wedge \mi{mtlsInfo}.1 \equiv m.\str{body}[\str{MTLS\_AuthN}]$}{\textbf{stop}\ \DefStop} \label{rs-retrieve-mtlsReq}
			\Let{$s'.\str{mtlsRequests}$}{$s'.\str{mtlsRequests} - \mi{mtlsInfo}$}
    \LetST{$\mi{id} \in s'.\mi{ids}$}{$\mathsf{check\_mtls\_AT}(\mi{id}, m.\mi{header}[\str{Authorization}], $\breakalgohook{3}$ \mi{mtlsInfo}.2, s'.\str{authServ}) \equiv \True$}		{\textbf{stop}\ \DefStop} \label{rs-check-mtls-at}
		\Else
  		\LetND{$\mi{ekmInfo}$}{$s'.\str{oautbEKM}$} \label{rs-rw-retrieve-oautbEKM} 
   		\Let{$\mi{TB\_Msg\_provided}$}{$m.\str{headers}[\str{Sec\mhyphen{}Token\mhyphen{}Binding}][\str{prov}]$} 
   		\Let{$\mi{TB\_provided\_pub}$}{$\mi{TB\_Msg\_provided}[\str{id}]$} 
   		\Let{$\mi{TB\_provided\_sig}$}{$\mi{TB\_Msg\_provided}[\str{sig}]$}

     \If{$\mathsf{checksig}(\mi{TB\_provided\_sig}, \mi{TB\_provided\_pub}) \not \equiv \True $} \label{rs-rw-check-oautb-sig}
							\Stop{\DefStop}
			\EndIf

			\If{$\mathsf{extractmsg}(\mi{TB\_provided\_sig}) \not \equiv \mi{ekmInfo} $} \label{rs-rw-check-oautb-ekm}
							\Stop{\DefStop}
			\EndIf

    \LetST{$\mi{id} \in s'.\mi{ids}$}{$\mathsf{check\_oautb\_AT}(\mi{id}, m.\mi{header}[\str{Authorization}], $\breakalgohook{2}$ 
            \mi{TB\_provided\_pub},
            s'.\str{authServ}) 
               \equiv \True$}		{\textbf{stop}\ \DefStop} \label{rs-check-oautb-binding} %

			\Let{$s'.\str{oautbEKM}$}{$s'.\str{oautbEKM} - \mi{ekmInfo}$} %
		\EndIf

    \LetND{$\mi{read}$}{$\{\True, \bot\}$}
    \If{$\mi{read} \equiv \True$}
		  \LetND{$\mi{resource}$}{$s'.\mi{rNonce}[\mi{id}]$}
    \Else
		  \LetND{$\mi{resource}$}{$s'.\mi{wNonce}[\mi{id}]$}
		\EndIf
		\Let{$m'$}{$\encs{\an{\cHttpResp, m.\str{nonce}, 200, \an{}, [\str{resource}{:}\mi{resource}]}}{k}$}      
		\Stop{\StopWithMPrime} \label{rs-send-wNonce}
	\EndIf
\EndFunction
\end{algorithmic} %
\end{algorithm}

\FloatBarrier %

\clearpage
\subsection{OpenID FAPI with Network Attacker} \label{fapi-web-system}

The formal model of the FAPI is based on web system as defined in Definition 27 of \cite{FettKuestersSchmitz-TR-OIDC-2017}.
Unless otherwise specified, we adhere to the terms as defined in \cite{FettKuestersSchmitz-TR-OIDC-2017}.

A web system 
$\fapiwebsystem=(\bidsystem, \scriptset, \mathsf{script}, E^0)$
is called a \emph{FAPI web system with a network attacker}. The components of the web system are defined 
in the following.

\begin{itemize}
\item $\bidsystem=\mathsf{Hon} \cup \mathsf{Net}$
consists of a network attacker process (in $\mathsf{Net}$),
a finite set $\fAP{B}$
of web browsers, a finite set $\fAP{C}$
of web servers for the clients, a finite set $\fAP{AS}$
of web servers for the authorization servers and
a finite set $\fAP{RS}$
of web servers for the resource servers, 
with
$\mathsf{Hon} := \fAP{B} \cup \fAP{C} \cup \fAP{AS} \cup \fAP{RS}$.
DNS servers are subsumed by the network attacker and are therefore not modeled explicitly.
\item $\scriptset$ contains the scripts shown in Table~\ref{tab:scripts-in-fapi}, with string representations defined
by the mapping $\mathsf{script}$.
\item $E^0$ contains only the trigger events as specified in
Definition 27 of \cite{FettKuestersSchmitz-TR-OIDC-2017}.
\end{itemize}

\begin{table}[htb]
  \centering
  \begin{tabular}{|@{\hspace{1ex}}l@{\hspace{1ex}}|@{\hspace{1ex}}l@{\hspace{1ex}}|}\hline 
    \hfill $s \in \scriptset$\hfill  &\hfill  $\mathsf{script}(s)$\hfill  \\\hline\hline
    $\Rasp$ & $\str{att\_script}$  \\\hline
    $\mi{script\_c\_index}$ & $\str{script\_c\_index}$  \\\hline
    $\mi{script\_c\_get\_fragment}$ & $\str{script\_get\_fragment}$  \\\hline
    $\mi{script\_as\_form}$ &  $\str{script\_as\_form}$  \\\hline
  \end{tabular}
  
  \caption{List of scripts in $\scriptset$ and their respective string
    representations.}
  \label{tab:scripts-in-fapi}
\end{table}

In addition to the set of nonces defined in \cite{FettKuestersSchmitz-TR-OIDC-2017}, we 
specify an infinite sequence of nonces $N_\text{r}$ representing read access to some resources and
an infinite sequence of nonces $N_\text{w}$ representing write access to some resources.
We call these nonces \emph{resource access nonces}.

  \FloatBarrier\clearpage
  \section{Definitions} \label{fapi-def}

In the following, we will define terms used within the analysis.
All other terms are used as defined in
\cite{FettKuestersSchmitz-TR-OIDC-2017} unless stated otherwise.

\begin{definition}[Read Client]\label{def:readClient}
  A client $\mi{c}$ with client id $\mi{clientId}$ issued from authorization server
  $\mi{as}$ is called \emph{read client} (w.r.t.~$\mi{as}$)  if 
  $s^\mi{as}_0.\str{clients}[\mi{clientId}][\str{profile}] \equiv \str{r}$.
\end{definition}

\begin{definition}[Read-Write Client]\label{def:readwriteClient}
  A client $\mi{c}$ with client id $\mi{clientId}$ issued from authorization server
  $\mi{as}$ is called \emph{read-write client} (w.r.t.~$\mi{as}$)  if 
  $s^\mi{as}_0.\str{clients}[\mi{clientId}][\str{profile}] \equiv \str{rw}$.
\end{definition}

\begin{definition}[Web Server Client]\label{def:webServerClient}
  A client $\mi{c}$ with client id $\mi{clientId}$ issued from authorization server
  $\mi{as}$ is called \emph{web server client} (w.r.t.~$\mi{as}$)  if 
  $s^\mi{as}_0.\str{clients}[\mi{clientId}][\str{is\_app}]  \equiv \bot$.
\end{definition}

A client is an \emph{app client} if it is not a web server client.

\begin{definition}[Client Type]\label{def:clientType}
  A client $\mi{c}$ with client id $\mi{clientId}$ issued from authorization server
  $\mi{as}$ is of type $t$ (w.r.t.~$\mi{as}$) if 
  $s^\mi{as}_0.\str{clients}[\mi{clientId}][\str{client\_type}] \equiv t$,
  for $t \in \{ \str{pub}, \str{conf\_JWS}, \str{conf\_MTLS}, \str{conf\_OAUTB} \}$.
\end{definition}

\begin{definition}[Confidential Client]\label{def:confClient}
  A client $\mi{c}$ with client id $\mi{clientId}$ issued from authorization server
  $\mi{as}$ is called \emph{confidential} (w.r.t.~$\mi{as}$)  if 
  it is of type 
  $\str{conf\_JWS}$, $\str{conf\_MTLS}$ or $\str{conf\_OAUTB}$.
\end{definition}

A client is a \emph{public} client if it is not a confidential client.\\

\textbf{Remarks: } As stated in Section~\ref{fapi-informal-description-of-the-model}, we assume that
for all authorization servers, a client is either a read client or a read-write client.
This also holds true for the client type and for being a web server client.
We also note that these properties do not change, as an honest authorization
server $\mi{as}$ never changes the values of $s^\mi{as}_0.\str{clients}$. %

\begin{definition}[Token Endpoint of Authorization Server]\label{def:token-ep}
  A message is sent to the token endpoint of the authorization server $\mi{as}$
  if it is sent to the URL
  $\an{\tUrl, \https, \mi{dom_{\mi{as}}}, \str{/token}, \mi{param}, \mi{frag}}$ 
  for $\mi{dom_{\mi{as}}} \in \mi{dom}(\mi{as})$ and arbitrary $\mi{param}$ and $\mi{frag}$.
\end{definition}

\begin{definition}[Bound Access Token]\label{def:bound-token}
  An access token $t$ issued from authorization server $\mi{as}$ is bound to 
  the client with client id $\mi{clientId}$ 
  in the configuration $(S, E, N)$  of a run 
  of a FAPI web system $\fapiwebsystem$
  with a network attacker, 
  if either 
  $\an{\str{MTLS}, \mi{id}, \mi{clientId}, \mi{t}, \mi{mtlsKey}, \str{rw}} \in S(\mi{as}).\str{accessTokens}$ or 
  $\an{\str{OAUTB}, \mi{id}, \mi{clientId}, \mi{t}, \mi{tbKey}, \str{rw}} \in S(\mi{as}).\str{accessTokens}$, 
  for some values of $\mi{id}, \mi{mtlsKey}$ and $\mi{tbKey}$.
\end{definition}

The access token is bound via mTLS or via OAUTB, depending on the first entry of the sequence.

\begin{definition}[Access Token associated with Client, Authorization Server and Identity]\label{def:AT-associated-with-c}
  Let $c$ be a client with client id $\mi{clientId}$ issued
  to $c$ by the authorization server $\mi{as}$, and let $\mi{id} \in \mathsf{ID}^\mi{as}$.
  We say that an \emph{access token $t$ is associated with $c$, $\mi{as}$ and $\mi{id}$}
  in state $S$ of the configuration $(S, E,N)$ 
  of a run $\rho$
  of a FAPI web system,
  if there is a sequence $s \in S(\mi{as}).\str{accessTokens}$
  such that
  $s \equiv \an{\mi{id}, \mi{clientId}, t, \str{r}}$, 
  $s \equiv \an{\str{MTLS}, \mi{id}, \mi{clientId}, t, \mi{key}, \str{rw}}$ or
  $s \equiv \an{\str{OAUTB}, \mi{id}, \mi{clientId}, t, \mi{key'}, \str{rw}}$,
  for some $\mi{key}$ and $\mi{key'}$.
\end{definition}

\subsection*{Validity of Access Tokens}

For checking the validity of an access token $t$ in state $S$ of a configuration
$(S, E, N)$  of a run,
we define the following functions,
where $\mathsf{dom}_\mi{as} \in \mathsf{dom}(\mi{as})$:

\begin{itemize}
\item $\mathsf{check\_read\_AT}(\mi{id}, t, \mathsf{dom}_\mi{as}) \equiv \True 
  \Leftrightarrow
      \exists \mi{clientId} $ s.t. $\an{\mi{id}, \mi{clientId}, t , \str{r}} \in S(\mi{as}).\str{accessTokens}$.

\item $\mathsf{check\_mtls\_AT}(\mi{id}, t, \mi{key}, \mathsf{dom}_\mi{as}) \equiv \True 
  \Leftrightarrow
      \exists \mi{clientId} $ s.t. $\an{\str{MTLS}, \mi{id}, \mi{clientId}, t , \mi{key}, \str{rw}} \in S(\mi{as}).\str{accessTokens}$.

\item $\mathsf{check\_oautb\_AT}(\mi{id}, t, \mi{key}, \mathsf{dom}_\mi{as}) \equiv \True 
  \Leftrightarrow
      \exists \mi{clientId} $ s.t. $\an{\str{OAUTB}, \mi{id}, \mi{clientId}, t , \mi{key}, \str{rw}} \in S(\mi{as}).\str{accessTokens}$.
\end{itemize}

  \FloatBarrier\clearpage
  \section{Security Properties} \label{fapi-formal-properties}

As the profiles of the FAPI are essentially regular OIDC
flows secured by additional mechanisms, they follow the same goals.
Therefore, the following security definitions
are similar to the definitions given in Appendix H of \cite{FettKuestersSchmitz-TR-OIDC-2017}.
Notably, the conditions under which these properties hold true
are not the same as in the case of regular OIDC, as
detailed in Section~\ref{sec-assumptions}.

We show that these properties hold true in
Theorem~\ref{thm:theorem-1}.

\section*{Authorization} \label{sec:fprop-authorization}

Intuitively, authorization means that an attacker should not be able
to get read or write access to a resource of an honest identity.

To capture this property, we extend the definition given in \cite{FettKuestersSchmitz-TR-OIDC-2017},
where the authorization property states that no access token 
associated with an honest identity may leak. As we assume that in the Read-Write
flow, access tokens may leak to an attacker, our definition states 
that honest resource servers may not provide access to resources of honest identities to 
the attacker.

More precisely, we require that if an honest resource server provides access to a resource
belonging to an honest user whose identity is governed by an honest authorization server,
then this access is not provided to the attacker. This includes the
case that the resource is not directly accessed by the attacker, but
also that no honest client provides the attacker access to such a resource.

\begin{definition}[Authorization Property]\label{def:property-authz-a} 
  We say that the FAPI web system
	with a network attacker \emph{$\fapiwebsystem$
  is secure w.r.t.~authorization} iff for every run $\rho$
  of $\fapiwebsystem$,
  every configuration $(S, E, N)$
  in $\rho$,
  every authorization server $\mi{as} \in \fAP{AS}$
  that is honest in $S$
  with $s^\mi{as}_0.\str{resource\_servers}$ being
  domains of honest resource servers,
  every identity $\mi{id} \in \mathsf{ID}^\mi{as}$ 
  with $b = \mathsf{ownerOfID}(\mi{id})$
  being an honest browser in $S$,
  every client $c \in \fAP{C}$ that is honest in $S$
  with client id $\mi{clientId}$ issued to $c$ by $\mi{as}$,
  every resource server $\mi{rs} \in \fAP{RS}$
  that is honest in $S$ such that
  $\mi{id} \in s^\mi{rs}_0.\str{ids}$,
  $s^\mi{rs}_0.\str{authServ} \in \mathsf{dom}(\mi{as})$ and
  with $\mi{dom}_\mi{rs} \in s_0^\mi{as}.\str{resource\_servers}$
  (with $\mi{dom}_\mi{rs} \in \mathsf{dom}(\mi{rs})$),
  every access token $t$ associated with $c$, $\mi{as}$ and $\mi{id}$
  and every resource access nonce $r \in s^\mi{rs}_0.\str{rNonce}[\mi{id}] \cup s^\mi{rs}_0.\str{wNonce}[\mi{id}]$
  it holds true that: 

  If $r$ is contained in a response to a request $m$ sent to $\mi{rs}$
  with $t \equiv m.\mi{header}[\str{Authorization}]$,
  then 
  $r$ is not derivable from the attackers knowledge in $S$
  (i.e., $r \not\in d_{\emptyset}(S(\fAP{attacker}))$).
\end{definition}

We require that the preconfigured domains of the resource servers
of an authorization server are domains of honest resource servers,
as otherwise, an access token of a read client can trivially leak to
the attacker.

\section*{Authentication}

Intuitively, an attacker should not be able to log in 
at an honest client under the identity of an honest user, where the identity
is governed by an honest authorization server. 
All relevant participants are required to be honest, as otherwise,
the attacker can trivially log in at a client, for example, if the attacker controls
the authorization server that governs the identity.

\begin{definition}[Service Sessions]\label{def:service-sessions}
  We say that there is a \emph{service session identified
  by a nonce $n$
  for an identity $\mi{id}$
  at some client $c$} in a configuration $(S, E, N)$
  of a run $\rho$
  of a FAPI web system
  iff there exists some session id $x$
  and a domain $d \in \mathsf{dom}(\mathsf{governor}(\mi{id}))$
  such that
  $S(r).\str{sessions}[x][\str{loggedInAs}] \equiv \an{d, \mi{id}}$
  and $S(r).\str{sessions}[x][\str{serviceSessionId}] \equiv n$.
\end{definition}

\begin{definition}[Authentication Property]\label{def:property-authn-a}
  We say that the FAPI web system
	with a network attacker \emph{$\fapiwebsystem$
  is secure w.r.t.~authentication} iff for every run $\rho$
  of $\fapiwebsystem$,
  every configuration $(S, E, N)$
  in $\rho$,
  every $c\in \fAP{C}$
  that is honest in $S$,
  every identity $\mi{id} \in \mathsf{ID}$
  with $\mi{as} = \mathsf{governor}(\mi{id})$ %
  being an honest AS
  and
  with $b = \mathsf{ownerOfID}(\mi{id})$
  being an honest browser in $S$,
  every service session identified by some nonce
  $n$
  for $\mi{id}$
  at $c$,
  $n$
  is not derivable from the attackers knowledge in $S$
  (i.e., $n \not\in d_{\emptyset}(S(\fAP{attacker}))$).
\end{definition}

\section*{Session Integrity for Authentication and Authorization}

There are two session integrity properties that capture that an honest
user should not be logged in under the identity of the attacker and
should not use resources of the attacker.

We first define notations for the processing steps that
represent important events during a flow of a FAPI web system, similar to
the definitions given in \cite{FettKuestersSchmitz-TR-OIDC-2017}.

\begin{definition}[User is logged in]\label{def:user-logged-in}
  For a run $\rho$
  of a FAPI web system with a network attacker $\fapiwebsystem$
  we say that a browser $b$
  was authenticated to a client $c$
  using an authorization server $\mi{as}$
  and an identity $u$
  in a login session identified by a nonce $\mi{lsid}$
  in processing step $Q$ in $\rho$ with
  $$Q = (S, E, N) \xrightarrow[r \rightarrow  
  E_{\text{out}}]{} (S', E', N')$$  (for some $S$, $S'$, $E$, $E'$, $N$, $N'$)
  and some event $\an{y,y',m} \in E_{\text{out}}$
  such that $m$
  is an HTTPS response matching an HTTPS request sent by $b$
  to $c$
  and we have that in the headers of $m$
  there is a header of the form
  $\an{\str{Set\mhyphen{}Cookie},
    [\an{\str{\_\_Secure}, \str{serviceSessionId}}{:}\an{\mi{ssid},\top,\top,\top}]}$ for
  some nonce $\mi{ssid}$
  such that 
  $S(c).\str{sessions}[\mi{lsid}][\str{serviceSessionId}] \equiv \mi{ssid}$
  and       $S(c).\str{sessions}[\mi{lsid}][\str{loggedInAs}] \equiv \an{d, u}$
  with $d \in \mathsf{dom}(\mi{as})$.
  We then write $\mathsf{loggedIn}_\rho^Q(b, c, u, \mi{as}, \mi{lsid})$.
\end{definition}

\begin{definition}[User started a login flow]
  For a run $\rho$
  of a FAPI web system with a network attacker $\fapiwebsystem$
  we say that the user of the browser $b$
  started a login session identified by a nonce $\mi{lsid}$
  at the client $c$
  in a processing step $Q$
  in $\rho$
  if (1) in that processing step, the browser $b$
  was triggered, selected a document loaded from an origin of $c$,
  executed the script $\mi{script\_client\_index}$
  in that document, and in that script, executed the
  Line~\ref{line:script-client-index-start-fapi-session} of
  Algorithm~\ref{alg:script-client-index}, and (2) $c$
  sends an HTTPS response corresponding to the HTTPS request sent by
  $b$
  in $Q$
  and in that response, there is a header of the form
  $\an{\str{Set\mhyphen{}Cookie},
    [\an{\str{\_\_Secure},\str{sessionId}}{:}\an{\mi{lsid},\top,\top,\top}]}$. We then write
  $\mathsf{started}_\rho^Q(b, c, \mi{lsid})$.
\end{definition}

\begin{definition}[User authenticated at an AS]
  For a run $\rho$
  of a FAPI web system with a network attacker $\fapiwebsystem$
  we say that the user of the browser $b$
  authenticated to an authorization server $\mi{as}$
  using an identity $u$
  for a login session identified by a nonce $\mi{lsid}$
  at the client $c$
  if there is a processing step $Q$
  in $\rho$   with
  $$Q = (S, E, N) \xrightarrow[]{}
  (S', E', N')$$ (for some $S$, $S'$, $E$, $E'$, $N$, $N'$) in which the browser $b$
  was triggered, selected a document loaded from an origin of $\mi{as}$,
  executed the script $\mi{script\_as\_form}$
  in that document, and in that script, (1) in
  Line~\ref{line:script-as-form-select-id} of
  Algorithm~\ref{alg:fapi-script-as-form}, selected the identity $u$,
  and (2) we have that the $\mi{scriptstate}$
  of that document, when triggered, contains a nonce $s$
  such that $\mi{scriptstate}[\str{state}] \equiv s$
  and $S(r).\str{sessions}[\mi{lsid}][\str{state}] \equiv s$.
  We then write
  $\mathsf{authenticated}_\rho^Q(b, c, u, \mi{as}, \mi{lsid})$.
\end{definition}

\begin{definition}[Resource Access]\label{def:browser-accesses-resource}
	   For a run $\rho$
	   of a FAPI web system with a network attacker $\fapiwebsystem$
	   we say that a browser $b$
	   accesses a resource of identity $u$
	   stored at resource server $\mi{rs}$
	   through the session of client $c$
     identified by the nonce $\mi{lsid}$
	   in processing step $Q$ in $\rho$ with
	   $$Q = (S, E, N) \xrightarrow{} (S', E', N')$$  (for some $S$, $S'$, $E$, $E'$, $N$, $N'$)
     with (1) 
    $\an{\an{\str{\_\_Secure}, \str{sessionid}}, \an{\mi{lsid}, y, z, z'}} \inPairing S(b).\str{cookies}[d]$
     for $d \in \mathsf{dom}(c)$,
     $y, z, z' \in \terms$,
     (2) $S(c).\str{sessions}[\mi{lsid}][\str{resource}] \equiv r$
	   with $r \in s^\mi{rs}_0.\str{rNonce}[\mi{u}] \cup s^\mi{rs}_0.\str{wNonce}[\mi{u}]$
     and (3) $S(c).\str{sessions}[\mi{lsid}][\str{resource\_server}] 
     \in \mathsf{dom}(\mi{rs})$.
     We then write $\mathsf{accessesResource}_\rho^Q(b, r, u, c, \mi{rs}, \mi{lsid})$.
\end{definition}

\subsection*{Session Integrity Property for Authentication for Web Server Clients with OAUTB}
This security property captures that, if the client is a web server
client with OAUTB, then (a) a user should only be logged
in when the user actually expressed the wish to start a FAPI flow
before, and (b) if a user expressed the wish to start a FAPI flow
using some honest authorization server and a specific identity, then user
is not logged in under a different identity.

\begin{definition} [Session Integrity for
  Authentication for Web Server Clients with OAUTB]\label{def:property-si-authn}
  Let $\fapiwebsystem$
  be an FAPI web system with a network attacker. We say that
  \emph{$\fapiwebsystem$
    is secure w.r.t.~session integrity for authentication} iff for
  every run $\rho$
  of $\fapiwebsystem$, every processing step $Q$ in $\rho$ with
  $$Q = (S, E, N) \xrightarrow[]{}
  (S', E', N')$$ (for some $S$,
  $S'$,
  $E$,
  $E'$,
  $N$,
  $N'$),
  every browser $b$
  that is honest in $S$,
  every $\mi{as} \in \fAP{AS}$,
  every identity $u$,
  every web server client $c\in \fAP{C}$
  of type $\str{conf\_OAUTB}$
  that is honest in $S$,
  every nonce $\mi{lsid}$,
  and $\mathsf{loggedIn}_\rho^Q(b, c, u, \mi{as}, \mi{lsid})$
  we have that (1) there exists a processing step $Q'$
  in $\rho$
  (before $Q$)
  such that $\mathsf{started}_\rho^{Q'}(b, c, \mi{lsid})$,
  and (2) if $\mi{as}$
  is honest in $S$,
  then there exists a processing step $Q''$
  in $\rho$
  (before $Q$)
  such that
  $\mathsf{authenticated}_\rho^{Q''}(b, c, u, \mi{as}, \mi{lsid})$.
\end{definition}

\subsection*{Session Integrity Property for Authorization for Web Server Clients with OAUTB}
This security property captures that, if the client is a web server
client with OAUTB, then (a) a user should only access resources
when the user actually expressed the wish to start a FAPI flow
before, and (b) if a user expressed the wish to start a FAPI flow
using some honest authorization server and a specific identity, then user
is not using resources of a different identity.
We note that for this, we require that the resource server which the
client uses is honest, as otherwise, the attacker can trivially return any resource
and receive any resource (for write access).

\begin{definition} [Session Integrity for
  Authorization for Web Server Clients with OAUTB]\label{def:property-si-authz}
  Let $\fapiwebsystem$
  be a FAPI web system with a network attackers. We say that
  \emph{$\fapiwebsystem$
    is secure w.r.t.~session integrity for authorization} iff for
  every run $\rho$
  of $\fapiwebsystem$, every processing step $Q$ in $\rho$ with
  $$Q = (S, E, N) \xrightarrow[]{}
  (S', E', N')$$  (for some $S$, $S'$, $E$, $E'$, $N$, $N'$),
  every browser $b$
  that is honest in $S$,
  every $\mi{as} \in \fAP{AS}$,
  every identity $u$, %
  every web server client $c\in \fAP{C}$
  of type $\str{conf\_OAUTB}$
  that is honest in $S$,
  every $\mi{rs} \in \fAP{RS}$ that is honest in $S$,
  every nonce $r$, 
  every nonce $\mi{lsid}$,
  we have that if
  $\mathsf{accessesResource_\rho^Q(b, r, u, c, \mi{rs}, \mi{lsid})}$
  and $s_0^\mi{rs}.\str{authServ} \in \mathsf{dom}(\mi{as})$, then
  (1)
  there exists a processing step $Q'$
  in $\rho$
  (before $Q$)
  such that $\mathsf{started}_\rho^{Q'}(b, c, \mi{lsid})$,
  and (2) if $\mi{as}$
  is honest in $S$,
  then there exists a processing step $Q''$
  in $\rho$
  (before $Q$)
  such that
  $\mathsf{authenticated}_\rho^{Q''}(b, c, u, \mi{as}, \mi{lsid})$.
\end{definition}

  \FloatBarrier\clearpage
  \newpage
\section{Proofs} \label{fapi-formal-proofs}

\subsection{General Properties}
\begin{lemma}[Host of HTTP Request]\label{lemma:https-server-correct-hosts}
 For any run $\rho$
 of a FAPI web system $\fapiwebsystem$
 with a network attacker,
 every configuration $(S, E, N)$  in $\rho$ and
 every process $p \in \fAP{C} \cup \fAP{AS} \cup \fAP{RS}$ that is honest in $S$
 it holds true that if the generic HTTPS server calls $\mathsf{PROCESS\_HTTPS\_REQUEST}(m_{dec}, k, a, f, s)$
  in Algorithm~\ref{alg:generic-server-main}, then $m_{dec}.\str{host} \in \mathsf{dom}(p)$, for all values of $k$, $a$, $f$ and $s$. 
\end{lemma}

\begin{proof}

  $\mathsf{PROCESS\_HTTPS\_REQUEST}$ is called only in Line~\ref{line:gen-server-first-req}
  of Algorithm~\ref{alg:generic-server-main}. 
  The input message $m$ is encrypted asymmetrically. Intuitively,
  such a message is only decrypted if the process knows the private TLS key, where this 
  private key is chosen (non-deterministically) according to the host of the decrypted message.

  More formally, when $\mathsf{PROCESS\_HTTPS\_REQUEST}$ is called, 
  the $\str{stop}$ in Line~\ref{line:gen-server-asym-check} is not called. Therefore, it holds true that
  \begin{align*}
        & \exists \; \mi{inDomain}, k': \an{\mi{inDomain},k'} \in S(p).\str{tlskeys} \wedge m_{dec}.\str{host} \equiv \mi{inDomain} \\
        & \Rightarrow \exists \; \mi{inDomain}, k': \an{\mi{inDomain},k'} \in \str{tlskeys}^p 
          \wedge m_{dec}.\str{host} \equiv \mi{inDomain} \\
        & \overset{\text{Def.}}{\Rightarrow} \exists \; \mi{inDomain}, k': \an{\mi{inDomain},k'} \in \{ \an{d, \mi{tlskey}(d)} | d \in \mathsf{dom}(p) \}  
          \wedge m_{dec}.\str{host} \equiv \mi{inDomain} 
      \end{align*}

      From this, it follows directly that $ m_{dec}.\str{host} \in \mathsf{dom}(p)$.

  The first step holds true due to
  $S(p).\str{tlskeys} \equiv s^p_0.\str{tlskeys} \equiv \str{tlskeys}^{p}$,
  as this sequence is never changed by any honest process $p$. \QED

\end{proof}

\begin{lemma}[Honest Read Client sends Token Request only to $\mathsf{dom}(\mi{as})$]\label{lemma:token-req-to-id-domain}
  For any run $\rho$
  of a FAPI web system $\fapiwebsystem$
  with a network attacker,
  every configuration $(S, E, N)$  in $\rho$,
  every authorization server $\mi{as}$ that is honest in $S$,
  every identity $\mi{id} \in \mathsf{ID}^\mi{as}$,
  every read client $c$ that is honest in $S$
  and every $\mi{sid}$,
  it holds true that if
Algorithm~\ref{alg:client-token-request} 
($\mathsf{SEND\_TOKEN\_REQUEST}$) is called with
$\mi{sessionId} \equiv \mi{sid}$ and
$S(c).\str{sessions}[\mi{sid}][\str{identity}] \equiv \mi{id} $,
then the messages in 
$\mathsf{SEND\_TOKEN\_REQUEST}$ are sent only to $d \in \mathsf{dom}(\mi{as})$. %
\end{lemma}

\begin{proof}
In Algorithm~\ref{alg:client-token-request}, the client sends messages either in
Line~\ref{line:r-client-sends-token-req-pub}, Line~\ref{line:r-client-sends-token-req-jws} or
Line~\ref{line:client-sends-token-req-mtls}. The $\mathsf{HTTPS\_SIMPLE\_SEND}$ in Line~\ref{line:jws-oautb-cont-send}
can only be reached by a read-write client.

In all three cases, the message is sent to 
$\mi{url}.\str{domain}$, 
which is equal to \\
$S(c).\str{oidcConfigCache}[\mi{issuer}][\str{token\_ep}].\str{domain}$ (Line~\ref{line:r-client-chooses-right-tep}),
due to $\mi{session}[\str{misconfiguredTEp}] \equiv \bot$ in Line~\ref{str-check-if-misconf-tep}
 (for read clients, this is always set to $\bot$ in Line~\ref{ptr-r-client-set-misconf-tep-false} 
of Algorithm~\ref{alg:client-prepare-token-request}).

Let $\overline{dom}$ and  $\overline{dom'}$ be from $\mathsf{dom}(\mi{as})$.

With this, it holds true that:
\begin{align}
 \notag & \mi{url}.\str{domain}  \\
  \notag & = S(c).\str{oidcConfigCache}[\mi{issuer}][\str{token\_ep}].\str{domain}  \\
  \shortintertext{(Line~\ref{line:str-client-chooses-issuer} of Alg.~\ref{alg:client-token-request})}
  \notag & = S(c).\str{oidcConfigCache}[S(c).\str{issuerCache}[\mi{identity]}][\str{token\_ep}].\str{domain} \\
  \shortintertext{(Line~\ref{line:str-client-chooses-identity} of Alg.~\ref{alg:client-token-request})}
  \notag & = S(c).\str{oidcConfigCache}[S(c).\str{issuerCache}[S(c).\str{sessions}[\str{sid}][\str{identity}]]][\str{token\_ep}].\str{domain} \\
  \shortintertext{(per assumption)}
  \notag  & = S(c).\str{oidcConfigCache}[S(c).\str{issuerCache}[\mi{id}]][\str{token\_ep}].\str{domain} \\
  \shortintertext{($\str{issuerCache}$ is never modified)}
  \notag & = S(c).\str{oidcConfigCache}[s^c_0.\str{issuerCache}[\mi{id}]][\str{token\_ep}].\str{domain} \\
  \shortintertext{(per definition)}
  \notag & = S(c).\str{oidcConfigCache}[\overline{dom}]][\str{token\_ep}].\str{domain} \\
  \shortintertext{(oidcConfigCache is never modified)}
  \notag & = s^c_0.\str{oidcConfigCache}[\overline{dom}]][\str{token\_ep}].\str{domain} \\
  \shortintertext{(per definition)}
  \notag & = \mi{\overline{dom'}} 
\end{align}

Therefore, the token request is always sent to a domain of $\mi{as}$.
\QED

\end{proof}

\begin{lemma}[Code used in Token Request was received at Redirection Endpoint]\label{lemma:code-token-req-from-redirect-ep}
  For any run $\rho$
  of a FAPI web system $\fapiwebsystem$
  with a network attacker,
  every configuration $(S, E, N)$  in $\rho$,
  every client $c$ that is honest in $S$
  it holds true that 
  if Algorithm~\ref{alg:client-fapi-http-response}
  ($\mathsf{PROCESS\_HTTPS\_RESPONSE}$) is called with 
  $\mi{reference}[\str{responseTo}] \equiv \str{TOKEN}$, then \\
  $\mi{request}.\str{body}[\str{code}] \equiv
   S(c).\str{sessions}[\mi{reference}[\str{session}]][\str{redirectEpRequest}][\str{data}][\str{code}]$,
  with $\mi{request}$ being an input parameter of 
  $\mathsf{PROCESS\_HTTPS\_RESPONSE}$.
\end{lemma}

\begin{proof}
  
  Let $\mi{sid} := \mi{reference}[\str{session}]$ be the session id 
  with which $\mathsf{PROCESS\_HTTPS\_RESPONSE}$ is called.

  Due to $\mi{reference}[\str{responseTo}] \equiv \str{TOKEN}$, the corresponding request was sent
  in Algorithm~\ref{alg:client-token-request} ($\mathsf{SEND\_TOKEN\_REQUEST}$), as this is the only algorithm
  that uses this reference when sending a message. The code included in the request is always the input parameter
  of $\mathsf{SEND\_TOKEN\_REQUEST}$ (due to Line~\ref{c-STR-main-body}).

  $\mathsf{SEND\_TOKEN\_REQUEST}$ is called in one of the following lines:
  (1) in Line~\ref{line:client-https-response-mtls-as-send-token-req} of 
  Algorithm~\ref{alg:client-fapi-http-response} ($\mathsf{PROCESS\_HTTPS\_RESPONSE}$),
  (2) in Line~\ref{line:PHResp-STR-oautb-call} of Algorithm~\ref{alg:client-fapi-http-response}
      ($\mathsf{PROCESS\_HTTPS\_RESPONSE}$) or
  (3) in Line~\ref{c-PTR-call-STR} of Algorithm~\ref{alg:client-prepare-token-request}
      ($\mathsf{PREPARE\_TOKEN\_REQUEST}$).

      \emph{Case 1: Algorithm~\ref{alg:client-fapi-http-response}:} 
The authorization code used for calling $\mathsf{SEND\_TOKEN\_REQUEST}$ is taken from 
$S(c).\str{sessions}[\mi{sid}][\str{code}]$. 
$S(c).\str{sessions}[\mi{sid}][\str{code}]$ is only set in Line~\ref{c-PTR-save-code} %
of Algorithm~\ref{alg:client-prepare-token-request} 
($\mathsf{PREPARE\_TOKEN\_REQUEST}$),
where this value is taken from the input of 
$\mathsf{PREPARE\_TOKEN\_REQUEST}$.

      \emph{Case 2: Algorithm~\ref{alg:client-prepare-token-request}:} 
In this case, the authorization code used for calling $\mathsf{SEND\_TOKEN\_REQUEST}$ 
is directly taken from the input of  $\mathsf{PREPARE\_TOKEN\_REQUEST}$.
Algorithm~\ref{alg:client-prepare-token-request} ($\mathsf{PREPARE\_TOKEN\_REQUEST}$)
is either directly called at the redirection endpoint (Line~\ref{c-redirection-ep-call-PTR}
of Algorithm~\ref{alg:client-fapi-http-request}),
called in Line~\ref{c-second-PTR} of Algorithm~\ref{alg:client-check-first-id-token}
($\mathsf{CHECK\_FIRST\_ID\_TOKEN}$) or
in Line~\ref{c-second-PTR-JARM} of Algorithm~\ref{alg:client-check-response-jws}
($\mathsf{CHECK\_RESPONSE\_JWS}$).

In both functions (Algorithm~\ref{alg:client-check-first-id-token} and
Algorithm~\ref{alg:client-check-response-jws}),
$\mi{code}$ is an input argument.

As both functions are called only at the
redirection endpoint (Line~\ref{line:c-call-check-id-token-immediately}
or Line~\ref{line:c-call-check-response-jws}
of Algorithm~\ref{alg:client-fapi-http-request}), it follows that 
the authorization code used as an input argument %
of
$\mathsf{PREPARE\_TOKEN\_REQUEST}$ 
is in all cases
the authorization code 
originally contained in the request that was received at 
$\str{/redirect\_ep}$.
In
Line~\ref{line:client-set-redirect-ep-request-record} of 
Algorithm~\ref{alg:client-fapi-http-request}, 
the data contained in the request is stored in the session
under the key $\str{redirectEpRequest}$.
\QED

\end{proof} 

\begin{lemma}[Receiver of Token Request and Authorization Request for Read Clients] \label{lemma:req-jws-aud-token-req-as}
For any run $\rho$
of a FAPI web system $\fapiwebsystem$
with a network attacker,
every configuration $(S, E, N)$  in $\rho$,
every authorization server $\mi{as}$ that is honest in $S$,
every read client $c$
that is honest in $S$
with client id $\mi{clientId}$
that has been issued to $c$
by $\mi{as}$,
every $\mi{sid}$ being a session identifier for
sessions in $S(c).\str{sessions}$
it holds true that if
a request JWS $\mi{reqJWS}$ 
is created in Line~\ref{client-slf-requestJWS-sign-1} of Algorithm~\ref{alg:client-fapi-start-login-flow}
(called with input argument $\mi{sessionId} \equiv \mi{sid}$)
with
$\mathsf{extractmsg}(\mi{reqJWS}[\str{aud}]) \in \mathsf{dom}(\mi{as})$,
then the token request of the session send in Algorithm~\ref{alg:client-token-request} (when called with
$\mi{sessionId} \equiv \mi{sid}$)
is sent only to a domain in $\mathsf{dom}(\mi{as})$.
\end{lemma}

\begin{proof}
Let $\mathsf{extractmsg}(\mi{reqJWS}[\str{aud}])
\equiv
\mi{authEndpoint}.\str{host} \in \mathsf{dom}(\mi{as})$ 
(Line~\ref{client-slf-reqJWS-aud} of Algorithm~\ref{alg:client-fapi-start-login-flow}),
which means that  \\
$s^c_0.\str{oidcConfigCache}[s^c_0.\str{issuerCache}[S(c).\str{sessions}[\mi{sid}][\str{identity}]]][\str{auth\_ep}].\str{host} 
\in \mathsf{dom}(\mi{as})$ \\
(Line~\ref{c-SLF-start} up to \ref{c-SLF-authep}; $\str{oidcConfigCache}$ and $\str{issuerCache}$ are never changed by the client).
Per definition of the $\str{auth\_ep}$ key of $\str{oidcConfigCache}$,
it follows that 
$s^c_0.\str{issuerCache}[S(c).\str{sessions}[\mi{sid}][\str{identity}]] 
\in \mathsf{dom}(\mi{as})$. 

Per definition, $\str{issuerCache}$ is a mapping from identities
to a domain of their governor. This means that
$S(c).\str{sessions}[\mi{sid}][\str{identity}] \in \mathsf{ID}^\mi{as}$. 

As all conditions of Lemma~\ref{lemma:token-req-to-id-domain} are fulfilled,
the token request sent in Algorithm~\ref{alg:client-token-request} (when called with the
input parameter
$\mi{sid}$) 
is sent only to 
a domain of 
$\mi{as}$.
\QED

\end{proof}

\begin{lemma}[EKM signed by Client does not leak] \label{lemma:ekm-signed-by-c-does-not-leak}
For any run $\rho$
of a FAPI web system $\fapiwebsystem$
with a network attacker,
every configuration $(S, E, N)$  in $\rho$,
every process $p \in \fAP{AS} \cup \fAP{RS}$ that is honest in $S$,
every domain $\mi{dom}_p \in \mathsf{dom}(p)$,
every key $k \in \mi{keyMapping}[\mi{dom}_p]$,
every client $c$ that is honest in $S$,
every domain $d \in \dns$,
every key $\mi{TB\_key} \in s^c_0.\str{tokenBindings}[d]$,
every terms $n_1$, $n_2$ and
every process $p'$ with $ p \ne p' \ne c$
it holds true that 
$\mathsf{sig}(\mathsf{hash}(\an{n_1, n_2, k}), \mi{TB\_key}) \not\in d_{\emptyset}(S(p'))$.
\end{lemma}

\begin{proof}

Let $\mathsf{hash}(\an{n_1, n_2, k})$ be signed by $c$.

The honest client signs only in one of the following places:
(1) in Line~\ref{client-slf-requestJWS-sign-1} of 
  Algorithm~\ref{alg:client-fapi-start-login-flow} ($\mathsf{START\_LOGIN\_FLOW}$),
(2) in the branch at Line~\ref{client-str-branch-oautb-sig-2} of 
  Algorithm~\ref{alg:client-token-request} ($\mathsf{SEND\_TOKEN\_REQUEST}$) or
(3) in Line~\ref{client-uat-sig-3} of Algorithm~\ref{alg:client-use-access-token} ($\mathsf{USE\_ACCESS\_TOKEN}$).

  \begin{description} 
    \item[Case 1]\strut

In Line~\ref{client-slf-requestJWS-sign-1} of Algorithm~\ref{alg:client-fapi-start-login-flow},
the client signs the request JWS, which has a different structure than the EKM value.

    \item[Case 2]\strut

If the client signs a value in the branch at Line~\ref{client-str-branch-oautb-sig-2} of Algorithm~\ref{alg:client-token-request},
it follows that $\mathsf{SEND\_TOKEN\_REQUEST}$ was called in Line~\ref{line:PHResp-STR-oautb-call} 
of Algorithm~\ref{alg:client-fapi-http-response} \\
($\mathsf{PROCESS\_HTTPS\_RESPONSE}$)
(due to $\mi{responseValue}[\str{type}] \equiv \str{OAUTB}$, L.~\ref{client-str-resp-value-oautb} of
Alg.~\ref{alg:client-token-request}).

The client signs $\mi{responseValue}[\str{ekm}]$ (Line~\ref{line:str-oautb-respVal1} of Algorithm~\ref{alg:client-token-request}),
which is an input argument. This value is created in Line~\ref{line:PHResp-STR-oautb-sign-ekm} 
of Algorithm~\ref{alg:client-fapi-http-response}, where the key of the sequence is
chosen by the client ($k \equiv \mi{keyMapping}[\mi{request}.\str{host}]$).

From this, it follows that $\mi{request}.\str{host} \in \mathsf{dom}(p)$.
Due to the reference value $\str{OAUTB\_AS}$ (Line~\ref{line:ref-oautb-as})
it follows that $\mi{request}$ was sent in Line~\ref{line:PTR-oautb-as-send} of
Algorithm~\ref{alg:client-prepare-token-request}  \\ 
($\mathsf{PREPARE\_TOKEN\_REQUEST}$),
as this is the only place where this reference value is used.

As $\mi{request}.\str{host} \in \mathsf{dom}(p)$, it follows that the message
in $\mathsf{PREPARE\_TOKEN\_REQUEST}$ is sent to $p$ via 
$\str{HTTPS\_SIMPLE\_SEND}$.
Therefore, the signed EKM value that is sent in
$\mathsf{SEND\_TOKEN\_REQUEST}$ is also send to 
$p$ via 
$\str{HTTPS\_SIMPLE\_SEND}$,
as the messages are sent to the same 
domain in both algorithms 
(the value of
$\mi{misconfiguredTEp}$
set in Line~\ref{c-chooses-misconfTEp} ff.~of
Algorithm~\ref{alg:client-prepare-token-request} is
the same in both algorithms, and therefore,
the domains to which the messages are sent are
the same).

    \item[Case 3]\strut

For an EKM value signed in Line~\ref{client-uat-sig-3} of 
Algorithm~\ref{alg:client-use-access-token} ($\mathsf{USE\_ACCESS\_TOKEN}$),
the same reasoning as in Case 2 holds true.

\end{description}

The signed EKM value is sent only in the header
of a request. 
As neither an honest authorization server
nor an honest resource server sends
any message containing a header value
of a received request,
it follows that 
$\mathsf{sig}(\mathsf{hash}(\an{n_1, n_2, k}), \mi{TB\_key})$ does
not leak to any other process. 
\QED

\end{proof}

\subsection{Client Authentication}

In the next lemmas, we prove that only the legitimate client 
can authenticate itself to the token endpoint.
More precisely, we show that if all checks pass at the token
endpoint of an authorization server, then 
the legitimate confidential client send the corresponding token request.

In case of a confidential read client using a JWS for authentication,
the message authentication code of the JWS is created with a key that only the client
knows (besides the authorization server). We show that the JWS never leaks
to an attacker, which implies that if a message is received at the token endpoint
containing the JWS, then the message was sent by the honest client (as otherwise,
the JWS must have leaked).

In case of a read-write client, we assume that the token endpoint might be
misconfigured. In this case, it seems to be possible that the attacker
uses the message received from the client to authenticate at the 
token endpoint of the honest authorization server. 

We show that this is not possible for both OAUTB and mTLS clients. %
In case of a confidential read-write client using OAUTB,
the audience value of the assertion
has the same value as 
the host of the
token request,
which means that the attacker cannot use a received assertion for
any other audience.

Similarly, the message that is decrypted by the client
for mTLS
(to prove possession of the private key) contains the domain
of the process that created the message, and therefore,
the decrypted nonce is sent to the same process that encrypted the nonce.

\begin{lemma}[JWS Client Assertion created by Client does not leak to Third Party]\label{lemma:jws-assertion-does-not-leak}
  For any run $\rho$
  of a FAPI web system $\fapiwebsystem$
  with a network attacker,
  every configuration $(S, E, N)$  in $\rho$,
  every authorization server $\mi{as}$ that is honest in $S$,
  every domain $d \in \mathsf{dom}(\mi{as})$,
  every client $c$ that is honest in $S$
  with client id $\mi{clientId}$ and client secret $\mi{clientSecret}$
  that has been issued to $c$
  by $\mi{as}$,
  every client assertion
  $t \equiv \mathsf{mac}([\str{iss}{:}\mi{clientId}, \str{aud}{:}d] , \mi{clientSecret})$ with $\mi{clientSecret} \equiv 
    S(\mi{as}).\str{clients}[\mi{clientId}][\str{client\_secret}]$
  and every process $p$ with $c \ne p \ne \mi{as}$ %
  it holds true that $t \not\in d_{\emptyset}(S(p))$.
\end{lemma}

\begin{proof} ~ 

  Let $t \equiv \mathsf{mac}([\str{iss}{:}\mi{clientId},  \str{aud}{:} d] , \mi{clientSecret})$
  be the assertion that is created by the client.

  \begin{description} 
    \item[Honest read client: sends assertion only to $d$.]\strut

    In case of read clients, the JWS client assertion is only created within
      the branch at Line~\ref{branch:r-client-creates-assertion} of Algorithm~\ref{alg:client-token-request}
      and send in Line~\ref{line:r-client-sends-assertion}.

      As the host value of the message in Line~\ref{line:r-client-assertion-rec}
      has the same value as the audience value of the assertion,
      it follows that the message is sent to $d$.

    \item[Honest read-write client: sends assertion only to $d$.]\strut

      A read-write client using OAUTB creates a client assertion only in
      Line~\ref{line:rw-client-creates-assertion} of Algorithm~\ref{alg:client-token-request}.
      As the client sets 
      $\mi{jwt}[\str{aud}] := d $ 
      in Line \ref{line:jws-oautb-set-aud}, it follows that %
      $\mi{url}.\str{domain} \equiv d$. 
      As above, the assertion is sent to $d$ with $\mathsf{HTTPS\_SIMPLE\_SEND}$.

    \item[Honest authorization server: never sends an assertion.]\strut

      The authorization server receiving the assertion never sends it out, and as it never
      creates any assertions, we conclude that the assertion never leaks to $p$ %
      due to the
      authorization server.

  \end{description}

  Combining both points, we conclude that $p$ %
  never receives
  an assertion which is created by an honest client and send to an honest authorization server.

  As the client secret is unique for each client id, and also neither send out by an honest %
  client nor by the honest authorization server, we conclude that $p$ %
  is never in possession of a client assertion
  with a valid message authentication code.
  \QED
\end{proof}

\begin{lemma}[mTLS Nonce created by AS does not leak to Third Party]\label{lemma:as-mtls-nonce-does-not-leak}
For any run $\rho$ of a FAPI web system $\fapiwebsystem$
with a network attacker,
every configuration $(S, E, N)$
in $\rho$,
every authorization server $\mi{as}$ that is honest in $S$,
every client $c$ that is honest in $S$ %
with client id $\mi{clientId}$ issued by $\mi{as}$,
every $\mi{mtlsNonce}$ 
created in Line~\ref{line:as-creates-mtls-nonce} of Algorithm~\ref{alg:as-fapi}
in consequence of a request 
$m$ received at the $\str{/MTLS\mhyphen{}prepare}$ path of
the authorization server (L.~\ref{line:as-mtls-ep}
of Alg.~\ref{alg:as-fapi})
with $m.\str{body}[\str{client\_id}] \equiv \mi{clientId}$
and every process $p$ with $\mi{as} \ne p \ne c$ it holds true that
$\mi{mtlsNonce} \not \in d_{\emptyset}(S(p))$.
\end{lemma}

\begin{proof}
  The authorization server sends a $\mi{mtlsNonce}$ created in
  Line~\ref{line:as-creates-mtls-nonce} of Algorithm~\ref{alg:as-fapi}
  only 
  in Line~\ref{as-send-mtls-prepare-response}, where it is asymmetrically encrypted with the public key
  \begin{align*}
    & \mi{clientKey}  & \\
    \equiv & \; S(\mi{as}).\str{clients}[\mi{clientId}][\str{mtls\_key}] & \text{ (Line~\ref{line:as-choses-mtls-key})} \\
    \equiv & \; s^\mi{as}_0.\str{clients}[\mi{clientId}][\str{mtls\_key}]  &\text{ (value is never changed)} \\
    \equiv & \; \mi{keyMapping}[\mi{dom_{c}}] & \text{ (def.)}
  \end{align*}
  with $\mi{dom_{c}} \in \mathsf{dom}(c)$.
  The corresponding private key is $\mi{tlskeys}[\mi{dom_{c}}] \in \mi{tlskeys}^c$,
  which is only known to $c$.
  (The $\mi{mtlsNonce}$ saved in $\str{mtlsRequests}$ is not sent in any other place).

  This implies that the encrypted nonce can only be decrypted by $c$.
  Such a message is decrypted either in Line~\ref{line:client-decrypt-mtls}
  or Line~\ref{line:client-decrypt-mtls-rs} of Algorithm~\ref{alg:client-fapi-http-response}.
  (The only other places where a message is decrypted asymmetrically by $c$ is
  in the generic HTTPS server (Line~\ref{line:gen-server-asym-check} of Algorithm~\ref{alg:generic-server-main}), 
  but this message is not decrypted there due to the requirement
  that the decrypted message must begin with $\str{HTTPReq}$).

  We also note that the encrypted message created by the authorization server
  containing the nonce also contains a public TLS key of $\mi{as}$.
  (This holds true due to Lemma~\ref{lemma:https-server-correct-hosts}).

  \begin{description}
    \item[Case 1: Line~\ref{line:client-decrypt-mtls-rs}]\strut

      In this case, it follows that 
      $\mi{reference}[\str{responseTo}] \equiv
      \str{MTLS\_RS}$ (Line~\ref{branch-ref-mtls-rs}).
      The only place where this value is used as a reference is in Line~\ref{line:c-sends-ref-mtls-rs}
      of Algorithm~\ref{alg:client-prepare-use-access-token} ($\mathsf{PREPARE\_USE\_ACCESS\_TOKEN}$).
      The corresponding request is sent to
  \begin{align*}
    & \mi{message}.\str{host}  & \text{ (L.~\ref{line:c-PUAT-message} of Alg.~\ref{alg:client-prepare-use-access-token})}\\
    \equiv & \; \mi{rsHost} & \text{ (L.~\ref{line:c-PUAT-message})} \\  
    \equiv & \; S'(c).\str{sessions}[\mi{sid}][\str{RS}] & \text{ (L.~\ref{line:c-PUAT-rsHost} and L.~\ref{line:c-PUAT-session}, for some $\mi{sid}$)}   
  \end{align*}

  for a previous state $S'$ within the run.

  This value is only set in Line~\ref{line:c-sets-session-RS} 
  of Algorithm~\ref{alg:client-fapi-start-login-flow} ($\mathsf{START\_LOGIN\_FLOW}$),
  where it is chosen from the set of preconfigured resource servers \\
  $s^c_0.\str{oidcConfigCache}[s^c_0.\str{issuerCache}[\mi{id}]][\str{resource\_server}]$, 
  for some $\mi{id} \in \IDs$. (The values of $\str{oidcConfigCache}$ and $\str{isserCache}$ 
  are not changed by the client and are therefore the same as in the initial state).

  These values are required to be from $\mathsf{dom}(\mi{rs})$ by definition, %
  for $\mi{rs} \in \fAP{RS}$, which means that the initial request with reference being
  $\str{MTLS\_RS}$ is sent to $\mi{rsHost} \in \mathsf{dom}(\mi{rs})$, $\mi{rs} \in \fAP{RS}$.

  Therefore, the value of $\mi{request}.\str{host}$ in Line~\ref{c-PHResp-check-mtls-rs-pubkey}
  of Algorithm~\ref{alg:client-fapi-http-response} is from $\mathsf{dom}(\mi{rs})$, and the check in
  this line fails, which means that the corresponding $\mi{stop}$ is executed and 
  the decrypted nonce is not sent.

    \item[Case 2: Line~\ref{line:client-decrypt-mtls}]\strut

    The reference $\str{MTLS\_AS}$ is only used in Line~\ref{line:client-https-simple-mtls-as}
    of Algorithm~\ref{alg:client-prepare-token-request} \\
    ($\mathsf{PREPARE\_TOKEN\_REQUEST}$).
    The corresponding request is sent to
    $\mi{url}.\str{domain}$ 
    (Line~\ref{line-client-https-simple-mtls-as-msg})
    via $\str{HTTPS\_SIMPLE\_SEND}$.
    If $\mi{url}.\str{domain} \not \in \mathsf{dom}(\mi{as})$ for $\mi{as} \in \fAP{AS}$,
    then the check of the public key fails as in Case 1. Otherwise, the initial request
    is sent to the honest authorization server, which means that the 
    decrypted nonce is sent to the authorization server in $\mathsf{SEND\_TOKEN\_REQUEST}$
    (as in both algorithms, the messages are sent to the same domain).

  \end{description}
    Summing up, the client sends the nonce encrypted by the authorization server
    only back to the authorization server.
    As an honest authorization server never sends out such a nonce received
    in a token request, we conclude that the nonce never leaks to any other process.
    \QED

\end{proof}

\begin{lemma}[Client Authentication]\label{lemma:client-authentication}
  For any run $\rho$
  of a FAPI web system $\fapiwebsystem$
  with a network attacker,
  every configuration $(S, E, N)$  in $\rho$,
  every authorization server $\mi{as}$ that is honest in $S$,
  every domain $d \in \mathsf{dom}(\mi{as})$,
  every confidential client $c$ that is honest in $S$
  with client id $\mi{clientId}$ issued by $\mi{as}$,
   it holds true that:

  If a response is sent in Line~\ref{as-send-token-ep-readwrite} or in Line~\ref{as-send-token-ep-read}
  of Algorithm~\ref{alg:as-fapi} due to a request $m$ received at the token endpoint of $\mi{as}$
  with $m.\str{body}[\str{client\_id}] \equiv \mi{clientId}$,
  then $m$ was sent by $c$.
\end{lemma}

\begin{proof}
    We assume that $m$ was sent by $p \ne c$ and that the authorization server sends a response,
    which implies that all (applicable) checks are passed.
    We also note that the authorization server never sends messages to itself.
    We distinguish the following cases:

  \begin{description} 
    \item[Case 1:] $\;\; S(\mi{as}).\str{clients}[\mi{clientId}][\str{client\_type}] \equiv \str{conf\_JWS}$ \strut \\
    In this case, the client id belongs to a read client. As the response in Line~\ref{as-send-token-ep-read}
    is sent,
    all relevant checks passed successfully, in particular, the checks 
    in Line~\ref{line:token-ep-check-jws-sig}
    and Line~\ref{line:token-ep-check-jws-values} of Algorithm~\ref{alg:as-fapi}, which means that that $p$ possesses a term
    $t \equiv \mathsf{mac}([\str{iss}{:}\mi{clientId}, \str{aud}{:}d] , \mi{clientSecret})$ with  \\
    $\mi{clientSecret}= 
    S(\mi{as}).\str{clients}[\mi{clientId}][\str{client\_secret}]$,
    contradicting Lemma~\ref{lemma:jws-assertion-does-not-leak}. 
    (As shown in Lemma~\ref{lemma:https-server-correct-hosts}, the host of the message is a domain 
    of the authorization server). %

    \item[Case 2:] $\;\; S(\mi{as}).\str{clients}[\mi{clientId}][\str{client\_type}] \equiv \str{conf\_OAUTB}$ \strut \\
    As in case 1, this is a contradiction to 
    Lemma~\ref{lemma:jws-assertion-does-not-leak}.

   \item[Case 3:] $\;\; S(\mi{as}).\str{clients}[\mi{clientId}][\str{client\_type}] \equiv \str{conf\_MTLS}$ \strut \\
   Here, 
   the check in Line~\ref{as-token-ep-check-mtls} passes,
   which means that $p$ knows $\mi{mtlsInfo}.1$, which is taken from 
   $S(\mi{as}).\str{mtlsRequests}[\mi{clientId}]$ (Line~\ref{as-token-ep-retrieve-mtlsInfo}).
   This sequence was added to $\str{mtlsRequests}$ in Line~\ref{as-add-mtlsInfo}, 
   as this is the only place where a term is added to $\str{mtlsRequests}$ (in the initial state,
   $\str{mtlsRequests}$ is empty).
   This implies that the nonce created in Line~\ref{line:as-creates-mtls-nonce}
   due to a request $m'$ to $\str{/MTLS\mhyphen{}prepare}$
   with $m'.\str{body}[\str{client\_id}] \equiv \mi{clientId}$
   is known to $p$, which is a contradiction to 
   Lemma~\ref{lemma:as-mtls-nonce-does-not-leak}.  \QED

  \end{description} 
\end{proof}

\subsection{PKCE Challenge} %

\begin{lemma}[PKCE Challenge of Public Client]\label{lemma:pub-rw-pkce-cc-created-by-c}
  For any run $\rho$
  of a FAPI web system $\fapiwebsystem$
  with a network attacker,
  every configuration $(S, E, N)$  in $\rho$,
  every authorization server $\mi{as}$ that is honest in $S$,
  every client $c$ 
  of type  $\str{pub}$
  that is honest in $S$
  with client id $\mi{clientId}$ issued to $c$ 
  by $\mi{as}$
  it holds true that:
  
  If $\mi{record} \in S(\mi{as}).\str{records}$ with 
  $\mi{record}[\str{client\_id}] \equiv \mi{clientId}$,
  then $\mi{record}[\str{pkce\_challenge}]$ was created by $c$.

\end{lemma}

\begin{proof}

Intuitively, this holds true because the authorization request
is authenticated, %
which means that the request
and therefore the PKCE challenge were created by the client.

More formally, we note that  
$S(\mi{as}).\str{records}$ is empty in the initial state.
New records are only added in $\str{/auth2}$, or more precisely,
in Line~\ref{as-adds-record} of Algorithm~\ref{alg:as-fapi}.

Let $m$ be the corresponding request to $\str{/auth2}$.
Due to $\mi{record}[\str{client\_id}] \equiv \mi{clientId}$ and Line~\ref{as-auth2-client-id-from-body},
it follows that $m.\str{body}[\str{client\_id}] \equiv \mi{clientId}$. As the record is added to the state,
all (applicable) checks in $\str{/auth2}$ have passed. 

This means that 
$\mathsf{checksig}(m.\str{body}[\str{request\_jws}], S(\mi{as}).\str{clients}[\mi{clientId}][\str{jws\_key}] ) 
\equiv \True $ (Line~\ref{as-check-sig-of-jws}), which implies that the request JWS was created by $c$ 
(as the key is preconfigured and not send to any other process).

Due to $\mi{record}[\str{pkce\_challenge}]$ being set from the request JWS 
(Line~\ref{as-retrieves-requestData-from-JWS} and Line~\ref{line:as-auth2-requestData-read-2}),
it follows that the PKCE code challenge was created by $c$.  \QED

\end{proof}

\subsection{Authorization Response}

\begin{lemma}[ID Token contained in Authorization Response for Web Server Client does not leak]\label{lemma:auth-resp-does-not-leak-isApp}
  For any run $\rho$
  of a FAPI web system $\fapiwebsystem$
  with a network attacker,
  every configuration $(S, E, N)$  in $\rho$,
  every authorization server $\mi{as}$ that is honest in $S$,
  every domain $d \in \mathsf{dom}(\mi{as})$,
  every identity $\mi{id} \in \mathsf{ID}^\mi{as}$
  with $b = \mathsf{ownerOfID}(\mi{id})$
  being an honest browser in $S$,
  every web server client $c$ that is honest in $S$
  with client id $\mi{clientId}$ 
  that has been issued to $c$
  by $\mi{as}$,
  every term $n$,
  every term $h$,
  every id token $idt \equiv \sig{[\str{iss}{:} d, \str{sub}{:}\mi{id}, \str{aud}{:}\mi{clientId},
    \str{nonce}{:}n, \str{s\_hash}{:} h]}{k}$
  with $k \equiv S(\mi{as}).\str{jwk}$ %
  and every attacker process $a$ %
  it holds true that:
  
  If a request $m$ is sent to the path $\str{/auth2}$ of $\mi{as}$
  with $m.\str{body}[\str{client\_id}] \equiv \mi{clientId}$ and
  $m.\str{body}[\str{password}] \equiv \mathsf{secretOfID}(\mi{id})$,
  then for the corresponding response $r$ it holds true that
  $\mi{idt} \equiv r.\str{body}[\str{id\_token}]$ 
  does not leak to $a$,
  i.e.~, $\mi{idt} \not \in  d_{\emptyset}(S(a))$.
\end{lemma}

\begin{proof}

We first highlight that $k \equiv s^\mi{as}_0.\str{jwk}$, as this value 
is never changed by the authorization server. Only $\mi{as}$ knows this value,
as it is preconfigured and never transmitted.

Furthermore, we assume that an id token is sent in the response, otherwise, it trivially cannot leak.

  \begin{description} 

    \item[Authorization Response does not leak.]\strut

Intuitively, the authorization response does not leak 
because only the honest browser knows the password of the
identity, which means that the response is sent back to the honest browser.
As the redirection URIs are preregistered, the browser redirects
the response directly to the honest client. 

This does not hold true in the case of app clients, as 
we assume that the response can be sent to a wrong app by the
operating system.

The formal proof is analogous to the proof of Lemma 6 in \cite{FettKuestersSchmitz-TR-OIDC-2017}.
We note that the changes made to the browser algorithms do 
not change the results of this proof.

    \item[ID Token does not leak.]\strut

The client never sends an id token. Therefore, we conclude that an id token created
at the authorization endpoint does not leak to an attacker.

\end{description} 

\QED

\end{proof} 

\begin{lemma}[Code of Read Web Server Clients does not leak]\label{lemma:r-web-server-client-code}
  For any run $\rho$
  of a FAPI web system $\fapiwebsystem$
  with a network attacker,
  every configuration $(S, E, N)$  in $\rho$,
  every authorization server $\mi{as}$ that is honest in $S$,
  every domain $d \in \mathsf{dom}(\mi{as})$,
  every identity $\mi{id} \in \mathsf{ID}^\mi{as}$
  with $b = \mathsf{ownerOfID}(\mi{id})$
  being an honest browser in $S$,
  every read web server client $c$ that is honest in $S$
  with client id $\mi{clientId}$ 
  that has been issued to $c$
  by $\mi{as}$,
  every code $\mi{code}$ 
  for which there is a record
  $\mi{rec} \inPairing S(\mi{as}).\str{records}$ with
  $\mi{rec}[\str{code}] \equiv \mi{code}$,
  $\mi{rec}[\str{client\_id}] \equiv \mi{clientId}$,
  $\mi{rec}[\str{subject}] \equiv \mi{id}$,
  and every attacker process $a$ %
  it holds true that 
  $\mi{code}$ 
  does not leak to $a$,
  i.e., $\mi{code} \not \in  d_{\emptyset}(S(a))$.
\end{lemma}

\begin{proof}

The code contained in the authorization response
is sent only to the client (analogous to Lemma~\ref{lemma:auth-resp-does-not-leak-isApp}).

Let $m$ be the request received at $\str{/auth2}$ (Line~\ref{line:as-receive-auth2}
of Algorithm~\ref{alg:as-fapi}) which led to the creation of
$\mi{rec}$. 

As the record was added to the state, it follows that the checks in 
the branch at Line~\ref{as-auth2-check-jws-start} are fulfilled. %
From this, it follows that the $\str{aud}$ value of the request JWS is a
domain in $\mathsf{dom}(\mi{as})$ (due to Lemma~\ref{lemma:https-server-correct-hosts}).

As shown in Lemma~\ref{lemma:req-jws-aud-token-req-as},
the corresponding token request send by the client is sent only to 
$\mi{as}$.

(This does not hold true for read-write clients, as we assume that the token endpoint might
be misconfigured, which means that the code might be sent to the attacker.)

As the authorization server does not send the code received at the token endpoint
(and only sends newly generated codes at the authorization endpoint), we
conclude that $\mi{code}$ does not leak to $a$.
\QED

\end{proof} 

\begin{lemma}[Response JWS created for Web Server Client does not leak]\label{lemma:req-jws-does-not-leak}
  For any run $\rho$
  of a FAPI web system $\fapiwebsystem$
  with a network attacker,
  every configuration $(S, E, N)$  in $\rho$,
  every authorization server $\mi{as}$ that is honest in $S$,
  every domain $d \in \mathsf{dom}(\mi{as})$,
  every identity $\mi{id} \in \mathsf{ID}^\mi{as}$
  with $b = \mathsf{ownerOfID}(\mi{id})$
  being an honest browser in $S$,
  every web server client $c$ that is honest in $S$
  with client id $\mi{clientId}$ 
  that has been issued to $c$
  by $\mi{as}$,
  every term $n$ for which a record $\mi{rec} \in S(\mi{as}).\str{records}$
  exists with $\mi{rec}[\str{code}] \equiv n$ and 
  $\mi{rec}[\str{subject}] \equiv \mi{id}$,
  every term $s$,
  every request JWS $t \equiv \sig{[\str{iss}{:} d,  \str{aud}{:}\mi{clientId},
    \str{code}{:}\mi{n}, 
    \str{state}{:} s]}{k}$
  with $k \equiv S(\mi{as}).\str{jwk}$ 
  and every attacker process $a$ 
  it holds true that
  $\mi{t} \not \in  d_{\emptyset}(S(a))$.
\end{lemma}

\begin{proof}

  Such a term is only created at the authorization endpoint of $\mi{as}$
  (in Line~\ref{as-creates-response-jws} of Algorithm~\ref{alg:as-fapi}), where the corresponding authorization
  request sent to $\mi{as}$ by $b$ (as only $b$ knows the secret of 
  $\mi{id}$). Therefore, the request JWS follows the same
  path  as an id token created at the authorization endpoint. %
  Due to this, the proof of this lemma is analogous to the proof of 
  Lemma~\ref{lemma:auth-resp-does-not-leak-isApp}.
\QED
\end{proof}

\subsection{ID Token}

\begin{lemma}[ID Token for Confidential Client created at the Token Endpoint]\label{lemma:id-token-conf-from-tokenep}
  For any run $\rho$
  of a FAPI web system $\fapiwebsystem$
  with a network attacker,
  every configuration $(S, E, N)$  in $\rho$,
  every authorization server $\mi{as}$ that is honest in $S$,
  every domain $d \in \mathsf{dom}(\mi{as})$,
  every identity $\mi{id} \in \mathsf{ID}^\mi{as}$,
  every confidential client $c$ that is honest in $S$ %
  with client id $\mi{clientId}$ 
  that has been issued to $c$
  by $\mi{as}$,
  every term $n$,
  every term $h$,
  every id token $t \equiv \sig{[\str{iss}{:} d, \str{sub}{:}\mi{id}, \str{aud}{:}\mi{clientId},
    \str{nonce}{:}n, \str{at\_hash}{:}h]}{k}$ with
  $k \equiv s^\mi{as}_0.\str{jwk}$  %
  and every process $p$ with $c \ne p \ne \mi{as}$ %
  it holds true that 
  $t \not \in  d_{\emptyset}(S(p))$.
\end{lemma}

\begin{proof}

Such an id token $t$ is only created in Line~\ref{as-token-ep-create-id-token} of Algorithm~\ref{alg:as-fapi},
which is the token endpoint of the authorization server (the only other place where an id token is
created by the authorization server is in $\str{/auth2}$ (Line~\ref{line:as-create-id-token-auth2}),
but the id token created there has different attributes, like $\str{s\_hash}$).

The id token created at the token endpoint is sent back to the sender of the 
request in Line~\ref{as-send-token-ep-readwrite} or Line~\ref{as-send-token-ep-read}. 
As shown in Lemma~\ref{lemma:client-authentication},
the request was sent by $c$ (as the client is confidential; more formally, the 
conditions of the lemma are fulfilled as the request contained $m.\str{body}[\str{client\_id}] \equiv \mi{clientId}$ 
due to Lines~\ref{as-token-ep-cid-from-body}, \ref{line:tokenep-checks-clientId} 
and \ref{as-token-ep-add-client-id-to-id-token}).

As an honest client never sends out an id token, it follows that $t$ does not leak to $p$. 
\QED

\end{proof} 

\begin{lemma}[ID Token created for Web Server Clients]\label{lemma:id-token-web-server-client}
  For any run $\rho$
  of a FAPI web system $\fapiwebsystem$
  with a network attacker, 
  every configuration $(S, E, N)$  in $\rho$,
  every authorization server $\mi{as}$ that is honest in $S$,
  every identity $\mi{id} \in \mathsf{ID}^\mi{as}$
  with $b = \mathsf{ownerOfID}(\mi{id})$
  being an honest browser in $S$,
  every web server client $c$ that is honest in $S$ 
  with client id $\mi{clientId}$ 
  that has been issued to $c$
  by $\mi{as}$, %
  every id token $t$ with
  $\mathsf{checksig}(t, \mi{pub}(s^\mi{as}_0.\mi{jwk})) \equiv \True$,
  $\mathsf{extractmsg}(t)[\str{aud}] \equiv \mi{clientId}$  
  and
  $\mathsf{extractmsg}(t)[\str{sub}] \equiv \mi{id}$ 
  and every attacker process $a$ %
  it holds true that 
  $t \not \in  d_{\emptyset}(S(a))$.
\end{lemma}

\begin{proof}

The private signing key used by the authorization server
is preconfigured and never send to other processes. As the signature of the id token is
valid, we conclude that it was created by $\mi{as}$ and that its issuer
parameter has the value $\mi{dom}_\mi{as}$, for $\mi{dom}_\mi{as} \in \mathsf{dom}(\mi{as})$.
More precisely, this value is equal to $\mi{rec}[\str{issuer}]$ (L.~\ref{line:as-create-id-token-auth2}
or L.~\ref{line:as-create-id-token-from-code} of Alg.~\ref{alg:as-fapi}, where $\mi{rec}$ is a record
from $\str{records}$.) This value is only set in Line~\ref{as-rec-set-iss} and is a domain
of $\mi{as}$, as shown in Lemma~\ref{lemma:https-server-correct-hosts}.

If the id token is created at the authorization endpoint (in Line~\ref{line:as-create-id-token-auth2}
of Algorithm~\ref{alg:as-fapi}), then $t \not \in  d_{\emptyset}(S(a))$,
as shown in Lemma~\ref{lemma:auth-resp-does-not-leak-isApp}.
(The conditions of the lemma are fulfilled as the audience value
is $\mi{clientId}$ and taken from the request that was
sent to the authorization endpoint and as the request contained
the password of $\mi{id}$). %

Otherwise, the id token was created at the token endpoint (in Line~\ref{as-token-ep-create-id-token}).
As noted in Section~\ref{fapi-informal-description-of-the-model}, we assume that a web server client is a confidential client.
As shown in Lemma~\ref{lemma:id-token-conf-from-tokenep}, 
it holds true that $t \not \in  d_{\emptyset}(S(a))$.
\QED

\end{proof}

\subsection{Access Token}

\begin{lemma}[Read Access Token does not leak to Attacker]\label{lemma:r-at-does-not-leak}
For any run $\rho$
of a FAPI web system $\fapiwebsystem$
with a network attacker,
every configuration $(S, E, N)$  in $\rho$,
every authorization server $\mi{as}$ that is honest in $S$,
every identity $\mi{id} \in \mathsf{ID}^\mi{as}$
with $b = \mathsf{ownerOfID}(\mi{id})$
being an honest browser in $S$,
every read client $c$
that is honest in $S$
with client id $\mi{clientId}$ %
that has been issued to $c$
by $\mi{as}$,
every access token $t$
with
$\an{\mi{id}, \mi{clientId}, t, \str{r}} \in S(\mi{as}).\str{accessTokens}$,
every resource server $\mi{rs}$ that is honest in $S$ with $\mi{dom}_\mi{rs} \in s^\mi{as}_0.\str{resource\_servers}$
(with $\mi{dom}_\mi{rs} \in \mathsf{dom}(\mi{rs})$)
and
every attacker process $a$
it holds true that if
the request to the resource server is sent to $\mi{dom}_\mi{rs}$ 
(in Line~\ref{r-c-sends-req-to-rs}
of Algorithm~\ref{alg:client-use-access-token}), 
then
$t \not\in d_{\emptyset}(S(a))$.
\end{lemma}

\begin{proof}
  An honest resource server never sends out an access token.
  The client sends the access token only in 
  Line~\ref{r-c-sends-req-to-rs}
  of Algorithm~\ref{alg:client-use-access-token}, where it is sent to an honest 
  resource server (per assumption).

  In the following, we show that the access token associated with the
  read client is never sent to the attacker by the authorization server.

  \begin{description}

    \item[Case 1: Confidential Client (Web Server Client or App)]\strut

Let $t$ be sent in 
Line~\ref{as-send-token-ep-read} of Algorithm~\ref{alg:as-fapi}
($\mathsf{PROCESS\_HTTPS\_REQUEST}$) and let
$m$ be the corresponding request made to the token endpoint of $\mi{as}$.
It holds true that $m$ contains the client id $\mi{clientId}$
(due to Line~\ref{as-token-ep-cid-from-body} and $\mi{clientId}$ being
included in 
$\an{\mi{id}, \mi{clientId}, t, \str{r}}$).

From Lemma~\ref{lemma:client-authentication}, it follows
that $m$ was sent by $c$. This means that the token response containing
the access token is sent back directly to $c$.

    \item[Case 2: Public Client (App)]\strut

Let $m$ be the message received at the token endpoint of $\mi{as}$
with $m.\str{body}[\str{client\_id}] \equiv \mi{clientId}$
such that the access token $t$ is sent back in Line~\ref{as-send-token-ep-read}
of Algorithm~\ref{alg:as-fapi}.

This implies that the sender of $m$ knows a verifier
$\mi{pkceCV}$ such that
$\mathsf{hash}(\mi{pkceCV}) \equiv \mi{record}[\str{pkce\_challenge}]$
(due to the check done in Line~\ref{line:as-checks-r-pkce}),
with $\mi{record}[\str{client\_id}] \equiv \mi{clientId}$.

As shown in Lemma~\ref{lemma:pub-rw-pkce-cc-created-by-c}, %
the corresponding PKCE challenge ($\mi{record}[\str{pkce\_challenge}]$)
was created by the honest client $c$ (as the challenge was signed by $c$ with
preconfigured keys).
The challenge was created in Line~\ref{r-client-creates-pkce-cc}
of Algorithm~\ref{alg:client-fapi-start-login-flow} ($\mathsf{START\_LOGIN\_FLOW}$)
as this is the only place where a read client creates a PKCE challenge.
Let $\mi{reqJWS}$ be the corresponding request JWS that was created in
Line~\ref{client-slf-requestJWS-sign-1}. 

The authorization server checks the request JWS at the authorization endpoint for
each request (Line~\ref{as-auth2-check-jws-start} to Line~\ref{as-auth2-check-jws-end}).
This check also includes the $\mi{aud}$ value of the request JWS,
which is required to be from a domain in $\mathsf{dom}(\mi{as})$ 
(due to Lemma~\ref{lemma:https-server-correct-hosts}).

This implies that $\mathsf{extractmsg}(\mi{reqJWS}[\str{aud}])
\equiv
\mi{authEndpoint}.\str{host} \in \mathsf{dom}(\mi{as})$ 
(Line~\ref{client-slf-reqJWS-aud} of Algorithm~\ref{alg:client-fapi-start-login-flow}).

As shown in Lemma~\ref{lemma:req-jws-aud-token-req-as},
the corresponding token request of the session is sent only
to a domain in $\mathsf{dom}(\mi{as})$ (in Algorithm~\ref{alg:client-token-request}, $\mathsf{SEND\_TOKEN\_REQUEST}$).

The PKCE verifier is a nonce chosen by the client (in Line~\ref{r-client-chooses-pkce-cv} 
of Algorithm~\ref{alg:client-fapi-start-login-flow}),
and send only in the token request. 
Therefore, only $c$ and $\mi{as}$ know the PKCE verifier. As the authorization server
never sends messages to itself, it follows that $m$ was sent by $c$. Thus, the response
(containing the access token) is sent back to $c$.
\QED

  \end{description}
\end{proof}

\begin{lemma}[Access Token bound via mTLS can only be used by Honest Client]\label{lemma:mtls-bound-at}
  For any run $\rho$
  of a FAPI web system $\fapiwebsystem$
  with a network attacker, 
  every configuration $(S, E, N)$  in $\rho$,
  every authorization server $\mi{as}$ that is honest in $S$,
  every read-write client $c$ 
  of type  $\str{conf\_MTLS}$
  that is honest in $S$
  with client id $\mi{clientId}$ issued to $c$ 
  by $\mi{as}$,
  every access token $t$ bound to $c$ (via mTLS, as defined in Appendix~\ref{fapi-def}),
  every resource server $\mi{rs}$ that is honest in $S$
  with $\mi{dom}_\mi{rs} \in s_0^\mi{as}.\str{resource\_servers}$
  (with $\mi{dom}_\mi{rs} \in \mathsf{dom}(\mi{rs})$)
  and every message $m$ received at an URL 
  $\an{\tUrl, \https, \mi{dom_{\mi{rs}}}, \str{/resource\mhyphen{}rw},
  \mi{param}, \mi{frag}}$ with arbitrary $\mi{param}$ and $\mi{frag}$
  and with $\str{MTLS\_AuthN} \in m.\str{body}$ and access token $t \equiv m.\str{header}[\str{Authorization}]$
  it holds true that:

  If a response to $m$ is sent in Line~\ref{rs-send-wNonce} of Algorithm~\ref{alg:rs-oidc},
  then the receiver is $c$.
\end{lemma}

\begin{proof}
  Let $m$ be a message with
  $\str{MTLS\_AuthN} \in m.\str{body}$ containing the access token $t \equiv m.\str{header}[\str{Authorization}]$
  such that a response is sent in Line~\ref{rs-send-wNonce} of Algorithm~\ref{alg:rs-oidc}.

  This implies that the $\str{stop}$ in Line~\ref{rs-retrieve-mtlsReq} 
  and Line~\ref{rs-check-mtls-at} are not executed, and therefore, it holds true that
  \begin{align}
    \notag & m.\str{body}[\str{MTLS\_AuthN}] \equiv S(\mi{rs}).\str{mtlsRequests}.\mi{ptr}.1 \\
    \notag \wedge \; &  \mathsf{check\_mtls\_AT}(\mi{id}, t, S(\mi{rs}).\str{mtlsRequests}.\mi{ptr}.2,
    s^\mi{rs}_0.\str{authServ}) \equiv \True 
  \end{align}
  for some $\mi{ptr}$ and $\mi{id}$, and as $s^\mi{rs}_0.\str{authServ}$ is never
  changed. This is equivalent to 
  \begin{align}
    & m.\str{body}[\str{MTLS\_AuthN}] \equiv S(\mi{rs}).\str{mtlsRequests}.\mi{ptr}.1 \\
    \wedge \;  &  \an{\str{MTLS}, \mi{id}, \mi{clientId}, t, S(\mi{rs}).\str{mtlsRequests}.\mi{ptr}.2, \str{rw}} \in 
    S(\mathsf{dom}^{-1}(s^\mi{rs}_0.\str{authServ})).\str{accessTokens} \label{mtls-bound-at-proof-1}
  \end{align}
  for the client id $\mi{clientId}$ of $c$ (as the access token is bound to $c$) and per definition
  of $\mathsf{check\_mtls\_AT}$ in Appendix~\ref{fapi-def}.

  An entry as in \ref{mtls-bound-at-proof-1} is only created in Line~\ref{as-creates-mtls-token-binding} 
  of Algorithm~\ref{alg:as-fapi} ($\mathsf{PROCESS\_HTTPS\_REQUEST}$).

  As shown in Lemma~\ref{lemma:client-authentication}, the request $m'$ that led to the creation of this sequence
  was sent from the honest client $c$ with client id $\mi{clientId}$ (intuitively, the client of type
  $\str{conf\_MTLS}$ has authenticated itself at the token endpoint; more formally, the preconditions of the lemma
  are fulfilled as there is always a response sent
  by the authorization server after adding an entry to $\mi{accessTokens}$ and the corresponding request contained
  $m'.\str{body}[\str{client\_id}] \equiv \mi{clientId}$). 

  The value $\mi{mtlsInfo}.2$ in Line~\ref{as-creates-mtls-token-binding} of Algorithm~\ref{alg:as-fapi} is retrieved
  from $S'(\mi{as}).\str{mtlsRequests}[\mi{clientId}]$ (in Line~\ref{as-token-ep-retrieve-mtlsInfo}, for
  a state $S'$ prior to $S$ within the run). 
  The entries of $\str{mtlsRequests}$ are only created in Line~\ref{as-add-mtlsInfo} of Algorithm~\ref{alg:as-fapi},
  where the entry corresponding to the key $\mi{clientKey}$ is taken from $s^{\mi{as}}_0.\str{clients}[\mi{clientId}][\str{mtls\_key}]$ 
  (due to Line~\ref{line:as-choses-mtls-key} and the entries of $\mi{clients}$ not being changed by the authorization server).
  Per definition, this key has the value $\mi{keyMapping}[\mi{dom_{c}}]$ with $\mi{dom_c} \in \mathsf{dom}(c)$.

  This key is equal to $S(\mi{rs}).\str{mtlsRequests}.\mi{ptr}.2$ (due to \ref{mtls-bound-at-proof-1}). This sequence
  was added to
  $S(\mi{rs}).\str{mtlsRequests}$ in Line~\ref{rs-adds-mtlsreq} of 
  Algorithm~\ref{alg:rs-oidc}, where the key was taken from the corresponding request
  (Line~\ref{rs-get-pub-key-mtls}). The nonce chosen in Line~\ref{rs-chooses-mtls-nonce} is sent as \\
  $\mathsf{enc_a}(\an{\mi{mtlsNonce}, \mi{keyMapping}(m.\str{host})},\mi{clientKey})$ (Line~\ref{rs-encrypts-mtlsnonce}),
  where the asymmetric key is the public key of $c$ (as shown above; the key is the same as 
  the key that the authorization server used for encrypting a nonce for $c$).
  As shown in Lemma~\ref{lemma:https-server-correct-hosts}, $m.\str{host}$ is a domain of the resource server.

  Analogous to Lemma~\ref{lemma:as-mtls-nonce-does-not-leak}, this nonce is only known to $c$ 
  (and $\mi{rs}$) and
  does not leak to any other process.
  It follows that the request $m$ was sent by $c$,
  and therefore, the response is sent back to $c$.
  \QED

\end{proof}

\begin{lemma}[Private OAUTB Key does not leak]\label{lemma:oautb-priv-key-does-not-leak}
  For any run $\rho$
  of a FAPI web system $\fapiwebsystem$
  with a network attacker, 
  every configuration $(S, E, N)$  in $\rho$,
  every authorization server $\mi{as}$ that is honest in $S$,
  every read-write client $c$ 
  of type  $\str{pub}$ or $\str{conf\_OAUTB}$
  that is honest in $S$
  with client id $\mi{clientId}$ issued to $c$ 
  by $\mi{as}$,
  every $\mi{id} \in \mathsf{ID}^\mi{as}$, every access token $t$ %
  and every process $p$
  it holds true that:
  
  If 
  $\an{\str{OAUTB}, \mi{id},  \mi{clientId}, t, %
    \mathsf{pub}(\mi{TB\_ref\_key}), \str{rw}} 
  \in S(\mi{as}).\str{accessTokens}$ and
  $\mi{TB\_ref\_key} \in d_{\emptyset}(S(p))$, then  $p = c$.

\end{lemma}

\begin{proof} 

Let $m$ be the token request that led to the creation of the sequence 
in $S(\mi{as}).\str{accessTokens}$. It holds true that 
$m.\str{body}[\str{client\_id}] \equiv \mi{clientId}$, as this value is included in the
sequence 
$\an{\str{OAUTB}, \mi{id},  \mi{clientId}, t,  \mathsf{pub}(\mi{TB\_ref\_key}), \str{rw}}$
(due to Line~\ref{as-token-ep-cid-from-body} of Algorithm~\ref{alg:as-fapi}).

  \begin{description} 
    \item[Case 1: Client of Type $\str{conf\_OAUTB}$]\strut

After the sequence is added to 
$S(\mi{as}).\str{accessTokens}$ (in Line~\ref{as-creates-oautb-token-binding} of Algorithm~\ref{alg:as-fapi}),
there are no further checks that lead to a $\str{stop}$, which means that 
a response is sent in Line~\ref{as-send-token-ep-readwrite}.
As shown in Lemma~\ref{lemma:client-authentication},
$m$ was sent by $c$.

The only
place where an honest client of type $\str{conf\_OAUTB}$
sends a message to the token endpoint of an authorization server 
is in Line~\ref{line:jws-oautb-cont-send} of Algorithm~\ref{alg:client-token-request}.
This means that $m$ was created by $c$ in Line~\ref{line:jws-oautb-msg}
of Algorithm~\ref{alg:client-token-request}.

Let $\mi{TB\_ref\_pub\_key}$ be the key stored in the sequence shown above (in $S(\mi{as}).\str{accessTokens}$).

It holds true that
\begin{align*}
  & \mi{TB\_ref\_pub\_key} & \\
  \equiv \; & m.\str{headers}[\str{Sec\mhyphen{}Token\mhyphen{}Binding}][\str{ref}][\str{id}] 
  & \text{(Line~\ref{as-token-ep-ref-tb-key-1}, Line~\ref{as-token-ep-ref-tb-key-2} of Alg.~\ref{alg:as-fapi})} \\
  \equiv \; & \mathsf{pub}(s^c_0.\str{tokenBindings}[\mi{t}]) & \text{(Line~\ref{c-str-retrieve-rs-tb}, Line~\ref{c-str-compose-ref-tb-msg}, Line~\ref{c-str-sec-tb-header} of Alg.~\ref{alg:client-token-request})} 
\end{align*}

for some term $t$ (which is the domain of a resource server, but not relevant at this point).
We also note that the value of $\str{tokenBindings}$ of the client state is the same as in the
initial state as it is not changed by the client.

    \item[Case 2: Client of Type $\str{pub}$]\strut

Let $\mi{record}$ be the record chosen in Line~\ref{as-token-ep-retrieves-record}
of Algorithm~\ref{alg:as-fapi} when $m$ is received at the token endpoint and the 
sequence is added to $S'(\mi{as}).\mi{accessTokens}$ (for a state $S'$
prior to $S$ within the run).
Due to Lines~\ref{as-token-ep-cid-from-body}, \ref{line:tokenep-checks-clientId} and \ref{as-creates-oautb-token-binding},
it holds true that $\mi{record}[\str{client\_id}] \equiv \mi{clientId}$.

As shown in Lemma~\ref{lemma:pub-rw-pkce-cc-created-by-c}, the PKCE challenge contained in
$\mi{record}[\str{pkce\_challenge}]$ was created by $c$. This happens only in
Line~\ref{pub-rw-client-creates-pkce-cc} of Algorithm~\ref{alg:client-fapi-start-login-flow}
($\mathsf{START\_LOGIN\_FLOW}$).

This implies that $\mi{record}[\str{pkce\_challenge}]$ has the value
$\mathsf{hash}(\mathsf{pub}(s^c_0.\str{tokenBindings}[t']))$,
for some term $t'$.
As the private keys in $s^c_0.\str{tokenBindings}$ are preconfigured and never send to any process,
they are only known to $c$.

As $m.\str{body}[\str{client\_id}]$ is the client id of a read-write client of type
$\str{pub}$, 
$m$ is required to contain an OAUTB provided message 
$\mi{TB\_Msg\_provided} \equiv m.\str{headers}[\str{Sec\mhyphen{}Token\mhyphen{}Binding}][\str{prov}]$ 
such that \\
$\mathsf{checksig}(\mi{TB\_Msg\_provided}[\str{sig}], \mi{TB\_Msg\_provided}[\str{id}]) \equiv \True$
(due to Lines~\ref{as-token-ep-get-prov-id}, \ref{as-token-ep-get-prov-sig} and \ref{as-token-ep-check-prov-tb-msg}).

Due to Line~\ref{as-rw-pub-check-pkce}, it holds true that
$\mathsf{hash}(\mi{TB\_Msg\_provided}[\str{id}]) \equiv \mi{record}[\str{pkce\_challenge}]$,
which is equal to $\mathsf{hash}(\mathsf{pub}(s^c_0.\str{tokenBindings}[t']))$.
Therefore, 
$\mi{TB\_Msg\_provided}[\str{id}] \equiv \mathsf{pub}(s^c_0.\str{tokenBindings}[t'])$.
This means that 
$\mi{TB\_Msg\_provided}[\str{sig}]$ was signed by $c$ (as only $c$ knows the corresponding 
private key). 

Due to Line~\ref{as-token-ep-check-prov-tb-msg}, it follows that 
the provided Token Binding message is equal to $\mi{ekmInfo}$, which is
taken from $S'(\mi{as}).\str{oautbEKM}$. These values are only
created in Line~\ref{as-create-and-add-oautbEKM} of Algorithm~\ref{alg:as-fapi}.
Therefore, $\mi{ekmInfo}$ is equal to $\mathsf{hash}(\an{n_1, n_2, \mathsf{keyMapping}(\mi{dom}_\mi{as})})$,
for some values $n_1$, $n_2$.

As shown in Lemma~\ref{lemma:ekm-signed-by-c-does-not-leak}, 
the signed EKM value does not leak, and therefore, $m$ was sent by $c$.

To conclude this case, we note that $m$ was created in Algorithm~\ref{alg:client-token-request},
where the key of the referred Token Binding is exactly the same as in Case 1. 

  \end{description} 
In both cases, the private key $\mi{TB\_ref\_key}$ is taken from 
$s^c_0.\str{tokenBindings}$, which is preconfigured and never sent
to any other process.
Consequently, only $c$ knows this value. 
\QED

\end{proof}

\begin{lemma}[Access Token bound via OAUTB can only be used by Honest Client]\label{lemma:oautb-bound-at}
  For any run $\rho$
  of a FAPI web system $\fapiwebsystem$
  with a network attacker, 
  every configuration $(S, E, N)$  in $\rho$,
  every authorization server $\mi{as}$ that is honest in $S$,
  every read-write client $c$ 
  of type  $\str{conf\_OAUTB}$ or $\str{pub}$
  that is honest in $S$
  with client id $\mi{clientId}$ issued to $c$ 
  by $\mi{as}$,
  every access token $t$ bound to $c$ (via OAUTB, as defined in Appendix~\ref{fapi-def}),
  every resource server $\mi{rs}$ that is honest in $S$
  with $\mi{dom}_\mi{rs} \in s_0^\mi{as}.\str{resource\_servers}$
  (with $\mi{dom}_\mi{rs} \in \mathsf{dom}(\mi{rs})$)
  and every message $m$ received at an URL 
  $\an{\tUrl, \https, \mi{dom_{\mi{rs}}}, \str{/resource\mhyphen{}rw},
  \mi{param}, \mi{frag}}$ with arbitrary $\mi{param}$ and $\mi{frag}$
  with $\str{MTLS\_AuthN} \not \in m.\str{body}$ and access token $t \equiv m.\str{header}[\str{Authorization}]$
  it holds true that:
  
  If a response to $m$ is sent, then the receiver is $c$.

\end{lemma}

\begin{proof}
  Let $m$ be a message 
  with $\str{MTLS\_AuthN} \not \in m.\str{body}$ and access token $t \equiv m.\str{header}[\str{Authorization}]$
  such that a response to this message is sent in Line~\ref{rs-send-wNonce} of Algorithm~\ref{alg:rs-oidc}.

  This implies that all applicable checks until Line~\ref{rs-send-wNonce} have passed successfully.
  Therefore, it holds true that 
  \begin{align}
     & \mathsf{checksig}(\mi{TB\_prov\_sig}, \mi{TB\_prov\_pub}) \equiv \True
         & \text{(L.~\ref{rs-rw-check-oautb-sig})} \label{lemma-oautb-binding-check-1}\\ %
     & \mi{TB\_prov\_msg} \equiv S(\mi{rs}).\str{oautbEKM}.\mi{ptr}
         & \text{(L.~\ref{rs-rw-check-oautb-ekm} and \ref{rs-rw-retrieve-oautbEKM})} \label{lemma-oautb-binding-check-2}\\
     & \mathsf{check\_oautb\_AT}(\mi{id}, t, \mi{TB\_prov\_pub}, s^\mi{rs}_0.\str{authServ}) 
               \equiv \True	
          & \text{(L.~\ref{rs-check-oautb-binding})} \label{lemma-oautb-binding-check-3} 
  \end{align}

  for some $\mi{ptr}$, $\mi{id}$ and 
  with $\mi{TB\_prov\_sig} := m.\str{headers}[\str{Sec\mhyphen{}Token\mhyphen{}Binding}][\str{prov}][\str{sig}]$, \\
  $\mi{TB\_prov\_pub} := m.\str{headers}[\str{Sec\mhyphen{}Token\mhyphen{}Binding}][\str{prov}][\str{id}]$ and \\
  $\mi{TB\_prov\_msg} := \mathsf{extractmsg}(\mi{TB\_prov\_sig})$. 
  (In this case, $\mi{id}$ can be an arbitrary identity. Even if this is the identity of
  an attacker, the token is still bound to the client). We also note that 
  $s^\mi{rs}_0.\str{authServ}$ is never changed by the resource server.

  Due to \ref{lemma-oautb-binding-check-3}, it follows that \\
  $\an{\str{OAUTB}, \mi{id},  \mi{clientId}, t, \mi{TB\_prov\_pub}, \str{rw}}  \in S(\mathsf{dom}^{-1}(s^\mi{rs}_0.\str{authServ})).\str{accessTokens}$
  (per definition of $\mathsf{check\_oautb\_AT}$), for the client id $\mi{clientId}$ of $c$ (as the access token is bound to $c$).

  Combining this with \ref{lemma-oautb-binding-check-1} %
  and Lemma~\ref{lemma:oautb-priv-key-does-not-leak}, it follows that
  the provided Token Binding message was created by $c$ (as the signature is valid and the corresponding private key is only
  known to $c$).

  We conclude the proof by showing that such a Token Binding message is directly sent from the client to the
  resource server, and therefore, cannot leak to and be used by another process. 

  Due to \ref{lemma-oautb-binding-check-2}, $c$ has signed 
	$\mathsf{hash}(\an{n_1, n_2, \mi{keyMapping}(dom_\mi{rs})})$ 
  for some values $n_1, n_2$ and with $dom_\mi{rs} \in \mathsf{dom}(\mi{rs})$
  (the only place where these values are created and added to the state of the
  resource server is in Line~\ref{rs-creates-and-adds-oautbEKM-to-state}. As shown
  in Lemma~\ref{lemma:https-server-correct-hosts}, the host of messages returned by
  the HTTPS generic server is a domain of the resource server).

  As shown in Lemma~\ref{lemma:ekm-signed-by-c-does-not-leak}, the signed EKM value
  does not leak to another process. Therefore, $m$ was sent by $c$, which means that the 
  response containing the resource access nonce is sent back to $c$.
  \QED

\end{proof}

\subsection{Authorization}
\begin{lemma}[Authorization]\label{theorem:authorization}
  For every run $\rho$
  of a FAPI web system $\fapiwebsystem$
  with a network attacker, 
  every configuration $(S, E, N)$  in $\rho$,
  every authorization server $\mi{as} \in \fAP{AS}$
  that is honest in $S$ with $s^\mi{as}_0.\str{resource\_servers}$ being
  domains of honest resource servers,
  every identity $\mi{id} \in \mathsf{ID}^\mi{as}$ 
  with $b = \mathsf{ownerOfID}(\mi{id})$
  being an honest browser in $S$,
  every client $c \in \fAP{C}$ that is honest in $S$
  with client id $\mi{clientId}$ issued to $c$ by $\mi{as}$,
  every resource server $\mi{rs} \in \fAP{RS}$
  that is honest in $S$ such that
  $\mi{id} \in s^\mi{rs}_0.\str{ids}$,
  $s^\mi{rs}_0.\str{authServ} \in \mathsf{dom}(\mi{as})$ and
  with $\mi{dom}_\mi{rs} \in s_0^\mi{as}.\str{resource\_servers}$
  (with $\mi{dom}_\mi{rs} \in \mathsf{dom}(\mi{rs})$),
  every access token $t$ associated with $c$, $\mi{as}$ and $\mi{id}$
  and every resource access nonce $r \in s^\mi{rs}_0.\str{rNonce}[\mi{id}] \cup s^\mi{rs}_0.\str{wNonce}[\mi{id}]$
  it holds true that:

  If $r$ is contained in a response to a request $m$ sent to $\mi{rs}$
  with $t \equiv m.\mi{header}[\str{Authorization}]$,
  then 
  $r$ is not derivable from the attackers knowledge in $S$
  (i.e., $r \not\in d_{\emptyset}(S(\fAP{attacker}))$).

\end{lemma}

\begin{proof}

Let $c$, $\mi{as}$, $\mi{rs}$, $t$, $r$ and $m$ be given as in the description of the lemma.

\begin{description} 
 \item[Resource Server never sends Resource Access Nonce $r$ to Attacker.]\strut

We assume that the resource server sends $r$ to the attacker.
Consequently, the attacker sent the message $m$ with the access token $t$ to either
the
$\str{/resource\mhyphen{}r}$ or the 
$\str{/resource\mhyphen{}rw}$ path of the resource server.

\textbf{Case 1:  $\str{resource\mhyphen{}r}$}

As the attacker receives $r$, we conclude that the check done in Line~\ref{rs-resource-r-check-at}
of Algorithm~\ref{alg:rs-oidc} ($\mathsf{PROCESS\_HTTPS\_REQUEST}$) passes successfully
and the identity chosen is this line is $\mi{id}$ (as a resource nonce associated with $\mi{id}$ is returned).

Therefore, it holds true that
$\mathsf{check\_read\_AT}(\mi{id}, t, s^\mi{rs}_0.\str{authServ}) \equiv \True$
($S(\mi{rs}).\str{authServ}$ is the same as in the initial state, as it is not modified by the resource server),
and it follows that
$\an{\mi{id},  \mi{clientId}, t,  \str{r}}  \in %
S(\mi{as}).\str{accessTokens}$
(per definition of $\mathsf{check\_read\_AT}$ and as $t$ is associated with $c$). 

This sequence is only added to $\str{accessTokens}$ by the authorization server
if $c$ is a read client (L.~\ref{as-adds-sequence-for-r-client} of Alg.~\ref{alg:as-fapi}).

This means that the attacker is in possession of a valid access token for the read client
$c$, which is a contradiction to Lemma~\ref{lemma:r-at-does-not-leak}.

\textbf{Case 2:  $\str{resource\mhyphen{}rw}$}

If $\str{MTLS\_AuthN} \in m.\str{body}$, then 
the check done in Line~\ref{rs-check-mtls-at} passes successfully,
which means that 
$\an{\str{MTLS}, \mi{id}, \mi{clientId}, t , \mi{key}, \str{rw}} \in S(\mi{as}).\str{accessTokens}$,
for some $\mi{key}$. This means that
the access token $t$ is bound to $c$
via mTLS and that the client $c$ is a read-write client 
(due to L.~\ref{as-adds-sequence-for-rw-client} of Alg.~\ref{alg:as-fapi}). 
Therefore, Lemma~\ref{lemma:mtls-bound-at} holds true,
which is a contradiction to the assumption that the response %
is sent to the attacker. 

Otherwise, $\str{MTLS\_AuthN} \not \in m.\str{body}$, and the access token $t$ is bound to $c$
via OAUTB (with the same reasoning as in the case of mTLS shown above).
Furthermore, $c$ is a read-write client of
type $\str{conf\_OAUTB}$ or $\str{pub}$.
Now, Lemma~\ref{lemma:oautb-bound-at} holds true, which is again a contradiction to 
the assumption that the response containing $r$
is sent to the attacker. 

\end{description} 

We highlight that if $c$ is a read-write client, the resource server sends $r$ 
only to $c$, which follows directly from 
Lemma~\ref{lemma:mtls-bound-at} and
Lemma~\ref{lemma:oautb-bound-at}.

\begin{description} 

 \item[Client never Sends Resource Access Nonce $r$ to Attacker.]\strut

In the following, we will show that the resource nonce $r$ received in 
Line~\ref{client-receives-resource}
of Algorithm~\ref{alg:client-fapi-http-response}
($\mathsf{PROCESS\_HTTPS\_RESPONSE}$)
is not sent to the attacker.

\textbf{Case 1:} $c$ is an app client %

In this case, the resource nonce is not sent at all in Algorithm~\ref{alg:client-fapi-http-response},
due to the check in Line~\ref{ws-client-send-r-to-browser}.
(Intuitively, the resource is used directly by the app).

\textbf{Case 2:} $c$ is a web server client %

We assume that the resource nonce received in Line~\ref{client-receives-resource}
of Algorithm~\ref{alg:client-fapi-http-response} is sent to the attacker.

The only place where the resource nonce is sent by the client is in Line~\ref{line:client-send-resource}
of Algorithm~\ref{alg:client-fapi-http-response}, as this is the only
place where the client uses the resource nonce. The nonce saved in the session
in Line~\ref{client-saves-resource} is not used by the client at any other place.
The resource nonce is sent to $\mi{request}[\str{sender}]$, where
$\mi{request}$ is retrieved from
$S(c).\str{sessions}[\mi{sid}][\str{redirectEpRequest}]$,
for some $\mi{sid}$.

The only place where 
$S(c).\str{sessions}[\mi{sid}][\str{redirectEpRequest}]$
is set by the client is in Line~\ref{line:client-set-redirect-ep-request-record} of 
Algorithm~\ref{alg:client-fapi-http-request} (at the redirection endpoint).
Intuitively, this means that the resource access nonce is sent back to the sender of the request to the redirection endpoint.

Let $m'$ be the corresponding request that was received at 
$\str{/redirect\_ep}$ (L.~\ref{line:client-redir-endpoint}
of Alg.~\ref{alg:client-fapi-http-request}). As we
assume that the resource nonce is sent to the attacker, it follows
that $m'$ was sent by the attacker.
The values of $\str{redirectEpRequest}$ are set to the corresponding values 
of $m'$.

\textbf{Subcase 2.1: read client}\\
If the client is a read client, then the token endpoint is 
chosen correctly. More precisely, 
the access token $t$ that is used by the client is associated with 
$c$, $\mi{as}$ and $\mi{id}$. Per definition, it follows
that
$\an{\mi{id}, \mi{clientId}, t, \str{r}} \in S(\mi{as}).\str{accessTokens}$.
As shown in Lemma~\ref{lemma:r-at-does-not-leak}, it holds true that
the access token $t$ does not leak to the attacker, and therefore, the client
sent the token request to $\mi{as}$ in order to get the access token.

The code included in the token request
can only be provided 
at the redirection endpoint (as shown in 
Lemma~\ref{lemma:code-token-req-from-redirect-ep}), which means that $m'$
contained the code that will be used by the client.
This means that there is a record $\mi{rec}$ at the authorization server
that associates the code with the identity and the client id.
More precisely, it holds true that 
$\mi{rec} \inPairing S'(\mi{as}).\str{records}$ with
$\mi{rec}[\str{code}] \equiv \mi{code}$,
$\mi{rec}[\str{client\_id}] \equiv \mi{clientId}$ and
$\mi{rec}[\str{subject}] \equiv \mi{id}$ (as we assume that
the resource is owned by $\mi{id}$ and due to 
L.~\ref{as-token-ep-retrieves-record}, L.~\ref{line:tokenep-checks-clientId} and L.~\ref{line:as-creates-r-at-sequence}
of Alg.~\ref{alg:as-fapi}),
for a state $S'$ prior to $S$.
As the read client is a web server client, 
all conditions of Lemma~\ref{lemma:r-web-server-client-code}
are fulfilled, which means that the attacker cannot know such a code.

\textbf{Subcase 2.2: read-write client}\\ 
As the resource nonce $r$ was sent to $c$ by $\mi{rs}$, it follows
that $m$ was sent by $c$.

\textbf{Subcase 2.2.1: OpenID Hybrid Flow: }
As Line~\ref{rs-send-wNonce} of Algorithm~\ref{alg:rs-oidc} is executed,
it follows that the check in Line~\ref{rs-check-at-iss} passes successfully,
and therefore, 
$m.\str{body}[\str{at\_iss}] \equiv s^\mi{rs}_0.\str{authServ} \in \mathsf{dom}(\mi{as})$
(the value of $\str{authServ}$ is not changed by an honest resource server and stays the same 
as in the initial state).

The client sends messages to the $\str{resource\mhyphen{}rw}$ path of a resource server
only in Line~\ref{c-sends-msg-to-resource-rw} of Algorithm~\ref{alg:client-use-access-token}
($\mathsf{USE\_ACCESS\_TOKEN}$), hence, 
$m.\str{body}[\str{at\_iss}] \equiv S(c).\str{sessions}[\mi{sid'}][\str{idt2\_iss}]$,
for some $\mi{sid'}$ (Line~\ref{c-sets-at-iss} of 
Algorithm~\ref{alg:client-use-access-token}). %

The value of 
$\str{idt2\_iss}$ is only set in Line~\ref{c-set-idt2-iss}
of Algorithm~\ref{alg:client-fapi-http-response} ($\mathsf{PROCESS\_HTTPS\_RESPONSE}$),
and therefore, the id token received in the token response 
has an issuer value from $\mathsf{dom}(\mi{as})$.

Due to Line~\ref{c-token-resp-check-issuer} of Algorithm~\ref{alg:client-fapi-http-response},
it follows that $\mi{issuer} \in \mathsf{dom}(\mi{as})$, and therefore,
$s^c_0.\str{jwksCache}[\mi{issuer}] \equiv \mathsf{pub}(s^\mi{as}_0.\str{jwk})$
(per definition of $\str{jwksCache}$).

This means that the second id token contained in the token response
was signed by $\mi{as}$ (Line~\ref{c-checksig-second-idt}), as the private key is
only known to $\mi{as}$. Therefore, this id token was created by $\mi{as}$.

As the id token has the correct hash value for the access token $t$ (Line~\ref{c-check-at-hash})
and as the id token was created by $\mi{as}$ (in Line~\ref{as-token-ep-create-id-token} of 
Algorithm~\ref{alg:as-fapi}), it follows that either
$\an{\str{MTLS}, \mi{id'}, \mi{clientId}, t, \mi{key}_1 , \str{rw}} \in S(\mi{as}).\str{accessTokens}$ or
$\an{\str{OAUTB}, \mi{id'}, \mi{clientId}, t, \mi{key}_2 , \str{rw}} \in S(\mi{as}).\str{accessTokens}$,
for some values $\mi{key}_1$, $\mi{key}_2$.

As the access token $t$ is associated with $c$, $\mi{as}$ and $\mi{id}$, it follows
that $\mi{id'} \equiv \mi{id}$, and therefore, the subject attribute of the second id token
is equal to $\mi{id}$ (due to Line~\ref{as-token-ep-add-sub-to-id-token}).

As shown in Lemma~\ref{lemma:id-token-web-server-client}, such an id token does
not leak to an attacker, therefore, we conclude that the token response was sent by $\mi{as}$.

Let $\mi{id\_token\_authep}$ be the id token contained in $m'$
(the message that is received at the redirection endpoint). 

Due to Line~\ref{c-compare-sub-of-both-id-tokens}
of Algorithm~\ref{alg:client-fapi-http-response},
it follows that 
$\mi{id\_token\_authep}[\str{sub}] \equiv \mi{id}$.
Furthermore, 
$\mi{id\_token\_authep}[\str{iss}]$ has the same value
as in the second id token (due to Line~\ref{c-compare-iss-of-both-id-tokens}),
and therefore, 
$\mi{id\_token\_authep}$ was signed by $\mi{as}$ (with the same reasoning as above).

As the client always checks that the audience value of
id tokens are equal to its own client id,
again %
Lemma~\ref{lemma:id-token-web-server-client}
holds true, which is a contradiction to the assumption that 
the client sends the resource access nonce to the attacker, as otherwise,
an id token with a valid signature from $\mi{as}$ with $\mi{id}$
being its subject value and the client id of $c$ 
being its audience value
would have leaked to the attacker.

\textbf{Subcase 2.2.2: Code Flow with JARM:}
As in the case of the Hybrid Flow, the check done in Line~\ref{rs-check-at-iss}
of Algorithm~\ref{alg:rs-oidc} passes successfully and it holds true that
$m.\str{body}[\str{at\_iss}] \equiv s^\mi{rs}_0.\str{authServ} \in \mathsf{dom}(\mi{as})$.

The client sets this value only in Line~\ref{c-sets-at-iss-JARM} of 
Algorithm~\ref{alg:client-use-access-token}, where it is set to
$S(c).\str{sessions}[\mi{sid'}][\str{JARM\_iss}]$ (with $\mi{sid'}$ being the 
session identifier of the corresponding session).

This value is only set in Line~\ref{c-set-JARM-iss} of 
Algorithm~\ref{alg:client-check-response-jws}, where it is set to the issuer
of the response JWS $\mi{respJWS}$.  %
As Algorithm~\ref{alg:client-check-response-jws}
is only called at the redirection endpoint of the client (in Line~\ref{line:c-call-check-response-jws}
of Algorithm~\ref{alg:client-fapi-http-request}), it follows that this response 
JWS was sent by the attacker.

Therefore, we conclude that 
$\mathsf{extractmsg}(\mi{respJWS})[\str{iss}] \in \mathsf{dom}(\mi{as})$.

Due to the checks done in Lines~\ref{c-check-respJWS-retrieve-jwks} 
and \ref{c-check-response-jws-aud} of Algorithm~\ref{alg:client-check-response-jws},
it follows that 
$\mi{jwks} \equiv s_0^c.\str{jwksCache}[\mi{dom\_as}] \equiv \mathsf{pub}(s_0^\mi{as}.\str{jwk})$
(with $\mi{dom\_as} \in \mathsf{dom}(\mi{as})$), which means that the response JWS
is signed by $\mi{as}$. \\
The request to the resource server was sent by $c$, which means that 
the checks done by $c$ prior to sending the request passed successfully,
i.e., the hash of the access token received in the token response was contained in
the response JWS.
More precisely, it holds true 
\begin{align}
  \notag & \mathsf{extractmsg}(\mi{respJWS})[\str{at\_hash}]  \\
  \shortintertext{(Line~\ref{line:client-set-redirect-ep-request-record} of Alg.~\ref{alg:client-fapi-http-request})}
  \notag & \equiv \mathsf{extractmsg}(S''(c).\str{sessions}[\mi{sid'}][\str{redirectEpRequest}][\str{data}][\str{responseJWS}])[\str{at\_hash}]  \\
  \shortintertext{(Line~\ref{JARM-get-reqJWS-from-redireEpReq-1} and \ref{JARM-c-check-at-hash} of Alg.~\ref{alg:client-fapi-http-response})}
  \notag & \equiv \mathsf{extractmsg}(m''.\str{body}[\str{access\_token}])
\end{align}
for a state $S'$ prior to $S$ and $m''$ being the token response.
The token sent to $\mi{rs}$ in Line~\ref{c-sends-msg-to-resource-rw}
of Algorithm~\ref{alg:client-use-access-token}
is the input argument of the algorithm. 
In case of the Read-Write profile, Algorithm~\ref{alg:client-use-access-token} is only called in 
Line~\ref{mtls-authn-rs-use-access-token} or Line~\ref{line:PHResp-UAT-oautb-call} of
Algorithm~\ref{alg:client-fapi-http-response}. In both cases, the token is taken from
$S'''(c).\str{sessions}[\mi{sid'}][\str{token}]$ %
which is only set in Line~\ref{c-save-access-token-in-session} of 
Algorithm~\ref{alg:client-prepare-use-access-token}.
Here, the token is the input argument of the algorithm, 
which is only called in Line~\ref{line:client-call-use-access-token-2}
of Algorithm~\ref{alg:client-fapi-http-response}.
This is the access token received in the token response %
(i.e., equal to $m''.\str{body}[\str{access\_token}]$).
Therefore, we conclude that the response JWS contains the hash of the
access token that the client sent to the resource server.

As the resource was sent by $\mi{rs}$ due to an access token $t$ associated with
$c$, $\mi{as}$ and $\mi{id}$, it follows per definition that either
$\an{\str{MTLS}, \mi{id}, \mi{clientId}, t, \mi{key}, \str{rw}} \in S(\mi{as}).\str{accessTokens}$ or
$\an{\str{OAUTB}, \mi{id}, \mi{clientId}, t, \mi{key'}, \str{rw}} \in S(\mi{as}).\str{accessTokens}$,
for some values $\mi{key}$ and $\mi{key'}$.

The values of the identity and client identifier are both taken from a record %
$\mi{rec} \in \bar{S}(\mi{as}).\str{records}$ %
with $\mi{rec}[\str{access\_token}] \equiv t$
(Lines~\ref{line:tokenep-checks-clientId},
\ref{as-tep-get-at-from-record},
\ref{as-creates-mtls-token-binding} and
\ref{as-creates-oautb-token-binding}
of Algorithm~\ref{alg:as-fapi}).

This record is created at the authorization endpoint of the 
authorization server, directly before the response JWS is created
in Line~\ref{as-creates-response-jws} of Algorithm~\ref{alg:as-fapi}.

For creating access tokens, the authorization server chooses fresh nonces for each authorization
request (Line~\ref{as-creates-at}). Therefore, the hash of the access token 
included in the JWS is unique for each authorization response %
Consequently, the response JWS $\mi{respJWS}$ (from above), contains the code $\mi{rec}[\str{code}]$
(the code that contained in the same record as the access token $t$). 
However, this record also contains the identity $\mi{id}$.
This is a contradiction to Lemma~\ref{lemma:req-jws-does-not-leak}, as the attacker cannot
know a response JWS that contains a code associated with $\mi{id}$, signed by $\mi{as}$
and with the client id of $c$. %

\QED

\end{description}

\end{proof}

\subsection{Authentication}
\begin{lemma}[Authentication]\label{theorem:authentication}
  For every run $\rho$
  of a FAPI web system $\fapiwebsystem$
  with a network attacker, %
  every configuration $(S, E, N)$  in $\rho$,
  every client $c \in \fAP{C}$ that is honest in $S$,
  every identity $\mi{id} \in \mathsf{ID}$
  with $\mi{as} = \mathsf{governor}(\mi{id})$ %
  being an honest authorization server
  and
  with $b = \mathsf{ownerOfID}(\mi{id})$
  being an honest browser in $S$,
  every service session identified by some nonce
  $n$
  for $\mi{id}$
  at $c$,
  $n$
  is not derivable from the attackers knowledge in $S$
  (i.e., $n \not\in d_{\emptyset}(S(\fAP{attacker}))$).
\end{lemma}

\begin{proof}
  Let $\mi{clientId}$ be the client id that has been issued to $c$
  by $\mi{as}$.

  If the client is an app client, %
  then no service session id is sent 
  due to the check done in Line~\ref{c-ssid-check-isApp} of Algorithm~\ref{alg:client-check-id-token}
  ($\mathsf{CHECK\_ID\_TOKEN}$).
  In the following, we look at the case of a web server client. %

  We assume that the service session id is sent to the attacker by
  the client. This happens only in Line~\ref{line:client-send-set-service-session}
  of Algorithm~\ref{alg:client-check-id-token}. 
  
  This message is sent to $\mi{request}[\str{sender}]$, where 
  $\mi{request} \equiv
  S(c).\str{sessions}[\mi{sessionId}][\str{redirectEpRequest}]$,
  for some $\mi{sessionId}$ (Line~\ref{c-CIT-retrieves-request}).
  This value is only set at the redirection endpoint (in Line~\ref{line:client-set-redirect-ep-request-record}
  of Algorithm~\ref{alg:client-fapi-http-request}, $\mathsf{PROCESS\_HTTPS\_REQUEST}$) and is the
  sender of the message $m$ received at the redirection endpoint of the client. In other words, the service session
  id is sent back to the sender of the redirection message, and it follows that $m$ was sent by the attacker.

  Per Definition~\ref{def:service-sessions} (Service Sessions),
  it holds true that
  $S(c).\str{sessions}[\mi{sessionId}][\str{loggedInAs}] \equiv \an{d, \mi{id}}$, 
  where 
  $d \in \mathsf{dom}(\mathsf{governor}(\mi{id}))$.
  As the identity is governed by $\mi{as}$, it follows that $d$ is a domain
  of $\mi{as}$, therefore, the value for $\mi{issuer}$ used in Algorithm~\ref{alg:client-check-id-token} is a domain of $\mi{as}$
  (Line~\ref{c-CIT-retrieve-sub}).

  $\mathsf{CHECK\_ID\_TOKEN}$ is only called in Line~\ref{line:client-call-check-id-token-after-code}
          of Algorithm~\ref{alg:client-fapi-http-response}

  As shown in Lemma~\ref{lemma:code-token-req-from-redirect-ep}, the authorization code
  used for the corresponding token request was included in the request to the redirection endpoint,
  i.e., the attacker knows this code.

  If the client is a read client, then the token request is sent to the authorization server $\mi{as}$
  (as the token endpoint might only be misconfigured in the read-write flow).
  More precisely, it holds true that $\mi{issuer} \in \mathsf{dom}(\mi{as})$, as shown above.
  Therefore, $s^c_0.\str{issuerCache}[\mi{identity}] \in \mathsf{dom}(\mi{as})$ (Line~\ref{c-cit-chooses-issuer}
  of Algorithm~\ref{alg:client-check-id-token}).
  Per definition of $\str{issuerCache}$, it follows that the identity
  $\mi{identity}$ chosen in Line~\ref{c-cit-chooses-identity}
  is from $\mathsf{ID}^\mi{as}$. In other words, it holds true that
  $S(c).\str{sessions}[\mi{sessionId}][\str{identity}] \in \mathsf{ID}^\mi{as}$.
  As shown in Lemma~\ref{lemma:token-req-to-id-domain}, the corresponding token request is
  sent to $d' \in  \mathsf{dom}(\mi{as})$.
  
  This means that the attacker knows a code 
  such that there is a record
  $\mi{rec} \inPairing S(\mi{as}).\str{records}$ with
  $\mi{rec}[\str{code}] \equiv \mi{code}$,
  $\mi{rec}[\str{client\_id}] \equiv \mi{clientId}$ and
  $\mi{rec}[\str{subject}] \equiv \mi{id}$. This contradicts 
  Lemma~\ref{lemma:r-web-server-client-code}.

  If the client is a read-write client, then it is possible 
  that the code leaks due to a wrongly configured token endpoint. 
  Here, we distinguish between the following cases:

  \textbf{Case 1: OpenID Connect Hybrid Flow:}
  We first look at the case that the client uses the OIDC Hybrid flow
  (for the current flow in which we assume that the client sent the attacker
  a service session id).
  Let $\mi{id\_token\_tep}$ 
  be the id token the client receives
  in response to the token request.
  Let $\mi{id\_token\_auth}$ be the id token received at the redirection endpoint.
  This implies that the attacker knows $\mi{id\_token\_auth}$.
  Due to the check done in Line~\ref{c-compare-sub-of-both-id-tokens} of 
  Algorithm~\ref{alg:client-fapi-http-response} 
  ($\mathsf{PROCESS\_HTTPS\_RESPONSE}$), it holds true
  that 
  $\mathsf{extractmsg}(\mi{id\_token\_auth})[\str{sub}] \equiv \mi{id}$.
  When the token request is sent in the read-write flow, the first id token is always checked
  in Algorithm~\ref{alg:client-check-first-id-token} ($\mathsf{CHECK\_FIRST\_ID\_TOKEN}$).
  Therefore, the audience value of $\mi{id\_token\_auth}$ is $\mi{clientId}$ 
  (L.~\ref{c-check-first-idt-aud} of Alg.~\ref{alg:client-check-first-id-token}).
  The signature of the id token is checked with the same key as in Algorithm~\ref{alg:client-check-id-token},
  which means that $\mi{id\_token\_auth}$ is signed with the key 
  $s^\mi{as}_0.\str{jwk}$. This contradicts 
  Lemma~\ref{lemma:id-token-web-server-client}, as the attacker cannot be in possession of such
  an id token.

  \textbf{Case 2: Code Flow with JARM:}
  As noted above, the value of $\mi{issuer}$ in Line~\ref{c-CIT-retrieve-sub} of 
  Algorithm~\ref{alg:client-check-id-token} is a domain of $\mi{as}$ and 
  the id token received in the token response was signed with the key
  $s^\mi{as}_0.\str{jwk}$.
  As the checks done in Algorithm~\ref{alg:client-check-id-token} pass successfully,
  it follows that the $\str{iss}$ value of the id token is a domain of $\mi{as}$
  and the $\mi{aud}$ value is the client id of $c$ %
  (Line~\ref{c-check-id-token-aud}).
  Furthermore, the $\mi{sub}$ value is $\mi{id}$ (as the SSID for this identity is
  sent to the attacker).
  As shown in Lemma~\ref{lemma:id-token-web-server-client}, such an id token does not leak
  to the attacker, and therefore, we conclude that the token response was sent by $\mi{as}$.
  This means that the token request sent by $c$ contains a code $\mi{code}$ such that the 
  authorization server $\mi{as}$ creates the id token with the values depicted above.
  The values for the audience and subject attributes of the id token are taken from a 
  record $\mi{rec} \in S'(\mi{as}).\str{records}$ (for a state $S'$ prior to $S$)
  with $\mi{rec}[\str{code}]$ being the code received in the token request
  (due to Lines~\ref{as-token-ep-get-code-from-body},  \ref{as-token-ep-retrieves-record} 
  \ref{as-token-ep-add-sub-to-id-token} and  \ref{as-token-ep-add-client-id-to-id-token} 
  of Algorithm~\ref{alg:as-fapi}).
  As the issuer value chosen in Algorithm~\ref{alg:client-check-id-token} is a domain of $\mi{as}$,
  it follows that $s_0^\mi{c}.\str{issuerCache}[\mi{session}[\str{identity}]] \in \mathsf{dom}(\mi{as}$
  (Lines~\ref{c-cit-chooses-identity} and \ref{c-cit-chooses-issuer} of 
  Algorithm~\ref{alg:client-check-id-token}). Therefore, the issuer value of the request JWS received
  in the authorization response is also a domain of $\mi{as}$ 
  (Lines~\ref{c-check-respJWS-select-identity-from-session}, ~\ref{c-check-respJWS-select-get-issuer-with-identity}
  of Algorithm~\ref{alg:client-check-response-jws}) and the JWS was created and signed by $\mi{as}$
  (with the same reasoning as above). We also note that the value of the issuer stored in the session
  does not change, as the value of the identity is only set in Line~\ref{c-create-initial-session}
  of Algorithm~\ref{alg:client-fapi-http-request}.
  However, this contradict Lemma~\ref{lemma:req-jws-does-not-leak}, as 
  the attacker sent the message to the redirection endpoint of the client containing
  a request JWS created by $\mi{as}$ and with a code that is associated with an honest
  identity.
  \QED
  
\end{proof}

\subsection{Session Integrity}
In the following, we show that the Read-Write profile of the FAPI,
when used with web server clients and OAUTB, provides session integrity for
both authentication and authorization. We highlight that this holds true
under the assumption that the state value (which is used for preventing
CSRF attacks; Section 10.12 of \cite{rfc6749-oauth2}) leaks to the attacker.

\begin{lemma} [Session Integrity Property for Authentication for Web Server Clients with OAUTB]\label{theorem:si-authn}
  For every run $\rho$
  of a FAPI web system $\fapiwebsystem$
  with a network attacker, 
  every processing step $Q$ in $\rho$ with
  $$Q = (S, E, N) \xrightarrow[]{}
  (S', E', N')$$ (for some $S$,
  $S'$,
  $E$,
  $E'$,
  $N$,
  $N'$),
  every browser $b$
  that is honest in $S$,
  every $\mi{as} \in \fAP{AS}$,
  every identity $u$,
  every web server client $c\in \fAP{C}$
  of type $\str{conf\_OAUTB}$
  that is honest in $S$,
  every nonce $\mi{lsid}$,
  and $\mathsf{loggedIn}_\rho^Q(b, c, u, \mi{as}, \mi{lsid})$
  we have that (1) there exists a processing step $Q'$
  in $\rho$
  (before $Q$)
  such that $\mathsf{started}_\rho^{Q'}(b, c, \mi{lsid})$,
  and (2) if $\mi{as}$
  is honest in $S$,
  then there exists a processing step $Q''$
  in $\rho$
  (before $Q$)
  such that
  $\mathsf{authenticated}_\rho^{Q''}(b, c, u, \mi{as}, \mi{lsid})$.
\end{lemma}

\begin{proof}
  ~

  \begin{description}
    \item[\textbf{Part (1):}]\strut

      This part of the proof is analogous to the proof given in 
      Lemma 10 of \cite{FettKuestersSchmitz-TR-OIDC-2017}. For completeness, we 
      give the full proof for the FAPI model.
      
      Per definition of  
      $\mathsf{loggedIn}_\rho^Q(b, c, u, \mi{as}, \mi{lsid})$ 
      (Definition~\ref{def:user-logged-in}), 
      it holds true that the client $c$ sent the service session id
      to the browser $b$. This happens only in Line~\ref{line:client-send-set-service-session}
      of Algorithm~\ref{alg:client-check-id-token} ($\mathsf{CHECK\_ID\_TOKEN}$),
      where the service session id is sent to $S(c).\str{sessions}[\mi{lsid}][\str{redirectEpRequest}][\str{sender}]$
      (Lines~\ref{c-CIT-retrieves-request} and \ref{line:client-send-set-service-session}
      of Algorithm~\ref{alg:client-check-id-token}). 
      
      This value is only set in Line~\ref{line:client-set-redirect-ep-request-record}
      of Algorithm~\ref{alg:client-fapi-http-request} (at the redirection endpoint of the
      client), where it is set to sender of the redirection request. In other words,
      the browser $b$ sent the request to the redirection endpoint.
      
      This request contains the nonce $\mi{lsid}$ as a session id (in a cookie;
      Line~\ref{c-redir-ep-get-sid-from-cookie} of Algorithm~\ref{alg:client-fapi-http-request}).
      
      As this cookie contains the secure prefix,
      it follows that it was set by the client (i.e., it was not set by the network attacker
      e.g., over a previous HTTP connection).
      
      This means that the client previously sent a response to $b$ in Line~\ref{line:client-send-authorization-redir}
      of Algorithm~\ref{alg:client-fapi-start-login-flow}, as this is the only algorithm in which the client
      sets a cookie containing a login session id (Line~\ref{client-setCookie-lsid}).
      This response is sent to 
      $S''(c).\str{sessions}[\mi{lsid}][\str{startRequest}][\str{sender}]$
      (for a state $S''$ prior to $S$ within the same run)
      (Line~\ref{c-redir-ep-retrieve-startReq}), which is only set in Line~\ref{c-create-initial-session}
      of Algorithm~\ref{alg:client-fapi-http-request}, i.e., the browser $b$ sent a POST request to the path
      $\str{/startLogin}$. This request contains a origin header with an origin of the client
      (checked in Line~\ref{c-startLogin-check-origin-header} of
      Algorithm~\ref{alg:client-fapi-http-request} and due to Lemma~\ref{lemma:https-server-correct-hosts}).

      From the two scripts that could send such a request ($\mi{script\_c\_get\_fragment}$ and
      $\mi{script\_client\_index}$), only $\mi{script\_client\_index}$ (Algorithm~\ref{alg:script-client-index})
      sends such a request. Therefore, it holds true that 
      $\mathsf{started}_\rho^{Q'}(b, c, \mi{lsid})$ (for a processing step $Q'$ that happens before
      $Q$).

    \item[\textbf{Part (2):}]\strut

    \textbf{Login with ID Token from Token Response:}
    From the definition of 
    $\mathsf{loggedIn}_\rho^Q(b, c, u, \mi{as}, \mi{lsid})$
    (Definition~\ref{def:user-logged-in}), 
    it follows that the client $c$ sent a response $m$ to $b$ containing
    the header 
    $\an{\str{Set\mhyphen{}Cookie},
    [\an{\str{\_\_Secure}, \str{serviceSessionId}}{:}\an{\mi{ssid},\top,\top,\top}]}$ for
    some nonce $\mi{ssid}$,
    and it also holds true that
    $S(c).\str{sessions}[\mi{lsid}][\str{serviceSessionId}] \equiv \mi{ssid}$ and
    $S(c).\str{sessions}[\mi{lsid}][\str{loggedInAs}] \equiv \an{d,u}$
    (with $d \in \mathsf{dom}(\mi{as})$).
    Let $m_\text{redir}^\text{req}$ be the request corresponding to the response
    $m$.
    
    The cookie contains the secure-prefix, which means that it was set in a connection
    to the client, i.e., it was set by the client.
    An honest web server client sends such a response only in 
    Line~\ref{line:client-send-set-service-session} of Algorithm~\ref{alg:client-check-id-token}
    ($\mathsf{CHECK\_ID\_TOKEN}$).
    This algorithm is only called in Line~\ref{line:client-call-check-id-token-after-code}
    of Algorithm~\ref{alg:client-fapi-http-response} ($\mathsf{PROCESS\_HTTPS\_RESPONSE}$),
    and therefore, the id token that is used in Algorithm~\ref{alg:client-check-id-token}
    is received in a response with the reference value $\str{TOKEN}$. Let 
    $m_\text{token}^\text{resp}$ denote this token response.
    
    \textbf{Token Response was sent by as:}
    Due to $S(c).\str{sessions}[\mi{lsid}][\str{loggedInAs}] \equiv \an{d,u}$,
    it follows that the id token received in the token response was signed by $\mi{as}$:
    The value of $\mi{issuer}$ chosen in Line~\ref{c-cit-chooses-issuer} of 
    Algorithm~\ref{alg:client-check-id-token} is $d$ (as this value is used
    in $S(c).\str{sessions}[\mi{lsid}][\str{loggedInAs}]$ in Line~\ref{c-CIT-retrieve-sub}).
    Therefore, the public key used for checking the signature of the
    id token is 
    $\mi{jwks}
    \overset{\text{L.~\ref{c-check-id-token-choose-key}, Alg.~\ref{alg:client-check-id-token}}}{\equiv} S(c).\str{jwksCache}[d]  
    \equiv s_0^c.\str{jwksCache}[d]  
    \overset{\text{Def.}}{\equiv} \mathsf{pub}(s^\mi{as}_0.\str{jwk})$, 
    and the id token was signed with the corresponding private key (as checked in Line~\ref{c-CIT-sig};
    we note that the value of $\str{jwksCache}$ is never changed by the client, and therefore, is the same as
    in the initial state).
    This private key is only known to $\mi{as}$ and never sent to any other process,
    which means that the id token was created by $\mi{as}$. 

    As shown in Lemma~\ref{lemma:id-token-conf-from-tokenep}, %
    such an id token
    does not leak to any process other than $c$ and $\mi{as}$, %
    which means that the token response was sent
    by $\mi{as}$. More precisely, all pre-conditions of the lemma are fulfilled, as 
    the authorization server is honest, 
    the client is a web server client (and therefore, it is confidential), the signature of the id token is valid
    and it has the client id of $c$ as its audience value (which is checked in 
    Line~\ref{c-check-id-token-aud} of Algorithm~\ref{alg:client-check-id-token}).
    Furthermore, it contains the attribute $\str{at\_hash}$ (checked in Line~\ref{c-check-at-hash}
    of Algorithm~\ref{alg:client-fapi-http-response}, when receiving the token response), which means
    that it was created at the token endpoint of the authorization server.

    \textbf{Code included in Token Request was provided by b:}
    As the client sends the service session id to $b$, it follows that 
    $S(c).\str{sessions}[\mi{lsid}][\str{redirectEpRequest}][\str{sender}]$ 
    is an IP address of $b$ (Lines~\ref{c-CID-retrieve-session}, \ref{c-CIT-retrieves-request} 
    and \ref{line:client-send-set-service-session} of Algorithm~\ref{alg:client-check-id-token}).
  
    Let $m_\text{token}^\text{req}$ be the token request corresponding to the token response $m_\text{token}^\text{resp}$.
    As shown in Lemma~\ref{lemma:code-token-req-from-redirect-ep}, it holds true that
    $m_\text{token}^\text{req}.\str{body}[\str{code}]  \equiv
    S''(c).\str{sessions}[\mi{lsid}][\str{redirectEpRequest}][\str{data}][\str{code}]$
    (for a state $S''$ prior to $S$).

    This value is only set at the redirection endpoint of the client (Line~\ref{line:client-set-redirect-ep-request-record} 
    of Algorithm~\ref{alg:client-fapi-http-request}),
    which means that the code used for the token request was sent by $b$. (We note that
    as the state value is invalidated at the session of the client (Line~\ref{c-invalidates-state} 
    of Algorithm~\ref{alg:client-fapi-http-request}), for each session, only one request to the redirection endpoint is
    accepted.)
  
    \textbf{Together with the code, b also included a TB-ID:}
    Along with the code, the browser sent a provided Token Binding message
    with the ID $S(c).\str{sessions}[\mi{lsid}][\str{browserTBID}]$ %
    (Line~\ref{c-set-browserTBID} of Algorithm~\ref{alg:client-fapi-http-request}).
    Here, this value is taken from a provided Token Binding message, which the honest browser
    only sets in Line~\ref{browser-TB-prov-msg} of Algorithm~\ref{alg:processresponse}.
    The private key used by the browser is only used for the client (i.e., for each domain, the browser uses
    a different key).  (Here, we again highlight that the response to the redirection endpoint was
    sent by the browser, which is honest).

    We conclude that the values for the code and the PKCE verifier (which
    is the Token Binding ID used by the browser for the client; Line~\ref{line:rw-client-pkce-cv-browserTBID}
    of Algorithm~\ref{alg:client-token-request}) included in the token request
    were both provided by the browser $b$. More precisely, the token request $m_\text{token}^\text{req}$ (with the reference
    value $\str{TOKEN}$) is only sent in Line~\ref{line:jws-oautb-cont-send} 
    of Algorithm~\ref{alg:client-token-request} ($\mathsf{SEND\_TOKEN\_REQUEST}$), as the client is
    a web server client of type $\str{conf\_OAUTB}$ (only read-write clients can be of this type).
   
    \textbf{The identity u was authenticated by b:}
    As noted above, the token response was sent by $\mi{as}$, which means that all checks done by $\mi{as}$
    passed successfully. Therefore, it holds true that 
    $m_\text{token}^\text{req}.\str{body}[\str{pkce\_verifier}] \equiv \mi{record}[\str{pkce\_challenge}]$
    (Line~\ref{line:as-checks-conf-OAUTB-pkce} of Algorithm~\ref{alg:as-fapi}),
    with $\mi{record} \in S'''(\mi{as}).\mi{records}$ such that
    $\mi{record}[\str{code}] \equiv m_\text{token}^\text{req}.\str{body}[\str{code}]$ (Lines~\ref{as-token-ep-get-code-from-body}
    and \ref{as-token-ep-retrieves-record}
    of Algorithm~\ref{alg:as-fapi}) (for a state $S'''$ prior to $S$).  %
    As noted above, $m_\text{token}^\text{req}.\str{body}[\str{pkce\_verifier}] \equiv 
    S(c).\str{sessions}[\mi{lsid}][\str{browserTBID}]$ 
    (Line~\ref{line:rw-client-pkce-cv-browserTBID} of Algorithm~\ref{alg:client-token-request}),
    which is a Token Binding ID used by $b$.
  
    The value of $\mi{record}[\str{pkce\_challenge}]$ is only set in 
    Line~\ref{line:as-auth2-conf-OAUTB-pkce-cc} of Algorithm~\ref{alg:as-fapi}
    (as the client is a web server client of type $\str{conf\_OAUTB}$; we note that this client id is 
    included in the record), %
    where it is set to the value
    $\mi{TB\_referred\_pub} \equiv 
    \bar{m}.\str{headers}[\str{Sec\mhyphen{}Token\mhyphen{}Binding}][\str{ref}][\str{id}]$
    (Lines~\ref{as-auth2-get-ref-msg} and
    \ref{as-auth2-get-ref-id}
    of  Algorithm~\ref{alg:as-fapi}), with $\bar{m}$ being the message which the AS receives at the authorization endpoint.
  
    To sum up the previous paragraphs, it holds true that the message
    $\bar{m}$ contains a valid Token Binding message (i.e., with a valid signature,
    as this is always checked by the AS) %
    with the Token Binding ID
    $\bar{m}.\str{headers}[\str{Sec\mhyphen{}Token\mhyphen{}Binding}][\str{ref}][\str{id}]
    \equiv \mi{record}[\str{pkce\_challenge}]
    \equiv m_\text{token}^\text{req}.\str{body}[\str{pkce\_verifier}] 
    \equiv S(c).\str{sessions}[\mi{lsid}][\str{browserTBID}]
    $, which is a Token Binding ID of the browser $b$.
    As only $b$ knows the corresponding private key (and does not reveal this key
    to another process), we conclude that 
    $\bar{m}$ was sent by $b$.

    When sending the token response, $\mi{as}$ does not only check the
    PKCE verifier, but also retrieves the identity that is then included in the id token
    from $\mi{record}$ (Line~\ref{as-token-ep-add-sub-to-id-token} of Algorithm~\ref{alg:as-fapi}). 

    The identity is added to the record in the $\str{/auth2}$ path and taken from
    $\bar{m}$ (Lines~\ref{line:as-auth2-id1} and \ref{line:as-auth2-id2} of Algorithm~\ref{alg:as-fapi}),
    which means that the identity $u$ was authenticated by $b$, i.e., $b = \mathsf{ownerOfID}(u)$.

    \textbf{The redirection request was sent from as (to b):}

    \emph{Case 1: OIDC Hybrid Flow:} Let $\mi{idt}_1$ be the id token contained in $m_\text{redir}^\text{req}$.
    (As the flow used for the session with the session identifier $\mi{lsid}$ 
    is the OIDC Hybrid flow, such an id token is always required to be included in the request).
    When receiving the token response, the client checks if the $\str{sub}$ and $\str{iss}$
    attributes have the same values in both id tokens, and only continues the flow if the values are the same
    (Lines~\ref{c-compare-sub-of-both-id-tokens} and \ref{c-compare-iss-of-both-id-tokens}
    of Algorithm~\ref{alg:client-fapi-http-response}).

    As we know that the second id token (which is used for logging in the end-user) contains the
    subject $u$ and an issuer being a domain of $\mi{as}$, it follows that $\mi{idt}_1$ contains the same values.
    Furthermore, the id token is signed by $\mi{as}$ (with the same reasoning as above)
    and contains the client identifier $\mi{clientId}$ as its audience value.

    As shown in Lemma~\ref{lemma:id-token-web-server-client}, such an id token does not leak to the attacker
    (we note that $b$ is honest and the client is a web server, i.e., all conditions of the lemma
    are fulfilled).

    Analogous to the proof of Lemma 10 of ~\cite{FettKuestersSchmitz-TR-OIDC-2017},
    it follows that the request $m_\text{redir}^\text{req}$ was caused by a redirection from $\mi{as}$.
    Instead of the state value, the id token $\mi{idt}_1$ is a secret value that does not leak to 
    the attacker. In short, the request $m_\text{redir}^\text{req}$ 
    was not caused by the attacker, as the id token does not leak.
    The redirect was also not caused by the client $c$, as $c$ does not send messages containing
    an id token.

    As the request $m_\text{redir}^\text{req}$ contains a state value (which is checked in 
    Line~\ref{c-redirect-ep-check-state} of Algorithm~\ref{alg:client-fapi-http-request}),
    it follows that 
    the request was not created by the scripts $\mi{script\_client\_index}$
    or $\mi{script\_as\_form}$, as these scripts do not send messages containing a state parameter.

    The script $\mi{script\_c\_get\_fragment}$ sends only data that is contained in the fragment 
    part of its own URI, and only to itself. This means that the script was sent from the client to
    the browser, which happens only in Line~\ref{line:c-send-script-get-fragment}
    of Algorithm~\ref{alg:client-fapi-http-request}, i.e., at the redirection endpoint.

    Altogether, we conclude that there was a location redirect which was sent 
    from $\mi{as}$ to $b$ containing the id token.

    \emph{Case 2: Authorization Code Flow with JARM:}
    Let the flow used in the session with the session identifier $\mi{lsid}$ 
    be a Code Flow with JARM.
    As shown above, it holds true that $b = \mathsf{ownerOfID}(u)$. 
    The code that the client uses at the token endpoint was sent by the browser, 
    and as an id token with the identity of $u$ is returned by the authorization server,
    it follows that there is a record that contains the code provided by the browser
    and the identity $u$ (in $S(\mi{as}).\str{records}$). 
    Furthermore, the code is contained in a response JWS. %
    The issuer value of the JWS is a domain of $\mi{as}$ (as the id token was signed by $\mi{as}$,
    it follows that for this particular session, the client is using a domain of $\mi{as}$
    as the expected issuer (in $\mi{issuerCache}[\mi{session}][\str{identity}]$), to which the
    client also compares the issuer of the response JWS.
    Therefore, its audience value is $\mi{clientId}$. Now, all conditions of Lemma~\ref{lemma:req-jws-does-not-leak}
    are fulfilled, which means that such a response JWS does not leak to the attacker. %
    The remaining argumentation is the same as in the first case.

  \QED

  \end{description}

\end{proof}

\begin{lemma} [Session Integrity Property for Authorization for Web Server Clients with OAUTB]\label{theorem:si-authz}
  For every run $\rho$
  of a FAPI web system $\fapiwebsystem$
  with a network attacker, 
  every processing step $Q$ in $\rho$ with
  $$Q = (S, E, N) \xrightarrow[]{}
  (S', E', N')$$  (for some $S$, $S'$, $E$, $E'$, $N$, $N'$),
  every browser $b$
  that is honest in $S$,
  every $\mi{as} \in \fAP{AS}$,
  every identity $u$, %
  every web server client $c\in \fAP{C}$
  of type $\str{conf\_OAUTB}$
  that is honest in $S$,
  every $\mi{rs} \in \fAP{RS}$ that is honest in $S$,
  every nonce $r$, 
  every nonce $\mi{lsid}$,
  we have that if
  $\mathsf{accessesResource_\rho^Q(b, r, u, c, \mi{rs}, \mi{lsid})}$
  and $s_0^\mi{rs}.\str{authServ} \in \mathsf{dom}(\mi{as})$, then
  (1)
  there exists a processing step $Q'$
  in $\rho$
  (before $Q$)
  such that $\mathsf{started}_\rho^{Q'}(b, c, \mi{lsid})$,
  and (2) if $\mi{as}$
  is honest in $S$,
  then there exists a processing step $Q''$
  in $\rho$
  (before $Q$)
  such that
  $\mathsf{authenticated}_\rho^{Q''}(b, c, u, \mi{as}, \mi{lsid})$.
\end{lemma}

\begin{proof}
  ~

  \begin{description}
    \item[\textbf{Part (1):}]\strut

      Per definition of $\mathsf{accessesResource}$
      (Definition~\ref{def:browser-accesses-resource}), it holds true that 
      the browser $b$
      has a cookie with the session identifier $\mi{lsid}$ for the origin of the client $c$.
      As this cookie has the secure prefix set, it follows that the cookie was set by $c$, which
      happens only in Line~\ref{line:client-send-authorization-redir}
      of Algorithm~\ref{alg:client-fapi-start-login-flow}. The 
      remaining reasoning is the same as in the proof of Lemma~\ref{theorem:si-authn}.

    \item[\textbf{Part (2) (using the OIDC Hybrid Flow)}]\strut

      Here, we also first prove the 
      property for the OIDC Hybrid Flow and 
      then show the parts that differ when using the Authorization Code Flow
      in conjunction with JARM.

      \textbf{Resource was sent from rs:}

      Per Definition of $\mathsf{accessesResource}$, it holds true that
      $c$ saved the resource access nonce $r$ in 
      $S(c).\str{sessions}[\mi{lsid}][\str{resource}]$.
      An honest client stores a resource access nonce only in 
      Line~\ref{client-saves-resource} of Algorithm~\ref{alg:client-fapi-http-response}.
      Here, $r$ was contained in response to a request $m_\text{resource}^\text{req}$
      with the reference value $\str{RESOURCE\_USAGE}$ (Line~\ref{client-PHResp-case-resource-usage}
      of Algorithm~\ref{alg:client-fapi-http-response}).
      The client sends requests with this reference value only in 
      Line~\ref{c-sends-msg-to-resource-rw} of Algorithm~\ref{alg:client-use-access-token}
      (this holds true as $c$ is a read-write client).

      $m_\text{resource}^\text{req}$ is sent to the ${/resource\mhyphen{}rw}$ path of
      $S(c).\str{sessions}[\mi{lsid}][\str{RS}]$ (Lines~\ref{c:UAT-retrieve-session},
      \ref{c:UAT-retrieve-rsHost} and \ref{c:UAT-create-rw-message} 
      of Algorithm~\ref{alg:client-use-access-token}). 

      Per Definition of $\mathsf{accessesResource}$, this is a domain of
      $\mi{rs}$, which means that $r$ was sent to the client by $\mi{rs}$.
      More precisely, the response of the resource server was sent in 
      Line~\ref{rs-send-wNonce} of Algorithm~\ref{alg:rs-oidc} (within
      the ${resource\mhyphen{}rw}$ path).

      \textbf{Second ID Token was created by as:} 
      
      As the checks done in Lines~\ref{rs:check-if-at-iss-in-body}
      and \ref{rs-check-at-iss} of Algorithm~\ref{alg:rs-oidc} passed successfully,
      it follows that $m_\text{resource}^\text{req}.\str{body}[\str{at\_iss}]
      \in \mathsf{dom}(\mi{as})$ (per assumption, it holds
      true that $s_0^\mi{rs}.\str{authServ} \in \mathsf{dom}(\mi{as})$).

      This means that the id token contained in the token response was signed by
      $\mi{as}$. More precisely, $S''(c).\str{sessions}[\mi{lsid}][\str{idt2\_iss}]
      \in \mathsf{dom}(\mi{as})$ (Line~\ref{c-sets-at-iss} of 
      Algorithm~\ref{alg:client-use-access-token}; here, we are only considering
      the OIDC Hybrid Flow), for some state $S''$ prior to $S$. 

      This value is only set in Line~\ref{c-set-idt2-iss} of 
      Algorithm~\ref{alg:client-fapi-http-response}, where it is set to 
      $\mathsf{extractmsg}(m_\text{token}^\text{resp}.\str{body}[\str{id\_token}])[\str{iss}]$,
      where $m_\text{token}^\text{resp}$ is the token response received in 
      Algorithm~\ref{alg:client-fapi-http-response}.
      Therefore, the $\str{iss}$ value of this id token is a domain of $\mi{as}$.
      Due to Lines~\ref{c-checksig-second-idt} and \ref{c-token-resp-check-issuer}
      of Algorithm~\ref{alg:client-fapi-http-response}, it follows that (with
      $\mi{dom_\mi{as}} \in \mathsf{dom}(\mi{as})$)
      $$\mathsf{checksig}(m_\text{token}^\text{resp}.\str{body}[\str{id\_token}], 
      s_0^\mi{as}.\str{jwksCache}[\mi{dom_\mi{as}}]) \equiv \top$$
      $$\overset{\text{Def.}}{\Rightarrow} \mathsf{checksig}(m_\text{token}^\text{resp}.\str{body}[\str{id\_token}], 
      \mathsf{pub}(s^\mi{as}_0.\str{jwk})) \equiv \top$$

      Therefore, we conclude that the id token was signed by $\mi{as}$. (The values
      of $\str{jwksCache}$ are never changed by the client, which means that they are the same
      as in the initial state).

      \textbf{Token Response was sent by as: }

      This id token contains a value for the attribute $\str{at\_hash}$ (Line~\ref{c-check-at-hash}
      of Algorithm~\ref{alg:client-fapi-http-response}). As shown in Lemma~\ref{lemma:id-token-conf-from-tokenep},
      such an id token does not leak to the attacker (the id token contains the client id of $c$ due to
      the check done in Line~\ref{c-hybrid-token-resp-check-aud} of Algorithm~\ref{alg:client-fapi-http-response}).
      Thus, we conclude that the token response $m_\text{token}^\text{resp}$ was sent by $\mi{as}$ (clients do not send any messages containing id tokens).

      \textbf{Access Token used by c:}

      The access token $t$ used by the client in the request to $rs$ was contained in the token response.
      More precisely, the message $m_\text{resource}^\text{req}$ was sent by the client in Algorithm~\ref{alg:client-use-access-token}
      (as noted above). In Line~\ref{c:UAT-set-authorization-header} of this algorithm, the client includes the
      access token in the header of the message. The access token used here is an input parameter of
      of the function ($\mathsf{USE\_ACCESS\_TOKEN}$). A read-write client calls this function either
      in Line~\ref{mtls-authn-rs-use-access-token}  %
      or Line~\ref{line:PHResp-UAT-oautb-call} 
      of Algorithm~\ref{alg:client-fapi-http-response}. In both cases, the access token
      is taken from $\overline{S}.\str{sessions}[\mi{lsid}][\str{token}]$ (Line~\ref{PHResp-mtls-rs-retrieve-token} or 
      Line~\ref{PHResp-oautb-rs-retrieve-token}), for some state $\overline{S}$.
      This value is only set in Line~\ref{c-save-access-token-in-session} of 
      Algorithm~\ref{alg:client-prepare-use-access-token}, which is only called in 
      Line~\ref{line:client-call-use-access-token-2} of Algorithm~\ref{alg:client-fapi-http-response},
      where the access token is taken from the body of $m_\text{token}^\text{resp}$.
      Therefore, the access token was sent to the client by $\mi{as}$. 

      \textbf{Sequence in $\str{accessTokens}$ (state of as):}

      We note that the token response was sent by $\mi{as}$ to $c$, which means that the corresponding token request was sent
      by $c$ and contains the client identifier $\mi{clientId}$ (as $c$ is honest). 

      Before sending the token response in Line~\ref{as-send-token-ep-readwrite} of
      Algorithm~\ref{alg:as-fapi}, the authorization server adds a sequence
      to $\bar{S}'(\mi{as}).\str{accessTokens}$ (for some state $\bar{S}'$ prior to $S$).

      Let $\mi{ATSeq}$ be the sequence added to the state directly before sending the token response
      $m_\text{token}^\text{resp}$. As the client identifier received in the token request belongs
      to a client using OAUTB, this sequence is equal to %
      $\an{\str{OAUTB}, u', \mi{clientId}, t, k, \str{rw}}$, 
      for some identity $u'$ %
      and key $k$. %
      The access token $t$ is the same that is included in the token response.
      In the following, we will show that the $u' \equiv u$. %

      We first note that for each access token, there is at most one sequence in $\str{accessTokens}$ (in the state
      of the authorization server) containing this access token. This holds true because the authorization server
      creates fresh authorization codes and access tokens in the $\str{/auth2}$ path (Lines~\ref{as-creates-code}
      and \ref{as-creates-at} of Algorithm~\ref{alg:as-fapi})
      for each authorization request received at $\str{/auth2}$. These values are stored in a record in $\str{records}$.
      When the authorization server receives a request to the token endpoint, it chooses the record depending on the
      authorization code contained in the token request and invalidates the authorization code contained in the record
      before creating the sequence for $\str{accessTokens}$ (Line~\ref{as-invalidates-code}).
      Therefore, the access token can only be added once to such a sequence. 

      As the resource server sent a response in Line~\ref{rs-send-wNonce}
      of Algorithm~\ref{alg:rs-oidc} to $c$, it follows that all (applicable) checks passed
      successfully. As the request $m_\text{resource}^\text{req}$ was sent by $c$, it does not
      contain the key $\str{MTLS\_AuthN}$ in its body (Line~\ref{client-uat-if-client-type-mtls}
      of Algorithm~\ref{alg:client-use-access-token}). Therefore, %
      it holds true that
      $\mathsf{check\_oautb\_AT}(u, t, k, S(\mi{rs}).\str{authServ}) \equiv \True$ (Line~\ref{rs-check-oautb-binding}
      of Algorithm~\ref{alg:rs-oidc}).
      (This holds true as the resource server provides access to a resource of the identity $u$). 
      From the definition of
      $\mathsf{check\_oautb\_AT}$, it follows that this identity is
      contained in the sequence $\mi{ATSeq}$, and therefore, $u' \equiv u$. 

      \textbf{The identity u was authenticated by b:}

      The browser $b$
      has a cookie with the session identifier $\mi{lsid}$ for the origin of the client $c$,
      which means that the request to the redirection endpoint of the client was sent by 
      $b$ (as this value is only known to $b$) %
      With the same reasoning as in Lemma~\ref{theorem:si-authn}, it follows that the identity
      $u$ was authenticated by $b$. We briefly summarize the argumentation:
      As the request to the redirection endpoint of $c$ %
      was sent by $b$, it follows that the code used in the token request was provided by $b$.
      Furthermore, $b$ proved possession of a Token Binding ID.
      As the token response was sent by $as$, it follows that the process that authenticated the identity $u$
      proved possession of the Token Binding ID used by the
      browser when sending the request to the redirection endpoint.
      As the private key of this Token Binding ID is only known to $b$, it follows
      that $b$ authenticated $u$.

      \textbf{ID Token contained in the Redirection Request :}

      As we are looking at the Hybrid Flow, the authorization response is required
      to contain an id token. 

      Let $\mi{idt_1}$ be this id token. When receiving this id token in the redirection request,
      the client stores it in the session (using the key $\str{redirectEpRequest}$;
      Line~\ref{line:client-set-redirect-ep-request-record}
      of Algorithm~\ref{alg:client-fapi-http-request}).

      In Lines~\ref{c-compare-sub-of-both-id-tokens} and \ref{c-compare-iss-of-both-id-tokens}
      of Algorithm~\ref{alg:client-fapi-http-response}
      (i.e., after receiving the token response), the client only continues the flow
      if the subject and issuer values of both id tokens have the same value.
      As shown above, the issuer of the second id token is a domain of $\mi{as}$,
      and the 
      signature of $\mi{idt_1}$ is checked in Line~\ref{c-check-first-idt-sig}
      of Algorithm~\ref{alg:client-check-first-id-token}. With the same reasoning as
      above, it can be seen that this id token is signed by $\mi{as}$ (as the issuer
      value is a domain of $\mi{as}$). 
      Furthermore, we note that this function is called in 
      Line~\ref{line:c-call-check-id-token-immediately} of 
      Algorithm~\ref{alg:client-fapi-http-request} (in the Hybrid flow, this function is always called
      when receiving the redirection request).
      The client identifier contained in the id token is $\mi{clientId}$ (checked
      in Line~\ref{c-check-first-idt-aud}
      of Algorithm~\ref{alg:client-check-first-id-token}, and as the issuer is a domain of $\mi{as}$).

      As showed above, the token response (containing both the id token and the access token)
      was sent by $\mi{as}$. Therefore, the id token contained in the token response
      contains the identity associated with the access token, which is $u$. 
      This means that $\mi{idt_1}$ also has the subject $u$. %
      As shown in Lemma~\ref{lemma:id-token-web-server-client}, such an id token does not
      leak to the attacker. 

      \textbf{Redirection Request was sent to from as (to b):}

      Analogous to Lemma~\ref{theorem:si-authn},
      it follows that the redirection request was sent to the browser by $\mi{as}$, 
      as the id token contained in the redirection request does not leak to the attacker.

    \item[\textbf{Part (2) (using the Code Flow with JARM)}]\strut

      Here, we focus on the parts that are different from the proof for the OIDC Hybrid flow.

      \textbf{Resource was sent from rs:} This part is the same as in the Hybrid Flow.

      \textbf{Response JWS was created by as:} 

      As above, it holds true that 
      $m_\text{resource}^\text{req}.\str{body}[\str{at\_iss}]
      \in \mathsf{dom}(\mi{as})$ (as this is checked at the resource server).

      From this, it follows that
      $S''(c).\str{sessions}[\mi{lsid}][\str{JARM\_iss}] \in \mathsf{dom}(\mi{as})$,
      as the value of $\str{at\_iss}$ of the request to the resource server is set from this
      value
      (Line~\ref{c-sets-at-iss-JARM} of  Algorithm~\ref{alg:client-use-access-token};
      here, we are considering the Code Flow with JARM).

      This value is only set in Line~\ref{c-set-JARM-iss} of 
      Algorithm~\ref{alg:client-check-response-jws} ($\mathsf{CHECK\_RESPONSE\_JWS}$),
      where it is set to 
       $ \mi{data}[\str{iss}]
        \equiv \mathsf{extractmsg}(\mi{respJWS})[\str{iss}] $
        (Line~\ref{c-check-respJWS-retrieve-data-from-resp-jws} of Alg.~\ref{alg:client-check-response-jws}),
      with $\mi{respJWS}$ being the input argument of 
      Algorithm~\ref{alg:client-check-response-jws}.

      Here, the first argument of the function (i.e., the session identifier) is $\mi{lsid}$,
      as the issuer of the response JWS is saved in the session identified by $\mi{lsid}$
      (again Line~\ref{c-sets-at-iss-JARM} of  Algorithm~\ref{alg:client-use-access-token}).

      This algorithm is only called in Line~\ref{line:c-call-check-response-jws}
      of Algorithm~\ref{alg:client-fapi-http-request} (at the redirection endpoint),
      where the $\mi{respJWS}$ is set to $\mi{data}{[\str{responseJWS}]}$.
      Let 
      $m_\text{redirect}^\text{req}$ be the request which the client received at the redirection 
      endpoint (i.e., at the path $\str{/redirect\_ep}$ and for the session identifier $\mi{lsid}$,
      i.e., the cookie $\str{sessionId}$ contained in the request has the value $\mi{lsid}$).

      As we are looking at the Code Flow using JARM, the value of $\mi{data}$
      is equal to $m_\text{redirect}^\text{req}.\str{parameters}$ (Line~\ref{c-redirect-ep-data-from-param}),
      hence, we conclude that the $\str{iss}$ value of the JWS $m_\text{redirect}^\text{req}.\str{parameters}[\str{responseJWS}]$ 
      is a domain of $\mi{as}$. 

      During the checks that is done in $\mathsf{CHECK\_RESPONSE\_JWS}$, 
      the client also checks the signature of the JWS (Line~\ref{c-check-response-jws-check-signature} of 
      Algorithm~\ref{alg:client-check-response-jws}). With the same reasoning as in the case of the
      Hybrid Flow, it follows that the response JWS contained in 
      $m_\text{redirect}^\text{req}$ was signed by $\mi{as}$. 

      \textbf{Token Response was sent by as:} 
      As the client used the access token $t$ it received in the token response at
      the resource server, it follows that the check of the hash of the access token 
      done in Line~\ref{JARM-c-check-at-hash}
      of Algorithm~\ref{alg:client-fapi-http-response}
      passed successfully. Furthermore, the $\str{aud}$ value of $\mi{respJWS}$ 
      is the client identifier $\mi{clientId}$ (checked in Line~\ref{c-check-response-jws-aud} of
      Algorithm~\ref{alg:client-check-response-jws}).

      In the following, we assume that the token response $m_\text{resp}^\text{token}$ was sent by the attacker.

      As the hash of the access token was included in a response JWS signed by 
      $\mi{as}$, and as this JWS contains the client identifier $\mi{clientId}$, we conclude
      that $\mi{as}$ created this access token for $c$, i.e., there is a record
      $\mi{rec}$ within the state of the AS with $\mi{rec}[\str{access\_tokens}] \equiv t$
      and $\mi{rec}[\str{aud}] \equiv \mi{clientId}$. %

      As we assume that this access token was sent from the attacker, it follows that it previously
      leaked to the attacker. In order to leak, the access token must first be sent from the AS.

      Let $Q''$ be the processing step in which $\mi{as}$ sent the access token $t$. 
      As the access token is contained in a record which also contains the client identifier
      $\mi{clientId}$, and as this is the client identifier of a web server client (which means that the
      client is confidential), it follows that the corresponding code $\mi{code}$ 
      was sent by $c$ (due to Lemma~\ref{lemma:client-authentication}). 

      However, this means that the client received the code $\mi{code}$ at the redirection endpoint
      (for this particular flow in which the access token leaks). As $c$ is a read-write client, it only
      accepts signed authorization responses, i.e, the response is either %
      a response JWS or an id token.

      We note that this response (either response JWS or id token) was signed by $\mi{as}$:
      The client sent (in this particular flow) the code (i.e., the token request) to $\mi{as}$.
      The token request is required to contain the redirection uri at which the redirection response
      was received (checked at the AS in Line~\ref{as-check-redir-uri-at-token-endpoint} of Algorithm~\ref{alg:as-fapi}
      (in the FAPI flows, this is  always contained in the token request, as the client is required to 
      previously include this in the authorization request). The check of $\mi{as}$ passed, which means that this
      URI is a redirection URI used by the client for $\mi{as}$. The FAPI requires these sets of URIs to be disjunct
      (i.e., for each AS, the client has a different set of redirection URIs). Therefore, this URI belongs to the set of
      URIs the client uses for $\mi{as}$. This means that (for this particular session), the issuer (i.e. the authorization server)
      the client expects
      is a domain of $\mi{as}$.  %

      From this, we conclude that the response JWS or id token which the client
      received in this (previous) flow at the redirection endpoint
      was signed by $\mi{as}$ (as both id tokens and response JWS always
      contain an issuer, which is the same that the client stored at the corresponding session. The signature
      is checked with a key of this issuer).

      The request to the redirection endpoint contained a response JWS,
      as the AS $\mi{as}$ included the access token corresponding to the code in
      a response JWS. (Otherwise, it would mean that the client would have created an additional
      id token with the attribute $\mi{c\_hash}$ being the hash of the code
      (due to the check done at the client in Line~\ref{c-check-first-idt-check-c-hash}
      of Algorithm~\ref{alg:client-check-first-id-token}.
      However, the AS creates either a response JWS or an id token).

      The response JWS the client received (in the main flow, i.e., the one 
      in which the client receives the resource)
      was created by $\mi{as}$ and contains the authorization code $\mi{code}$.
      However, the AS creates exactly one response JWS with a particular code
      (as codes are nonces that are freshly chosen). 

      This implies that in both flows that we are looking at
      (i.e., the original flow in which the honest $\mi{as}$ sends out the access token,
      and the flow for which we assume that the token endpoint is controlled by the attacker),
      the client received the same response JWS, and in particular, the same state value
      (which is contained in the response JWS).

      However, the state is unique to each session, as it is chosen by the client as a fresh nonce
      (Line~\ref{c-slf-chooses-state-value-here} of Algorithm~\ref{alg:client-fapi-start-login-flow}).

      This is a contradiction to the assumption that the check done in Line 
      Line~\ref{JARM-c-check-at-hash}
      of Algorithm~\ref{alg:client-fapi-http-response}
      was executed successfully, as this would mean that the client previously accepted the check
      done in Line~\ref{c-check-if-state-equiv-bot} of Algorithm~\ref{alg:client-fapi-http-request} (where it is checked if
      the state was already invalidated) for the response JWS received in the second flow.

      Therefore, we conclude that the token endpoint is controlled by the honest $\mi{as}$,
      i.e., the token response was sent by $\mi{as}$.

      \textbf{Access Token used by c:}
      As above, the access token that the client uses for the request
      to the resource server was contained in 
      the token response, i.e., the access token was sent by $\mi{as}$.

      \textbf{Sequence in $\str{accessTokens}$ (state of as):}
      As in the case of the Hybrid flow,
      the state of the authorization server $\mi{as}$ contains
      the sequence
      $\an{\str{OAUTB}, u, \mi{clientId}, t, k, \str{rw}}$.

      \textbf{The identity $u$ was authenticated by $b$:}
      As above, the identity $u$ was authenticated by $b$
      (due to the check of the Token Binding ID used by the browser for the client,
      which happens at the AS).

      \textbf{Response JWS contained in the Redirection Request:}
      As the identity $u$ is governed by an honest browser, and due to 
      Lemma~\ref{lemma:req-jws-does-not-leak}, it follows that the response JWS does not leak to the attacker.

      The rest of the proof is the same as above, as now, the response JWS is 
      a value that does not leak to the attacker (instead of an id token).

  \QED

  \end{description}

\end{proof}

\subsection{Proof of Theorem}
Theorem~\ref{thm:theorem-1} follows immediately from Lemmas~\ref{theorem:authorization},~\ref{theorem:authentication},~\ref{theorem:si-authn}, 
and~\ref{theorem:si-authz}.

\fi

\end{document}

